\documentclass{article}
\usepackage[utf8]{inputenc} % input file can use UTF8, e.g. ö is recognised
\usepackage[english]{babel}
\usepackage{geometry}
\usepackage[numbers]{natbib}
\usepackage{amsmath,amsfonts,amssymb,amsthm,mathrsfs}
\usepackage[hidelinks]{hyperref}
\hypersetup{colorlinks, citecolor=blue, linkcolor=blue, urlcolor=blue}
\usepackage{paralist}
\usepackage{enumitem}
\usepackage{textcase} % defines '\MakeTextUppercase'
\usepackage{comment, authblk}
\usepackage{algorithm}
\usepackage{algpseudocode}
\usepackage{tikz}
\usetikzlibrary{arrows}
\usetikzlibrary{decorations.pathreplacing}
\tikzset{mybrace/.style={decoration={brace,raise=1.8mm},decorate}}
\usepackage{titlesec} % Used to redefine format of section heading
\usepackage{etoolbox} % Use '\AfterEndEnvironment' to avoid indent after theorems
\usepackage{float}
\usepackage{graphics}
\usepackage{booktabs}
\usepackage{multirow}
\usepackage{subfigure}
\usepackage{xcolor}
\usepackage{sectsty}
\sectionfont{\bfseries\Large\raggedright}

\makeatletter

\usepackage{color}

\numberwithin{equation}{section}
\allowdisplaybreaks

\def\@noindentfalse{\global\let\if@noindent\iffalse}
\def\@noindenttrue {\global\let\if@noindent\iftrue}
\def\@aftertheorem{%
	\@noindenttrue
	\everypar{%
		\if@noindent%
		\@noindentfalse\clubpenalty\@M\setbox\z@\lastbox%
		\else%
		\clubpenalty \@clubpenalty\everypar{}%
		\fi}}

\theoremstyle{plain}
\newtheorem{theorem}{Theorem}[section]
\AfterEndEnvironment{theorem}{\@aftertheorem}
\newtheorem{definition}[theorem]{Definition}
\AfterEndEnvironment{definition}{\@aftertheorem}
\newtheorem{lemma}[theorem]{Lemma}
\AfterEndEnvironment{lemma}{\@aftertheorem}
\newtheorem{corollary}[theorem]{Corollary}
\AfterEndEnvironment{corollary}{\@aftertheorem}
\newtheorem{proposition}[theorem]{Proposition}
\AfterEndEnvironment{proposition}{\@aftertheorem}

\theoremstyle{definition}
\newtheorem{remark}[theorem]{Remark}
\AfterEndEnvironment{remark}{\@aftertheorem}
\newtheorem{example}[theorem]{Example}
\AfterEndEnvironment{example}{\@aftertheorem}

\AfterEndEnvironment{proof}{\@aftertheorem}

\newtheorem{assumption}[theorem]{Assumption}
\AfterEndEnvironment{assumption}{\@aftertheorem}

\titleformat{name=\section}
{\center\small\sc}{\thesection\kern1em}{0pt}{\MakeTextUppercase}
\titleformat{name=\section, numberless}
{\center\small\sc}{}{0pt}{\MakeTextUppercase}

\titleformat{name=\subsection}
{\bf\mathversion{bold}}{\thesubsection\kern1em}{0pt}{}
\titleformat{name=\subsection, numberless}
{\bf\mathversion{bold}}{}{0pt}{}

\titleformat{name=\subsubsection}
{\normalsize\itshape}{\thesubsubsection\kern1em}{0pt}{}
\titleformat{name=\subsubsection, numberless}
{\normalsize\itshape}{}{0pt}{}

\def\note#1{\par\smallskip%
	\noindent\kern-0.01\hsize%
	{\setlength\fboxrule{0pt}\fbox{\setlength\fboxrule{0.5pt}\fbox{%
				\llap{$\boldsymbol\Longrightarrow$ }%
				\vtop{\hsize=0.98\hsize\parindent=0cm\small\rm #1}%
				\rlap{$\enskip\,\boldsymbol\Longleftarrow$}
	}}}%
}

\usepackage{delimset}
\def\given{\mskip 0.5mu plus 0.25mu\vert\mskip 0.5mu plus 0.15mu}
\newcounter{bracketlevel}%
\def\@bracketfactory#1#2#3#4#5#6{%
	\expandafter\def\csname#1\endcsname##1{%
		\global\advance\c@bracketlevel 1\relax%
		\global\expandafter\let\csname @middummy\alph{bracketlevel}\endcsname\given%
		\global\def\given{\mskip#5\csname#4\endcsname\vert\mskip#6}\csname#4l\endcsname#2##1\csname#4r\endcsname#3%
		\global\expandafter\let\expandafter\given\csname @middummy\alph{bracketlevel}\endcsname%
		\global\advance\c@bracketlevel -1\relax%
	}%
}
\def\bracketfactory#1#2#3{%
	\@bracketfactory{#1}{#2}{#3}{relax}{0.5mu plus 0.25mu}{0.5mu plus 0.15mu}
	\@bracketfactory{b#1}{#2}{#3}{big}{1mu plus 0.25mu minus 0.25mu}{0.6mu plus 0.15mu minus 0.15mu}
	\@bracketfactory{bb#1}{#2}{#3}{Big}{2.4mu plus 0.8mu minus 0.8mu}{1.8mu plus 0.6mu minus 0.6mu}
	\@bracketfactory{bbb#1}{#2}{#3}{bigg}{3.2mu plus 1mu minus 1mu}{2.4mu plus 0.75mu minus 0.75mu}
	\@bracketfactory{bbbb#1}{#2}{#3}{Bigg}{4mu plus 1mu minus 1mu}{3mu plus 0.75mu minus 0.75mu}
}
\bracketfactory{clc}{\lbrace}{\rbrace}
\bracketfactory{clr}{(}{)}
\bracketfactory{cls}{[}{]}
\bracketfactory{abs}{\lvert}{\rvert}
\bracketfactory{norm}{\Vert}{\Vert}
\bracketfactory{floor}{\lfloor}{\rfloor}
\bracketfactory{ceil}{\lceil}{\rceil}
\bracketfactory{angle}{\langle}{\rangle}

\newcounter{ctr}\loop\stepcounter{ctr}\edef\X{\@Alph\c@ctr}%
\expandafter\edef\csname s\X\endcsname{\noexpand\mathscr{\X}}
\expandafter\edef\csname c\X\endcsname{\noexpand\mathcal{\X}}
\expandafter\edef\csname b\X\endcsname{\noexpand\boldsymbol{\X}}
\expandafter\edef\csname I\X\endcsname{\noexpand\mathbb{\X}}
\ifnum\thectr<26\repeat

\let\@IE\IE\let\IE\undefined
\newcommand{\IE}{\mathop{{}\@IE}\mathopen{}}
\let\@IP\IP\let\IP\undefined
\newcommand{\IP}{\mathop{{}\@IP}}

\def\^#1{\relax\ifmmode {\mathaccent"705E #1} \else {\accent94 #1}\fi}
\def\~#1{\relax\ifmmode {\mathaccent"707E #1} \else {\accent"7E #1}\fi}
\def\*#1{\relax#1^\ast}
\edef\-#1{\relax\noexpand\ifmmode {\noexpand\bar{#1}} \noexpand\else \-#1\noexpand\fi}
\def\>#1{\vec{#1}}
\def\.#1{\dot{#1}}

\def\atop{\@@atop}

\renewcommand{\leq}{\leqslant}
\renewcommand{\geq}{\geqslant}

\newcommand{\ub}{\mathbf{u}}

\newcommand{\ro}{\mathrm{o}}
\newcommand{\rO}{\mathrm{O}}

\newcommand{\ri}{\mathrm{i}}

\newcommand\indep{\protect\mathpalette{\protect\@indep}{\perp}}
\def\@indep#1#2{\mathrel{\rlap{$#1#2$}\mkern2mu{#1#2}}}

\def\parsetime#1#2#3#4#5#6{#1#2:#3#4}
\def\parsedate#1:20#2#3#4#5#6#7#8+#9\empty{20#2#3-#4#5-#6#7 \parsetime #8}
\def\moddate{\expandafter\parsedate\pdffilemoddate{\jobname.tex}\empty}

\theoremstyle{definition}

\theoremstyle{remark}

\theoremstyle{definition}

\theoremstyle{plain}

\theoremstyle{plain}

\theoremstyle{plain}

\theoremstyle{plain}

\allowdisplaybreaks[4]

\makeatother

\providecommand{\conditionname}{Condition}
\providecommand{\definitionname}{Definition}
\providecommand{\lemmaname}{Lemma}
\providecommand{\propositionname}{Proposition}
\providecommand{\remarkname}{Remark}
\providecommand{\corollaryname}{Corollary}
\providecommand{\theoremname}{Theorem}

\begin{document}
	
	\title{\bf \Large{Extreme eigenvalues of sample covariance matrices under generalized elliptical models with applications}} 
	
\author[1]{Xiucai Ding \thanks{E-mail: xcading@ucdavis.edu. XCD is partially supported by NSF-DMS 2113489 and a grant from UC Davis COVID-19 Research Accelerator Funding Track.  The author also wants to thank Zhigang Bao, Hong Chang Ji, Jiang Hu, Miles Lopes, Debashis Paul and Fan Yang  for many helpful discussions. }}
\author[2]{Jiahui Xie \thanks{E-mail:  jiahui.xie@u.nus.edu.}}
\author[3]{Long Yu  \thanks{E-mail: yulong@mail.shufe.edu.cn}}
\author[2]{Wang Zhou \thanks{E-mail: wangzhou@nus.edu.sg.}}

\affil[1]{Department of Statistics, University of California, Davis}

\affil[2]{Department of Statistics and Data Science, National University of Singapore}

\affil[3]{School of Statistics and Management, Shanghai University of Finance and Economics}

	\date{}
	
	\maketitle
	\begin{abstract}
We consider the extreme eigenvalues of the sample covariance matrix $Q=YY^*$ under the generalized elliptical model that $Y=\Sigma^{1/2}XD.$ Here $\Sigma$ is a bounded $p \times p$ positive definite deterministic matrix representing the population covariance structure, $X$ is a $p \times n$ random matrix containing either independent columns sampled from the unit sphere in $\mathbb{R}^p$ or i.i.d. centered entries with variance $n^{-1},$ and $D$ is a diagonal random matrix containing i.i.d. entries and independent of $X.$ Such a model finds important applications in statistics and machine learning. For example, when $X$ contains independent samples from the unit sphere, $Q$ is the sample covariance matrix of elliptically distributed data. For another instance, when $X$ contains i.i.d. entries, then $Q$ can be understood as the bootstrapped sample covariance matrix for $\Sigma^{1/2}X$ where $D$ represents the resampling scheme or random weights.

In this paper, assuming that $p$ and $n$ are comparably large, we prove that the extreme edge eigenvalues of $Q$  can have several types of distributions  depending on $\Sigma$ and $D$ asymptotically. These distributions include: Gumbel, Fr{\'e}chet, Weibull, Tracy-Widom, Gaussian or their mixtures. On the one hand, when the random variables in $D$ have unbounded support, the edge eigenvalues of $Q$ can have either Gumbel or Fr{\'e}chet distribution depending on the tail decay property of $D.$ On the other hand, when the random variables in $D$ have bounded support, under some mild regularity assumptions on $\Sigma,$ the edge eigenvalues of $Q$ can exhibit Weibull, Tracy-Widom, Gaussian or their mixtures. The phase transitions rely on the behavior of the random variables of $D$ near the edges of their supports and the aspect ratio $p/n.$ Based on our theoretical results,  we consider two important applications. First, we propose some statistics and procedure based on edge statistics
to detect and estimate the possible spikes for elliptically distributed data. We also justify their superior theoretical properties. Second, in the context of a factor model, by using the multiplier bootstrap procedure via selecting the weights in $D,$ we establish the second order asymptotic relation of the eigenvalues between the bootstrapped and unbootstrapped sample covariance matrices. Based on the results, we propose a new algorithm to infer and estimate the number of factors in the factor model. Numerical simulations also confirm the accuracy and powerfulness of our proposed methods and illustrate better performance compared to some existing methods in the literature.     
	\end{abstract}

\maketitle

%\setcounter{tocdepth}{1}

%\tableofcontents

\section{Introduction}
Covariance matrix plays prominent roles in almost every aspect of multivariate data analysis. In the last few decades, due to technological advancements and availability of massive data
collected from novel resources, there has been a growing interest in developing methodologies
and tools to response to this high-dimensionality and complexity. This situation is certainly not
suited for the classical multivariate statistics, but rather calls for techniques from high dimensional
statistics \cite{yao2015sample}. Consider $\mathbf{y}_i \sim \mathbf{y} \in \mathbb{R}^p,  1 \leq i \leq n,$ are i.i.d. observations of a random vector $\mathbf{y} $ that
\begin{equation}\label{eq_generatingmodel}
\mathbf{y}=\xi  T \mathbf{x} \in \mathbb{R}^p, 
\end{equation}
where $\xi \in \mathbb{R}$ is a random variable, $T^*T=\Sigma \in \mathbb{R}^{p \times p}$ is some positive definite deterministic matrix, and $\mathbf{x} \in \mathbb{R}^p$ is a random vector that is independent of $\xi.$ For high dimensionality, we mean that $p$ and $n$ are comparably large.

The model (\ref{eq_generatingmodel}) is referred to as the \emph{generalized elliptical model} \cite{Karoui2009} which finds important applications in statistics.  We now list but a few examples. First, when $\mathbf{x}$ is distributed on the unit sphere $\mathbb{S}^{p-1},$ it becomes the elliptically distributed data as in \cite{fang1990} which includes multivariate Pearson and multivariate student-$t$ distributions as special examples. The model is commonly used in finance, robust statistics and signal processing to model the heterogeneity and heavy tailness \cite{Hu2019,hu2019central,kariya2014robustness, owen1983class,schmidt2003credit,usseglio2018estimation}. Second, when $\mathbf{x}$ contains i.i.d. centered random variables with variance $n^{-1},$ it has been widely used in financial econometrics, multivariate data analysis and modern statistical learning theory \cite{el2018impact, el2013robust,li2018structure,  yang2021testing}. Finally, when $\xi$ is regarded as a random weight or sampling, (\ref{eq_generatingmodel}) is closely related to high dimensional bootstrap method and deep learning theory. For example,  in  \cite{el2019non,2022arXiv220206188Y}, the authors consider the problem
of non-parametric bootstrap for individual eigenvalues where $\xi$ follows the multinomial distribution. For another instance, in \cite{lopes2019bootstrapping,naumov2019bootstrap}, the authors consider
the construction of bootstrapped confidence intervals for spectral projectors or spectrum where $\xi$ follows either Gaussian distribution or Pareto distribution.  Additionally, in the analysis of neural networks \cite{PLone,PLthree,PLtwo}, the columns of the input-output
Jacobian matrices have the form of (\ref{eq_generatingmodel}); 
see Examples \ref{exam_elliptical} and \ref{exam_boostrapping} for more details. 

Given the samples $\mathbf{y}_i=\xi_i T \mathbf{x}_i, 1 \leq i \leq n,$ we can write the data matrix  $Y=TXD,$ where $X=(\mathbf{x}_i)$ and $D$ is a diagonal matrix containing $\{\xi_i\}.$ Then the sample covariance matrix can be constructed as follows 
\begin{equation}\label{eq_samplecov}
Q:=YY^* \equiv TXD^2X^*T^*.
\end{equation}
We refer the readers to Section \ref{sec_modelandassumption} for more precise definitions. An important topic in the statistical study of sample covariance matrices is the
asymptotics of the largest eigenvalues of $Q.$ They are of great interests to signal processing \cite{Bao2015,9779233,4493413,silverstein1992signal}, principal component analysis and factor model analysis \cite{bai2002determining,fan2008high,fan2022estimating,onatski2009testing,onatski2010determining,passemier2014estimation}, covariance matrix testing \cite{johnstone2020testing,johnstone2018pca}, statistical learning theory \cite{fan2019spectral, fan2020spectra, han2021eigen} and time series analysis \cite{bi2022spiked,7347462,zhang2018clt}, to name but a few. Motivated by these applications, in this paper, we study the largest eigenvalues of sample covariance matrices $Q$ in (\ref{eq_samplecov}) under the generalized elliptical model (\ref{eq_generatingmodel}) and consider various applications in signal detection and high dimensional bootstrap. In what follows, we first provide a summary of some related results in Section \ref{subsec_existingresult}. Then we offer an overview of our results in Section \ref{sec_overviewofresults}.

%{\color{red} introduce why interested in the extreme values. }  

\subsection{Summary of some existing related results}\label{subsec_existingresult}
In this subsection, we summarize the results related to the model (\ref{eq_samplecov}) in random matrix theory and high dimensional statistics literature, with a focus on the extreme eigenvalues and their applications. 

It can be seen that (\ref{eq_samplecov}) has a \emph{separable covariance} structure with $D$ being random and possibly unbounded. In the literature, such a model has been studied to some extent when both $T$ and $D$ are bounded and deterministic. Under this setting, on the global scale, the empirical spectral distribution (ESD) of $Q$ can be best formulated by its Stieltjes transform whose limit can be described via a system of two equations. For example, when $X$ contains i.i.d. entries satisfying some moment assumptions, the system of equations has been derived and studied in various works, for example \cite{couillet2014analysis,Karoui2009,paul2009no,Zhanggeneral}. A special case is when $D=I,$ the two equations will be reduced to a single equation whose solution is known as the (deformed) \emph{Marchenko-Pastur} (MP) law as studied in \cite{MarchenkoandPastur1967}. Similar results have been obtained when $X$ contains i.i.d. columns sampled from the unit sphere \cite{Karoui2009,hu2019central}.  
 
 On the local scale, individual edge eigenvalues have also been studied in various contexts. Especially, under some regularity conditions and moment assumptions, the edge eigenvalues of $Q$ will obey the \emph{Tracy-Widom} (TW) distribution \cite{tracy1994level} asymptotically. Such results have first been established for the case when $D=I$ under various moment assumptions in \cite{Bao2015,Ding&Yang2018,el2007tracy,hachem2016large,Joha2000,John2001,knowles2017anisotropic,lee2016tracy,PillaiandYin2014} and then extended to the general $D$ settings in \cite{9779233,yang2019edge}.  We emphasize that in order to see the TW law, the limiting ESD near the edge needs to exhibit a square root decay behavior. Additionally, motivated by statistical applications, another research line is to study the outlier eigenvalues of $Q$ if a few spikes are added on $T^*T$ or $D^2.$ When $X$ has i.i.d. entries, the \emph{spiked covariance matrix} model (i.e., $D=I$)  has been studied under different setups, for example, in \cite{BY, Baik2005,Baik2006,bao2022statistical,benaych2011eigenvalues,bloemendal2016principal,ding2021spiked11,John2001, zhang2022asymptotic} and the \emph{spiked separable covariance matrix} model (i.e., general $D$) has been investigated in \cite{ding2021spiked}. In these spiked models, roughly speaking, when the spikes are larger than some threshold, the outlier eigenvalues will detach from the bulk spectrum and follow  Gaussian asymptotically. The non-outlier eigenvalues will stick to the right-most edge of the spectrum and follow TW law asymptotically.

In contrast, less is touched when $D$ is random. On the macroscopic scale, the limiting ESD of $Q$ has been studied to some extend in \cite{Karoui2009,hu2019central,Zhanggeneral}. Especially, when $X$ contains independent samples from the unit sphere or i.i.d. entries satisfying certain moment conditions, the Stieltjes transforms of the limiting ESD can also be characterized by two equations. Much less is known for each individual edge eigenvalue.  On the one hand, when $X$ has i.i.d. columns sampled from the unit sphere, under strong conditions on $D^2$ so that the Stieltjes transform of the liming ESD of $Q$ is governed by only one equation, \cite{Wen2021} proved that the TW law held for the extreme eigenvalues and  \cite{Hu2019} studied an associated spiked model. We emphasize that even though the aforementioned two papers allow certain randomness on $D,$ it requires that its entries are infinitely divisible satisfying certain moment assumptions. In this sense, after being properly scaled, $D$ is almost deterministic and isotropic. This also explains the reason that this assumption reduces two equations to only one for the limiting ESD. To relax the assumptions on $D$ for elliptically distributed data, we mention that in a recent work \cite{DX_ell}, the first two authors of this paper prove that when $D$ has bounded support and the limiting ESD exhibits a square-root decay behavior, the edge eigenvalues of $Q$ also follow TW asymptotically. On the other hand, when $X$ has i.i.d. entries, $\Sigma=I$ and the entries of $D$ have bounded support,   it was shown in \cite{Kwak2021} that the edge eigenvalues can have Weibull or Gaussian distributions depending on the edge behavior of the density of the entries of $D.$

%{\color{red} Hu and Zhou's paper: Random is not real random. Fake. Connect the results with 
%}   
 
%Similar results have been obtained when $X$ has i.i.d. columns in  under strong assumptions on $D^2.$           

\subsection{An overview of our results, contributions and novelties}\label{sec_overviewofresults}

In this subsection, we provide an informal overview of our results and outline the main contributions and novelties of our paper. Motivated by the statistical applications discussed above and the challenges summarized in Section \ref{subsec_existingresult}, we study the distributions of the edge eigenvalues (cf. the first few largest eigenvalues) of $Q$ in (\ref{eq_samplecov}) when $D^2$ is random and possibly unbounded, and $X$ either contains independent columns from unit sphere or i.i.d. entries.  We prove that depending on $D^2,$ $\Sigma$ and the aspect ratio $p/n,$ the edge eigenvalues can have various types of distributions asymptotically, including the three extreme value distributions for sequences of i.i.d. random variables \cite{beirlant2004statistics} (cf. Gumbel, Fr{\' e}chet and Weibull), the TW law, Gaussian, or the mixture of TW law and Gaussian. To guide the readers, we now give a heuristic
description of our results.

On the one hand, when $D$ has unbounded support, for convenience, as in Assumption \ref{assum_D}, we impose some mild assumptions on the tail decay behavior due to their popular usage which cover most of the commonly used random variables with unbounded supports. In Theorem \ref{thm_main_unbounded} below, we prove that after being properly scaled by some constant $\varphi$ (cf. (\ref{eq_defnvarphi})), with probability tending to one, $\varphi^{-1} \lambda_1(Q)$ will be sufficiently close to $\xi_{(1)}^2.$ Furthermore, depending on the behavior of $D$, since $\xi_{(1)}^2$ follows Fr{\' e}chet distribution for polynomial type tail decay (cf. (\ref{ass3.1})) and Gumbel distribution for exponential type tail decay (cf. (\ref{ass3.2})) (see Lemma \ref{lem_summaryevt} below), we see that $\varphi^{-1} \lambda_1(Q)$ will follow either Fr{\' e}chet or Gumbel distribution.           

On the other hand, when $D$ has bounded support, say $(0,l],$  we prove that under some mild regularity conditions on $\Sigma$ (cf. Assumption \ref{assum_additional_techinical}), $\lambda_1(Q)$ can have different behaviors depending on the properties of the cumulative distribution function (CDF) of the entries of $D^2$ near $l,$ $\Sigma$ and $p/n.$ We assume that the CDF near $l$ exhibits polynomial decay with exponent $d+1$ as in (\ref{ass3.4}). In Theorem \ref{thm_main_bounded} below, we prove that if $d>1,$ there exists a threshold $\varsigma_3 \equiv \varsigma_3(D, \Sigma)$ in (\ref{eq_phasetransition}) so that when $n/p>\varsigma_3,$ there exist some constants $c_1, c_2$ and $c_3,$  with probability tending to one, $c_1n^{\frac{1}{d+1}}(\lambda_1(Q)-c_2)$ will be sufficiently close to $c_3 n^{\frac{1}{d+1}}(\xi_{(1)}^2-l).$ Since the distribution of the latter statistic can be studied using extreme value theory (cf. Lemma \ref{lem_summaryevt}), we conclude that after being properly scaled and centered, $\lambda_1(Q)$ will follow Weibull distribution asymptotically. Moreover, when $d>1$ but $n/p<\varsigma_3,$ a transition will occur in the sense that $\lambda_1(Q)$ will be influenced by an average of all $\{\xi_i^2\}$ so that $\lambda_1(Q)$ will follow Gaussian asymptotically. Finally, when $-1<d \leq 1,$ the asymptotic distribution of $\lambda_1(Q)$ can be characterized by a convolution of TW law and Gaussian with potentially vanishing variance. Especially, we will see a sharp transition between TW limit and Gaussian limit when the variance of the Gaussian part crosses the order of $n^{-1/6}.$

Our results also find important applications in high dimensional statistics. In Section \ref{sec_statapplication}, we consider two such applications. First, based on the edge statistics, we propose some methodology to detect and estimate the number of spikes in $\Sigma$ for the generalized elliptical model (\ref{eq_generatingmodel}) for general settings of $\xi$ with possibly unbounded support and heavy tails. Second, in the context of high-dimensional factor model, we show that the multiplier bootstrap procedure together with a resampling scheme can still work if the random sampling weights are properly chosen. For better theoretical understanding of our methodologies, we establish the first order convergence results and transitions for outlier and nonoutlier eigenvalues in Theorem \ref{thm_main_spike}, and further theoretical justification of our applications can be found in Corollary \ref{sec_spikeellipticaldata}.  Moreover, in Theorem \ref{sec_boostrapping}, we establish the results on the eigenvalues of the sample covariance matrices with and without multiplier bootstrap. We remark that our theory and applications are highly compatible. The theoretical results are not only interesting and natural on their
own, they are also highly motivated by and indispensable in the applications, which are all fundamental problems in the statistics literature.

We highlight several technical components, insights and novelties of our paper. We refer the readers to Section \ref{sec_sec_ofproof} for more details.  The arguments for $D$ with unbounded and bounded supports are quite different. First, when $D$ has unbounded support, the eigenvalues will be divergent. In this setting, we utilize a perturbation argument. However, as in this case, the limiting ESD of $Q$ may also have unbounded support, the perturbation approach developed in \cite{bloemendal2016principal,ding2021spiked,knowles2013isotropic} cannot be applied directly. Instead, we modify the perturbation arguments by isolating $\mathbf{y}_i$ corresponding to $\xi_{(1)}^2 $ from the data matrix $Y$ as in (\ref{eq_samplecov}); see Section \ref{sec_sec_ofproof} for more detailed discussion. We mention that in our discussion, the scaling in ESD is still of the order $n^{-1}$ so that the smaller eigenvalues of $Q$ are bounded and the larger eigenvalues will be divergent. By doing so, the Stieltjes transforms are still governed by a system of two equations. In contrast, for some other random matrix models with i.i.d. heavy tailed entries \cite{auffinger2009poisson11}, where the first few eigenvalues are also divergent, the authors used the scaling so that the larger eigenvalues were bounded and the smaller eigenvalues would be vanishing. This will result in Poisson convergence. Due to the complicatedness of our model (\ref{eq_generatingmodel}), we find that it is more involved to apply the idea of \cite{auffinger2009poisson11}, if applicable. Second, when $D$ has bounded support, our arguments are non-perturbative and generalize those used in \cite{lee2016extremal,Kwak2021}. In this case, the limiting ESD is bounded and we denote its rightmost edge as $L_+.$ When $d>1$ and $n/p>\varsigma_3,$  $L_+$ can be fully characterized by only one equation (cf. (\ref{eq_onlyoneequationdecide})) and $\lambda_1$ can be connected to $\xi_{(1)}^2$ by a sophisticated understanding of the Stieltjes transforms of the limiting ESDs. On the other hand, when $n/p<\varsigma_3,$ $L_+$ and some other quantities involving the Stieltjes transforms will be determined together by a system of equations (cf. (\ref{eq_edgeequationstwodecide})). Then the distribution of $\lambda_1$ can be connected with an average of all $\{\xi_i^2\}.$ Similar arguments apply to the setting $-1<d \leq 1;$ see Section \ref{sec_sec_ofproof} for more details. Finally, our actual proof relies on two technical inputs. One is the detailed analysis of the Stieltjes transforms of the limiting ESD on local scales in some carefully chosen spectral domains. The other one is the local laws of the resolvents; see Section \ref{sec_asymptoticlocalaveragedlocal} for more details.     

%Third, [mention the ESD behavior and local laws]          
%{\color{red} also need to highlight some other results we have proved. }

%{\color{red} some comments on future works: Moreover, when $X$ is a Haar matrix, the systems of equations can be best written in terms of the subordination functions and the ESD will converge to the free convolution of $T$ and $D$ \cite{}.}

The rest of this article is organized as follows. In Section \ref{sec_modelsetup}, we define the generalized elliptical models, provide some examples and some basic assumptions. We also give the asymptotic laws. In Section \ref{sec_mainresults}, we provide the main results and offer the sketch of the proof strategies. In Section \ref{sec_statapplication}, we study two statistical applications using our established results. Theoretical justifications and numerical simulations are also provided.   The technical proofs are deferred to the appendices. In particular, in Appendix \ref{appendix_preliminary}, we provide some preliminary results and the averaged local laws.  In Appendix \ref{appendxi_locallawproof}, we prove the averaged local laws near the edges. In Appendix \ref{sec_locationofeigenvalues}, we provide the asymptotic locations of the edge eigenvalues and prove the main results and  other results related to our statistical applications. Finally, some auxiliary lemmas are proved in Appendix \ref{appendix_last}.  

%and additional simulation results are enclosed in Appendix  \ref{appendix_lastlast}.    

\vspace{10pt}
\noindent {\bf Conventions.} Let $\mathbb{C}_+$ be the complex upper half plane. We denote $C>0$ as a generic constant whose value may change from line to line. For two sequences of deterministic positive values $\{a_n\}$ and $\{b_n\},$ we write $a_n=\rO(b_n)$ if $a_n \leq C b_n$ for some positive constant $C>0.$ In addition, if both $a_n=\rO(b_n)$ and $b_n=\rO(a_n),$ we write $a_n \asymp b_n.$ Moreover, we write $a_n=\ro(b_n)$ if $a_n \leq c_n b_n$ for some positive sequence $c_n \downarrow 0.$ Moreover, for a sequence of random variables $\{x_n\}$ and positive real values $\{a_n\},$ we use $x_n=\rO_{\mathbb{P}}(a_n)$ to state that $x_n/a_n$ is stochastically bounded. Similarly, we use $x_n=\ro_{\mathbb{P}}(a_n)$ to say that $x_n/a_n$ converges to zero in probability. For a sequence of positive random variables $\{y_n\},$ we use $y_{(k)}, 1 \leq k \leq n,$ for its order statistics with $y_{(1)} \geq y_{(2)} \geq \cdots \geq y_{(n)}>0.$

\section{Generalized elliptical models and asymptotic laws}\label{sec_modelsetup}

\subsection{The model, motivating examples and basic assumptions}\label{sec_modelandassumption}
In this subsection, we introduce our model and some assumptions. Throughout the paper, we consider data matrix of the following form
\begin{equation}\label{eq_datamatrix}
Y=TXD,
\end{equation}
where $T$ is a $p_1 \times p, \ p_1 \geq p,$ deterministic matrix, $D$ is an $n \times n$ diagonal random matrix containing i.i.d. random variables, and $X$ is a $p \times n$ random matrix independent of $D.$

Motivated by statistical applications and for the purpose of definiteness, we consider the following two general classes of random matrix models in the form of (\ref{eq_datamatrix}). To avoid repetition, we summarize the model settings as follows. 

\begin{assumption}\label{assum_model}
Throughout the paper, we consider the following two model settings in the form of (\ref{eq_datamatrix}):
\begin{enumerate}
\item[(1).] {\bf Elliptically distributed data.} In this setting, we assume that the columns of $X$ are i.i.d. distributed on the unit sphere; that is to say 
\begin{equation*}
X=(\ub_1, \cdots, \ub_n), \ \ub_i \overset{\mathrm{i.i.d.}}{\sim} \mathsf{U}(\mathbb{S}^{p-1}).
\end{equation*}  
Moreover, we assume that $T^*T=\Sigma.$ Due to rotational invariance, without loss of generality, we assume that $\Sigma$ is a diagonal matrix so that  $\Sigma=\operatorname{diag}\left\{ \sigma_1, \cdots, \sigma_p \right\}. $
\item[(2).] {\bf Separable covariance i.i.d. data.} In this setting, we assume that the entries of $X=(x_{ij})$ are centered  i.i.d. random variables satisfying that for $1 \leq i \leq p, 1 \leq j \leq n,$
\begin{equation}\label{eq_standard1n}
\mathbb{E} x_{ij}=0, \ \mathbb{E} x_{ij}^2=\frac{1}{n}. 
\end{equation} 
Moreover, we assume that for all $k \in \mathbb{N},$ there exists some constant $C_k>0$ so that $\mathbb{E}|\sqrt{n} x_{ij}|^k \leq C_k. $ Finally, for $T,$ we assume that $p_1=p$ and $T=\Sigma^{1/2}$ for some positive definite matrix $\Sigma.$
\end{enumerate} 
\end{assumption}

To illustrate the generality and usefulness of the concerned models in Assumption \ref{assum_model}, we provide a few examples and discuss their concrete applications in the statistical literature. 

\begin{example}\label{exam_elliptical} For model (1) of Assumption \ref{assum_model}, according to \cite{CAMBANIS1981368}, the columns of $Y, \{\mathbf{y}_i\}, 1 \leq i \leq n,$ follow the elliptical distributions. More specifically, we can denote 
\begin{equation}\label{eq_ellipiticalmodel}
\mathbf{y}_i=\xi_i T \mathbf{u}_i,
\end{equation}
where $\xi_i \geq 0$ are some nonnegative random variables independent of $\mathbf{u}_i.$ The class of elliptical distributions are  natural generalization of the multivariate normal distributions which remains the
simple linear dependence structure but allows for heavy tails. For example, when $\xi_i \sim \sqrt{F_{p, \nu}},$ where $F_{p,\nu}$ is an $F$ distributed random variable with $p$ and $\nu$ degrees of freedom, then $\mathbf{y}_i$ has a multivariate student-t distribution with $\nu$ degree of freedom and dispersion matrix $\Sigma$ provided it has full rank; see Example 4 of \cite{frahm2004generalized} for more details. Such a distribution has found important applications in finance \cite{schmidt2003credit}.

\end{example}

\begin{example}\label{exam_boostrapping}
For model (2) of Assumption \ref{assum_model}, it has been used in many different contexts. We can write the columns of $Y, \{\mathbf{y}_i\}, 1 \leq i \leq n,$ as follows
\begin{equation*}
\mathbf{y}_i=\xi_i \Sigma^{1/2} \mathbf{x}_i. 
\end{equation*}
First, when $\mathbf{x}_i$'s are Gaussian, the data has been used in \cite{el2018impact,el2013robust} to study the performance of high dimensional robust regressions and used in \cite{hu2019central} to study the large MIMO systems in wireless communications. Second, when the entries $\mathbf{x}_i$ have more general distributions,  in \cite{li2018structure,yang2021testing}, the model has been utilized to study the covariance structures in various settings. Third, when $\xi_i$'s are chosen as the random sampling weights, the model has been used to study the high dimensional bootstrap in \cite{el2019non,naumov2019bootstrap,2022arXiv220206188Y}. Finally, model (2) with the data matrix (\ref{eq_datamatrix}) appears frequently in deep neural networks and are closely related to the  input-output Jacobian matrices \cite{PLone, PLthree, PLtwo}.       
\end{example}

\quad In what follows, we study the extreme singular values of $Y,$ i.e., the first few largest eigenvalues of the $p \times p$ sample covariance matrix $Q$ in (\ref{eq_samplecov}). Or equivalently, the first few largest eigenvalues of its $n \times n$ companion $\mathcal Q$ 
\begin{equation}\label{eq_gram}
\mathcal Q:=D X^*T^* T X D \equiv DX^*\Sigma XD.
\end{equation}  
\noindent In the rest of this subsection, we introduce the  two main technical assumptions. The first assumption (cf. Assumption \ref{assum_D}) is imposed on the random diagonal matrix $D.$
\begin{assumption}\label{assum_D}
Let $D^2=\operatorname{diag} \left\{ \xi_1^2, \cdots, \xi_n^2 \right\}.$ Moreover, for its entries, we assume $\xi_i^2 \sim \xi^2, 1 \leq i \leq n,$ are i.i.d. generated from a nonnegative and non-degenerated random variable $\xi^2$ satisfying the following assumptions.
\begin{enumerate}
\item[(i)] {\bf Unbounded support case.} We assume that $\xi^2$ has an unbounded support and  satisfies either of the following two conditions:  \\
 (a). $\xi^2$ is a  regularly varying random variable  \cite{resnick2008extreme}  that 
	\begin{equation}\label{ass3.1}
	    \mathbb{P}(\xi^2>x)=\frac{L(x)}{x^{\alpha}},\quad x\rightarrow\infty,
	\end{equation}
	for some $\alpha\in[2,+\infty)$, where $L(x)$ is a slowly varying function in the sense that for all $t>0,$ $\lim_{x \rightarrow \infty} L(tx)/L(x)=1.$

\vspace{3pt}
	
(b). $\xi^2$ has an exponential decay tail in the sense that for some constant $\beta>0$ and any fixed constant $t>0$
\begin{equation}\label{ass3.2}
\mathbb{E}e^{t\xi^{2\beta}}<\infty. 
\end{equation}
%{\color{red} [need to add more assumption here, in some cases, the distribution is Tracy-Widom]}
\item[(ii)] {\bf Bounded support case}.  We assume that
 $\xi^2$ has a bounded support on $(0,l]$ for fixed some constant
$l>0.$ Moreover, for some constant $d>-1,$ we assume that  
%the tail probability around the hard edge $l$ satisfies
\begin{equation}\label{ass3.4}
\mathbb{P}( l-\xi^2 \leq x) \asymp x^{d+1}.
\end{equation}
Finally, let $F(x)$ be the cumulative distribution function (CDF) of $\xi^2,$ we assume that  
\begin{equation}\label{eq_defnbfrak}
0< \mathfrak{b}:= \lim_{x \uparrow l} \frac{1-F(x)}{(l-x)^{d+1}}<\infty. 
\end{equation}

%for some small constant $C>0$ and $d>0.$ Moreover, when $d>1,$ for definiteness, we assume that 
%\begin{equation}\label{ass3.3}
%\sigma_1=\mathrm{o}(l). 
%\end{equation}
%{\color{purple} [This assumption is surely redundant.] }
%{\color{red} need to add some discussion on the above condition}
\end{enumerate}

%Without loss of generality, we assume that 

\end{assumption}

\begin{remark}
Several remarks are in order. First, for the unbounded case, (\ref{ass3.1}) indicates that the tails of $\xi^2$ decay polynomially. Many commonly used distributions are included in this category. To name but a few, Pareto distribution, $F$ distribution and student-$t$ distribution. Moreover, according to extreme value theory (see Lemma \ref{lem_summaryevt} below), when (\ref{ass3.1}) is satisfied, $\xi_{(1)}^2$ follows Fr{\' e}chet distribution asymptotically. Second, for the unbounded setting, (\ref{ass3.2}) implies that the tails of $\xi^2$ decay exponentially. In fact, by elementary calculations \cite{hansen2020three}, it is not hard to see that when (\ref{ass3.2}) holds, it is necessarily that the CDF of $\xi^2$ admits 
\begin{equation}\label{eq_defng}
\mathbb{P}(\xi^2>x)=\exp(-\mathsf{g}(x)),
\end{equation}
for some positive decreasing function $\mathsf{g}(x)>0.$ Furthermore, if \begin{equation}\label{assum_gg}
\mathsf{g} \in C^{\infty}([0, \infty)), \ \ \lim_{x \uparrow \infty} (1/\mathsf{g}'(x))'=0,
\end{equation}     
we see from Lemma \ref{lem_summaryevt} that $\xi_{(1)}^2$ follows Gumbel distribution asymptotically. In fact, many commonly used distributions, for instance, Chi-squared distribution, exponential distribution and Gamma distribution, satisfy these conditions. Third, for the bounded case, (\ref{ass3.4}) indicates that $\xi^2$ has a possible polynomial decay behavior near the edge. Under the assumption of (\ref{eq_defnbfrak}), we see from Lemma \ref{lem_summaryevt} that $\xi_{(1)}^2$ obeys Weibull distribution asymptotically. The conditions allow for many distributions like (shifted) Beta distribution, uniform distribution and U-quadratic distribution. In summary, we emphasize that our assumptions in Assumption \ref{assum_D} are general and mild and cover many commonly used examples. In contrast, as mentioned in Section \ref{subsec_existingresult}, existing literature only handles deterministic or nearly deterministic $\xi_i^2, 1 \leq i \leq n.$

%{\color{red} add some remarks: 1. on the assumptions and their implications on extreme value theory; also examples like Pareto distribution, Gamma distribution and Beta distribution. 2. some examples from extreme value theory and then combining with the examples in elliptical distribution and bootstrapping for more data examples. }
%
%{\color{red}  According to the discussion and summary in Remark \ref{rem_extremevalue_appendix},  In order...., we assume that 
%
%
%add comments that (a) is frequently used in other models like i.i.d. assuming entries have different moments. \cite{Heiny2017,Heiny2018}  and Others' paper like Ben Nahous. 
%}
%{\color{blue}
%The regularly varying random variable for $\alpha\in[2,+\infty)$ involves some heavy tailed distributions, typical examples include the Pareto distribution with parameter $\alpha$ and $t$-distribution with $\alpha$ degrees of freedom. Moreover, we may observe that $\mathbb{E}(\xi^4)<\infty$ when $\alpha=2$. Since we have the identity
%\[
%\mathbb{E}(\xi^4)=4\int_{0}^{+\infty}\mathbb{P}(\xi^4\ge r)r^3dr,
%\]
%then the fourth moment of $\xi$ is finite if and only if $\mathbb{P}(\xi^4\ge r)r^3$ is integrable on the real line with respect to Lebesgue measure. Therefore, it is equivalent to say the tail probability of $\xi$ satisfies 
%\[
%\mathbb{P}(\xi^4\ge r)\le\frac{C}{r^4},
%\]
%for some positive constant $C$, then $\xi$ has finite fourth moment. Hence, the tail condition \eqref{eq: tail condition on xi} is sufficient to imply the finite fourth moment case of $\xi$ when $\alpha=2$.
%}
\end{remark}

\quad The second assumption (cf. Assumption \ref{assumption_techincial}) introduces some mild conditions on the aspect ratio $p/n$ and the population covariance matrix $\Sigma.$

\begin{assumption}\label{assumption_techincial}
We assume the following conditions hold true for some small constant $0<\tau<1$. 
\begin{enumerate}
\item[(i)] {\bf On dimensionality}. Throughout the paper, we consider the high dimensional regime that 
\begin{equation}\label{ass1}
 \tau \leq \phi:=\frac{p}{n} \leq \tau^{-1}. 
\end{equation}
\item[(ii)] {\bf On $\Sigma.$} For the eigenvalues of $\Sigma,$ denoted as $\sigma_i, 1 \leq i \leq p,$ we assume that
\begin{equation}\label{ass2}
\tau \leq \sigma_p \leq \sigma_{p-1} \leq \cdots \leq \sigma_2 \leq \sigma_1 \leq \tau^{-1}. 
\end{equation}
\end{enumerate}
\end{assumption}

\quad We remark that (\ref{ass1}) is commonly used in random matrix theory and high dimensional statistics literature for quantifying the high dimensionality. (\ref{ass2}) states the eigenvalues of the population covariance matrix are bounded from above and below. On the one hand, when $\xi^2$ has unbounded support as in Case (i) of Assumption \ref{assum_D}, (\ref{ass2}) is the only assumption imposed on $\Sigma.$  On the other hand, when $\xi^2$ has bounded support as in Case (ii) of Assumption \ref{assum_D}, we will provide an additional mild assumption, Assumption \ref{assum_additional_techinical}, after some necessary notations are introduced.

%\begin{remark}

%\end{remark}

\subsection{Resolvents and asymptotic laws}
%{\color{red}[Notations needed to be cleaned up]}. 
%
%{\color{red}[The discussions of the i.i.d. case and elliptical case are different. Need to revise here]}
%{\color{purple} [FROM HERE]}
In this subsection, we introduce some results on the limiting global laws of the eigenvalues of the sample covariance matrices. 
Recall that the empirical spectral distributions (ESD) of $Q$ and $\mathcal{Q}$ in (\ref{eq_samplecov}) and (\ref{eq_gram}) are defined as 
\begin{equation*}
\mu_{Q}:=\frac{1}{p}\sum_{i=1}^p \delta_{\lambda_i(Q)},\quad \mu_{\mathcal{Q}}:=\frac{1}{n}\sum_{j=1}^n \delta_{\lambda_j(\mathcal{Q})}.
\end{equation*}
For $z=E+\ri \eta \in \mathbb{C}_+,$ denote the resolvents  
%
%
% Define the Green functions for sample covariance matrices $Q,\mathcal{Q}$ as,
\begin{equation}\label{eq_resolvents}
G(z)=(Q-zI)^{-1} \in \mathbb{R}^{p \times p},\quad \mathcal{G}(z)=(\mathcal{Q}-zI)^{-1} \in \mathbb{R}^{n \times n}.
\end{equation}
Correspondingly, the Stieltjes transforms are denoted as %{\color{red}[we may need to change this a a little bit]}
\begin{equation}\label{eq_mq}
m_{Q}:=\int\frac{1}{x-z}\mu_{Q}=\frac{1}{p}{\rm tr}(G(z)),\quad m_{\mathcal{Q}}:=\int\frac{1}{x-z}\mu_{\mathcal{Q}}=\frac{1}{n}{\rm tr}(\mathcal{G}(z)). 
\end{equation}
Since $Q$ and $\mathcal{Q}$ share the same non-trivial eigenvalues, it suffices to study $\mu_Q$ and $m_Q.$ To characterize the limit of $\mu_Q,$ similarly to \cite{couillet2014analysis, ding2021spiked,Karoui2009, hu2019central,paul2009no, Zhanggeneral}, we consider a system of equations. To avoid repetitions, we summarize these equations in the following definition. 
\begin{definition}[Systems of consistent equations]\label{defn_couplesystem} For $z \in \mathbb{C}_+,$ we define the triplets $(m_{1n}, m_{2n}, m_n) \in \mathbb{C}^3_+,$  via the following systems of equations.
\begin{enumerate}
\item When $Y$ in (\ref{eq_datamatrix}) is generated from the elliptically distributed data as in {\normalfont Case (1)} of Assumption \ref{assum_model},  the equations denoted are as follows
\begin{align}\label{eq_systemequationsm1m2elliptical}
   & m_{1n}(z)=\frac{1}{p}\sum_{i=1}^p\frac{\sigma_i}{-z(1+\sigma_i  m_{2n}(z))},\quad m_{2n}(z)=\frac{1}{p}\sum_{i=1}^n\frac{\xi_i^2}{-z(1+\xi^2_im_{1n}(z))},\\
   & m_{n}(z)=\frac{1}{p}\sum_{i=1}^p\frac{1}{-z(1+\sigma_i  m_{2n}(z))}. \nonumber
\end{align}
\item When $Y$ in (\ref{eq_datamatrix}) is generated from the separable covariance i.i.d. data as in {\normalfont Case (2)} of Assumption \ref{assum_model}, the equations are denoted as follows
\begin{align}\label{eq_systemequationsm1m2}
   & m_{1n}(z)=\frac{1}{n}\sum_{i=1}^p\frac{\sigma_i}{-z(1+\sigma_im_{2n}(z))},\quad m_{2n}(z)=\frac{1}{n}\sum_{i=1}^n\frac{\xi_i^2}{-z(1+\xi^2_im_{1n}(z))},\\
   & m_{n}(z)=\frac{1}{p}\sum_{i=1}^p\frac{1}{-z(1+\sigma_im_{2n}(z))}. \nonumber
\end{align}
\end{enumerate}
\end{definition}

\quad For sufficiently large $n,$ we find that $\mu_Q$ has a nonrandom deterministic equivalent and can be uniquely characterized by the above consistent equations.
This is summarized by the following theorem.   
%
%\quad  {\color{red} existence and uniqueness? }
%
%Since we are mainly interested in the extreme eigenvalues, conditional on $D,$ we now focus our discussion on the domains 
%For matrix $\mathcal{Q}$, we define the limiting version of a system self-consistent equations,
%\begin{gather*}
%    m_{1\phi_0}(z)=\phi_0\int\frac{t}{-z(1+tm_{2\phi_0}(z))}dF_{\Sigma}(t),\quad m_{2\phi_0}(z)=\int\frac{t}{-z(1+tm_{1\phi_0}(z))}dF_{\xi^2}(t),\\
%    m_{\phi_0}(z)=\int\frac{1}{-z(1+tm_{2\phi_0}(z))}dF_{\Sigma}(t),\quad z\in\mathbb{C}^{+},
%\end{gather*}
%then $m_{1\phi_0}(z)$ and $m_{2\phi_0}(z)$ have unique solutions on $z\in\mathbb{C}^{+}$, \cite{}. And define the finite sample version,
%{\color{red}[from here]}
%then for any $z\in\mathbb{C}^{+}$, $m_{1n}(z)$ converges in probability to $m_{1\phi_0}(z)$, and similarly for $m_{2n}(z)$ and $m_{n}(z)$. In another word, the self-consistent equations for $m_{1n}$ and $m_{2n}$ also admit a unique solution pair on $\mathbb{C}^{+}$.

\begin{theorem}\label{lem_solutionsystem}
Suppose Assumptions \ref{assum_model}, \ref{assum_D} and \ref{assumption_techincial} hold. Then conditional on some event $\Omega \equiv \Omega_n$ that $\mathbb{P}(\Omega)=1-\ro(1),$ for any $z \in \mathbb{C}_+,$ when $n$ is sufficiently large, there exists a unique solution $(m_{1n}(z), m_{2n}(z), m_{n}(z)) \in \mathbb{C}_+^3$ to the systems of equations in (\ref{eq_systemequationsm1m2elliptical}) and (\ref{eq_systemequationsm1m2}). Moreover, $m_n(z)$ is the Stieltjes transform of some probability density function $\rho \equiv \rho_n$ defined on $\mathbb{R}$ which can be obtained using the inversion formula.
\end{theorem}
\begin{proof}
The proofs can be obtained by following lines of the arguments  of \cite[Theorem 2]{Karoui2009} and \cite[Theorem 1]{paul2009no}, or \cite[Theorem 2.4]{ding2021spiked} verbatim. We omit the details. 

% For self-completeness, we will provide some further detail in Appendix \ref{sec_sub_asymptotic}. 
\end{proof}

\begin{remark}\label{rmk_systemequationsremark}
Several remarks on Theorem \ref{lem_solutionsystem} are in order. First, the probability event $\Omega$ can be constructed explicitly as in Defition \ref{defn_probset} and Lemma \ref{lem_probabilitycontrol}. Second, we prove an unconditional counterpart for Theorem \ref{lem_solutionsystem} by integrating out the randomness of $\xi^2.$ Recall $F(x)$ is the CDF of $\xi^2.$ We take (\ref{eq_systemequationsm1m2}) for an example where the systems of equations are defined as follows
\begin{align}\label{eq_systemequationsm1m2intergrate}
   & m_{1n,c}(z)=\frac{1}{n}\sum_{i=1}^p\frac{\sigma_i}{-z(1+\sigma_im_{2n,c}(z))},\quad m_{2n,c}(z)=\int_0^l \frac{s}{-z(1+s m_{1n,c}(z))} \mathrm{d} F(s),\\
   & m_{n,c}(z)=\frac{1}{p}\sum_{i=1}^p\frac{1}{-z(1+\sigma_im_{2n,c}(z))}. \nonumber
\end{align}
In this setting, $(m_{1n,c}, m_{2n,c}, m_{n,c})$ is always deterministic. Especially, when $\xi^2$ has bounded support as in Case (ii), we can actually obtain stronger results as in \cite{paul2009no} that the support of the associated probability density function $\widetilde{\rho}$ is bounded and denoted as 
\begin{equation}\label{edge_edge_edge}
\text{supp}(\widetilde{\rho})=[L_-, L_+].   
\end{equation}
Finally, we will see later that the conditional and unconditional version are both useful in their own aspects. To be more specific, the conditional version is more powerful when $\xi^2$ has unbounded support whereas the unconditional version is more convenient when $\xi^2$ has bounded support. 
%{\color{red} remarks: 1. the probability event. and can be moved when it is bounded.  2. some key issues in the proof and literature. }
\end{remark}

\quad Thanks to Theorem \ref{lem_solutionsystem}, it is easy to see that the study of the systems of equations in (\ref{eq_systemequationsm1m2elliptical}) and (\ref{eq_systemequationsm1m2}) can be reduced to the analysis of 
\begin{equation}\label{eq_functionFequal}
F_n(m_{1n}(z),z)=0, \quad z\in\mathbb{C}_{+},
\end{equation}
where $F_n(\cdot, \cdot)$ are defined as follows corresponding to (\ref{eq_systemequationsm1m2elliptical}) and (\ref{eq_systemequationsm1m2}), respectively 
\begin{equation}\label{eq:F(m,z)1}
    F_n(m_{1n}(z),z)
    =\frac{1}{p}\sum_{i=1}^p\frac{\sigma_i}{-z+\frac{\sigma_i}{p}\sum_{j=1}^n\frac{\xi^2_j}{1+\xi^2_j m_{1n}(z)}}-m_{1n}(z),
\end{equation} 
and 
\begin{equation}\label{eq:F(m,z)}
    F_n(m_{1n}(z),z)
    =\frac{1}{n}\sum_{i=1}^p\frac{\sigma_i}{-z+\frac{\sigma_i}{n}\sum_{j=1}^n\frac{\xi^2_j}{1+\xi^2_j m_{1n}(z)}}-m_{1n}(z).
\end{equation}

\quad Finally, armed with the above notations, we introduce some additional assumption on $\Sigma$ which will be used when $\xi^2$ has bounded support in the sense of (ii) of Assumption \ref{assum_D}. Such an assumption has been frequently used in the random matrix theory literature, for example, see \cite{Bao2015,Ding&Yang2018,ding2021spiked,9779233,el2007tracy,knowles2017anisotropic,lee2016tracy}. Recall the notations $m_{2n,c}$ and $L_+$ in Remark \ref{rmk_systemequationsremark}. 

\begin{assumption}\label{assum_additional_techinical} When (ii) of Assumption \ref{assum_D} holds, for $\Sigma$ satisfying Assumption \ref{assumption_techincial}, we assume that for some constant $\tau>0$ 
\begin{equation*}
\min_{1 \leq i \leq p} |1+\sigma_i m_{2n,c}(L_+)| \geq \tau. 
\end{equation*}  
\end{assumption}

%\[
%,
%\]
%also admits the same solutions. 

\section{Main results and proof strategies}\label{sec_mainresults}
%\subsection{Extreme eigenvalues}
  
\subsection{Main results}
Our main results are summarized in Theorems \ref{thm_main_unbounded} and \ref{thm_main_bounded} below. We first provide the results for the extreme eigenvalues when $\xi^2$ has unbounded support in the sense that (i) of Assumption \ref{assum_D} holds. Denote
\begin{equation*}
\bar{\sigma}=\frac{1}{p} \sum_{i=1}^p \sigma_i.
\end{equation*}
Recall $F(x)$ is the cumulative distribution function (CDF) of $\xi^2_i, 1 \leq i \leq n.$ Denote 
\begin{equation}\label{eq_defnbn}
b_n:=\inf \left\{ x: 1-F(x) \leq \frac{1}{n} \right\}.
\end{equation}

%{\color{red} need to introduce the results of extreme value theory}

Let $\lambda_1$ be the largest eigenvalue of $Q$. Our first result is stated as follows. 
%In this paper, we always define $a_{np}^2:=n^{4/\alpha}\vee n$. It can be seen that when $\alpha\in(2,4)$, then $a_{np}^2=n^{4/\alpha}$.

\begin{theorem}[The unbounded support case]\label{thm_main_unbounded} Suppose Assumptions \ref{assum_model}, \ref{assumption_techincial} and (i) of Assumption \ref{assum_D} hold. Then we have that when $n$ is sufficiently large
\begin{equation*}
\frac{\lambda_1}{\xi_{(1)}^2}=\varphi+\ro_{\mathbb{P}}(1).
\end{equation*}
where 
\begin{equation}\label{eq_defnvarphi}
\varphi:=
\begin{cases}
\bar{\sigma}, & \text{if {\normalfont Case (1)} of Assumption \ref{assum_model} holds},\\
\phi \bar{\sigma}, & \text{if {\normalfont Case (2)} of Assumption \ref{assum_model} holds}. 
\end{cases}
\end{equation}
Consequently, when (\ref{ass3.1}) holds, $\varphi^{-1}\lambda_1 $ follows the Fr{\'e}chet distribution asymptotically in the sense that for $x \geq 0$
\begin{equation}\label{eq_mainresultunboundedfrechet}
\lim_{n \rightarrow \infty} \mathbb{P} \left( \frac{\lambda_1}{ \varphi b_n} \leq x \right)=\exp\left(-x^{-\alpha} \right).
\end{equation}
Moreover, when (\ref{ass3.2}) and (\ref{assum_gg}) hold, $ \varphi^{-1}\lambda_1 $ follows the Gumbel distribution asymptotically in the sense that for $x \in \mathbb{R}$
 \begin{equation}\label{eq_mainresultunboundedgumbel}
 \lim_{n \rightarrow \infty} \mathbb{P}\left( \mathsf{g}'(b_n) \left[\varphi^{-1}\lambda_1-b_n \right] \leq x \right)=\exp \left(-e^{-x} \right),
 \end{equation}
 where we recall $\mathsf{g}$ is defined in (\ref{eq_defng}). 

%\begin{enumerate}
%\item[(1).] , we have that
%%\begin{equation*}
%%\frac{\lambda_1}{\xi_{(1)}^2}=\phi \bar{\sigma}+\ro_{\mathbb{P}}(1).
%%%=\phi \bar{\sigma} \xi_{(1)}^2 \left(1+\rO_{\prec} (n^{2/\alpha-1}+n^{-1/2}) \right). 
%%\end{equation*}
%Consequently, 
%\begin{equation*}
%\lim_{n \rightarrow \infty} \mathbb{P} \left( \frac{(\phi \bar{\sigma})^{-1}\lambda_1}{b_n} \leq x \right)=\exp\left(-x^{-\frac{\alpha}{2}} \right).
%\end{equation*}
%\item[(2).] When , we have that with high
% probability 
% \begin{equation*}
% \lambda_1=\phi \bar{\sigma} \xi_{(1)}^2 \left(1+\rO(\log^{-1}n) \right). 
% \end{equation*}
% Consequently, $ (\phi \bar{\sigma})^{-1}\lambda_1 $ follows the Gumbel distribution asymptotically in the sense that 
%
%\end{enumerate}

\end{theorem}

\begin{remark}\label{rmk_mainresults_unbounded}
Three remarks are in order. First, Theorem \ref{thm_main_unbounded} states that when $\xi^2$ has unbounded support, $\lambda_1$ will be divergent. Moreover, after being properly centered and scaled, $\lambda_1$ will have a similar behavior to $\xi_{(1)}^2.$ Especially, when $\xi^2$ has a polynomial decay tail as in (\ref{ass3.1}), we can obtain the Fr{\' e}chet limit and when $\xi^2$ has an exponential decay tail as in (\ref{ass3.2}), we can get the Gumbel limit.    Second, the above results can be generalized to the joint distribution of $k$ largest eigenvalues for any fixed $k.$ That is, for all $s_i \in \mathbb{R}, 1 \leq i \leq k,$ (\ref{eq_mainresultunboundedfrechet}) can be generalized to 
\begin{equation*}
\lim_{n \rightarrow \infty} \mathbb{P}  \left( \left( \frac{\lambda_i}{\varphi b_n} \leq s_i \right)_{1 \leq i \leq k} \right)=\lim_{n \rightarrow \infty}\mathbb{P} \left( \left( \frac{\xi^2_{(i)}}{b_n} \leq s_i \right)_{1 \leq i \leq k} \right),
\end{equation*}
and (\ref{eq_mainresultunboundedgumbel}) can be generalized to 
\begin{equation*}
\lim_{n \rightarrow \infty} \mathbb{P}  \left( \mathsf{g}'(b_n)\left(\varphi^{-1} \lambda_i-b_n \leq s_i \right)_{1 \leq i \leq k} \right)=\lim_{n \rightarrow \infty} \mathbb{P}  \left( \mathsf{g}'(b_n)\left(\xi_{(i)}^2-b_n \leq s_i \right)_{1 \leq i \leq k} \right).
\end{equation*}
Since the joint distribution of the order statistics of $\{\xi_i^2\}$ can be computed explicitly \cite{coles2001introduction}, the above formulas give a complete description of the finite-dimensional correlation
functions of the extremal eigenvalues. Third, Theorem \ref{thm_main_unbounded} shows that even when $\Sigma$ has no spikes, due to the effect of extreme values of $\{\xi_i^2\},$ the first few eigenvalues of $Q$ can also be divergent. Consequently, in order to be properly detected, if exist, the true spikes of $\Sigma$ have to be divergent; see Theorem \ref{thm_main_spike} for more detail.

\end{remark}

%\begin{remark}
%{\color{red} add some concrete examples}
%\end{remark}

\quad Next, we state the results when $\xi^2$ has bounded support in the sense that (ii) of Assumption \ref{assum_D} holds. Recall $F(x)$ and $l$ from (\ref{eq_defnbfrak}) and $L_+$ from (\ref{edge_edge_edge}). Let 
\begin{equation*}
    \mathsf{s}_1:=\int_0^l\frac{l^2s^2}{(l-s)^2}\mathrm{d}F(s),\quad \mathsf{s}_2:= \int_0^l\frac{ls}{l-s}\mathrm{d}F(s),\quad \mathsf{s}_3:=\frac{1}{p} \sum_{i=1}^p \frac{\sigma_i^2 \varsigma_1}{(L_+-\sigma_i \varsigma_2)^2}, \ \mathsf{s}_4:=\frac{1}{n}\sum_{i=1}^p \frac{\sigma_i}{(-L_++\sigma_i \varsigma_2)^2},
\end{equation*}
and 
\begin{equation*}
\mathsf{v}:=\left( \frac{s}{1+s m_{1n,c}(L_+)} \right)^2 \mathrm{d}F(s)-\left( \int \frac{s}{1+s m_{1n,c}(L_+)}  \mathrm{d}F(s) \right)^2.
\end{equation*}
Using the above notations, we further denote that for $k=1,2,3,4,$
\begin{equation}\label{eq_phasetransition}
      \varsigma_k:=
    \begin{cases}
    \phi^{-1}\mathsf{s}_k, & \text{when (1) of Assumption \ref{assum_model} holds} \\
    \mathsf{s}_k,  & \text{when (2) of Assumption \ref{assum_model} holds} \\
    \end{cases}.
\end{equation}
%
%\begin{equation}\label{eq_phasetransition}
%    \varsigma_1=\int_0^l\frac{l^2s^2}{(l-s)^2}\mathrm{d}F(s),\quad \varsigma_2=\int_0^l\frac{ls}{l-s}\mathrm{d}F(s),\quad \varsigma_3=\frac{1}{p} \sum_{i=1}^p \frac{\sigma_i^2 \varsigma_1}{(L_+-\sigma_i \varsigma_2)^2}, \ \varsigma_4=\frac{1}{n}\sum_{i=1}^p \frac{\sigma_i}{(-L_++\sigma_i \varsigma_2)^2}. 
%%       \int\frac{t^2\varsigma_1}{(v-t\varsigma_2)^2}dF_{\Sigma}(t),
%\end{equation}
Moreover, we denote 
\begin{equation}\label{eq_cltvariance}
\vartheta:=
\begin{cases}
\phi^{-2} \mathsf{v}, & \text{when (1) of Assumption \ref{assum_model} holds} \\
 \mathsf{v}, & \ \text{when (2) of Assumption \ref{assum_model} holds}
\end{cases}.
\end{equation}
Our second result is summarized as follows. Recall the exponent $d$ in (\ref{ass3.4}).  
\begin{theorem}[The bounded support case]\label{thm_main_bounded} Suppose Assumptions \ref{assum_model}, \ref{assumption_techincial}, \ref{assum_additional_techinical} and (ii) of Assumption \ref{assum_D} hold. Then we have that when $n$ is sufficiently large, 
\begin{enumerate}
\item[(1).]  When $d>1$ and $\phi^{-1}>\varsigma_3$ in (\ref{eq_phasetransition}), we have that $L_+$ satisfies
\begin{equation*}
    1=\frac{1}{p}\sum_{i=1}^p \frac{-l\sigma_i}{-L_{+}+\sigma_i \varsigma_2}, \ \text{when (1) of Assumption \ref{assum_model} holds}, 
\end{equation*}
and 
\begin{equation}\label{eq_onlyoneequationdecide}
    1=\frac{1}{n}\sum_{i=1}^p \frac{-l\sigma_i}{-L_{+}+\sigma_i \varsigma_2}, \ \text{when (2) of Assumption \ref{assum_model} holds}. 
\end{equation}
Moreover, we have that 
\begin{equation}\label{eq_boundednnnnn}
n^{\frac{1}{d+1}}\left| \left( \frac{\lambda_1-L_+}{\varsigma_4^{-1}(1-\phi \varsigma_3)} \right)-\left( \xi_{(1)}^2-l \right) \right|=\ro_{\mathbb{P}}(1). 
\end{equation} 
Consequently, we have that $\lambda_1-L_+$ follows Weibull distribution with parameter $d+1$ asymptotically in the sense that for $x \leq 0$
\begin{equation}\label{eq_distributionresultweibull}
\lim_{n\rightarrow \infty} \mathbb{P} \left( \frac{(\mathfrak{b}n)^{d+1}}{\varsigma_4^{-1}(1-\phi \varsigma_3)}(\lambda_1-L_+) \leq x \right)=\exp\left(-|x|^{d+1} \right),
\end{equation}
where $\mathfrak{b}$ is defined in (\ref{eq_defnbfrak}). 
\item[(2).] When $d>1$ and $\phi^{-1}<\varsigma_3,$ we have that $\lambda_1$ is asymptotically Gaussian in the sense that 
\begin{equation}\label{eq_cltresult}
\lim_{n\rightarrow \infty} \mathbb{P} \left( \sqrt{n \vartheta^{-1}} (\lambda_1-L_+) \leq x \right)=\Phi(x),
\end{equation}
where $\Phi(x)$ is the CDF of a real standard Gaussian random variable. 
\item[(3).] When $-1<d \leq 1,$ we have that 
\begin{equation*}
\lambda_1-L_+=\nu_1+\nu_2+\rO_{\mathbb{P}}(n^{-1}),
\end{equation*}
where for some $\gamma$ defined in (\ref{eq_gammadefinition}) below, $n^{2/3} \gamma \nu_1$ follows the type-1 Tracy-Widom law asymptotically  and $\sqrt{n}\nu_2/\vartheta^{1/2}$ follows standard Gaussian distribution asymptotically. More specifically, if $\vartheta=\ro(n^{-1/3}),$ $$ \lim_{n \rightarrow \infty}\left(n^{2/3}\gamma (\lambda_1-L_+) \leq x \right)=\mathrm{T}(x),$$ where $\mathrm{T}(x)$ is the CDF of the type-1 Tracy-Widom distribution. Moreover, if $\vartheta \gg n^{-1/3},$ (\ref{eq_cltresult}) holds. 
\end{enumerate}

\end{theorem}

\begin{remark}
When $\xi^2$ has bounded support as in (\ref{ass3.4}), $\lambda_1$ will be bounded and can have several phase transitions depending on the exponent $d,$ aspect ratio $\phi$ and the threshold $\varsigma_3$ which encodes the information of $\Sigma$ and the distribution of $\xi^2.$ We provide several remarks here. First, in the setting when $d>1,$ on the one hand, when $\phi^{-1}>\varsigma_3,$ after being properly centered and scaled, $\lambda_1$ will have similar asymptotics as $\xi_{(1)}^2$ and Weibull limit will be obtained. On the other hand when $\phi^{-1}<\varsigma_3,$ $\lambda_1$ will be influenced by all $\{\xi_i^2\}$ and hence asymptotically Gaussian. For the critical case $\phi^{-1}=\varsigma_3,$ we believe there will be a phase transition connecting Gaussian and Weibull. Since this is out of the scope of the paper which focuses on statistical applications, we will pursue this direction in the future works. Second, when $-1<d \leq 1,$  $\rho$ in Theorem \ref{lem_solutionsystem} will have a square root decay behavior.  In this setting, $\lambda_1$ will be influenced by two components, the TW part $\nu_1$ and the Gaussian part $\nu_2.$ The TW part is due to the square root behavior and the Gaussian is due to the fact that $\lambda_1$ will be potentially influenced by all $\{\xi_i^2\};$ see Section \ref{sec_sec_ofproof} for more details. We mention that the variance of the Gaussian part can potentially decay and $\nu_1$ and $\nu_2$ are in generally dependent. Finally, as discussed in Remark \ref{rmk_mainresults_unbounded}, we can generalize the results of Theorem \ref{thm_main_bounded} to the joint distribution of $k$ largest eigenvalues for any fixed $k.$ We omit the details.     
\end{remark}

\subsection{Strategy for the proof}\label{sec_sec_ofproof}
In this subsection, we provide a sketch of the proof strategies. The bounded and unbounded settings will require different treatments. For simplicity of the discussion, we focus on the separable covariance i.i.d. data as in Case (2) of Assumption \ref{assum_model}. Similar arguments apply to the elliptical data as in Case (1) of Assumption \ref{assum_model} with minor modifications. 

%{\color{red}[introduce and summarize the strategy of master equations in \cite{ding2021spiked} and related literature]}

%{\color{red}[add more once merge JHX's file]}

\quad We start with the discussion of the main idea of proof when $\xi$ is bounded, i.e., Theorem \ref{thm_main_bounded}. The arguments are non-perturbative and require a sophisticated understanding on the systems of equations as in (\ref{eq_systemequationsm1m2}) on the local scales. A crucial input is the local behavior of the asymptotic law near the rightmost edge of the spectrum (cf. Lemma \ref{localestimate2}) which is the generalization of \cite{lee2016extremal}. More specifically, for our concerned matrix (\ref{eq_samplecov}), it turns out that the behavior varies according to the combination of the exponent $d$ in (\ref{ass3.4}) and the aspect ratio $p/n.$ Especially, for the unconditionally density function as in Remark \ref{rmk_systemequationsremark},  we find that the asymptotic law has square root decay behavior (i.e., $\widetilde{\rho}(x)$ in (\ref{edge_edge_edge}) satisfies that $\widetilde{\rho}(x) \sim \sqrt{L_+-x}$) when  either $-1<d \leq 1$ or $d>1$ and $\phi^{-1}<\varsigma_3$ as in (\ref{eq_phasetransition}). Moreover, the square root behavior will be updated to general polynomial behavior (i.e., $\widetilde{\rho}(x) \sim (L_+-x)^{d}$) when $d>1$ and $\phi^{-1}>\varsigma_3$); see Lemma \ref{localestimate2} for a more precise statement.  

\quad In the actual proof, we need to introduce a quantity $\widehat{L}_+$ which is the edge of the density function $\rho$ as in Theorem \ref{lem_solutionsystem} for some fixed realization of $\{\xi_i^2\}$ in some high probability event $\Omega$. In fact, as will be proved in Lemma \ref{localestimate2}, $L_+=\widehat{L}_++\rO_{\mathbb{P}}(n^{-1/2+\delta})$ for some small constant $\delta>0$ and the properties of $\widetilde{\rho}$ can be inherited to $\rho$ when conditional on $\Omega.$ Then for the former case when square root decay behavior exhibits, we can check that the assumptions of \cite{DX_ell,9779233} are satisfied so that conditional on $\Omega,$ $n^{2/3}(\lambda_1-\widehat{L}_+)$ obeys TW law asymptotically with constant order variance. For the unconditional distribution, it suffices to analyze the fluctuation of $\widehat{L}_+.$           In general, due to the i.i.d. assumption of $D^2,$ by CLT, $\widehat{L}$ is asymptotically Gaussian whose variance can decay to zero. Therefore, overall, the distribution can be represented as a summation of TW law and Gaussian (with possibly vanishing variance) as in (3) of Theorem \ref{thm_main_bounded}. Especially, when $d>1$ and $\phi^{-1}<\varsigma_3,$ the variance of the Gaussian part is of constant order one so that $\lambda_1$ is asymptotically Gaussian as in (2) of Theorem \ref{thm_main_bounded}. We mention that in this setting, the edge eigenvalues are influenced by all $\{\xi_i^2\}$ and the edge $\widehat{L}_+$ is regular in the sense that (\ref{rem1nbound}) holds.

\quad For the latter case when square root behavior disappears, i.e., (1) of Theorem \ref{thm_main_bounded}, the discussion is more subtle. Our idea provides a nontrivial generalization of \cite{lee2016extremal}.  Similarly to the aforementioned argument, we shall mainly work with $\widehat{L}_+-\lambda_1$ for some fixed realization $\{\xi_i^2\}$ from $\Omega$.  In this setting, a key observation is that $\widehat{L}_+$ is irregular and  can be represented as the solution of the equation that (cf. see (\ref{eq_conditionaledgedefinition}))
\begin{equation}\label{eq_masterequationtwo}
m_{1n}(\widehat{L}_+)=-l^{-1}. 
\end{equation}
Moreover, for $z \in \mathbb{C}_+$ in a small neighborhood of $\widehat{L}_+,$ $m_{1n}(z)$ can be expanded linearly with an error much smaller than $n^{-1/(d+1)}$ as in part (a) of Lemma \ref{localestimate2}. That is, for some constant $\alpha,$ $m_{1n}(z)-m_{1n}(\widehat{L}_+)=\alpha (z-\widehat{L}_+)+\ro(n^{-1/(d+1)}).$ In order to prove (\ref{eq_boundednnnnn}), we separate our proof into two steps.  First, we show that $\lambda_1$ is close to $\widehat{L}_+.$ The discussion relies on two parts. In the first part, we introduce some auxiliary quantity $E_0$ as follows. For some fixed sufficiently small constant $\epsilon_d>0,$ let 
\begin{equation}\label{eq_eta0definition}
\eta_0=n^{-1/2-\epsilon_d},
\end{equation}
 and $z_0=E_0+\ri \eta_0$ as the solution of the following equation   
\begin{gather}\label{eq: def of gamma}
    \operatorname{Re}m_{1n}(E_0+\ri\eta_0)=-\xi^{-2}_{(1)}.
\end{gather}
In fact, as discussed in Remark \ref{rmk_boundedatleastonesolution}, there exists at least one solution satisfying (\ref{eq: def of gamma}) and we always choose the one with the largest $E_0$. Combining (\ref{eq_masterequationtwo}), (\ref{eq: def of gamma}) and the linear expansion around $m_{1n}(\widehat{L}_+)$ for $z_0$, we see that $\widehat{L}_+$ is close to $E_0;$ see the proof of Proposition \ref{prop_boundedsetting}. In the second part, we prove that $\lambda_1$ is close to $E_0$ by analyzing the Stieltjes transforms of the matrix (\ref{eq_samplecov}); see (\ref{eq_holdspartoneoneone}).  Second, we connect $\lambda_1$ with $\xi_{(1)}^2$ by approximating (\ref{eq: def of gamma}) that $\operatorname{Re} m_{1n}(\lambda_1+\ri \eta_0) \approx -\xi_{(1)}^{-2}$; see (\ref{eq_secondconnect}). Combining the linear expansion around $\widehat{L}_+$ for $\lambda_1+\ri \eta_0,$ we can conclude the proof of (\ref{eq_boundednnnnn})  for $\widehat{L}_+.$ For the unconditional result with $L_+,$ it follows directly from Lemma \ref{localestimate2} that $L_+=\widehat{L}_++\rO_{\mathbb{P}}(n^{-1/2+\delta})$ and $d>1.$ Consequently, we can conclude that $\lambda_1$ is only influenced by $\xi_{(1)}^2$ and asymptotically Weibull using the extreme value theory as summarized in Lemma \ref{lem_summaryevt}. We point out the proof of the second step relies on the averaged local law. Due to the non-square root decay behavior and lack of stability bound, the local laws cannot be analyzed as in \cite{ding2021spiked,knowles2017anisotropic,lee2016tracy,PillaiandYin2014,yang2019edge}. As mentioned in (\ref{eq_masterequationtwo}), even the characterization of the edge $\widehat{L}_+$ is different from the form in the aforementioned works. To address this issue, we need to adapt and generalize the strategies of \cite{Kwak2021,lee2016extremal}.  

%...{\color{red} from here}
%
% First, when $d<1,$ {\color{red}[finish here]}.  Second, when $d>1,$

\quad Then we discuss the unbounded $\xi^2$ case as in Theorem \ref{thm_main_unbounded}. In these settings, the proof strategy is perturbative since the edge is divergent. Instead of directly analyzing the systems as in Definition \ref{defn_couplesystem}, we introduce a real auxiliary quantity $\mu_1>0$ governed by some master equations. More specifically, for the  pair $(\mu_1, m_{1n}(\mu_1))$ satisfying $F_n(m_{1n}(\mu_1), \mu_1)=0$ for $F_n(\cdot, \cdot)$ defined in (\ref{eq:F(m,z)1}) and (\ref{eq:F(m,z)}), we set $\mu_1$ to be the largest solution of the following equation
\begin{equation}\label{eq: def of mu_1}
    1+(\xi^2_{(1)}+d_1)m_{1n}(\mu_{1})=0,
\end{equation}
where for some fixed sufficiently small constant $\epsilon>0,$ $d_1$ is defined as 
\begin{equation}\label{eq_firstddefinition}
d_1:=
\begin{cases}
n^{1/\alpha-\epsilon}, & \text{if (\ref{ass3.1}) holds}; \\
1, & \text{if (\ref{ass3.2}) holds}. 
\end{cases}
\end{equation} 
According to the results of  extreme value theory as summarized in Lemma \ref{lem_summaryevt}, we have that $d_1=\ro_{\mathbb{P}}(\xi_{(1)}^2).$ As will be seen in (\ref{eq: 2.2 beta=2}), $d_1$ is introduced mainly for some technical reasons.

On the one hand, in contrast to (\ref{eq: def of gamma}), we point out that $m_{1n}(\mu_1)=\lim_{\eta \downarrow 0} m_{1n}(\mu_1+\ri\eta)$ is valid.  In fact, we can replace $z_1 \in \mathbb{C}_+$ with $\mu_1$ in (\ref{eq: def of mu_1}), according to (\ref{eq_functionFequal}) and (\ref{eq:F(m,z)}), by an argument similar to (\ref{eq_mu1part}), we have that $z_1 \asymp \xi_{(1)}^2. $ Then the validity is guaranteed by Lemma \ref{lem: basic bounds} and Remark \ref{rem_keyrem}; see Remark \ref{rem_JHX} for more detail. Second, observe that
\begin{equation}\label{eq: 2.1}
    \mathsf{h}(\mu_1):=\frac{1}{n}\sum_{i=1}^p\frac{\sigma_i}{\frac{\mu_1}{\xi^2_{(1)}+d_{1}}-\frac{\sigma_i}{n}\sum_{j=1}^n\frac{\xi_j^2}{\xi^2_{(1)}+d_{1}-\xi^2_j}}=1.
\end{equation} 
It is clear that $\mathsf{h}(\mu_1)$ a continuous decreasing function on $\mu_1$ and not hard to see there always exists a solution on the interval,
\[
\left(\mathrm{a},+\infty \right), \ \mathrm{a} \equiv \mathrm{a}(n):=\frac{\sigma_1}{n}\sum_{j=1}^n\frac{\xi^2_j(\xi^2_{(1)}+d_{1})}{\xi^2_{(1)}+d_{1}-\xi^2_j}.
\]
In fact, using (\ref{eq: 2.2 beta=2}) below, we find that when $n$ is sufficiently large, $\mathsf{h}(\mathrm{a}) \gg 1$ and $\mathsf{h}(\infty)=0.$ 
%{\color{red} add some discussion on triangle inequality here}

Now we proceed to explain how to utilize (\ref{eq: def of mu_1}) to conclude the proof. On the one hand, (\ref{eq: def of mu_1}) provides a natural way to connect $\mu_1$ and $\xi_{(1)}^2.$ In fact, by (\ref{eq: 2.1}) and the properties of order statistics of $\{\xi_i^2\}$ as in (\ref{def1}) and (\ref{def3}), for $\varphi$ defined in (\ref{eq_defnvarphi}), we see from (\ref{eq_mu1part}) that $\mu_1/\xi_{(1)}^2=\varphi+\ro_{\mathbb{P}}(1).$ On the other hand, (\ref{eq: def of mu_1}) links $\mu_1$ with $\lambda_1$ with a perturbation argument similar to \cite{bloemendal2016principal,ding2021spiked11, ding2021spiked}. The discussion contains two steps. In step one, we see that according to (\ref{eq_masterequation}), $\lambda_1$ can be uniquely characterized by the equation $M(\lambda_1)=0$ where $M(\cdot)$ is a random quantity defined in (\ref{eq_defnmlambda}) by isolating the column in (\ref{eq_datamatrix}) associated with $\xi_{(1)}^2$. To connect $M(\mu_1)$ back with (\ref{eq: def of mu_1}), we need to establish the local laws. Thanks to the divergence of $\mu_1$ and the large deviation control as in Lemma \ref{lem:large deviation}, we only need to prove the averaged local law as in Theorem \ref{thm_unboundedcaselocallaw}; see equations (\ref{eq_conconconone}) and (\ref{eq_bigexpansion}) for more details. The above arguments basically yields that $1+(\xi_{(1)}^2+d_1) m_{1n}(\lambda_1) \approx 0.$ Together with (\ref{eq: def of mu_1}), by a continuity and stability analysis, we can show that $\lambda_1/\mu_1=1+\ro_{\mathbb{P}}(1).$ Combining the above arguments, we can build the connection between $\lambda_1$ and $\xi_{(1)}^2$ and conclude the proof together with extreme value theory as summarized in Lemma \ref{lem_summaryevt}.

We mention that the two crucial equations for the bounded and unbounded settings are (\ref{eq: def of gamma}) and (\ref{eq: def of mu_1}), respectively.  The proof of the bounded case is non-perturbative because the random counterpart of (\ref{eq: def of gamma}) (cf. (\ref{eq_secondconnect})) is introduced non-perturbatively. In contrast, the unbounded case is perturbative since the random counterpart of (\ref{eq: def of mu_1}) (cf. (\ref{eq_defnmlambda})) is introduced by isolating the column of the data matrix associated with $\xi_{(1)}^2.$  Even though the pertubative approach is usually used in the literature for studying the deformed random matrix models, to list but a few \cite{bao2022statistical,benaych2011eigenvalues,bloemendal2016principal,ding2020high,ding2021spiked} or see Theorem \ref{thm_main_spike} below, it is also natural for us to study our model in the unbounded setting since $\xi_{(1)}^2$ is unbounded and also well separated from other order statistics of $\{\xi_i^2\}$ with high probability; see (\ref{def1}) and (\ref{def3}).

\section{Statistical applications}\label{sec_statapplication}
In this section, we will consider several statistical applications to illustrate the usefulness of the established results. In Section \ref{sec_spikeellipticaldata}, we will provide some useful statistics and procedure to detect and estimate the number of spikes for data generated from a possibly spiked elliptically distributed model as in Example \ref{exam_elliptical}. In Section \ref{sec_boostrapping}, we examine the performance for high dimensional bootstrap with different weights as discussed in Example \ref{exam_boostrapping}. 

\quad Before proceeding to the discussion of the two applications, we first pause to introduce the spiked population covariance matrix model which has gained increasing popularity in high dimensional data analysis. To avoid confusion, in what follows, we will consistently use $\Sigma$ for a population covariance matrix without spikes and use $\widetilde{\Sigma}$ for a population covariance matrix possibly with spikes. Moreover,    following \cite{ding2021spiked11,ding2021spiked,John2001}, for $\Sigma=\sum_{i=1}^p \sigma_i \mathbf{v}_i \mathbf{v}_i^*,$ where $\{\mathbf{v}_i\}$ are the eigenvectors of $\Sigma,$ and some fixed integer $r \geq 0,$ we define the spiked model  $\widetilde{\Sigma}$ as follows  
\begin{equation}\label{eq_truemodelspiked}
\widetilde{\Sigma}=\sum_{i=1}^p \widetilde{\sigma}_i \mathbf{v}_i \mathbf{v}_i^*,
\end{equation}   
where $\widetilde{\sigma}_i=\sigma_i$ for $i>r$ and $\widetilde{\sigma}_1 \geq \widetilde{\sigma}_2 \geq \cdots \geq \widetilde{\sigma}_r>\widetilde{\sigma}_{r+1}$ are $r$  values representing the larger spikes. When $r=0,$ we notice that $\widetilde{\Sigma}=\Sigma.$

\subsection{Detection and estimation of the number of spikes for elliptically distributed data}\label{sec_spikeellipticaldata}
%\subsection{Signal detection under general noise structure}
Detection and estimation of the number of spikes is an important problem in many areas such as signal processing \cite{4493413,silverstein1992signal}, financial economics \cite{bai2002determining,fan2022estimating,onatski2009testing} and biomedical research \cite{MR3904784, ke2021estimation}. Most of the existing results and methods have focused on the spiked sample covariance matrix that $D=I$ in (\ref{eq_samplecov}) and $X$ has i.i.d. entries, to list but a few, \cite{bai2002determining,Bao2015,CHP, MR3904784,fan2022estimating,ke2021estimation,onatski2009testing, passemier2014estimation}. In this paper, we consider such a problem for elliptically distributed data when (1) of Assumption \ref{assum_D} holds. 

Motivated by the discussions in Example \ref{exam_elliptical}, for the matrix $D,$ we consider the unbounded support case as in (i) of Assumption \ref{assum_D} which allows high dimensional heavy tailed data sets. More specifically, we consider the data matrix $\widetilde{Y}=(\widetilde{\mathbf{y}}_i)$ that 
\begin{equation}\label{eq_datamodel}
\widetilde{\mathbf{y}}_i=\xi_i \widetilde{\Sigma}^{1/2} \mathbf{u}_i \in \mathbb{R}^p, \ 1 \leq i \leq n, 
\end{equation}      
where $\mathbf{u}_i, 1 \leq i \leq n,$ are independently distributed on the unit sphere and $\widetilde{\Sigma}$ is denoted in (\ref{eq_truemodelspiked}). Let the non-zero eigenvalues of $\widetilde{Y} \widetilde{Y}^*$ be 
\begin{equation}\label{eq_eigenvalueintheendspiked}
\mu_1 \geq \mu_2 \geq \cdots \geq \mu_{\min\{p,n\}}>0. 
\end{equation}
We are interested in proposing some statistics based on $\{\mu_i\}$ to study the following hypothesis testing problem  
\begin{equation}\label{eq_testingproblem}
\mathbf{H}_0: r=r_0 \ \ \text{vs} \ \ \mathbf{H}_a: r>r_0, 
\end{equation}
where $r_0$ is some pre-given integer representing our belief of the true value of $r.$ When $r_0=0,$ it reduces to the detection of the existence of signals. Note that based on (\ref{eq_testingproblem}), we can generate a natural sequential testing estimate of $r$, that is 
\begin{equation}\label{eq_sequencentialtestprocedure}
\widehat{r}:=\inf \left\{ r_0 \geq 0:  \mathbf{H}_0 \ \text{is accepted} \right\}. 
\end{equation}

In order to propose some data-adaptive statistics to test (\ref{eq_testingproblem}), we first analyze the theoretical proprieties of the first few largest eigenvalues of $\widetilde{Y} \widetilde{Y}^*.$ The results are summarized in the following theorem. Recall (\ref{eq_eigenvalueintheendspiked}) and denote 
\begin{equation*}
\mathsf{T}:= 
\begin{cases}
n^{1/\alpha} \log n, & \text{if (\ref{ass3.1}) holds}; \\
\log ^{1/\beta} n,  & \text{if (\ref{ass3.2}) holds}. 
\end{cases}
\end{equation*} 

\begin{theorem}\label{thm_main_spike} Suppose (i) of Assumption \ref{assum_D} and Assumption \ref{assumption_techincial} hold. For the spikes in (\ref{eq_truemodelspiked}), we assume that 
\begin{equation}\label{spiked_assumption}
\widetilde{\sigma}_r \gg \mathsf{T}. 
\end{equation}  
Then when $n$ is sufficiently large
\begin{enumerate}
\item[(1).] For $1 \leq i \leq r,$ we have that
\begin{equation*}
\frac{\mu_i}{\widetilde{\sigma}_i}=\phi^{-1} \mathbb{E} \xi^2+\mathrm{o}_{\mathbb{P}}(1). 
\end{equation*} 
\item[(2).] Recall $d_1$ in (\ref{eq_firstddefinition}) and  $\{\lambda_i\}$ are the eigenvalues of the non-spiked model $\Sigma^{1/2}XD^2X^* \Sigma^{1/2}$. For any fixed integer $k,$ we have that
\begin{equation}\label{eq_sticking}
\left|\mu_{r+i}-\lambda_i \right|=\rO_{\mathbb{P}} \left( n^{-1/2+2\epsilon} d_1 \right), \ 1 \leq i \leq k. 
\end{equation}
\end{enumerate} 

\end{theorem}

\begin{remark}\label{rem_multiplevaluespikedremark}
Two remarks are in order. First, on the one hand, part (1) of Theorem \ref{thm_main_spike} implies that the outlier eigenvalues have the same order as their associated spikes. On the other hand, part (2) shows that the extremal non-outlier eigenvalue will stick to the largest eigenvalue of the non-spiked model. Combining Theorem \ref{thm_main_spike} and Remark \ref{rmk_mainresults_unbounded}, we find that $\mu_{r+i}, 1 \leq i \leq k$ follow either Fr{\' e}chet or Gumbel distribution depending on the tail behavior of $\xi^2.$  Combining the results of (1) and (2), using (\ref{def1}) and (\ref{def3}), under the assumption of (\ref{spiked_assumption}), we find that the outliers and non-outliers are separated and this provides the theoretical guarantee for detecting the spikes. Second, in Theorem \ref{thm_main_spike}, motivated by our applications, we only prove the results for the elliptical data under the unbounded setting of $\xi^2$. However, similar results can be obtained for the seperable covariance i.i.d. data and bounded $\xi^2.$ Since this is out of the scope of the current paper, we will pursue this direction in the future, for example, see \cite{DXYZspiked}.  
\end{remark}

\quad Based on the discussions in Remark \ref{rem_multiplevaluespikedremark} that the distributions of $\{\mu_{r+i}\}$ are known, following the ideas of \cite{9779233,onatski2009testing}, we can make use of the non-outlier eigenvalues to test (\ref{eq_testingproblem}). However, in practice, the parameters in (\ref{eq_mainresultunboundedfrechet}) and (\ref{eq_mainresultunboundedgumbel}) are usually unknown. To resolve this issue, we use the following statistics
\begin{equation}\label{eq_statisticsdefinition1}
\mathbb{T} \equiv \mathbb{T}(r_0):=\max_{r_0< i \leq r_*} \frac{\mu_i-\mu_{i+1}}{\mu_{i+1}-\mu_{i+2}},
\end{equation}
and 
\begin{equation}\label{eq_statisticsdefinition2}
\mathbb{T}_{r_0}:=\frac{\mu_{r_0+1}-\mu_{r_0+2}}{\mu_{r_*+1}-\mu_{r_*+2}},
\end{equation}
where $r_*$ is a pre-chosen large fixed integer that is interpreted as the maximum possible number of spikes the model can have. It is easy to see from the discussions in Remark \ref{rem_multiplevaluespikedremark} that both $\mathbb{T}$ and $\mathbb{T}_{r_0}$ can be used to count the number of outlier eigenvalues that correspond to the spikes through a sequential testing procedure as in (\ref{eq_sequencentialtestprocedure}), except that (\ref{eq_statisticsdefinition2}) used fewer sample eigenvalues which can possibly have better finite sample performance as discussed in \cite{9779233}.

Corresponding to the above statistics, for some properly chosen critical value $\delta_n^{(1)}$ and $\delta_n^{(2)},$ we can follow (\ref{eq_sequencentialtestprocedure}) to define the sequential testing estimators 
\begin{equation}\label{eq_sequencialestimator}
 \widehat{r}_1:=\inf\left\{r_0 \geq 0: \ \mathbb{T}(r_0)<\delta_n^{(1)} \right\}, \  \widehat{r}_2:=\inf\left\{r_0 \geq 0: \ \mathbb{T}_{r_0}<\delta_n^{(2)} \right\}.
\end{equation}
Then we examine the properties of the proposed statistics and estimators and show that both of them will be consistent estimators of $r$ once the critical values are chosen properly. Denote 
\begin{equation*}
\mathbb{G}_1:=\max_{1 \leq i \leq r_*-r_0} \frac{\xi_{(i)}^2-\xi_{(i+1)}^2}{\xi_{(i+1)}^2-\xi_{(i+2)}^2}, \ \ \mathbb{G}_2:= \frac{\xi_{(1)}^2-\xi_{(2)}^2}{\xi_{(r_*-r_0+1)}^2-\xi_{(r_*-r_0+2)}^2}.
\end{equation*}

\begin{corollary}\label{coro_distribution}
Suppose the assumptions of Theorem \ref{thm_main_spike} hold and $r_*>r.$ Under the null hypothesis $\mathbf{H}_0$ in (\ref{eq_testingproblem}), we have that for all $x \in \mathbb{R}$
\begin{equation}\label{eq_universalityresult}
\lim_{n \rightarrow \infty} \mathbb{P} \left( \mathbb{T} \leq x \right)=\lim_{n \rightarrow \infty} \mathbb{P} (\mathbb{G}_1 \leq x), \ \text{and} \ \lim_{n \rightarrow \infty} \mathbb{P} \left( \mathbb{T}_{r_0} \leq x \right)=\lim_{n \rightarrow \infty} \mathbb{P} (\mathbb{G}_2 \leq x). 
\end{equation}
On the other hand, if $\delta_n^{(1)} \mathsf{T} \widetilde{\sigma}_r^{-1} \rightarrow 0, $ then 
\begin{equation}\label{eq_deltan(1)ha}
\lim_{n \rightarrow \infty} \mathbb{P}\left( \mathbb{T}>\delta_n^{(1)} \right)=1, \ \text{under} \ \mathbf{H}_a; 
\end{equation}
if $\delta_n^{(2)} \mathsf{T}(\widetilde{\sigma}_{r_0+1}-\widetilde{\sigma}_{r_0+2})^{-1} \rightarrow 0, $ then 
\begin{equation*}
\lim_{n \rightarrow \infty} \mathbb{P}\left( \mathbb{T}_{r_0}>\delta_n^{(2)} \right)=1, \ \text{under} \ \mathbf{H}_a; 
\end{equation*}
Consequently, if $\delta_n^{(1)} \rightarrow \infty$ and $\delta_n^{(1)} \mathsf{T} \widetilde{\sigma}_r^{-1} \rightarrow 0, $ then 
\begin{equation}\label{eq_deltan(1)harproperty}
\lim_{n \rightarrow \infty} \mathbb{P}\left( \widehat{r}_1=r \right)=1,  
\end{equation}
if $\delta_n^{(2)} \rightarrow \infty$ and $\delta_n^{(2)} \mathsf{T}(\widetilde{\sigma}_{r_0+1}-\widetilde{\sigma}_{r_0+2})^{-1} \rightarrow 0, $ then 
\begin{equation*}
\lim_{n \rightarrow \infty} \mathbb{P}\left( \widehat{r}_2=r \right)=1. 
\end{equation*}
\end{corollary}

\begin{remark}\label{remk_firstapplicationremark}
We remark that the conditions $\delta_n^{(1)} \rightarrow \infty$ and $\delta_n^{(2)} \rightarrow \infty$ are necessary and sufficient to guarantee that $\mathbb{T}$ and $\mathbb{T}_{r_0}$ have asymptotic zero type I errors. Moreover, for any fixed $r_*-r_0,$ the joint distributions of $\{\xi_{(i)}^2\}_{1 \leq i \leq r_*-r_0+2}$ can be calculated explicitly  as in \cite{resnick2008extreme}. In fact, in the study of extreme value theory, $\{\xi_{(i)}^2-\xi_{(i+1)}^2\}$ are called the \emph{spacings} of $\{\xi_{(i)}^2\}$ \cite{nagaraja2015spacings} whose distributions can also be calculated explicitly in some settings. For example, if $\xi_i^2$ are i.i.d. exponential distribution with parameter $t>0,$ then $\xi_{(i)}^2-\xi_{(i+1)}^2$ will be i.i.d. exponential distribution with parameter $t/(n-i).$ In general, it is hard to get explicit expressions for the limiting distributions of $\mathbb{G}_1$ and $\mathbb{G}_2,$ but it is easy to check that both distributions are supported on the whole positive real line so that it is necessary to let the critical values to diverge. Finally, following the ideas of \cite{ding2021spiked,9779233,onatski2009testing,passemier2014estimation},  since $\mathbb{G}_1$ and $\mathbb{G}_2$ only depend on $\{\xi_i^2\},$ we can numerically generate the critical values to calibrate the empirical distributions of $\mathbb{G}_1$ and $\mathbb{G}_2$. 
\end{remark}

\quad Inspired by Corollar \ref{coro_distribution} and Remark \ref{remk_firstapplicationremark}, we can test (\ref{eq_testingproblem}) using the statistics $\mathbb{T}$ and $\mathbb{T}_{r_0}.$ Given some type I error rate $\alpha$ (say $\alpha=0.1$), the critical values can be generated numerically through $\mathbb{G}_1$ and $\mathbb{G}_2,$ respectively as in the procedure below Theorem 4.3 of \cite{ding2021spiked}. More specifically, we can generate a sequence of $N,$ say $N=10^4,$ $\{\xi_{k,i}^2\}, 1 \leq k \leq N,$ and the associated sequence of statistics $\{\mathbb{G}_{k,1}\}$ and $\{\mathbb{G}_{k,2}\}.$ Given the level $\alpha,$ we can choose $\delta_n^{(t)}$ so that for $t=1,2,$ 
\begin{equation*}
\frac{\# \{ \mathbb{G}_{k,t} \leq \delta_n^{(t)}\}}{N} \geq 1-\alpha. 
\end{equation*}   
We point out that by the constructions of $\mathbb{G}_1$ and $\mathbb{G}_2,$ people need to know the distribution of $\{\xi_i^2\}.$ However, in the simulations below, we show that the performance of the proposed statistics and the choices of the critical values are relatively robust against the correct choices of $\{\xi_i^2\},$ as long as the tail behavior of the distributions are  specified.

\quad In what follows, we conduct Monte-Carlo simulations to 
demonstrate the accuracy, power and robustness of our proposed statistics for (\ref{eq_testingproblem}) for the elliptically distributed data under various settings of $\xi^2.$ More specifically, we will consider the following four setups:
(1). Gamma distribution with shape parameter $5$ and rate parameter $5$; (2). Pareto distribution with scale parameter $x_{\min}=0.75$ and shape parameter $3$; (3). Exponential distribution with rate parameter $1$; (4). Squares of student-$t$ distribution with 3 degrees of freedom.  For the possibly spiked population covariance matrix, we consider Johnstone's spiked model \cite{John2001}  
\begin{equation*}
\widetilde{\Sigma}=\text{diag}\left\{ \widetilde{\sigma}_1, \cdots, \widetilde{\sigma}_r, 1,1,\cdots, 1 \right\}. 
\end{equation*}

First, we study our proposed statistics.  We check the accuracy  under $\alpha=0.1$ when the null hypothesis $\mathbf{H}_0$ in (\ref{eq_testingproblem}) holds with the setting $r_0=2, \widetilde{\sigma}_1=25, \widetilde{\sigma}_2=20$  with various choices of $\phi=0.5,1,2.$ Under this null hypothesis, we also examine the power of the statistics when $\mathbf{H}_a$ in (\ref{eq_testingproblem}) holds with $r=3$ with $\widetilde{\sigma}_1=25, \widetilde{\sigma}_2=20$ and some $\widetilde{\sigma}_3>0$. We can conclude from Figures \ref{fig_test1typeis1} and \ref{fig_test1power} that both statistics are reasonably accurate and powerful. Especially for the power, once $\widetilde{\sigma}_3$ is in a certain range which depends on the behavior of $\xi^2$, i.e., (\ref{spiked_assumption}) is satisfied, our statistics will be able to reject $\mathbf{H}_a$. We can also make the same conclusion as \cite{9779233} that $\mathbb{T}_{r_0}$ has slightly better finite sample performance in terms of power since it uses a smaller number of eigenvalues.

\begin{figure}[!ht]
\subfigure[Accuracy of  $\mathbb{T}$ in (\ref{eq_statisticsdefinition1}).]{\label{fig:a}\includegraphics[width=8.3cm,height=5cm]{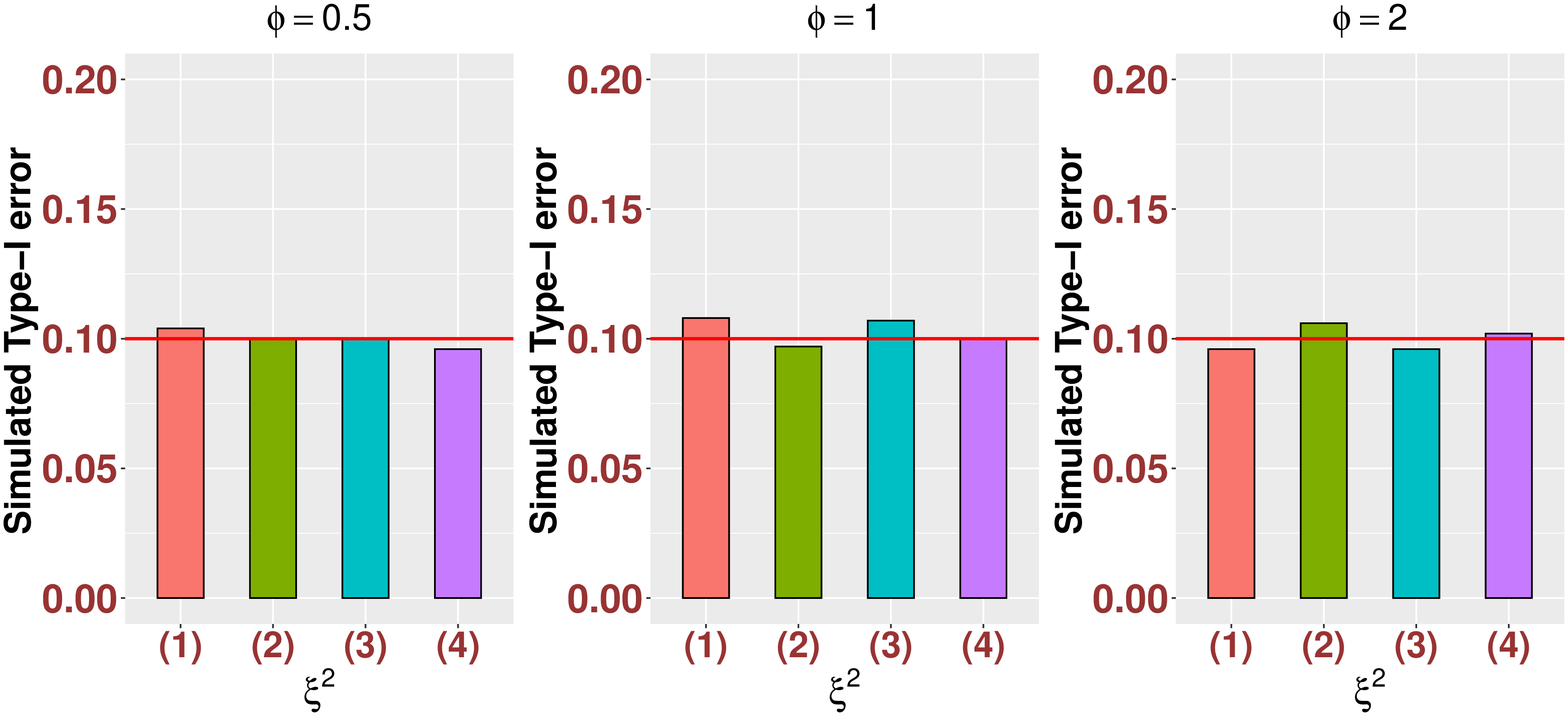}}
\hspace*{0.3cm}
\subfigure[Accuracy of  $\mathbb{T}_{r_0}$ in (\ref{eq_statisticsdefinition2}).]{\label{fig:b}\includegraphics[width=8.3cm,height=5cm]{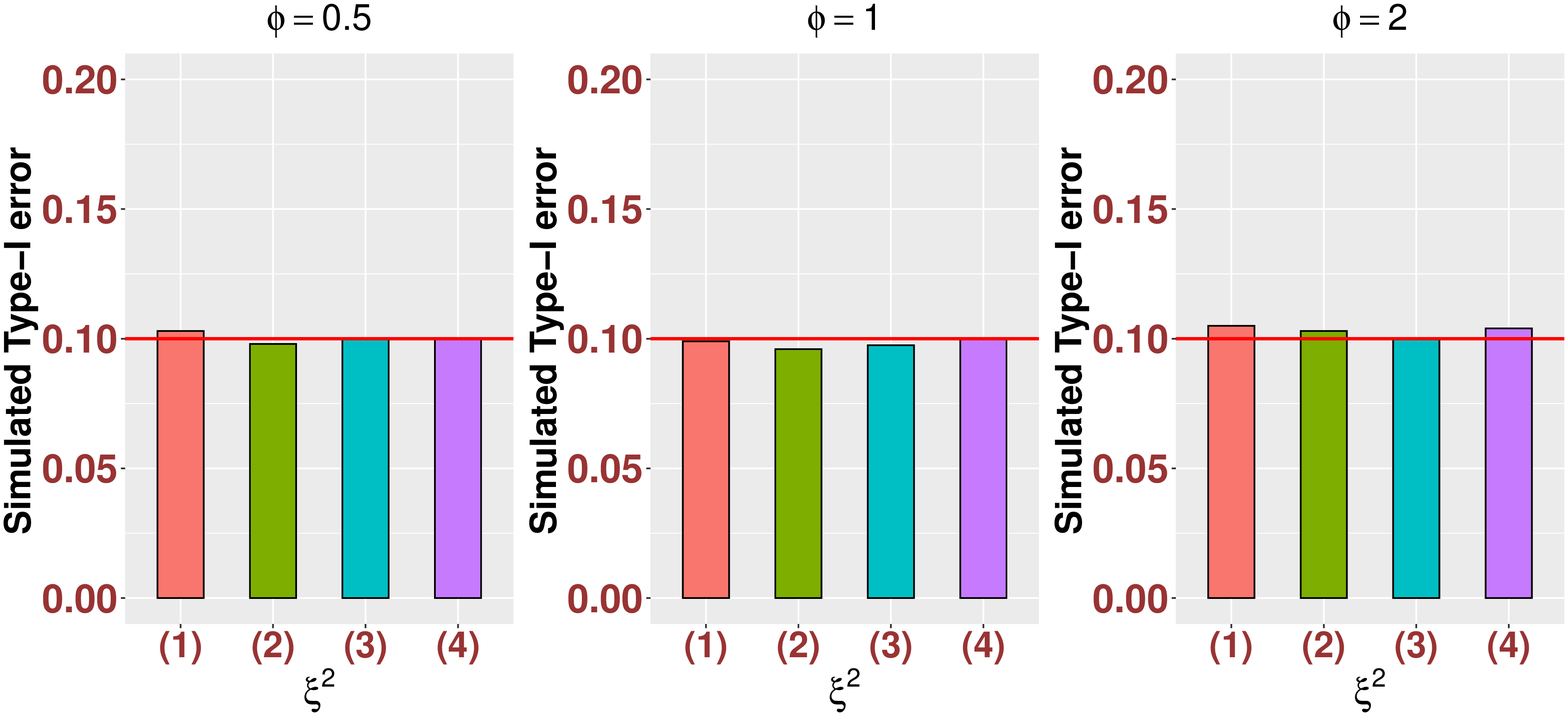}}
%\begin{subfigure}{0.3\textwidth}
%\includegraphics[width=8cm,height=5cm]{revfig/app1san2r3.eps}
%\caption{Accuracy of  $\mathbb{T}$ in (\ref{eq_ona intro}).}
%\end{subfigure}
%\hspace*{4cm}
%\begin{subfigure}{0.3\textwidth}
%\includegraphics[width=8cm,height=5cm]{revfig/app1san2r3new.eps}
%\caption{Accuracy of  $\mathbb{T}_{r_0}$ in (\ref{eq_ona intro1}).}
%\end{subfigure}
\caption{Simulated type I error rates under the nominal level 0.1 for $\mathbb{T}$ and $\mathbb{T}_{r_0}$. Here (1)-(4) corresponds to the four different settings of $\xi^2$. We take $n=400$ and report the results based on 2,000 Monte-Carlo simulations. }
\label{fig_test1typeis1}
\end{figure}  

\begin{figure}[!ht]
\hspace*{-1.5cm}\includegraphics[width=20cm,height=5.5cm]{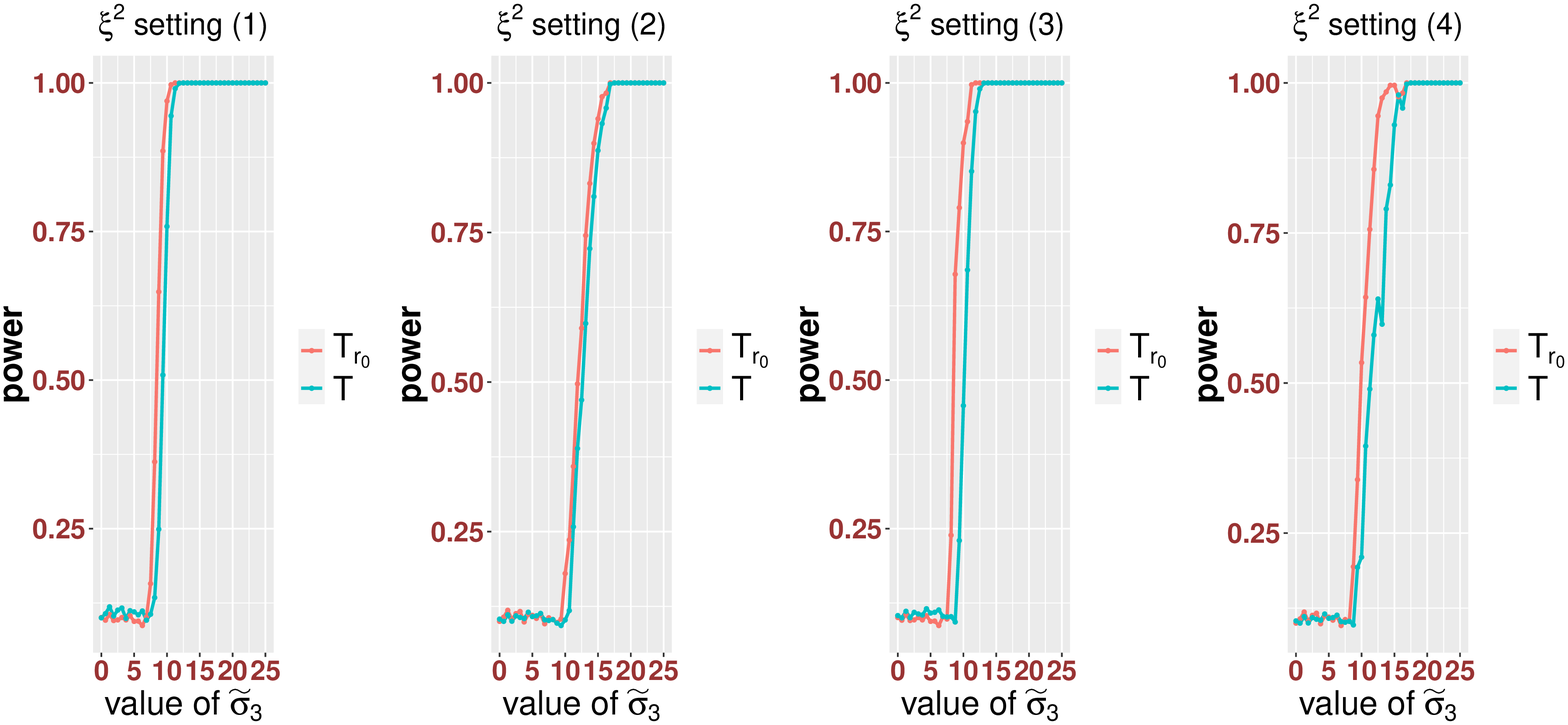}
%\begin{subfigure}{0.3\textwidth}
%\includegraphics[width=8cm,height=5cm]{revfig/app1san2r3.eps}
%\caption{Accuracy of  $\mathbb{T}$ in (\ref{eq_ona intro}).}
%\end{subfigure}
%\hspace*{4cm}
%\begin{subfigure}{0.3\textwidth}
%\includegraphics[width=8cm,height=5cm]{revfig/app1san2r3new.eps}
%\caption{Accuracy of  $\mathbb{T}_{r_0}$ in (\ref{eq_ona intro1}).}
%\end{subfigure}
\caption{Simulated power under the nominal level 0.1 for $\mathbb{T}$ and $\mathbb{T}_{r_0}$. We take $\phi=0.5, n=400$ and report the results based on 2,000 Monte-Carlo simulations. Similar results can be reported for the cases $\phi=1,2.$ }
\label{fig_test1power}
\end{figure}

Second, we compare the performance of our approaches with some existing ones. Since most of the existing literature focus on the estimation of the number of the spikes instead of inferring, we compare the performance of our inference based estimators $\widehat{r}_1$ and $\widehat{r}_2$ in (\ref{eq_sequencialestimator}) with the existing ones. For definiteness, we compare our estimators with the ones proposed in \cite{CHP,ke2021estimation,passemier2014estimation}. The above estimators are all developed to estimate the number of spikes when the samples are generated from $\widetilde{\Sigma}^{1/2} \mathbf{x}_i$ under various assumptions on $\widetilde{\Sigma}$ with $\mathbf{x}_i$ having mean zero i.i.d. entries. In Figure \ref{fig_comparsion1}, we compare the accuracy of our estimators when $r=1$ under setting (2) for $\xi^2$ and a wide range of $\widetilde{\sigma}_1.$ For comparison, we conduct 2,000 Monte-Carlo simulations  and report the correct detection ratio (CDR); that is, the ratio between the number of simulations that $r=1$ is estimated correctly and the total number of simulations 2,000. We can find that our estimators outperform the existing ones once the spikes are reasonably large.  

\begin{figure}[!ht]
\subfigure{\includegraphics[width=5cm,height=5.1cm]{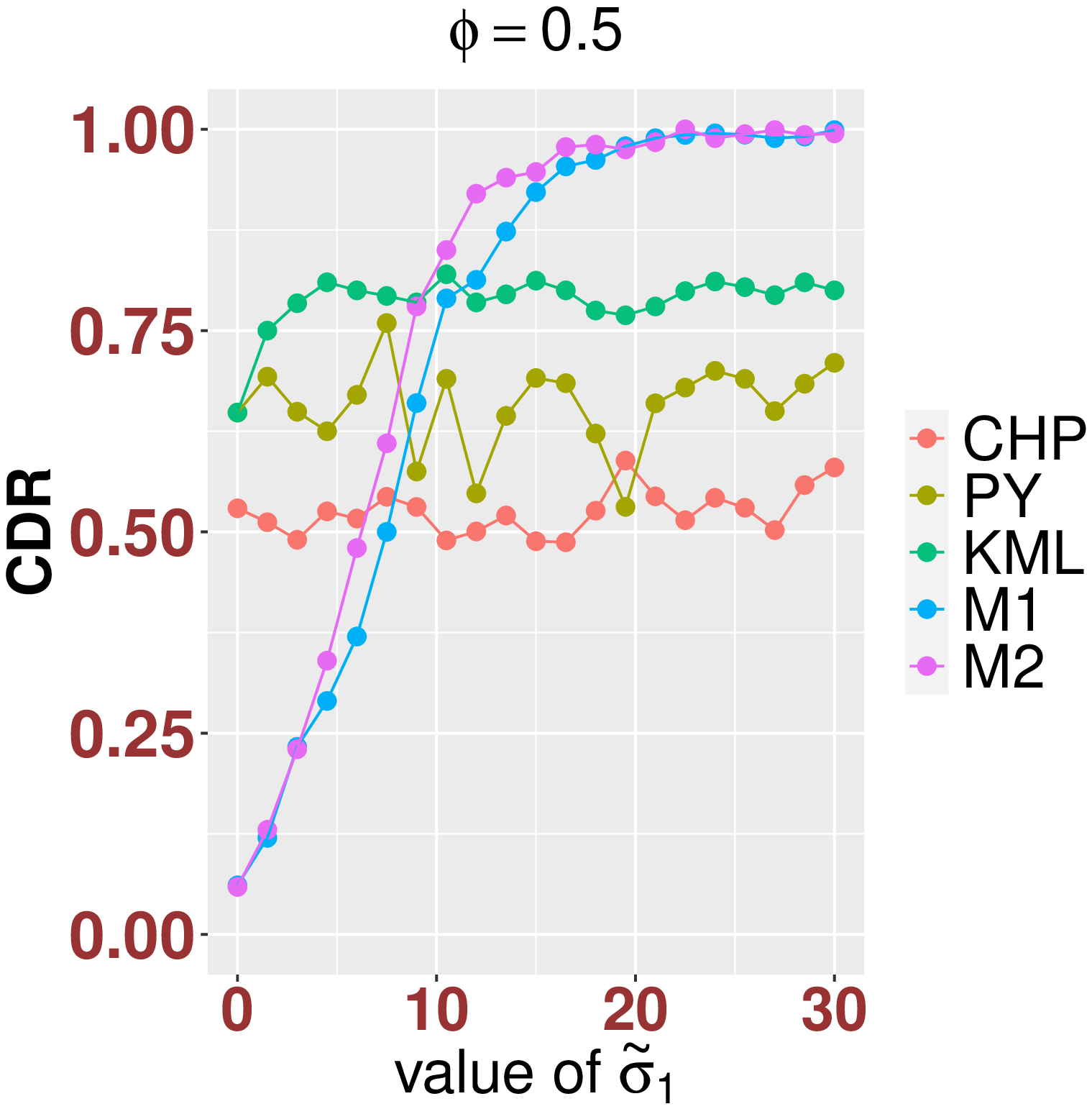}}
\hspace*{0.2cm}
\subfigure{\includegraphics[width=5cm,height=5.1cm]{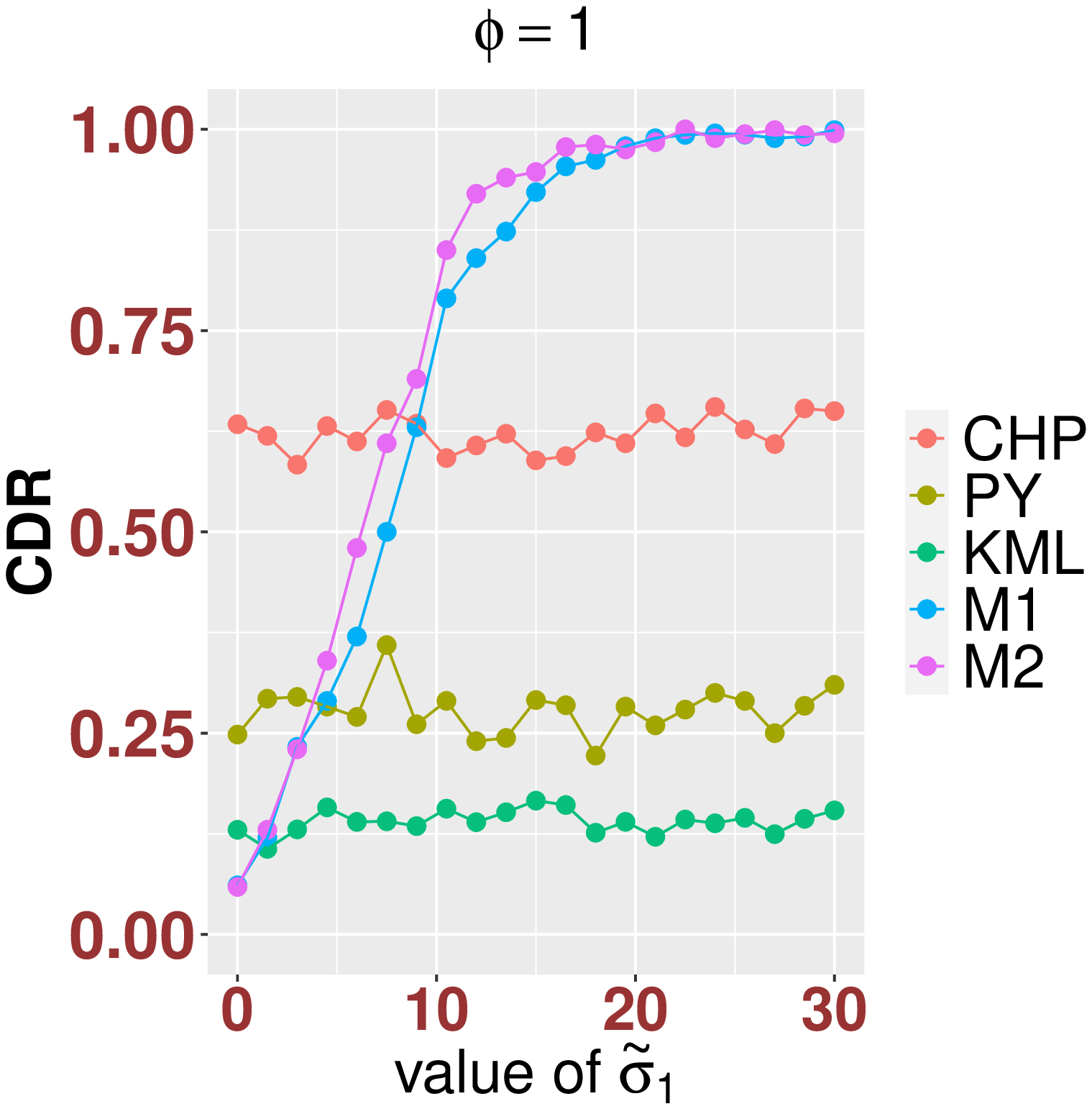}}
\hspace*{0.2cm}
\subfigure{\includegraphics[width=5cm,height=5.1cm]{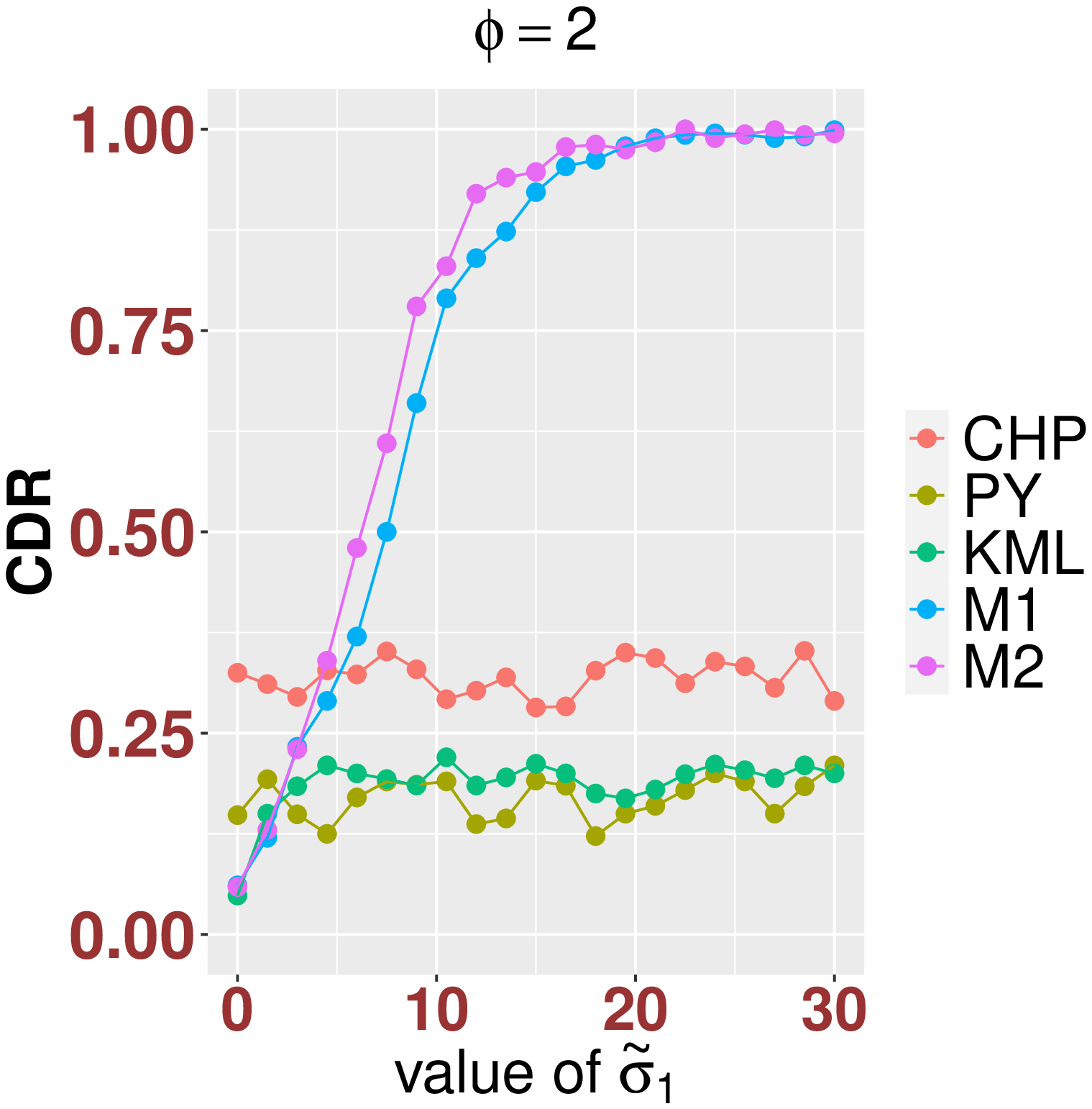}}
\caption{Comparison of estimation. In the above figures, "CHP" refers to the method from \cite{CHP}, "PY" refers to the method from \cite{passemier2014estimation}, "KML" refers to the method from \cite{ke2021estimation}, and "M1" and "M2" refer to using our method via $\widehat{r}_1$ and $\widehat{r}_2$ in (\ref{eq_sequencialestimator}), respectively. Here $n=400.$ }
\label{fig_comparsion1}
\end{figure}

Finally, we discuss the robustness of our proposed methods against the choices of $\xi^2.$ As mentioned earlier, the constructions of $\mathbb{G}_1$ and $\mathbb{G}_2$ require the knowledge of the distribution of $\{\xi_i^2\}$ which may be too restrictive. In Figure \ref{fig_test1typeis2}, we checked the robustness of statistics in terms of accuracy when applied to infer the null hypothesis in (\ref{eq_testingproblem}) with $r_0=2$. Especially,  the true distribution of $D$ follows setting (2), i.e., Pareto distribution with parameters $0.75$ and 3. However, when we construct the critical values from $\mathbb{G}_1$ and $\mathbb{G}_2,$ we choose misspecified distributions (a). Pareto distribution with parameters 1 and 4, (b). Squares of student-t distribution with 2 degrees of freedom.  We can conclude that the performance of the proposed statistics and the choices of the critical values are relatively robust. In fact, we have also conducted simulations to check the power and similar conclusions can also be made.

\begin{figure}[t]
\subfigure[Accuracy of  $\mathbb{T}$ in (\ref{eq_statisticsdefinition1}).]{\label{fig:aaaa}\includegraphics[width=8.3cm,height=5cm]{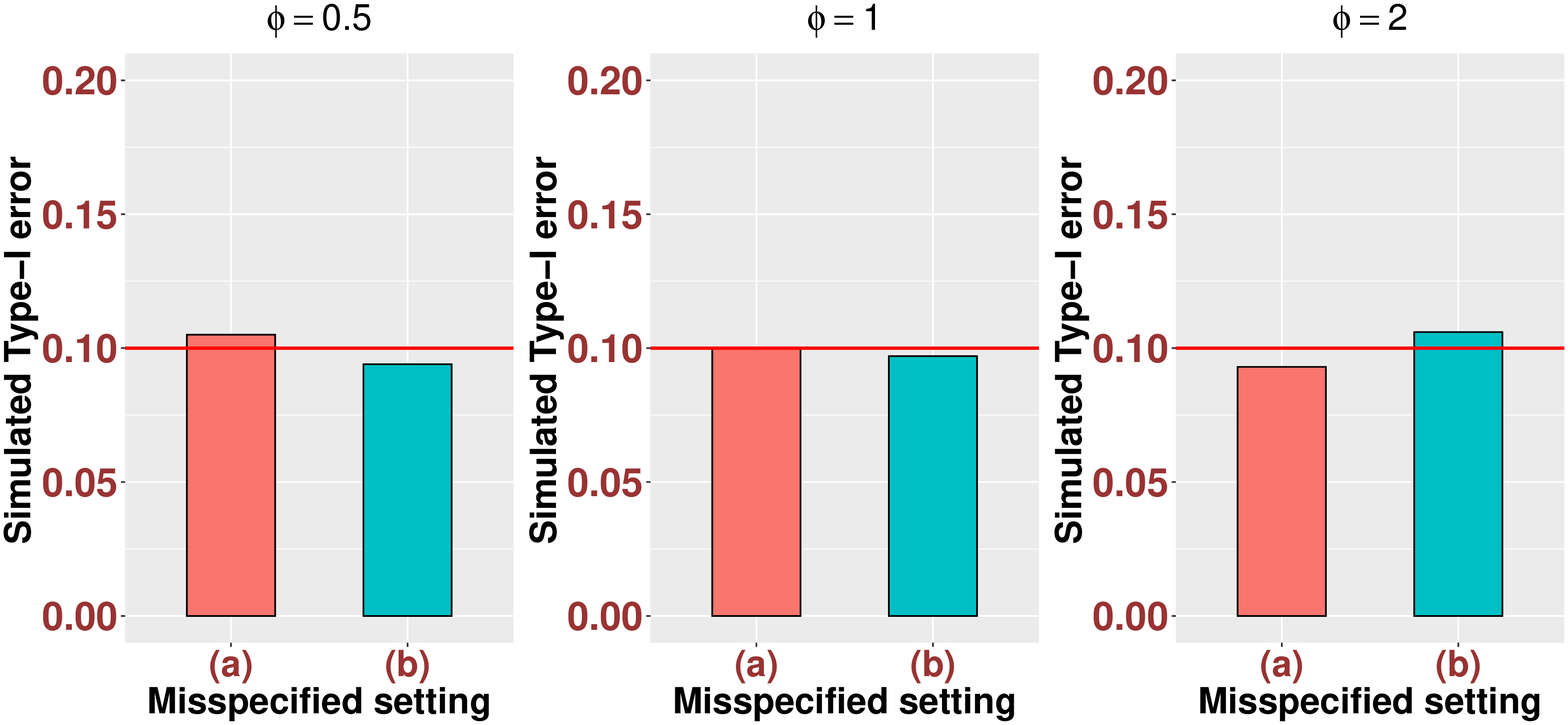}}
\hspace*{0.3cm}
\subfigure[Accuracy of  $\mathbb{T}_{r_0}$ in (\ref{eq_statisticsdefinition2}).]{\label{fig:bbbb}\includegraphics[width=8.3cm,height=5cm]{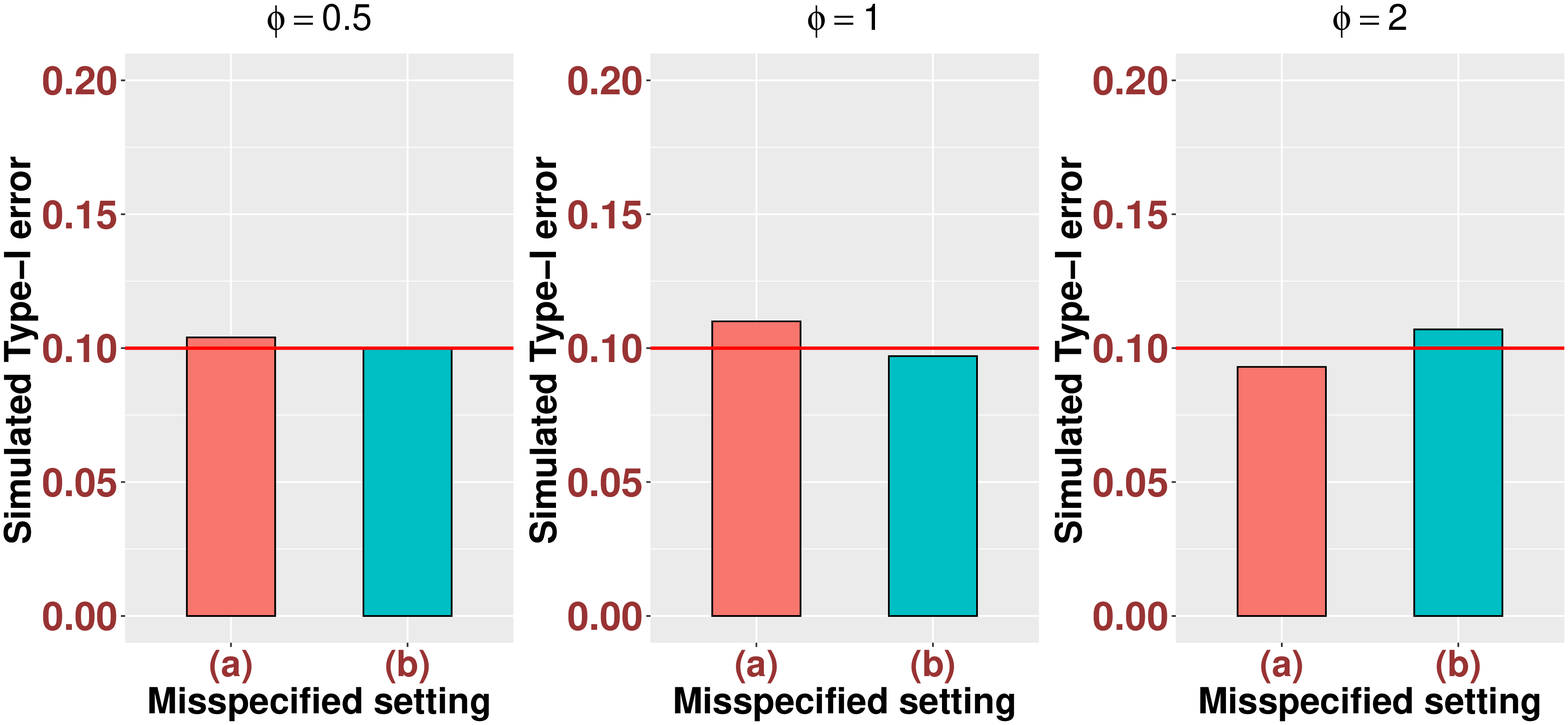}}
%\begin{subfigure}{0.3\textwidth}
%\includegraphics[width=8cm,height=5cm]{revfig/app1san2r3.eps}
%\caption{Accuracy of  $\mathbb{T}$ in (\ref{eq_ona intro}).}
%\end{subfigure}
%\hspace*{4cm}
%\begin{subfigure}{0.3\textwidth}
%\includegraphics[width=8cm,height=5cm]{revfig/app1san2r3new.eps}
%\caption{Accuracy of  $\mathbb{T}_{r_0}$ in (\ref{eq_ona intro1}).}
%\end{subfigure}
\caption{Simulated type I error rates under the nominal level 0.1 for $\mathbb{T}$ and $\mathbb{T}_{r_0}$ under misspecified settings. Here (a) and (b) correspond to the two misspecified settings. We take $n=400$ and report the results based on 2,000 Monte-Carlo simulations. }
\label{fig_test1typeis2}
\end{figure}

\subsection{High dimensional multiplier bootstrap with applications in common factors selection}\label{sec_boostrapping}

The bootstrap \cite{davison1997bootstrap,boostraporginialpaper} is a central tool in statistics which  resamples a single dataset to create many simulated samples to enable inference when very little is known about the properties of the data-generating distribution. In multivariate statistical analysis, bootstrap has witnessed many successful applications in real world problems related to principal component analysis (PCA), especially the ones related to solving inference problems on the population covariance matrices, to list but a few \cite{fisher2016fast,li2019local, wagner2015go} and see \cite{yao2021rates} for more references. Even though the asymptotic theory showing that bootstrap generally works in the context of PCA with low-dimensional data (i.e., when $p$ is much smaller than the size $n$) is well-established, see \cite{beran1985bootstrap} for example, less is touched in the high dimensional setting (\ref{ass1}). To our best knowledge, the existing works mostly focus on the spectral norm or the spectral projections  under various assumptions on the population covariance matrix, to list but a few, \cite{VK,lopes2022improved,xia2021normal,yao2021rates}. But generally,  much less is known about each individual eigenvalue.

In what follows, we study the performance of multiplier bootstrap when applied to investigate the first few eigenvalues of the population covariance matrix under a spiked covariance matrix model. For simplicity, we assume that the data matrix is generated as $\widehat{Y}:=\widetilde{\Sigma}^{1/2}X$ with $X$ satisfying (2) of Assumption \ref{assum_model}. Such a model finds important applications in financial economics especially in the large scale factor models where all economic variables can be explained by a few common components, see \cite{bai2008large,stock2016dynamic} for a survey of the factor models. For $\widehat{Y}=(\widehat{\mathbf y}_i)$, where $\widehat{\mathbf y}_i, 1 \leq i \leq n,$ are i.i.d. sampled from the factor model 
\begin{equation}\label{eq_factormodelmodelcase}
\widehat{\mathbf y}=L\mathbf f+\mathbf e \in \mathbb{R}^p,
\end{equation} 
where $\mathbf f$ is an $r \times 1$ low rank (unobserved) factor, $L$ is a $p \times r$ low rank loading matrix and $\mathbf e$ is the $p \times 1$ idiosyncratic error which is independent of $\mathbf{f}$. For the purpose of identifiability, following \cite{bai2002determining,bai2008large,fan2022estimating,stock2016dynamic}, we assume that  $\operatorname{Cov}(\mathbf{f}, \mathbf{f})=I_r.$ Therefore, the covariance structure of $\widehat{\mathbf{y}}$ can be written as $LL^*+\operatorname{Cov}(\mathbf{e}, \mathbf{e})$ which has a spiked structure as $\widetilde{\Sigma}.$  In financial economics, an important question is to understand how many common components are needed in order to understand the economic variables. Formally, we are interested in testing the value of $r$ via the hypothesis testing problem 
\begin{equation}\label{eq_hypothesistesttwo}
\mathbf{H}_0: r \geq r_0  \ \ \text{vs} \ \ \mathbf{H}_a: r<r_0,
\end{equation}   
where $r_0$ is some pre-given integer representing our belief of the value of $r.$ Corresponding to (\ref{eq_sequencentialtestprocedure}), we can propose the sequential testing estimator for $r$ as 
\begin{equation}\label{eq_factorestimator}
\widehat{r}:=\sup\left\{ r_0 \geq 0: \ \mathbf{H}_0 \ \text{is accepted}  \right\}. 
\end{equation}
We point out that the two testing problems (\ref{eq_testingproblem}) and (\ref{eq_hypothesistesttwo}) are closely related but different so that our proposed procedure and the associated estimators for $r$ are also different.

In the literature, many methods based on the eigenvalues of the sample covariance matrix, i.e., $\widehat{Q}= \widetilde{\Sigma}^{1/2}XX^* \widetilde{\Sigma}^{1/2},$ have been proposed for the hypothesis testing problem (\ref{eq_hypothesistesttwo}) in terms of factor models under our setting, for example, see \cite{ahn2013eigenvalue,bai2002determining,CHP,fan2022estimating,onatski2010determining}. In what follows, we propose a different approach based on multiplier bootstrap procedure. Our motivation is from \cite{2022arXiv220206188Y} where the authors used a standard bootstrap based approach. We see that   
under the above setup, $\widetilde{Q}=\widetilde{\Sigma}^{1/2}XD^2 X^* \widetilde{\Sigma}^{1/2}$ can be regarded as the bootstrapped sample covariance matrix with the entries in $D^2$ being the random weights.

We now state the main results on the relation between the eigenvalues of the sample covariance matrices with and without multiplier bootstrap. Denote the non-zero eigenvalues of $\widehat{Q}$ as  $\widehat{\lambda}_1 \geq \widehat{\lambda}_2 \geq \cdots \geq \widehat{\lambda}_{\min\{p,n\}}>0. $ Recall (\ref{eq_eigenvalueintheendspiked}). For simplicity and definiteness, we follow \cite{bai2008large,stock2016dynamic} and assume that the entries in the factor $\mathbf{f}$ and error $\mathbf{e}$ are independent so that $\widetilde{\Sigma}=\operatorname{diag}\{\widetilde{\sigma}_1, \cdots, \widetilde{\sigma}_p\}$ is diagonal. For simplicity and definiteness, we assume that for some constant $\tau>0$
\begin{equation}\label{eq_separation}
\frac{\widetilde{\sigma}_i}{\widetilde{\sigma}_{i+1}} \geq 1+\tau,  \  1 \leq i \leq r. 
\end{equation}

%Denote 
%\begin{equation*}
%\mathsf{V}_i:=
%\end{equation*}

\begin{theorem}\label{thm_boostrappingdiscussion}
Suppose Assumption \ref{assumption_techincial}, (\ref{spiked_assumption}) and (\ref{eq_separation})  hold. Moreover, for the entries of $X$, we assume (2) of Assumption \ref{assum_model} holds. Then for $1 \leq i \leq r,$ we have that when conditional on  the data matrix $\widetilde{\Sigma}^{1/2}X$ 
\begin{equation*}
\lim_{n \rightarrow \infty} \mathbb{P} \left( \sqrt{\frac{n}{\mathsf{V}}}\left(\frac{\mu_i}{\widehat{\lambda}_i}-\mathbb{E} \xi^2\right) \leq x \right)=\Phi(x), \ \text{where} \ \mathsf{V}:= \mathfrak{m}_4 \mathbb{E} \xi^4 -(\mathbb{E} \xi^2)^2. 
\end{equation*}
Here $\mathfrak{m}_4=\mathbb{E}(\sqrt{n}x_{11})^4>1$ due to (\ref{eq_standard1n}) so that $\mathsf{V}>0,$ and $\Phi(x)$ is the CDF of a real standard Gaussian random variable. 
\end{theorem}

\begin{remark}\label{rem_boostrappingremark}
Several remarks are in order. First, combining Theorem \ref{thm_boostrappingdiscussion} with Theorem \ref{thm_main_unbounded} and Remark \ref{rmk_mainresults_unbounded}, we see that we can propose an accurate and powerful procedure (see Algorithm \ref{alg1} below) to test (\ref{eq_hypothesistesttwo}) using the first few outlier eigenvalues. In fact, using the discussions in Remark \ref{rem_multiplevaluespikedremark}, we see that there exists a sharp phase transition between $i \leq r$ and $i>r,$ where the outliers are asymptotically Gaussian whereas the extremal non-outliers follow Gumbel or Fr{\'e}chet distribution.  Second, in Section \ref{sec_spikeellipticaldata}, instead of using the outlier eigenvalues, we utilize the extremal non-outlier eigenvalues in  (\ref{eq_statisticsdefinition1}) and (\ref{eq_statisticsdefinition2}).  The main reason is that in Section \ref{sec_spikeellipticaldata}, we considered a different hypothesis testing problem (\ref{eq_testingproblem}) instead of (\ref{eq_hypothesistesttwo}). But both approaches are accurate and powerful in terms of the estimation of $r.$ Third, similar to the discussions in Section 3.3 of \cite{2022arXiv220206188Y}, Theorem \ref{thm_boostrappingdiscussion} and its proof, in general,we see the multiplier bootstrap is biased in terms of replicating the distribution of $\widehat{\lambda}_i-\widetilde{\sigma}_i.$ Finally, in the current paper, we only consider the multiplier bootstrap when the entries in $D$ are independent. However, our approach can also apply to the standard bootstrap when $(\xi_1, \cdots, \xi_n)$ follows $n$-dimensional multinomial distribution with success probabilities $(n^{-1}, \cdots, n^{-1}).$ Since this is not the focus of the current paper, we will pursue this direction in the future works.   
\end{remark}

\quad Based on Theorems \ref{thm_boostrappingdiscussion} and Remark \ref{rem_boostrappingremark}, we can propose the following resampling procedure to test (\ref{eq_hypothesistesttwo}) utilizing the outlier eigenvalues. 
\begin{algorithm}[H]
	\caption{Resampling testing for (\ref{eq_hypothesistesttwo})}\label{alg1}
\normalsize
\begin{flushleft}
\noindent{\bf Inputs:}   $r_0, \mathsf{V}, \widehat{\lambda}_{r_0}$, the distribution of $\xi^2,$ number of resampling $B$ (say 1,000), type I error $\alpha$ and the standard $Z$-score $z_{1-\alpha/2}.$  

\noindent{\bf Step One:}  Generate $B$ i.i.d. copies of the matrices $D^2_k, \ k=1,2,\cdots, B.$ Compute the associated bootstrapped matrices $\widetilde{Q}_k=\widetilde{\Sigma}^{1/2}XD_k^2 X^* \widetilde{\Sigma}^{1/2}$ and the sequence of eigenvalues $\mu_{k,r_0}, 1 \leq k \leq B,$ where $\mu_{k,r_0}$ is the $r_0$-th eigenvalue of $\widetilde{Q}_{k}.$   

\noindent{\bf Step Two:} Compute $\widehat{\mathcal{T}}_k:=\sqrt{n/ \mathsf{V}}(\mu_{k,r_0}/\widehat{\lambda}_{r_0}-\mathbb{E} \xi^2)$ and let $B^*= \#\{k: |\widehat{\mathcal{T}}_k| \leq z_{1-\alpha/2}\}.$

\noindent{\bf Output:} p-value of the test can be computed as $\mathsf{p}:=1-\frac{B^*}{B}.$ Reject $\mathbf{H}_0$ in (\ref{eq_hypothesistesttwo}) if $\mathsf{p}<\alpha.$
\end{flushleft}
\end{algorithm}

\quad In what follows, we conduct Monte-Carlo simulations to  demonstrate the accuracy, power and robustness of our proposed Algorithm \ref{alg1} under the factor model setup (\ref{eq_factormodelmodelcase}). For simplicity, considering the setups in \cite{fan2022estimating,2022arXiv220206188Y}, in the simulations, for the data matrix $\widehat{Y} \in \mathbb{R}^{p \times n},$ we assume that 
\begin{equation*}
\widehat{Y}=\delta L' \mathbf{F}+\mathbf{E},
\end{equation*}
where $L' \in \mathbb{R}^{p \times 3}$ is the loading matrix whose rows are independent Gaussian random vectors in $\mathbb{R}^3$ with covariance matrix $\operatorname{diag}\{1.3, 0.8, 0.5\}$, $\mathbf{F} \in \mathbb{R}^{3 \times n}$ is the factor score matrix independent of $L'$ with i.i.d. standard Gaussian entries and $\mathbf{E} \in \mathbb{R}^{p \times n}$ is a standard Gaussian matrix independent of the factor loading and score matrices. Here $ \delta \geq 0$ is the factor strength. Under this setup, the null of (\ref{eq_hypothesistesttwo}) can be characterized as $\mathbf{H}_0: r=3$ which reduces to checking whether $\delta$ is large enough. The alternative of (\ref{eq_hypothesistesttwo})  can be expressed as $\mathbf{H}_a: r=0$ which reduces to checking whether $\delta=0$. 

First, we study our proposed statistics.  We check the accuracy  under $\alpha=0.1$ under the null that $r=3$ with $\delta=3.$ Moreover, we also examine the power of the statistics for the alternative when $r=0$ which implies $\delta=0$. We can conclude from Figure \ref{fig_test1typeis3} that our Algorithm \ref{alg1} is reasonably accurate and powerful for various choices of weights $\xi^2$ under different settings of $\phi.$

\begin{figure}[!ht]
\subfigure[Simulated type I error.]{\label{fig:a1aa}\includegraphics[width=8cm,height=5cm]{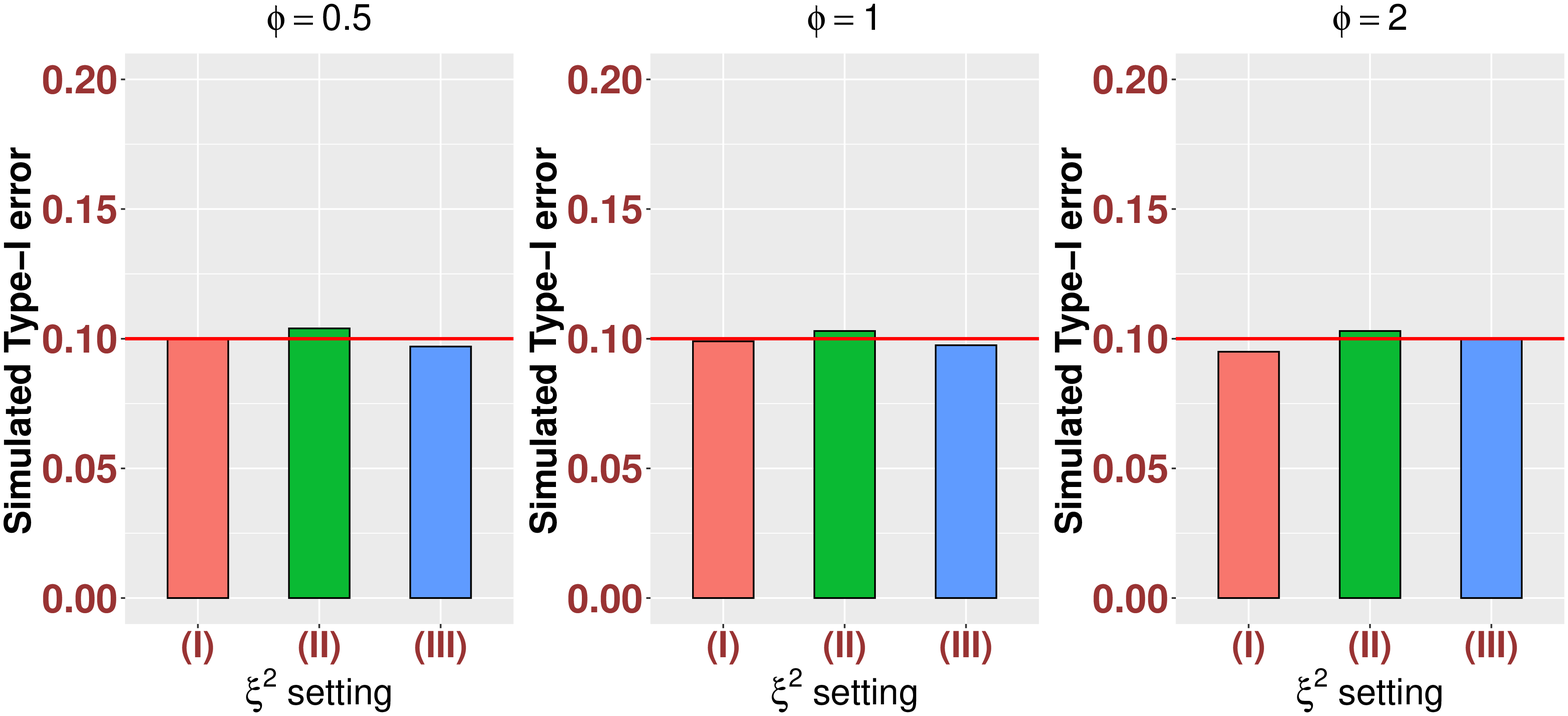}}
\hspace*{0.3cm}
\subfigure[Simulated power.]{\label{fig:b1aaa}\includegraphics[width=9cm,height=5cm]{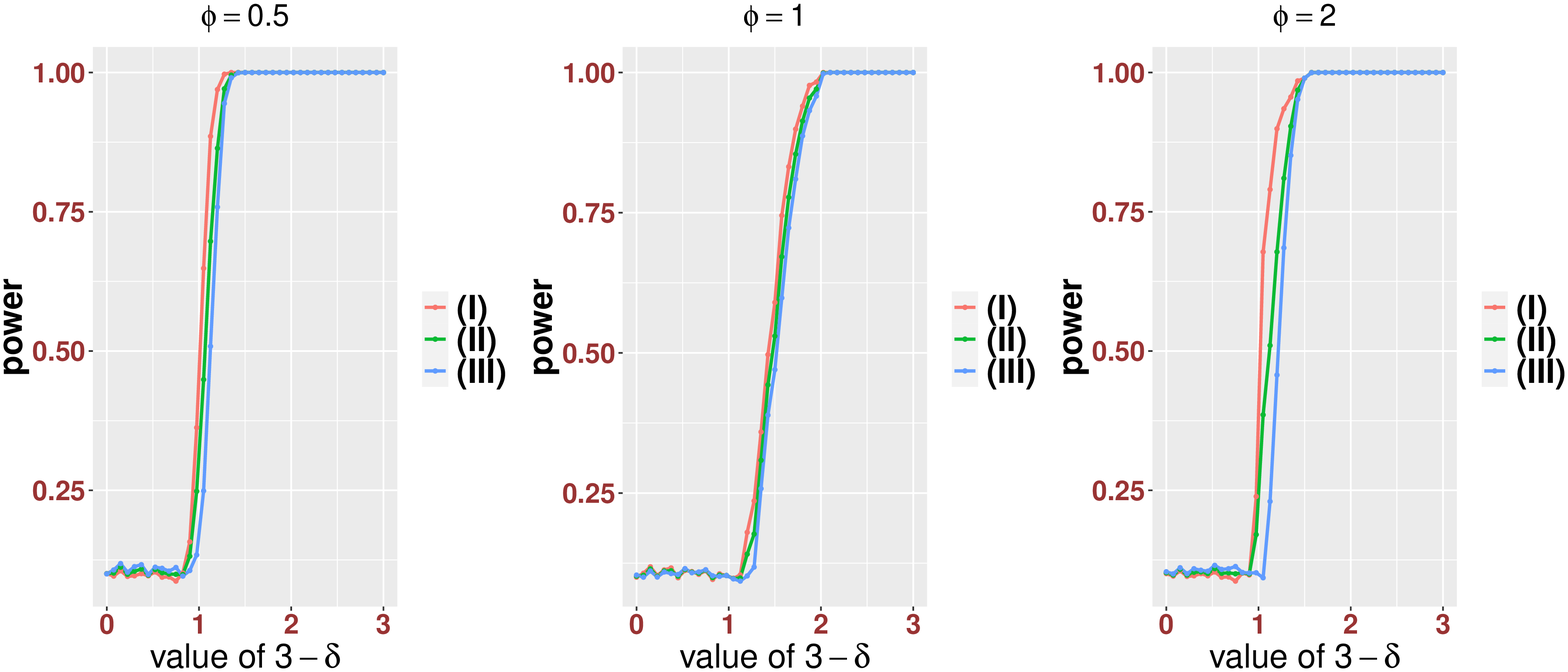}}
%\begin{subfigure}{0.3\textwidth}
%\includegraphics[width=8cm,height=5cm]{revfig/app1san2r3.eps}
%\caption{Accuracy of  $\mathbb{T}$ in (\ref{eq_ona intro}).}
%\end{subfigure}
%\hspace*{4cm}
%\begin{subfigure}{0.3\textwidth}
%\includegraphics[width=8cm,height=5cm]{revfig/app1san2r3new.eps}
%\caption{Accuracy of  $\mathbb{T}_{r_0}$ in (\ref{eq_ona intro1}).}
%\end{subfigure}
\caption{Simulated type I error rates and power under the nominal level 0.1 for our proposed Algorithm \ref{alg1}. Here we consider three different settings of multiplier $\xi^2:$ (I). Gamma distribution with parameters 15 and 15, (II). $\exp(1)$ distribution, and (III). $\chi^2_1$ distribution.   We take $n=400$ and report the results based on 2,000 Monte-Carlo simulations. The entries of $X$ are i.i.d. Gaussian with mean zero and variance $n^{-1}.$}
\label{fig_test1typeis3}
\end{figure}

Second, we compare the performance of our approaches with some existing ones. Again, since most of the existing literature focus on the estimation of the number of the spikes instead of inferring, we compare the performance of our inference based estimator $\widehat{r}$ in (\ref{eq_factorestimator}) with a few existing ones for estimating the number of factors in the context of factor model. For definiteness, we compare our estimators with the ones proposed in \cite{ahn2013eigenvalue,bai2002determining,CHP,fan2022estimating,onatski2010determining}. In Figure \ref{fig_comparsion2}, we compare the accuracy of our estimators when $r=3$ using CDR as in Figure \ref{fig_comparsion1}.  We can find that our estimators can outperform some of the existing ones especially when the spikes are not that large.

\begin{figure}[!ht]
\subfigure{\includegraphics[width=5cm,height=5.1cm]{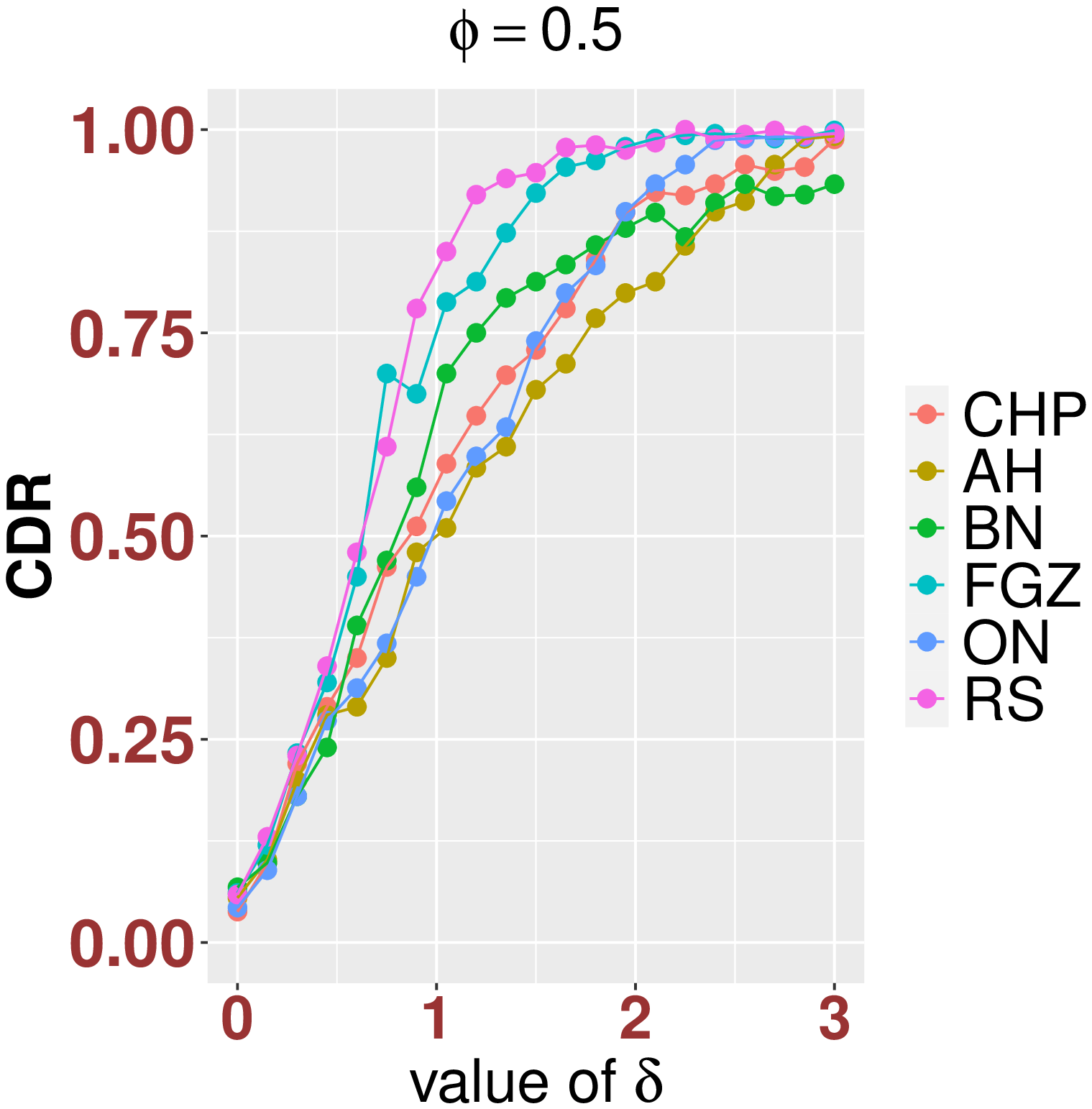}}
\hspace*{0.2cm}
\subfigure{\includegraphics[width=5cm,height=5.1cm]{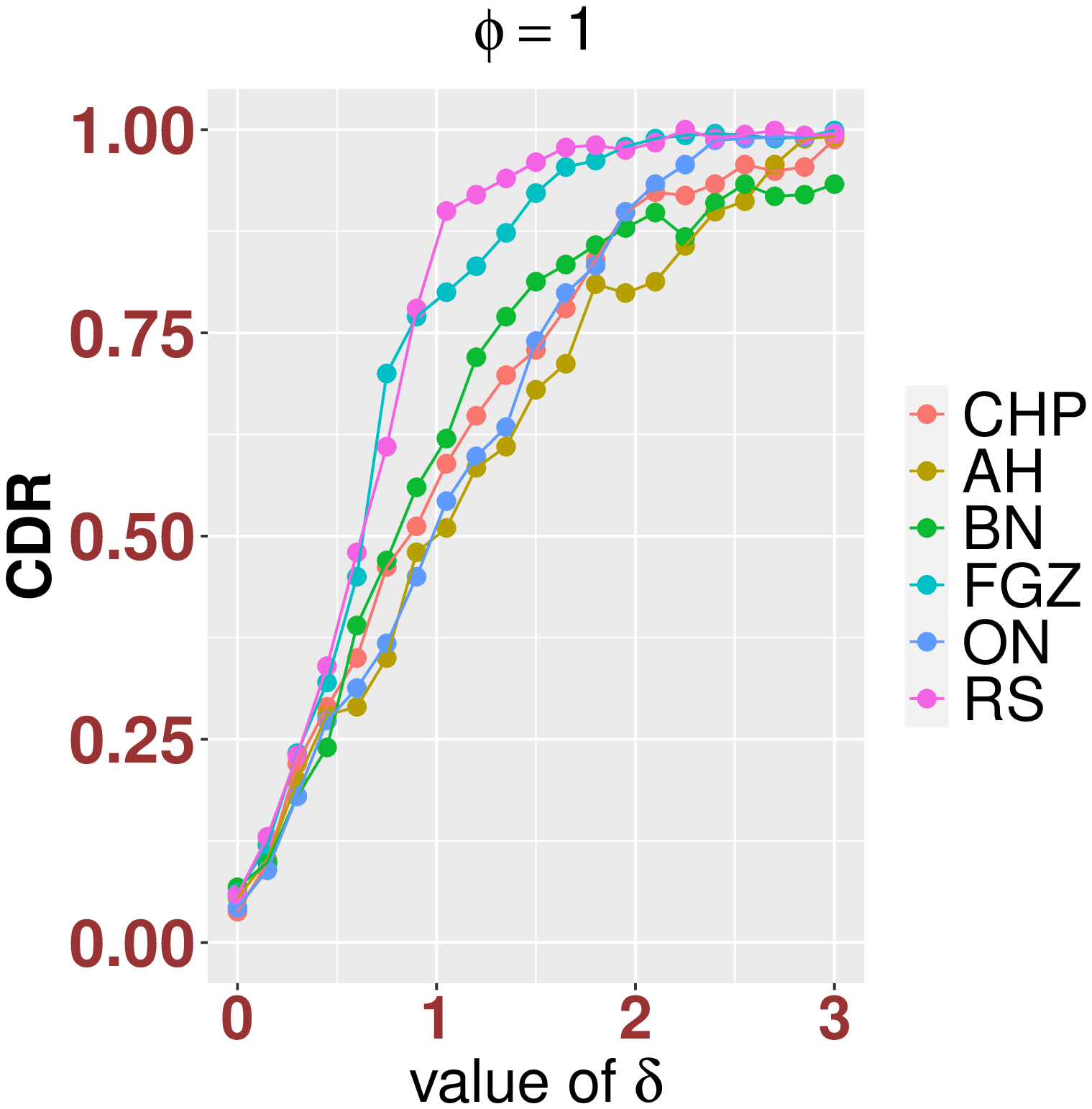}}
\hspace*{0.2cm}
\subfigure{\includegraphics[width=5cm,height=5.1cm]{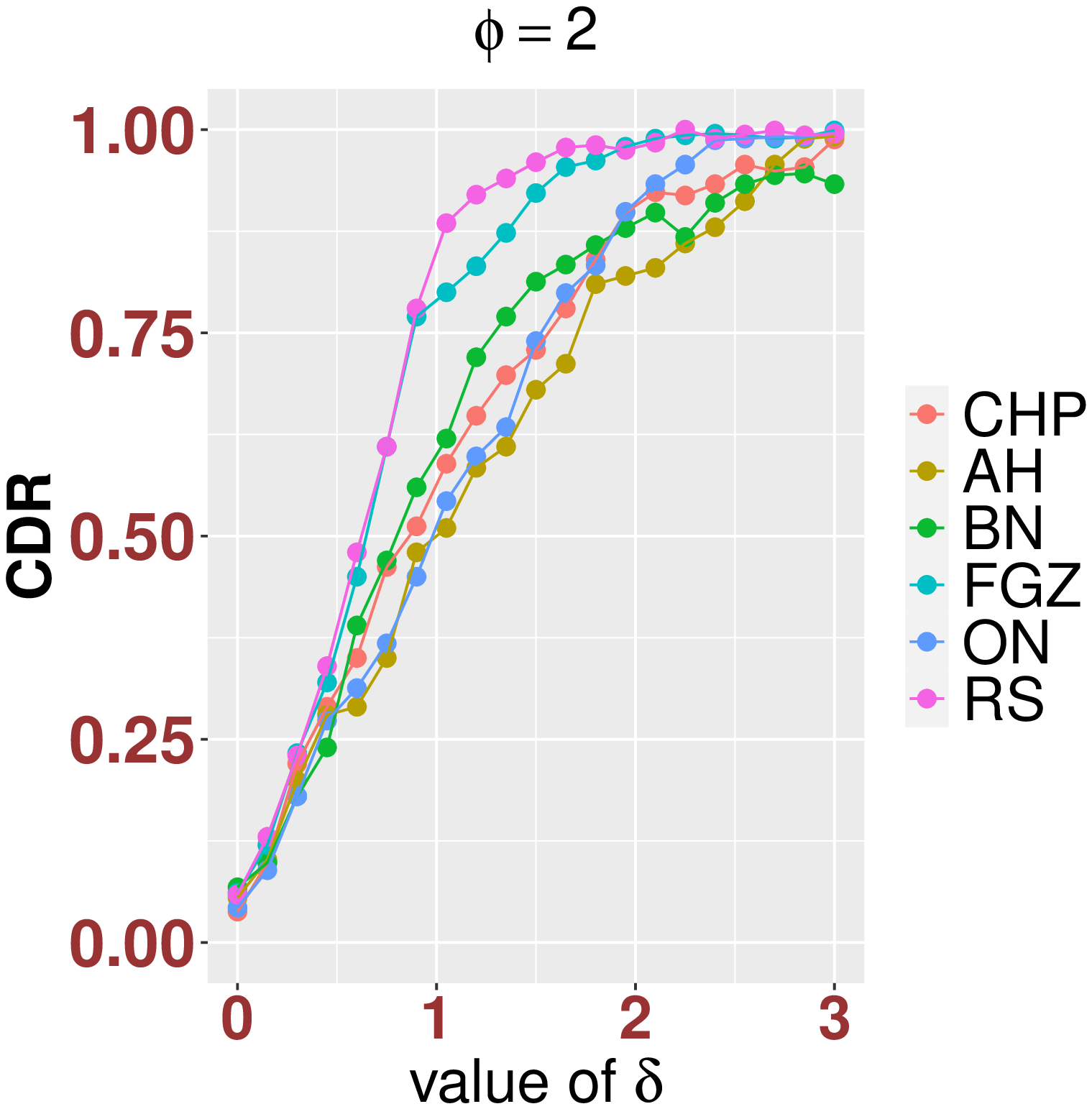}}
\caption{Comparison of estimation. In the above figures, "CHP" refers to the method from \cite{CHP}, "AH" refers to the method from \cite{ahn2013eigenvalue}, "BN" refers to the method from \cite{bai2002determining}, "FGZ" refers to the method from \cite{fan2022estimating}, "ON" refers to the method from \cite{onatski2010determining}, and "RS" refers to our proposed method in Algorithm \ref{alg1} where the entries of $D^2$ are i.i.d. $\exp(1)$ random variables. Here $n=400$ and the entries of $X$ are i.i.d. Gaussian with mean zero and variance $n^{-1}.$  The CDR is reported using 2,000 simulations.}
\label{fig_comparsion2}
\end{figure}

\appendix

\section{Some preliminary results}\label{appendix_preliminary}

In this section, we introduce some preliminary results which will be used in the proofs. First, in Section \ref{sec_asymptoticlocalaveragedlocal}, we provide the properties of the asymptotic local laws $m_{1n}, m_{2n}$ and $m_n$ from Definition \ref{defn_couplesystem} and establish the averaged local laws. Second, in Section \ref{appendix_goodconfiguration}, we examine the properties of the entries of $D^2$ and construct some probability events to which our arguments will be restricted. Finally, in Section \ref{sec_summaryofextremevalue}, we provide some useful lemmas and a short review of the extreme value theory.   

\subsection{Properties of asymptotic laws and averaged local laws}\label{sec_asymptoticlocalaveragedlocal}

We start with introducing the properties of the asymptotic local laws as in Definition \ref{defn_couplesystem}. We now define the sets of spectral parameters as follows.  For $\xi^2$ with unbounded support as in Case (i) of Assumption \ref{assum_D}, for $\mu_1$ defined in (\ref{eq: def of mu_1}) and $d_1$ defined in (\ref{eq_firstddefinition}) and some sufficiently large constant $\mathtt C>0,$ we denote
\begin{equation}\label{eq_spectraldomainone}
\mathbf{D}_{u} \equiv \mathbf{D}_{u}(\mathtt C):=\left\{z=E+\ri\eta \in \mathbb{C}_+: |E-\mu_1| \leq \mathtt C d_{1}, \ n^{-2/3}\le\eta\le \mathtt C \mu_1\right\}.
\end{equation}
For $\xi^2$ with bounded support as in Case (ii) of Assumption \ref{assum_D}, for some sufficiently small constants $\mathrm{c}, \ \epsilon_d>0,$ we denote (recall $L_+$ in (\ref{edge_edge_edge}))
\begin{equation}\label{eq_spectraldomaintwo}
\mathbf{D}_{b} \equiv \mathbf{D}_{b}(\mathtt c):=\left\{z=E+\ri\eta \in \mathbb{C}_+: L_{+}-\mathtt c \le E\le L_{+}+\mathtt c, \ n^{-1/2-\epsilon_d}\le\eta\le n^{-1/(d+1)+\epsilon_d} \right\}.
\end{equation}

Throughout the paper, we will frequently use the minors of a matrix.  For the data matrix $Y$ in (\ref{eq_datamatrix}), denote the index set ${\cal I}=\{1,\dots,n\}$. Given an index set $\mathcal{T}\subset{\cal I}$, we introduce
	the notation $Y^{(\mathcal{T})}$ to denote the $p \times(n-|\mathcal{T}|)$ minor of $Y$
	obtained from removing all the $i$th columns of $Y$ for $i\in \mathcal{T}$ and keep the original indices of $Y$. In particular, $Y^{(\emptyset)}=Y$. For convenience, we briefly write $(\{i\})$, $(\{i,j\})$ and $\{i,j\}\cup \mathcal{T}$ as $(i)$, $(i,j)$
	and $(ij\mathcal{T})$ respectively. Correspondingly, we denote their sample covariance matrices and resolvents as 
	\begin{equation}\label{eq_defnminor}
	Q^{(\mathcal{T})}=(Y^{(\mathcal{T})}) (Y^{(\mathcal{T})})^*,\ {\cal Q}^{(\mathcal{T})}=(Y^{(\mathcal{T})})^* (Y^{(\mathcal{T})}). 
	\end{equation}
	and
	\begin{equation}\label{eq_defnminorG}
	G^{(\mathcal{T})}(z)=(Q^{(\mathcal{T})}-zI)^{-1}, \quad{\cal G}^{(\mathcal{T})}(z)=({\cal Q}^{(\mathcal{T})}-zI)^{-1}.
	\end{equation}
Similar to (\ref{eq_mq}) and Definition \ref{defn_couplesystem}, we can define 	
$m_Q^{(\mathcal{T})}(z)$, $m_{\mathcal{Q}}^{(\mathcal{T})}(z),$ $m_{1n}^{(\mathcal{T})}(z)$, $m_{2n}^{(\mathcal{T})}(z)$ and  $m_n^{(\mathcal{T})}(z)$ by removing $\mathbf{y}_i, i \in \mathcal{T}$ or $\xi_i^2, i \in \mathcal{T}.$

\quad We begin with the summary of the results when $\xi^2$ has unbounded support as in Case (i) of Assumption \ref{assum_D}. The proofs will be deferred to Section \ref{sec_proofpreliminarystieltjes}. Conditional on some probability event, we provide some useful deterministic estimates for $m_{1n}, m_{2n}$ and $m_n(z)$ on the above concerned spectral domain (\ref{eq_spectraldomainone}).  Denote the control parameter $e$ as follows
\begin{equation}\label{e2_definition}
e:=
\begin{cases}
\frac{\log n}{n^{1/\alpha}}, & \ \text{if (\ref{ass3.1}) holds}; \\
\frac{1}{\log^{1/\beta} n}, & \ \text{if (\ref{ass3.2}) holds}.  
\end{cases}
\end{equation} 
\begin{lemma}\label{lem: basic bounds}
Suppose Assumptions \ref{assum_model}, \ref{assumption_techincial} and (i) of Assumption \ref{assum_D} hold. For any fixed realization $\{\xi_i^2\} \in \Omega$ where $\Omega \equiv \Omega_n$ is some  probability event that $\mathbb{P}(\Omega)=1-\ro(1)$, we have  
\begin{enumerate}
\item For $z \in \mathbf{D}_{u},$ we have that for some constants $C_1, C_2>0$ 
\begin{equation*}
\operatorname{Re} m_{1n}(z) \asymp -E^{-1}, \ C_1 \eta E^{-2} \leq \operatorname{Im} m_{1n}(z) \leq C_2 \eta E^{-1}.  
\end{equation*}
\item When $|E-\mu_1| \leq C d_1$ for  some sufficiently large constant $C>0,$ let $m_{1n}(E)=\lim_{\eta \downarrow 0} m_{1n}(E+\ri \eta),$ then we have that 
\begin{equation*}
    m_{1n}(E) \asymp -E^{-1}.
\end{equation*}
\item For $z \in \mathbf{D}_{u}$ and $e$ defined in (\ref{e2_definition}), we have that 
\begin{gather*}
%    -\frac{1}{\xi^2_{(1)}+d_{2}}<\operatorname{Re} m_{1n}(z)<-\frac{\phi\bar{\sigma}}{1+c}\frac{E}{E^2+\eta^2},\\
%    \frac{\phi\bar{\sigma}}{1+c}\frac{\eta}{E^2+\eta^2}<\operatorname{Im}m_{1n}(z)<\eta|\operatorname{Re}m_{1n}(z)|<\rO(\eta E^{-1}),\\
    |m_{2n}(z)|=\rO(e),\quad |m_n(z)|=\rO(E^{-1}),\\
    \operatorname{Im}m_{2n}(z)=\rO(\eta E^{-1}),\quad \operatorname{Im}m_n(z)=\rO(\eta E^{-2}).
\end{gather*}
\end{enumerate} 
%{\color{red}[NEED TO SEPERATE THE CASE WHEN IT IS EXPONENTIAL DECAY]}
\end{lemma}

\begin{remark}\label{rem_keyrem}
The above lemma provides some controls for the Stieltjes transforms. Three remarks are in order. First, the construction of the probability event $\Omega$ will be given in Section \ref{appendix_goodconfiguration}. Second, by a discussion similar to (\ref{eq_mu1part}), conditional on $\Omega$, we can replace $\mu_1$ with $\varphi \xi^2_{(1)}.$  Third, The above results hold when we replace $m_{1n}, m_{2n}$ and $m_n$ with $m_{1n}^{(\mathcal{T})}, m_{2n}^{\mathcal{(T)}}$ and $m_n^{\mathcal{(T)}}$ for any finite $\mathcal{T}.$
\end{remark}

\quad Then we state the results when $\xi^2$ has bounded support as in Case (ii) of Assumption \ref{assum_D}. As mentioned in Remark \ref{rmk_systemequationsremark}, for the bounded support case, it will be more convenient to use both the conditional and unconditional systems. For the conditional setting, when restricted to $\Omega,$ we denote the rightmost edge of $\rho$ as $\widehat{L}_+.$ Moreover, parallel to (\ref{eq_phasetransition}), we introduce the following quantities
\begin{align}\label{eq_finitesample123}
 &  \widehat{\varsigma}_1:=
   \begin{cases}
   \frac{1}{p}\sum_{j=1}^n\frac{l^2\xi^4_j}{(l-\xi^2_j)^2}, \\
   \frac{1}{n}\sum_{j=1}^n\frac{l^2\xi^4_j}{(l-\xi^2_j)^2}
   \end{cases}
   \quad \widehat{\varsigma}_2:=
   \begin{cases}
    \frac{1}{p}\sum_{j=1}^n\frac{l\xi^2_j}{l-\xi^2_j}, & \text{if (1) of Assumption \ref{assum_model} holds} \\
   \frac{1}{n}\sum_{j=1}^n\frac{l\xi^2_j}{l-\xi^2_j},& \text{if (2) of Assumption \ref{assum_model} holds}
   \end{cases}, \\
& \widehat{\varsigma}_3:=
   \begin{cases}
     \frac{\phi^{-1}}{p}\sum_{i=1}^p \frac{\sigma_i^2\hat{\varsigma}_1}{(\widehat{L}_+-\sigma_i\hat{\varsigma}_2)^2} \\
   \frac{1}{p}\sum_{i=1}^p \frac{\sigma_i^2\hat{\varsigma}_1}{(\widehat{L}_+-\sigma_i\hat{\varsigma}_2)^2},
   \end{cases} 
   \ \widehat{\varsigma}_4:=
   \begin{cases}
      \frac{1}{p}\sum_{i=1}^p \frac{\sigma_i}{(-\widehat{L}_++\sigma_i \widehat{\varsigma}_2)^2}, &  \text{if (1) of Assumption \ref{assum_model} holds} \\
   \frac{1}{n}\sum_{i=1}^p \frac{\sigma_i}{(-\widehat{L}_++\sigma_i \widehat{\varsigma}_2)^2}, &  \text{if (2) of Assumption \ref{assum_model} holds} \nonumber
   \end{cases}.
\end{align}
Furthermore, we need the following spectral parameter set 
\begin{equation}\label{eq_spectralparameterprime}
\mathbf{D}_b^\prime=\Big\{z\in\mathbf{D}_b:|1+\xi^2_jm_{1n,c}(z)|>\frac{1}{2}n^{-1/(d+1)-\epsilon_d},\quad \text{for all} \ 2 \leq j \leq n\Big\}.
\end{equation}

\begin{remark}
We will see from (\ref{eq_secondconnect}) and (\ref{def4}) that with probability $1-\ro_{\mathbb{P}}(1),$ $\lambda_1+\ri \eta_0 \in \mathbf{D}_b'.$ 
\end{remark}

\begin{lemma}\label{localestimate2}
Suppose Assumptions \ref{assum_model}, \ref{assumption_techincial}, \ref{assum_additional_techinical} and (ii) of Assumption \ref{assum_D} hold. Then for any fixed realization $\{\xi_i^2\} \in \Omega$ where $\Omega \equiv \Omega_n$ is some  probability event that $\mathbb{P}(\Omega)=1-\ro(1),$ for sufficiently large $n,$ we have that 
\begin{enumerate}
\item[(a).] If $d>1$ and $\phi^{-1}>\widehat{\varsigma}_3,$ $\widehat{L}_+$ can be expressed explicitly by the following equations
\begin{equation}\label{eq_conditionaledgedefinitionelliptical}
    1=\frac{1}{p}\sum_i\frac{-l\sigma_i}{(-\widehat{L}_{+}+\sigma_i\hat{\varsigma}_2)}, \ \text{when (1) of Assumption \ref{assum_model} holds},
\end{equation}
and 
\begin{equation}\label{eq_conditionaledgedefinition}
    1=\frac{1}{n}\sum_i\frac{-l\sigma_i}{(-\widehat{L}_{+}+\sigma_i\hat{\varsigma}_2)}, \ \text{when (2) of Assumption \ref{assum_model} holds}.
\end{equation}
Moreover, for any $0\le\kappa\le \widehat{L}_{+}$,
\begin{gather}\label{eq: concave decay of rho_Q}
    \rho(\widehat{L}_{+}-\kappa) \asymp \kappa^d,
\end{gather}
Moreover, for some sufficiently small constant $\epsilon>0$
  \begin{equation}\label{eq_closenessequation}
 \varsigma_k=\widehat{\varsigma}_k+\rO(n^{-1/2+\epsilon}), k=1,2,3,4;\ L_+=\widehat{L}_++\rO(n^{-1/2+\epsilon}).   
\end{equation}
%{\color{red} something is missing here, how the assumption from unconditional can be transferred to conditional setting.}
 In addition, let $z=\widehat{L}_{+}-\kappa+\ri\eta\in\mathbf{D}_b$, then 
 \begin{gather}\label{eq_expansionlinear}
     m_{1n}(\widehat{L}_{+})-m_{1n}(z)=\frac{\widehat{\varsigma}_4}{(1-\phi\widehat{\varsigma}_3)}\left(\widehat{L}_{+}-z\right)+\rO((\log n)(\kappa+\eta)^{\min\{d,2\}}).
 \end{gather}
 Similarly, for any $z,z^{\prime}\in\mathbf{D}_b$, we have
\begin{gather}\label{eq_a11coro}
    m_{1n}(z)-m_{1n}(z^{\prime})=\frac{\widehat{\varsigma}_4}{(1-\phi \widehat{\varsigma}_3)}(z-z^{\prime})+\rO((\log n)(n^{-1/(d+1)})^{\min\{d-1,1\}}|z-z^{\prime}|).
\end{gather}
Finally, for $z \in \mathbf{D}_b^{\prime}$ in (\ref{eq_spectralparameterprime}), we have that 
\begin{equation}\label{eq_zopointrate11}
\operatorname{Im} m_{1n} (z) =\rO\left( \max \left\{ \eta, \frac{1}{n \eta}\right\} \right), \ \ \operatorname{Im} m_{n} (z) =\rO\left( \max \left\{ \eta, \frac{1}{n \eta}\right\} \right).
\end{equation}
Moreover, for $z_0$ defined in (\ref{eq: def of gamma}), we have that 
\begin{equation}\label{eq_zopointrate}
\operatorname{Im} m_{1n}(z_0) \asymp n^{-1/2}, \ \operatorname{Im} m_n(z_0) \asymp n^{-1/2}.
\end{equation}
and for $z=E+\ri \eta_0 \in \mathbf{D}_b^\prime$ in (\ref{eq_spectralparameterprime}), we have that 
\begin{equation}\label{eq_oneregimeedgecontrol}
\operatorname{Im} m_{1n}(z) \asymp \eta_0, \ \operatorname{Im} m_n(z) \asymp \eta_0,
\end{equation}
if $|z-z_0| \geq C n^{-1/2+3\epsilon_d}$ for some constant $C>0.$ 
%{\color{red}[add other results on the control of the Stieltjes transforms.]}
\item[(b).] If $-1<d \leq 1$ or $d>1$ and $\phi^{-1}<\widehat{\varsigma}_3,$  we have that for some fixed constant $\tau>0$ and all $1 \leq i \leq n$
\begin{equation}\label{eq_boundedfrombelowimportant}
\left|1+\xi_{i}^2 m_{1n}(\widehat{L}_+) \right| \geq \tau. 
\end{equation}
Moreover, we have that  for any $0\le\kappa\le \widehat{L}_{+}$,
\begin{gather}\label{thm_main_squared_root_bounded}
    \rho(\widehat{L}_{+}-\kappa) \asymp \kappa^{1/2}.
\end{gather}
Equivalently, for some constant $\gamma>0,$ we have for $\kappa \downarrow 0$ 
\begin{equation}\label{eq_gammadefinition}
\rho(\widehat{L}_+-\kappa)=\frac{1}{\pi} \gamma^{3/2} \sqrt{\kappa}+\rO(\kappa). 
\end{equation}

\end{enumerate}

Finally,  the results of (a) an (b) still hold unconditionally when $m_{1n}$ is replaced by $m_{1n,c},$ $\rho$ is replaced by $\widetilde{\rho}$ and $\widehat{L}_+$ is replaced by $L_+$ as in Remark \ref{rmk_systemequationsremark}  where $\widehat{\varsigma}_3$ and $\widehat{\varsigma}_4$ should be replaced by $\varsigma_3$ and $\varsigma_4$ as in (\ref{eq_phasetransition}). 

\end{lemma}

\begin{remark}\label{rmk_boundedatleastonesolution}
Using discussions similar to the paragraphs around equation (5.1) of \cite{Kwak2021}, by (\ref{eq_expansionlinear}), (\ref{def4}) and the fact $m_{1n}(\widehat{L}_+)=-l^{-1},$ we see that (\ref{eq: def of gamma}) has at least one solution. 
\end{remark}

\quad Then we provide the results of the averaged local laws. 
Throughout the paper, we will consistently use the notion of \emph{stochastic domination} to systematize the statements of the form ``$\xi$ is bounded by $\zeta$ with high probability up to a small power of $n$."
\begin{definition}[Stochastic domination]

(i) Let
\[\xi=\left(\xi^{(n)}(u):n\in\mathbb{N}, u\in U^{(n)}\right),\hskip 10pt \zeta=\left(\zeta^{(n)}(u):n\in\mathbb{N}, u\in U^{(n)}\right),\]
be two families of nonnegative random variables, where $U^{(n)}$ is a possibly $n$-dependent parameter set. We say $\xi$ is stochastically dominated by $\zeta$, uniformly in $u$, if for any fixed (small) $\epsilon>0$ and (large) $D>0$, 
\[\sup_{u\in U^{(n)}}\mathbb{P}\left(\xi^{(n)}(u)>n^\epsilon\zeta^{(n)}(u)\right)\le n^{-D}\]
for large enough $n \ge n_0(\epsilon, D)$, and we shall use the notation $\xi\prec\zeta$. Throughout this paper, the stochastic domination will always be uniform in all parameters that are not explicitly fixed, such as the matrix indices and the spectral parameter $z$.  If for some complex family $\xi$ we have $|\xi|\prec\zeta$, then we will also write $\xi \prec \zeta$ or $\xi=\mathrm{O}_\prec(\zeta)$.
%\item[(ii)] 

\vspace{5pt}

\noindent (ii) We say an event $\Xi$ holds with high probability if for any constant $D>0$, $\mathbb P(\Xi)\ge 1- n^{-D}$ for large enough $n$.

\end{definition}

\quad Similar to  \cite{couillet2014analysis,ding2021spiked, paul2009no, yang2019edge}, instead of working directly with $m_Q$ and $m_{\mathcal{Q}}$ in (\ref{eq_mq}), it is more convenient to study the following quantities
\begin{equation}\label{eq_m1m2}
m_1(z)=
\begin{cases}
\frac{1}{p}{\rm tr}\left(G(z)\Sigma\right), \\
\frac{1}{n}{\rm tr}\left(G(z)\Sigma\right) ,\\
\end{cases}
\quad  m_2(z)=
\begin{cases}
\frac{1}{p}\sum_{i=1}^n\xi^2_i\mathcal G_{ii}(z), & \ \text{if {\normalfont Case (1)} of Assumption \ref{assum_model} holds},\\
\frac{1}{n}\sum_{i=1}^n\xi^2_i\mathcal G_{ii}(z) & \ \text{if {\normalfont Case (2)} of Assumption \ref{assum_model} holds},\\
\end{cases}
\end{equation}
Analogously, using the minors in (\ref{eq_defnminor}), we can define $m_1^{(\mathcal{T})}(z)$ and $m_2^{(\mathcal{T})}(z).$  The following Theorem \ref{thm_unboundedcaselocallaw} summarizes the averaged local laws for unbounded $\xi^2$ which will be used in our proof for the main results. Its proof can be found in Section \ref{appendix_sec_prooflocallawunbounded}.

%{\color{red} [from here: provide some definitions here] 
%We first prepare some notations. Denote 
%%\begin{equation}\label{e2_definition}
%%e_2:=\frac{\log^2 n}{n^{2/\alpha}}+\frac{C\log^2n}{n^{2/\alpha}}. 
%%\end{equation}
%
%and two related quantities
%
%Since $Q$ and $\mathcal{Q}$ have the same non-zero eigenvalues, we have
%\begin{equation}\label{eq_elementaryidentity}
%    nm_{Q}(z)=pm_{\mathcal{Q}}(z)-\frac{n-p}{z}.
%\end{equation}
%And we define the limiting spectral distribution of $Q$ and $\mathcal{Q}$ as $\tilde{\rho}_{Q}$ and $\tilde{\rho}_{\mathcal{Q}}$ respectively. We denote the rightmost endpoint of the support of $\tilde{\rho}_{Q}$(or $\tilde{\rho}_{\mathcal{Q}}$) as $L_{+}$.
%}

\begin{theorem}[Averaged local laws for unbounded support $\xi^2$] \label{thm_unboundedcaselocallaw} Suppose Assumptions \ref{assum_model}, \ref{assumption_techincial} and (i) of Assumption \ref{assum_D} hold. For any fixed realization $\{\xi_i^2\} \in \Omega$ where $\Omega \equiv \Omega_n$ is introduced in Lemma \ref{lem: basic bounds}, let $m_1^{(1)}(z)$ and $m_{1n}^{(1)} (z)$ be defined by removing the column or entries associated with $\xi_{(1)}^2.$ We have that the followings hold true uniformly for $z \in \mathbf{D}_{u}$ in (\ref{eq_spectraldomainone})
\begin{enumerate}
\item If {\normalfont Case (a)} of (i) of  Assumption \ref{assum_D} holds, we have that 
\begin{equation*}
%m_{1}(z)=m_{1n}(z)+\rO_{\prec}\left(n^{-1/2-2/\alpha}\right), \ 
m_{1}^{(1)}(z)=m_{1n}^{(1)}(z)+\rO_{\prec}\left(n^{-1/2-2/\alpha}\right). 
\end{equation*}

\item If {\normalfont Case (b)} of (i) of Assumption \ref{assum_D} holds, we have that 
\begin{equation*}
%m_{1}(z)=m_{1n}(z)+\rO_{\prec}\left(n^{-1/2}\right), \ 
m_{1}^{(1)}(z)=m_{1n}^{(1)}(z)+\rO_{\prec}\left(n^{-1/2}\right).
\end{equation*}
\end{enumerate}  

\end{theorem}

\quad Then we provide the averaged local laws for bounded $\xi^2.$ Recall $m_Q$ defined in (\ref{eq_mq}). 

%{\color{red} add a remark here on the probability event }

\begin{theorem}[Averaged local laws for bounded support $\xi^2$] \label{thm_boundedcaselocallaw} Suppose Assumptions \ref{assum_model}, \ref{assumption_techincial}, \ref{assum_additional_techinical} and (ii) of Assumption \ref{assum_D} hold. When $d>1$ and $\phi^{-1}>\varsigma_3,$ for any fixed realization $\{\xi_i^2\} \in \Omega$ where $\Omega \equiv \Omega_n$ is introduced in Lemma \ref{localestimate2}, for $\eta_0$ defined in (\ref{eq_eta0definition}),  we have that the followings hold true uniformly for $z \in \mathbf{D}_b^\prime$ in (\ref{eq_spectralparameterprime}) 
\begin{equation*}
\left| m_{1n}(z)-m_{1n,c}(z) \right| \leq n^{-1/2+\epsilon_d}, \ \left| m_1(z)-m_{1n}(z) \right|=\rO_{\prec} \left((n \eta_0)^{-1} \right),
\end{equation*}
and 
\begin{equation*}
\left| m_{n}(z)-m_{n,c}(z) \right| \leq n^{-1/2+\epsilon_d} , \ \left| m_Q(z)-m_{n}(z) \right|=\rO_{\prec} \left((n \eta_0)^{-1} \right).
\end{equation*}
\end{theorem}

\subsection{Characterization of "good configurations"}\label{appendix_goodconfiguration}

%Besides above assumptions, we need to specify the distribution among the largest elements of $\mathscr{D}\mathscr{D}^{*}$. Define $\mathcal{A}$ as the set of all permutation of $\{1,\dots, n\}$. Then the order of $\{\xi^2_1,\dots,\xi^2_n\}$ uniformly distributes from $\mathcal{A}$. We use $\{s_1,\dots,s_n\}\in\mathcal{A}$ and $\xi^2_{s_1}\ge\dots\ge\xi^2_{s_n}$ to denote the order. Now we give the definition of a series of event $\Omega_n$, on which $\mathscr{D}\mathscr{D}^{*}$ has a good configuration,

In this subsection, independent of Section \ref{sec_asymptoticlocalaveragedlocal}, we define some probability events which are some "good configurations" for the first few largest eigenvalues of $D^2.$ Our proofs will be restricted on these probability events. In fact, as will be seen in Lemma \ref{lem_probabilitycontrol}, under Assumption \ref{assum_D},  these probability events hold with high probability when $n$ is sufficiently large. 

Recall Assumption \ref{assum_D} and
\begin{equation*}
D^2=\operatorname{diag}\left\{ \xi_1^2, \cdots, \xi_n^2 \right\}.
\end{equation*}
Moreover, we define the order statistics of $\{\xi^2_i\}$ as 
\begin{equation*}
\xi_{(1)}^2 \geq \xi_{(2)}^2 \geq \cdots \geq \xi_{(n)}^2. 
\end{equation*}
In what follows, we define these probability events according to the various assumptions of $\{\xi_i^2\}$ in (\ref{ass3.1})--(\ref{ass3.4}).

\begin{definition}\label{defn_probset}
Denote $\Omega \equiv \Omega_n$ be the event on $\{\xi_i^2\}$ so that the following conditions hold:
\begin{itemize}
\item[] {\bf (a). Unbounded support with polynomial decay.}  When $\{\xi^2_i\}$ has unbounded support with polynomial decay tail as in (\ref{ass3.1}), we assume that for all $\epsilon\in(0,1/\alpha)$, $b\in(1/2,1]$  and some constants $C, c>1,$ the following holds on $\Omega$
\begin{equation}\label{def1}
\begin{aligned}
& \xi^2_{(1)}-\xi^2_{(2)}\ge C^{-1}n^{1/\alpha}\log^{-1}n, \\
& C^{-1}n^{1/\alpha}\log^{-1}n\le\xi^2_{(1)}\le Cn^{1/\alpha}\log n,  \\
& \xi^2_{(i)}-\xi^2_{(i+1)}\ge C^{-1}n^{\epsilon}\log^{-1}n, \ 1 \leq i<\sqrt{n},\\
&  \xi^2_{(1)}-\xi^2_{(\lceil n^b \rceil)}\ge c^{-1}n^{1/\alpha}\log^{-1} n, \\
&  \frac{1}{n}\sum_{i=1}^n\xi^2_i< \infty. 
%C\log^{1/\alpha} n.
\end{aligned}
\end{equation}
\item[] {\bf (b). Unbounded support with exponential decay.} When $\{\xi^2_i\}$ has unbounded support with polynomial decay tail as in (\ref{ass3.2}), we assume that for some constant $C>1,$ 
%and sequence of constants $\mathsf c_n$ that $\mathsf c_n=\ro (1),$ 
the following holds on $\Omega$
\begin{equation}\label{def3}
\begin{aligned}
& \xi^2_{(1)}-\xi^2_{(2)}\ge C^{-1}\log^{1/\beta} n, \\
&  C^{-1}\log^{1/\beta} n\le\xi^2_{(1)}\le C\log^{1/\beta} n, \\
&  \frac{1}{n}\sum_{i=1}^n\xi^2_i<\infty. 
%\mathsf c_{n} \log n.
\end{aligned}
\end{equation}
\item[] {\bf (c). Bounded support with $d>1.$}  When $\{\xi^2_i\}$ has bounded support satisfying (\ref{ass3.4}) with $d>1,$ we assume that for some sufficiently small constant $\epsilon>0,$ $\epsilon_d<1/8(1/2-1/(d+1))$, $0<b \leq 1$ and $0<C_l<l,$ the following holds on $\Omega$
\begin{equation}\label{def4}
\begin{aligned}
& n^{-1/(d+1)-\epsilon_d}<l-\xi^2_{(1)}<n^{-1/(d+1)}\log n, \\
& \xi^2_{(1)}-\xi^2_{(2)}> n^{-1/(d+1)-\epsilon_d}, \\
& l-\xi^2_{(\lfloor bn \rfloor)}>C_l, \\
& \frac{1}{n}\sum_{i=1}^n\xi^2_{i}\le l, \\
&  \left|\frac{1}{n}\sum_{i=1}^n\frac{\xi^2_i}{1+\xi^2_{i}m_{1 n,c}(z)}-\int\frac{t}{1+tm_{1 n,c}(z)}\mathrm{d} F(t)\right|\le\frac{Cn^{\epsilon}}{\sqrt{n}}, \ \text{for} \  z \in \mathbf{D}_b,
\end{aligned}
\end{equation}
where we recall that  $F(t)$ is the distribution of $\xi^2$ and $C>0$ is some generic constant. 
\end{itemize}
\end{definition}

\begin{remark}
Two remarks are in order. First, on the event $\Omega,$ for the unbounded support case, according to (a) and (b), we see that the first few largest $\xi_i^2$ are divergent and well separated from each other. Second, for the bounded support case, we only provide the results for $d>1$ in (c). Nevertheless, it is easy to see that similar results can be obtained for $-1<d \leq 1.$   
\end{remark}

\quad The following lemma shows that under Assumption \ref{assum_D}, the probability event $\Omega$ happens with high probability in all the four settings. The proof will be given in Section \ref{sec_appendxi_goodevent}. 

\begin{lemma}\label{lem_probabilitycontrol}
Let $\Omega$ be the events defined in Definition \ref{defn_probset}, suppose Assumption \ref{assum_D} holds, we then have that when $n$ is sufficiently large  
\begin{equation*}
\mathbb{P}(\Omega)=1-\rO(\log^{-D}n),
\end{equation*}
for some constant $D>0$.
\end{lemma}

%\begin{remark}
%{\color{red} 1. some more detailed remarks on sampling, and the form of $\ro(1)$ and refer the detail to the proof. 2. Even though we will not use it, we can also prove that
%{\bf (c). Bounded support with $0<d \leq 1.$ } When $\{\xi^2_i\}$ has bounded support satisfying (\ref{ass3.4}) with $0<d \leq 1$, we assume that for some sufficiently small constant $\epsilon>0$ the following holds on $\Omega$
%\begin{equation}\label{def5}
%\begin{aligned}
%& |l-\xi^2_{(1)}|\le n^{-1/(d+1)-\epsilon}, \\
%& \xi^2_{(1)}-\xi^2_{(2)}\le n^{-1/(d+1)-\epsilon} \\
%& \frac{1}{n}\sum_{i=1}^n\xi^2_{s_i}\le l.
%\end{aligned}
%\end{equation}}
%\end{remark}

\subsection{Some useful lemmas and a summary of extreme value theory}\label{sec_summaryofextremevalue}
In this subsection, we first provide some technical lemmas which will be used in our proof. The following resolvent identities play an important role in our proof. Recall the resolvents defined in (\ref{eq_resolvents}) and the minors defined in (\ref{eq_defnminor}). 
\begin{lemma}[Resolvent identities]\label{lem: Resolvent}
Let $\{\mathbf{y}_i\} \subset \mathbb{R}^p$ be the columns of $Y$ as in (\ref{eq_datamatrix}), then we have that 
\begin{eqnarray*}
			\mathcal G_{ii}(z) & = & -\frac{1}{z+z\mathbf{y}_{i}^{*}{ G}^{(i)}(z)\mathbf{y}_{i}},\\
			\mathcal G_{ij}(z) & = & z \mathcal{G}_{ii}(z)\mathcal{G}_{jj}^{(i)}(z)\mathbf{y}_{i}^{*}{G}^{(ij)}(z)\mathbf{y}_{j}\quad i\ne j,\\
			\mathcal G_{ij}(z) & = & \mathcal G_{ij}^{(k)}(z)+\frac{\mathcal G_{ik}(z)\mathcal G_{kj}(z)}{\mathcal G_{kk}(z)}\quad i,j\ne k.
		\end{eqnarray*}
%The results also hold for $G^{\mathcal{T}}$.
\end{lemma}

\begin{proof}
The proof is straightforward using Schur's complement formula; for example see \cite[Lemma 2.3]{PillaiandYin2014}. 
\end{proof}

\begin{lemma}[Some useful matrix identities] For any finite subset $\mathcal{T}\subset\{1,\dots,n\}$, we have that 
\begin{equation}\label{lem:Wald}
\left\|G^{(\mathcal{T})}\Sigma^{1/2} \right \|^2_F=\eta^{-1}\operatorname{Im}\operatorname{Tr} \left(G^{(\mathcal{T})} \Sigma \right). 
\end{equation}
Moreover, we have that 
	\begin{gather}
			\left|{\rm Tr}(\mathcal{G}^{(i)}-\mathcal{G}) \right|  \le  \eta^{-1}, \nonumber \\
			\left|{\rm Tr}({ G}^{(i)}-{ G}) \right| \le  |z|^{-1}+\eta^{-1}, \label{lem:trace_difference} \\
			\left|\operatorname{Im}{\rm Tr}({G}^{(i)}-{\cal G})\right|  \le  \eta|z|^{-2}+\eta^{-1}. \nonumber
		\end{gather}
\end{lemma}
\begin{proof}
Due to similarity, we focus our discussion on the separable covariance i.i.d. data, i.e., Case (2) of Assumption \ref{assum_model}. In fact, it is easier to handle Case (1) since $\Sigma$ can  be always assumed to be  diagonal. 

We start with the proof of (\ref{lem:Wald}). Recall (\ref{eq_defnminorG}). We can write
\begin{equation*}
G^{(\mathcal{T})}=\left( \Sigma^{1/2} X^{(\mathcal{T})} D^2 X^{(\mathcal{T})} \Sigma^{1/2}-z   \right)^{-1}.
\end{equation*}
Let the spectral decomposition $\Sigma=U\Lambda U^*.$ Observe that 
\begin{align*}
\| G^{(\mathcal{T})} \Sigma^{1/2} \|_F^2 & =\operatorname{Tr} \left( \left( \Sigma^{1/2} X^{(\mathcal{T})} D^2 X^{(\mathcal{T})} \Sigma^{1/2}-z   \right)^{-1} U\Lambda U^* \left( \Sigma^{1/2} X^{(\mathcal{T})} D^2 X^{(\mathcal{T})} \Sigma^{1/2}-\bar{z}   \right)^{-1}\right) \\
&=\operatorname{Tr} \left( U \left( \Lambda^{1/2} U^* X^{(\mathcal{T})} D^2 X^{(\mathcal{T})} U \Lambda^{1/2}-z   \right)^{-1} U^* U\Lambda U^* U \left( \Lambda^{1/2} U^* X^{(\mathcal{T})} D^2 X^{(\mathcal{T})} U \Lambda^{1/2}-\bar{z}   \right)^{-1} U^*\right) \\
&= \left\| \left( \Lambda^{1/2} U^* X^{(\mathcal{T})} D^2 X^{(\mathcal{T})} U \Lambda^{1/2}-z   \right)^{-1} \Lambda^{1/2} \right\|_F^2 \\
&=\eta^{-1} \operatorname{Im} \operatorname{Tr} \left[ \left( \Lambda^{1/2} U^* X^{(\mathcal{T})} D^2 X^{(\mathcal{T})} U \Lambda^{1/2}-z   \right)^{-1} \Lambda \right] , \\
&=\eta^{-1} \operatorname{Im} \operatorname{Tr} \left[ U^* U\left( \Lambda^{1/2} U^* X^{(\mathcal{T})} D^2 X^{(\mathcal{T})} U \Lambda^{1/2}-z   \right)^{-1} U^* U \Lambda \right] , \\
& =\eta^{-1} \operatorname{Im} \operatorname{Tr} \left( G^{(\mathcal{T})} \Sigma \right),
\end{align*}
where in the fourth step we used Ward's identity (see the equation below (4.42) of \cite{ding2020local}).

Second, the proof of (\ref{lem:trace_difference}) follows from the definitions of the resolvents; see \cite[Lemma A.4]{Ding&Yang2018} and the proof of \cite[Lemma 4.6]{Bao2015} for more detail.   
 
\end{proof}

%{\color{red} [revise here]}
%{\color{red} ALSO need to list the i.i.d. case}
\begin{lemma}[Large deviation bounds]\label{lem:large deviation}
Let $\mathbf{u}=(u_1, u_2, \cdots, u_p)^*, \widetilde{\mathbf{u}}=(\widetilde{u}_1, \widetilde{u}_2, \cdots, \widetilde{u}_p)^* \in \mathbb{R}^p$  be two real independent random vectors. Moreover, let $A$ be a $p \times p$ matrix independent of the above vectors. Then the following holds. 
\begin{enumerate}
\item[(1).] When the entries of the random vectors are centered i.i.d. random variables with variance $n^{-1}$ and  $\mathbb{E}|\sqrt{n} v_{i}|^k \leq C_k,$ where $v_i=u_i, \widetilde{u}_i, 1 \leq i \leq p,$  then we have that 
		\begin{align*}
		|\widetilde{\mathbf{u}}^{*}\mathbf{u}|  \prec  \sqrt{\frac{\|\mathbf{u}\|^{2}}{n},} \ \ \left|\mathbf{u}^{*}A\tilde{\mathbf{u}}\right| \prec  \frac{1}{n}\|A\|_{F}, \ \
		|\mathbf{u}^{*}A\mathbf{u}-\frac{1}{n}{\rm Tr}A|  \prec  \frac{1}{n}\|A\|_{F}. 		
		\end{align*}
		\item[(2).] When the random vectors	are sampled from  $\mathsf{U}(\mathbb{S}^{p-1}),$ then we have that 
				\begin{align*}
		|\widetilde{\mathbf{u}}^{*}\mathbf{u}|  \prec  \sqrt{\frac{\|\mathbf{u}\|^{2}}{p},} \ \ \left|\mathbf{u}^{*}A\tilde{\mathbf{u}}\right| \prec  \frac{1}{p}\|A\|_{F}, \ \
		|\mathbf{u}^{*}A\mathbf{u}-\frac{1}{p}{\rm Tr}A|  \prec  \frac{1}{p}\|A\|_{F}. 		
		\end{align*}
\end{enumerate}

%\noindent (2).  
%		random vectors, $A=(a_{ij})$  an $p\times p$ matrix and $\mathbf{b}=(b_{1},\dots,b_{p})^{*}$
%		an $p$-dimensional vector, where $A$ and $\mathbf{b}$ may be
%		complex-valued and $\mathbf{u},\tilde{\mathbf{u}},A,\mathbf{b}$ are
%		independent. Then as $p\to\infty$
%		\begin{eqnarray*}
%		|\mathbf{b}^{*}\mathbf{u}| & \prec & \sqrt{\frac{\|\mathbf{b}\|^{2}}{p},}\\
%		|\mathbf{u}^{*}A\mathbf{u}-\frac{1}{p}{\rm Tr}A| & \prec & \frac{1}{p}\|A\|_{F},\\
%		\Big|\mathbf{u}^{*}A\tilde{\mathbf{u}}\Big| & \prec & \frac{1}{p}\|A\|_{F}.
%		\end{eqnarray*}
%		Moreover, if $\mathbf{u},\tilde{\mathbf{u}},A,\mathbf{b}$ depend on an index $t\in \mathcal{T}$ for some set $\mathcal{T}\subset\{1,\dots,n\}$, then the above domination
%		bounds hold uniformly for $t$.
\end{lemma}
\begin{proof}
The proof of (1) can be found in Lemma 3.4 of \cite{PillaiandYin2014} or Lemma 5.6 of \cite{yang2019edge}; the proof of (2) can be found in Lemma I.3 of \cite{Wen2021}. 
\end{proof}

\quad In what follows, we provide a mini-review of the extreme value theory for a sequence of i.i.d. random variables following \cite{hansen2020three}. For more systematic treatments, we refer to the monographs \cite{beirlant2004statistics,coles2001introduction,fang1990,resnick2008extreme}.  

\begin{lemma}[Fisher-Tippett-Gnedenko Theorem]\label{lem_summaryevt} Let $\{x_i^2\}$ be a sequence of i.i.d. random variables and denote $M_n:=x_{(1)}^2$ as the largest order statistic. 
\begin{enumerate}
\item If there exist some constants $\alpha_n>0$ and $\beta_n \in \mathbb{R}$ and some non-degenerate {\normalfont cdf} G such that $\alpha_n^{-1}(M_n-\beta_n)$ converges in distribution to G, then $G$ belongs to the type of one of the following three {\normalfont cdfs:}
\begin{equation*}
\begin{split}
&\text{Gumbel}: \ G_0(x)=\exp(-e^{-x}), \ x \in \mathbb{R}, \\
&\text{Fr{\'e}chet}: \ G_{1, \alpha}(x)=\exp(-x^{-\alpha}), \ x \geq 0, \ \alpha>0, \\
& \text{Weibull}: \ G_{2, \alpha}(x)=\exp(-|x|^{\alpha}), \ x \leq 0, \ \alpha>0.
\end{split}
\end{equation*}
\item Recall (\ref{eq_defnbn}). First, if $\{x_i\}$ satisfies (\ref{ass3.1}), we have that 
\begin{equation*}
\frac{M_n}{b_n} \overset{d}{\Rightarrow}  G_{1, \alpha},
\end{equation*}
Moreover, if we futher assume $\lim_{x \uparrow \infty} L(x)=\mathsf{C}$ for some constant $\mathsf{C}>0,$ then $b_n=(\mathsf{C} n)^{1/\alpha}.$ Second, if $\{x_i\}$ satisfies (\ref{ass3.2}) and (\ref{assum_gg}), we have that
\begin{equation*}
\mathsf{g}'(b_n)(M_n-b_n) \overset{d}{\Rightarrow}  G_{0}.
\end{equation*} 
Finally, if (\ref{ass3.4}) holds, recall $\mathfrak{b}$ in (\ref{eq_defnbfrak}), we have that 
\begin{equation*}
(\mathfrak{b}n)^{1/(d+1)}(M_n-l) \overset{d}{\Rightarrow} G_{2, d+1}. 
\end{equation*}
 \end{enumerate}
\end{lemma}

\begin{proof}
The proof can be found in the standard textbook or review article regarding extreme value theory. For example,  see \cite{hansen2020three} and \cite{beirlant2004statistics}. 
\end{proof}

\section{Proof of averaged local laws}\label{appendxi_locallawproof}

In this section, we prove the local laws Theorems \ref{thm_unboundedcaselocallaw} and \ref{thm_boundedcaselocallaw}. Throughout the proof, due to similarity, we focus on the discussion of the separable covariance i.i.d. model as in Case (2) of Assumption \ref{assum_model} and only briefly explain the main differences from the elliptically distributed model as in  Case (1) of Assumption \ref{assum_model}.

%{\color{red} SAY SOMETHING THAT WE ONLY FOCUS ON THE I.I.D. CASE. AND THEN MENTION THE DIFFERENCE IN ELLIPTICAL CASE.}

\subsection{Unbounded support setting: proof of Theorem \ref{thm_unboundedcaselocallaw}}\label{appendix_sec_prooflocallawunbounded}

In this section, we will prove Theorem \ref{thm_unboundedcaselocallaw}. Due to similarity, we focus on the proof of part 1 and briefly discuss that of part 2. The proof contains two steps. In the first step, we will establish the results for the results of $Q$ outside the bulk of the spectrum on the domain $\widetilde{\mathbf{D}}_{u}$ denoted as follows   
\begin{equation}\label{eq_spectraldomainonereduced}
\widetilde{\mathbf{D}}_{u  } \equiv \widetilde{\mathbf{D}}_{u }(\mathtt C):=\left\{z=E+\ri\eta: 0<E-\mu_1 \leq \mathtt C d_{1}, \ n^{-2/3}\le\eta\le \mathtt C \mu_1 \right\}.
\end{equation}
That is, we will establish the following proposition. 
\begin{proposition} \label{eq_propooursidebulk} Under the assumptions of Theorem \ref{thm_unboundedcaselocallaw}, the following results hold uniformly on the spectral domain $\widetilde{\mathbf{D}}_{u}$ in (\ref{eq_spectraldomainonereduced}) when conditional on the event $\Omega$ in Lemma \ref{lem_probabilitycontrol}. 
\begin{enumerate}
\item[(1).] If {\normalfont Case (a)} of (i) of  Assumption \ref{assum_D} holds, we have that 
\begin{align}\label{eq_entrywiselocallaw}
\mathcal{G}_{ij}(z)=-\frac{\delta_{ij}}{z(1+m_{1n}(z)\xi_i^2)}+\rO_{\prec} \left( n^{-1/2-1/\alpha} \right),
\end{align}
where $\delta_{ij}$ is the Dirac delta function so that $\delta_{ij}=1$ when $i=j$ and $\delta_{ij}=0$ when $i \neq j.$ Moreover, we have that 
\begin{equation*}
m_1(z)=m_{1n}(z)+\rO_{\prec}\left( n^{-1/2-2/\alpha}\right), \  m_2(z)=m_{2n}(z)+\rO_{\prec}\left( n^{-1/2-1/\alpha}\right), 
\end{equation*}
and 
\begin{equation}\label{eq_averagedlawused}
m_{Q}(z)=m_{n}(z)+\rO_{\prec} \left( n^{-1/2-2/\alpha} \right).
\end{equation}
\item[(2).] If {\normalfont Case (b)} of (i) of Assumption \ref{assum_D} holds, we have that the results in part (1) hold by setting $\alpha=\infty.$
\end{enumerate}  
\end{proposition}
\quad  Once Proposition \ref{eq_propooursidebulk} is proved, we can quantify the rough locations of the eigenvalues of $Q$ as summarized in the following lemma. 
 
% {\color{red}[NEED TO BE MORE CAREFUL ON THE CASE $\alpha=\infty$. Need to go over the proof below again.]}
\begin{lemma}\label{lem: upper bound for eigenvalues}
Suppose Assumptions \ref{assum_model}, \ref{assumption_techincial} and (i) of Assumption \ref{assum_D} hold. For some sufficiently large constant $C>0,$ with high probability, for any fixed realization $\{\xi_i^2\} \in \Omega$ where $\Omega \equiv \Omega_n$ is introduced in Lemma \ref{lem_probabilitycontrol}, for all $1 \leq i \leq \min\{p,n\},$ we have that 
\begin{equation}\label{eq_boundestimateone}
\lambda_i(Q) \notin (\mu_1, C n^{1/\alpha} \log n), \ \text{if \normalfont{Case (i)-a} of Assumption \ref{assum_D} holds},  
\end{equation}
and 
\begin{equation}\label{eq_boundestimatetwo}
\lambda_i(Q) \notin (\mu_1, C  \log^{1/\beta} n), \ \text{if {\normalfont Case (i)-b of Assumption \ref{assum_D} holds}}.  
\end{equation}
%
%\ref{ass1},\ref{ass2},\ref{ass3.1},\ref{ass3.2} and events $\Omega_n$ hold. With high probability there is no eigenvalue of $\mathcal{W}$ in $(\lambda_{(1)},Cn^{2/\alpha}\log n)$ for $\alpha\in(0,+\infty)$ and no eigenvalues in $(\lambda_{(1)},C\log^{K}n)$ for $\alpha=+\infty$.
\end{lemma}
\begin{proof}
Due to similarity, we focus our arguments on (\ref{eq_boundestimateone}). We prove the results by contradiction. 
%We use contradiction to prove this lemma and we take the case $\alpha\in(0,+\infty)$ for example. 
Assume there is an eigenvalue of $Q$ lies in the interval as in (\ref{eq_boundestimateone}), denote as $\widehat{\lambda}$. Let $z=\hat{\lambda}+\ri n^{-2/3}.$ Since $z\in\widetilde{\mathbf{D}}_u \subset \mathbf{D}_u$ as in (\ref{eq_spectraldomainone}), by Lemma \ref{lem: basic bounds}, we obtain $\operatorname{Im}m_n(z)=\eta \widehat{\lambda}^{-2}$. According to \eqref{eq_mu1part} and (\ref{def1}), we have that on the event $\Omega$
\begin{equation}\label{eq_lowerboundmu1}
\mu_1 \gtrsim n^{1/\alpha} \log^{-1} n.
\end{equation} 
\quad Together with (\ref{eq_averagedlawused}), we readily see that 
% Therefore, for $\alpha\in(0,2]$,
%\begin{equation}\label{eq: value of m_W outside spectrum alpha=<2}
%    \begin{split}
%        \operatorname{Im}m_{\mathcal{W}}(z)&=\operatorname{Im}m_n(z)+\operatorname{Im}(m_{\mathcal{W}}(z)-m_n(z))\\
%        &\le O(n^{-2/\alpha-5/3}\log n)+\max_{z\in\mathbf{D}}|m_{\mathcal{W}}(z)-m_n(z)|\\
%        &=O_{\prec}(n^{-3/2-2/\alpha}),
%    \end{split}
%\end{equation}
%and
% for $\alpha\in(2,+\infty)$,
%{\color{red} revised here in the end, something's wrong}
\begin{equation}\label{eq: value of m_W outside spectrum}
    \begin{split}
        \operatorname{Im}m_{Q}(z)&=\operatorname{Im}m_n(z)+\operatorname{Im}(m_{Q}(z)-m_n(z))\\
        &\prec n^{-2/\alpha-2/3}+n^{-1/2-2/\alpha} \prec n^{-1/2-2/\alpha}.
    \end{split}
\end{equation}
On the other hand, we have
\begin{equation*}
\operatorname{Im}m_{Q}(z)=\frac{1}{n}\sum_i\frac{\eta}{(\lambda_i-\widehat{\lambda})^2+\eta^2}\ge\frac{1}{n\eta}=n^{-1/3},
\end{equation*}
which contradicts \eqref{eq: value of m_W outside spectrum}. Therefore, there is no eigenvalue in this interval. Similarly, we can prove (\ref{eq_boundestimatetwo}). The only difference is that (\ref{eq_lowerboundmu1}) should be replaced by $\mu_1 \gtrsim  \log^{1/\beta} n$ according to (\ref{def3}) so that the error rate in (\ref{eq: value of m_W outside spectrum}) should be updated to $n^{-1/2}.$ This completes the proof. 
\end{proof}

\quad Armed with the above lemma, we can proceed to the second step to conclude the proof of Theorem \ref{thm_unboundedcaselocallaw}.  In what follows, we first provide the proof of Proposition \ref{eq_propooursidebulk} in Section \ref{prop_subsubsubsubsubsub}. After that, we prove Theorem \ref{thm_unboundedcaselocallaw} in Section \ref{sec_proofa2}.

\subsubsection{Proof of Proposition \ref{eq_propooursidebulk}}\label{prop_subsubsubsubsubsub}

%{\color{red} NEED TO REWRITE THE RESULTS IN THE END. NEED TO MENTION IN THE PROOF THAT WE HAVE ALREADY GIVEN THE ORDER OF $\xi.$ ALSO EXPLAIN THE STRATEGY HERE}

We first prepare two lemmas. The first one is to establish Proposition \ref{eq_propooursidebulk} for large scale of $\eta.$
%In what follows, for the simplicity of statements, without we will condition  on the order of $\{\xi^2_i\}$ that $\xi_1^2 \geq \xi_2^2 \geq \cdots \geq \xi_n^2.$ 
%
%
%
%Recall that the order is independent of $T\mathcal{U}$ and from uniform distribution. Write $\mathbf{s}=\{s_1,\dots,s_n\}$ as a permutation of $\{1,\dots,n\}$ and $\mathbf{S}=\{\mathbf{s}_1,\dots,\mathbf{s}_{2^n}\}$. Then,
%\[
%\begin{split}
%    \mathbb{P}(|\lambda_1-\mu_1|&\le n^{-1/2+2/\alpha+c}\log^{-1}n)\sum_{i=1}^{2^n}\mathbb{P}(|\lambda_1-\mu_1|\le n^{-1/2+2/\alpha+c}\log^{-1}n,\mathbf{s}=\mathbf{s}_i)\\
%    &=\sum_{i=1}^{2^n}\mathbb{P}(|\lambda_1-\mu_1|\le n^{-1/2+2/\alpha+c}\log^{\alpha\delta(\alpha-4)-1}n\mid\mathbf{s}=\mathbf{s}_i)\mathbb{P}(\mathbf{s}=\mathbf{s}_i)\\
%    &\ge\frac{1}{2^n}\sum_{i=1}^{2^n}(1-n^{-d})\\
%    &\ge1-n^{-d},
%\end{split}
%\]
%for any $d>0$. This lemma is concluded now.

\begin{lemma}[Average local law for large $\eta$]\label{lem: average local law for large eta} Proposition \ref{eq_propooursidebulk} holds when $\eta=\mathtt{C} \mu_1.$
\end{lemma}

\begin{proof}
As mentioned before, due to similarity, we focus on the study of the separable covariance i.i.d. model as in Case (2) of Assumption \ref{assum_model} when $\xi^2$ satisfies (\ref{ass3.1}). The main differences from the other cases will be explained in the end of the proof. In what follows, without loss of generality, we assume that $\xi_1^2 \geq \xi_2^2 \geq \cdots \geq \xi_n^2.$

Note that according to (\ref{eq_mu1part}), (\ref{def1}) and the definition of $d_1$ in (\ref{eq_firstddefinition}), on the event $\Omega,$ we have that 
\begin{equation}\label{eq_Econtrol}
E \asymp \xi_1^2. 
\end{equation}
When $\eta=\mathtt{C} E$, we have $\max\left\{\|G^{(\mathcal{T})}\|,\|\mathcal{G}^{(\mathcal{T})}\|\right\}\le\eta^{-1}=\mathtt{C}^{-1} E^{-1}$ for any finite $\mathcal{T}\subset\{1,\dots, n\}$ by definition.  The main idea is to explore the relation of $m_1$ and $m_2$ using Lemma \ref{lem: Resolvent}. We start with $m_2.$ By Lemma \ref{lem: Resolvent} and definition of $m_2$ in (\ref{eq_m1m2}), we have
\begin{gather}\label{eq: decomp m_2}
    {
    m_2=\frac{1}{n}\sum_{i=1}^n\frac{\xi^2_i}{-z-z\mathbf{y}^{*}_i G^{(i)}\mathbf{y}_i}=\frac{1}{n}\sum_{i=1}^n\frac{\xi^2_i}{-z(1+\xi^2_in^{-1}{\rm tr}G^{(i)}\Sigma+Z_i)},}\\
    Z_i=\mathbf{y}^{*}_i G^{(i)}\mathbf{y}_i-\xi^2_in^{-1}{\rm tr}G^{(i)}\Sigma. \nonumber
\end{gather}
As $\mathbf{y}_i$ is independent of $G^{(i)}$, by (1) of  Lemma \ref{lem:large deviation}, we see that 
\begin{equation}\label{eq: bound for Z}
    Z_i\prec\frac{\xi^2_i}{n}\|G^{(i)}\Sigma\|_F\le\frac{\xi^2_i}{n}\|G^{(i)}\|\|\Sigma\|_F\prec\frac{\xi^2_i}{\sqrt{n}\eta}.
\end{equation}
Moreover, using the definition of $m_1$ in (\ref{eq_m1m2}) and the second resolvent identity, we readily obtain that for some constant $C>0$ 
\begin{equation}\label{eq_m1approximate}
    \frac{1}{n}{\rm tr}(G^{(i)}\Sigma)-m_1(z)=\frac{1}{n}\mathbf{y}_i^{*}G \Sigma G^{(i)}\mathbf{y}_i\le C\frac{\xi^2_i}{n \eta^2}.
\end{equation}

Moreover, by (\ref{eq_Econtrol}) and the form of $\eta,$ we find that for some constant $C>0$
\begin{equation*}
|1+\xi^2_i m_1|\ge 1-C \mathtt{C}^{-1}>0, 
\end{equation*} 
when $\mathtt{C}>0$ is chosen to be sufficiently large. 
Together with (\ref{eq: decomp m_2}), we obtain
\begin{equation}\label{eq: m_2 by m_1}
    m_2=\frac{1}{n}\sum_{i=1}^n\frac{\xi^2_i}{-z(1+\xi^2_im_1+\rO_{\prec}(\frac{\xi^2_i}{\sqrt{n}\eta}))}=\frac{1}{n}\sum_{i=1}^n\frac{\xi^2_i}{-z(1+\xi^2_im_1)}+\rO_{\prec}(n^{-1/2-1/\alpha}),
\end{equation}
where we used (\ref{def1}). 

Then we work with $m_1.$ Decompose that 
\[
Q-zI=\sum_{i=1}^n\mathbf{y}_i\mathbf{y}_i^{*}+zm_2(z)\Sigma-z(I+m_2(z)\Sigma).
\]
Applying resolvent expansion to the order one, we obtain that
\[
G=-z^{-1}(I+m_2(z)\Sigma)^{-1}+z^{-1}G\left(\sum_{i=1}^n\mathbf{y}_i\mathbf{y}_i^{*}+zm_2(z)\Sigma\right)(I+m_2(z)\Sigma)^{-1}.
\]
Furthermore, using Shernman-Morrison formula, we have that
\begin{equation}\label{eq_usefulformula}
G\mathbf{y}_i=\frac{G^{(i)}\mathbf{y}_i}{1+\mathbf{y}_i^{*} G^{(i)}\mathbf{y}_i}.
\end{equation}
Combining the above two identities and Lemma \ref{lem: Resolvent}, we can further write
\begin{equation}\label{eq: decomp of mathcal_G}
\begin{split}
    G&=-z^{-1}(I+m_2(z)\Sigma)^{-1}+\left[z^{-1}\sum_{i=1}^n\frac{G^{(i)}(\mathbf{y}_i\mathbf{y}_i^{*}-n^{-1}\xi^2_i\Sigma)}{1+\mathbf{y}_i^{*}G^{(i)}\mathbf{y}_i}(I+m_2(z)\Sigma)^{-1} \right]\\
    &+\left[z^{-1}\frac{1}{n}\sum_{i=1}^n\frac{(G^{(i)}-G)\xi^2_i\Sigma}{1+\mathbf{y}_i^{*}G^{(i)}\mathbf{y}_i}(I+m_2(z)\Sigma)^{-1} \right]\\
    &:=-z^{-1}(I+m_2(z)\Sigma)^{-1}+R_1+R_2.
\end{split}
\end{equation}
In what follows, we control the two error terms $R_1,R_2.$ For $R_1$, we notice that 
\begin{equation}\label{eq: decomp R_1}
    \begin{split}
    &\frac{z}{n}{\rm tr}(R_1\Sigma)=\frac{1}{n}\sum_i{\rm tr}\left(\frac{G^{(i)}(\mathbf{y}_i\mathbf{y}_i^{*}-n^{-1}\xi^2_i\Sigma)}{1+\mathbf{y}_i^{*} G^{(i)}\mathbf{y}_i}(I+m_2^{(i)}\Sigma)^{-1}\Sigma\right)\\
    &+\frac{1}{n}\sum_i{\rm tr}\left(\frac{ G^{(i)}(\mathbf{y}_i\mathbf{y}_i^{*}-n^{-1}\xi^2_i\Sigma)}{1+\mathbf{y}_i^{*} G^{(i)}\mathbf{y}_i}(I+m_2\Sigma)^{-1}(m_2^{(i)}-m_2)\Sigma(I+m^{(i)}_2\Sigma)^{-1}\Sigma\right)\\
    &:=\mathtt R_{11}+\mathtt R_{12}.
\end{split}
\end{equation}
Since $\|\mathcal G^{(i)}\| \le\eta^{-1}$, using (\ref{eq_Econtrol}), (\ref{def1}), with high probability, we have that for some constant $C>0,$ 
\begin{equation}\label{eq_eqcontrolcontrolcontrolcontrol}
|m_2^{(i)}(z)|\le\frac{1}{n}\sum_{j\neq i}\xi^2_j| \mathcal G^{(i)}_{jj}|\le C \frac{\log^2 n}{n^{1/\alpha}}. 
\end{equation}
Moreover, according to (\ref{eq: bound for Z}), with high probability, when $n$ is sufficiently large, we have that for some constant $C>0$
\begin{equation}\label{eq_controlower}
|1+\mathbf{y}_i^{*} G^{(i)}\mathbf{y}_i| \asymp  |1+\xi_i^2 n^{-1} {\rm tr} G^{(i)} \Sigma | \geq 1- C \mathtt{C} ^{-1}>0,
\end{equation}
whenever $\mathtt{C}$ is chosen sufficiently large. Consequently, for all $i$, we have that 
\begin{equation}\label{eq_controncontrolcontrol}
\begin{split}
&   {\rm tr}\left(\frac{ G^{(i)}(\mathbf{y}_i\mathbf{y}_i^{*}-n^{-1}\xi^2_i\Sigma)}{1+\mathbf{y}_i^{*} G^{(i)}\mathbf{y}_i}(I+m_2^{(i)}\Sigma)^{-1}\Sigma\right)  \asymp {\rm tr} \left(\xi^2_i G^{(i)}(\mathbf{u}_i\mathbf{u}_i^{*}-n^{-1}I) (I+m_2^{(i)}\Sigma)^{-1} \Sigma^2 \right)\\
   &= \xi_i^2 \left( \mathbf{u}_i^* G^{(i)} (I+m_2^{(i)}\Sigma)^{-1} \Sigma^2 \mathbf{u}_i -n^{-1} {\rm tr} \left( G^{(i)} (I+m_2^{(i)}\Sigma)^{-1} \Sigma^2 \right) \right)\\
 & \prec \xi_i^2 \frac{1}{\eta\sqrt{n}},  
%   &\le C_1  \xi^2_i\eta^{-1}\|(\mathbf{u}_i\mathbf{u}_i^{*}-n^{-1}I)\|_2\prec \xi^2_i\eta^{-1}n^{-1/2},
\end{split}
\end{equation}
where in the third step we used (1) of Lemma \ref{lem:large deviation}.  Together with (\ref{def1}) and (\ref{eq_Econtrol}), we find that 
\begin{equation*}
\mathtt{R}_{11} \prec n^{-1/2-1/\alpha}. 
\end{equation*}
For $\mathtt R_{12}$, using the definition in (\ref{eq_m1m2}) and the identity in Lemma \ref{lem: Resolvent} and the definition of $\mathcal{G}^{(i)}$ (see (\ref{eq_trivialcontrol}) below), we see that  
\begin{equation}\label{eq_decompositionleavoneout}
m_2(z)-m^{(i)}_2(z)=\frac{1}{n}\sum_{j=1}^n \xi^2_j(\mathcal G_{jj}-\mathcal G^{(i)}_{jj})=\frac{1}{n}\sum_{j\neq i}\xi^2_j\frac{ \mathcal G_{ji} \mathcal G_{ij}}{\mathcal G_{ii}}+\frac{\xi_i^2(\mathcal{G}_{ii}-|z|^{-1})}{n}.
\end{equation}
In addition, using Lemma \ref{lem: Resolvent} and a discussion similar to (\ref{eq_controlower}), we conclude that  
\begin{equation*}
  \frac{1}{\mathcal G_{ii}(z)}=-z-z\mathbf{y}_i^{*} G^{(i)}\mathbf{y}_i\prec |z|.
\end{equation*}
Moreover, by Lemmas \ref{lem: Resolvent} and \ref{lem:large deviation}, we have that 
\begin{equation*}  
  \mathcal G_{ij}(z)=z\mathcal G_{ii}(z)\mathcal G^{(i)}_{jj}(z)\mathbf{y}_i^{*} G^{(ij)}\mathbf{y}_j\prec |z|\eta^{-2} |\xi_i\xi_j| n^{-1}\|\mathcal{G}^{(ij)}\|_F\prec n^{-1/2}|z|\eta^{-3}|\xi_i\xi_j| ,\quad i\neq j.  
\end{equation*}
Combining the above bounds with (\ref{def1}), we see that
\begin{equation*}
 m_2(z)-m^{(i)}_2(z) \prec n^{-1-1/\alpha}.
\end{equation*}
%\[
%\begin{split}
%   m_2(z)-m^{(i)}_2(z)&\le\frac{C\log n}{n}\sum_{j\neq i}\xi^2_j\frac{|z|^3\xi^2_i\xi^2_j}{p\eta^6}\\
%   &\le \frac{C\log^4 n}{np}\sum_{j\neq i}\frac{\xi_j^2\xi_i^2}{\eta^2}\\
%   &\le\frac{C\log^5n}{p}\frac{\log n}{n^{2/\alpha}}\\
%   &\prec n^{-1-2/\alpha}.
%\end{split}
%\]
Together with (\ref{eq_controlower}) and (\ref{eq_controncontrolcontrol}), we arrive at 
\begin{equation*}
\mathtt R_{12} \prec  \frac{1}{\eta^2 n^{3/2}}.  
\end{equation*}
Using the above bounds, we see that 
\[
\frac{z}{n}{\rm tr}(R_1\Sigma)\prec n^{-1/2-1/\alpha}.
\]

For $R_2$, applying the Sherman–Morrison formula to $((G^{(i)})^{-1}+\mathbf{y}_i \mathbf{y}_i^*)^{-1},$ we obtain that 
\begin{gather}\label{eq_simimimimimi}
\begin{split}
    \frac{1}{n}\left|{\rm tr}\left(\frac{(G^{(i)}-G)\Sigma(I+m_2\Sigma)^{-1}\Sigma}{1+\mathbf{y}_i^{*} G^{(i)}\mathbf{y}_i}\right) \right|&=\frac{1}{n}\left|\frac{\mathbf{y}_i^{*} G^{(i)}\Sigma(I+m_2\Sigma)^{-1}\Sigma G \mathbf{y}_i}{1+\mathbf{y}_i^{*}G^{(i)}\mathbf{y}_i}\right|\\
 %   &\le \frac{\sigma_1^2\xi^2_i}{n\log n}\|\mathcal{G}\|^2_2\\
    &\prec \frac{\xi_i^2}{n \eta^2},
\end{split}
\end{gather}
where in the second step we used Lemma \ref{lem:large deviation} and (\ref{eq_controlower}) and a discussion similar to (\ref{eq_eqcontrolcontrolcontrolcontrol}). Together with the definition of $R_2$ in (\ref{eq: decomp of mathcal_G}), by (\ref{def1}), we find that
\[
\begin{split}
  \frac{z}{n} \left|{\rm tr}(R_2\Sigma) \right|
  %&=|\frac{1}{n^2}\sum_i\xi^2_i {\rm tr}\Big(\frac{(\mathcal{G}^{(i)}-\mathcal{G})\Sigma(I+m_2\Sigma)^{-1}\Sigma}{(1+\mathbf{y}_i^{*}\mathcal{G}^{(i)}\mathbf{y}_i)}\Big)|\\
  \le\frac{1}{n^2}\sum_i\frac{\xi^2_i}{\eta^2 }\prec n^{-1-2/\alpha}.
\end{split}
\]
As a result, in light of the definition $m_1$ in (\ref{eq_m1m2}),
we have
\begin{equation}\label{eq: m_1 by m_2}
\begin{split}
    m_1&=\frac{1}{n}{\rm tr}(G(z)\Sigma)=-z^{-1}\frac{1}{n}{\rm tr}((I+m_2(z)\Sigma)^{-1}\Sigma)+\rO_{\prec}(n^{-1/2-2/\alpha})\\
    &=-\frac{1}{n}\sum_{i=1}^p\frac{\sigma_i}{z(1+m_2\sigma_i)}+\rO_{\prec}(n^{-1/2-2/\alpha}).
\end{split}
\end{equation}

We first control $m_2(z)-m_{2n}(z).$ Recall $m_{2n}(z)$ in (\ref{eq_systemequationsm1m2}). Combing \eqref{eq: m_2 by m_1} and \eqref{eq: m_1 by m_2},  we have that 
\begin{align}\label{eq_closenessm2m2n}
%\begin{split}
    m_2(z)-m_{2n}(z)&=\frac{1}{n}\sum_{i=1}^n \left(\frac{\xi^2_i}{-z(1+\xi^2_im_1)}+\frac{\xi^2_i}{z(1+\xi^2_im_{1n})} \right)+\rO_{\prec}(n^{-1/2-1/\alpha})\\
    &=\frac{1}{n}\sum_{i=1}^n\frac{\xi^4_i(m_1-m_{1n})}{z(1+\xi^2_im_1)(1+\xi^2_im_{1n})}+\rO_{\prec}(n^{-1/2-1/\alpha})\nonumber \\
    &=\left(\frac{1}{n}\sum_{i=1}^n\frac{\xi^4_i}{z(1+\xi^2_im_1)(1+\xi^2_im_{1n})}\right)\left(\frac{1}{n}\sum_{i=1}^n\frac{\sigma_i^2(m_2-m_{2n})}{z(1+\sigma_im_2)(1+\sigma_im_{2n})}\right)+\rO_{\prec}(n^{-1/2-1/\alpha}). \nonumber
%\end{split}
\end{align}
By a discussion similar to (\ref{eq_eqcontrolcontrolcontrolcontrol}) and (\ref{eq_controlower}) with high probability, when $n$ is sufficiently large, we have that  
%
%Since $|z|\sim \xi^2_{s_1}$ with $|1+\xi^2_im_1|>0$,$|1+\xi^2_im_{1n}|>0$,$|1+\sigma_im_1|>0$,$|1+\sigma_im_{1n}|>0$, thereafter we have
\[
\begin{split}
  |m_2-m_{2n}|&=\rO\left(\frac{1}{n}\sum_{i=1}^n\frac{\xi^4_i}{z^2}|m_2-m_{2n}| \right)+\rO_{\prec}(n^{-1/2-1/\alpha})\\
  &=\rO(n^{-2/\alpha}|m_2-m_{2n}|)+\rO_{\prec}(n^{-1/2-1/\alpha}),
\end{split}
\]
where in the second step we used (\ref{def1}). Then we can conclude that $m_2-m_{2n}\prec n^{-1/2-1/\alpha}$. By a similar procedure, we also have $m_1-m_{1n}\prec n^{-1/2-2/\alpha}$.

Armed with the above two results, we proceed to finish the rest of the proof. Recall $m_Q$ in (\ref{eq_mq}). Using (\ref{eq: decomp of mathcal_G}) and a discussion similar to \eqref{eq: m_1 by m_2}, one can see that 
\begin{equation}\label{eq_verbalttt}
\begin{split}
    m_{\mathcal{Q}}&=\frac{1}{p}{\rm tr}(G(z))=-\frac{1}{p}\sum_{i=1}^p\frac{1}{z(1+m_2\sigma_i)}+\rO_{\prec}(n^{-1/2-2/\alpha})\\
    &=-\frac{1}{p}\sum_{i=1}^p\frac{1}{z(1+m_{2n}\sigma_i)}+\frac{1}{p}\sum_{i=1}^p\frac{(m_2-m_{2n})\sigma_i}{z(1+m_2\sigma_i)(1+m_{2n}\sigma_i)}+\rO_{\prec}(n^{-1/2-2/\alpha})\\
    &=m_{n}+\rO_{\prec}(n^{-1/2-2/\alpha}),
\end{split}
\end{equation}
where in the last step we recall $m_n(z)$ in (\ref{eq_systemequationsm1m2}). Finally, for the control of the matrix $\mathcal G,$ for the diagonal entries, by Lemma \ref{lem: Resolvent} and a discussion similar to  \eqref{eq: bound for Z} and \ref{eq_m1approximate},  we have
\[
\begin{split}
    \mathcal G_{ii}&=-\frac{1}{z(1+\mathbf{y}^{*}_i G^{(i)}\mathbf{y}_i)}=-\frac{1}{z(1+\xi^2_in^{-1}{\rm tr} G^{(i)}\Sigma+\rO_{\prec}(\frac{\xi^2_i}{\sqrt{n}\eta}))}\\
    &=-\frac{1}{z(1+\xi^2_im_1+\rO_{\prec}(\frac{\xi^2_i}{\sqrt{n}\eta}))}=-\frac{1}{z(1+\xi^2_im_{1n})}+\rO_{\prec}(n^{-1/2-1/\alpha}).
\end{split}
\]
For off-diagonal entries, together with Lemmas \ref{lem: Resolvent} and \ref{lem:large deviation}, we have
\[
\begin{split}
    |\mathcal G_{ij}|&\le|z||\mathcal G_{ii}||\mathcal G^{(i)}_{ii}||\mathbf{y}^{*}_i G^{(ij)}\mathbf{y}_j| \prec n^{-1/2-2/\alpha}. 
%    \\
%    &\le \frac{\log^4n}{p}\|\mathcal{G}^{(ij)}\|_F\prec n^{-1/2-2/\alpha}.
\end{split}
\]
This completes the proof when (\ref{ass3.1}) holds. For the case (\ref{ass3.2}), the main difference is to use the estimates of (\ref{def3}) instead of (\ref{def1}) whenever it is needed, for example, (\ref{eq: decomp m_2}). We omit further details.

For the elliptical model as in Case (1) of Assumption \ref{assum_model}, the discussions are similar by using the definitions in (\ref{eq_systemequationsm1m2elliptical}) and (\ref{eq_m1m2}) with (2) of Lemma \ref{lem:large deviation}. The main difference is that (\ref{eq: decomp of mathcal_G}) should be replaced by  
\begin{equation*}
\begin{split}
    G&=-z^{-1}(I+m_2(z)\Sigma)^{-1}+\left[z^{-1}\sum_{i=1}^n\frac{G^{(i)}(\mathbf{y}_i\mathbf{y}_i^{*}-p^{-1}\xi^2_i\Sigma)}{1+\mathbf{y}_i^{*}G^{(i)}\mathbf{y}_i}(I+m_2(z)\Sigma)^{-1} \right]\\
    &+\left[z^{-1}\frac{1}{p}\sum_{i=1}^n\frac{(G^{(i)}-G)\xi^2_i\Sigma}{1+\mathbf{y}_i^{*}G^{(i)}\mathbf{y}_i}(I+m_2(z)\Sigma)^{-1} \right].
\end{split}
\end{equation*}
The rest can then be proved verbatim with some minor changes. We omit further details due to similarity. This completes our proof.
%{\color{red}[The results above are uniformly hold for $z\in\mathbf{D}$, since $m_{\mathcal{Q}}$,$m_1$ and $m_2$ are $\eta^{-2}$-Lipschitz on $z$, we can construct a lattice from $z^{\prime}\in\mathbf{D}$ with $|z-z^{\prime}|\le n^{-3}$. Then the result follows a standard lattice argument.]}

\end{proof}
%{\color{blue}
%\begin{remark}
%When $\xi^2_i$'s follow the exponential decays(see Assumption \ref{ass3.2}), we have the typical rates of $\xi^2_{(1)}$ as $C\log n$ for $\beta=1$ (see Definition \ref{def3}). Then we may set the rate of $\eta$ as $\eta=\mathbf{c}\log n$. It follows that the error bounds from large deviation will be changed into $\xi^2_i/n^{-1/2}\log n$. Thereafter, we can obtain that 
%\begin{gather}
%    m_2=\frac{1}{n}\sum_{i=1}^n\frac{\xi^2_i}{-z(1+\xi^2_i m_1)}+O_{\prec}(n^{-1/2}\log^{-1}n),\\
%    m_1=\frac{1}{n}\sum_{i=1}^p\frac{\sigma_i}{-z(1+\sigma_i m_2)}+O_{\prec}(n^{-1/2}\log^{-1}n).
%\end{gather}
%The average local laws will follow from above estimations.
%
%In case of general $\beta>0$, one can gradually obtain the rate of $\xi^2_{(1)}$ as $C\log^{1/\beta}n$ and changes the above error rates as $O_{\prec}(n^{-1/2}\log^{-1/\beta}n)$.
%\end{remark}
%}

\quad  The second component is to prove Proposition \ref{eq_propooursidebulk} under a priori control of the resolvent which is summarized in the following lemma.

%Now by the results in Lemma \ref{lem: average local law for large eta}, we can get rid of the assistance from $\eta$.
\begin{lemma} \label{lem: self improvement}
Proposition \ref{eq_propooursidebulk} holds if (\ref{eq_entrywiselocallaw}) holds uniformly for $z \in \widetilde{\mathbf{D}}_{u}.$ 
%Suppose Assumptions \ref{ass1},\ref{ass2},\ref{ass3.1},\ref{ass3.2} and events $\Omega_n$ hold. For any $z\in\tilde{\mathbf{D}}$, if 
%\begin{gather}\label{eq_assumption}
%    \max_{i,j}(G+z^{-1}(I+m_{1n}(z)D^2)^{-1})_{ij}\prec n^{-1/2-2/\alpha}.
%\end{gather}
%We have
%\begin{gather}
%    |m_1-m_{1n}|+|m_{\mathcal{Q}}-m_n|\prec n^{-1/2-4/\alpha},\quad |m_2-m_{2n}|\prec n^{-1/2-2/\alpha}.
%\end{gather}
%\begin{itemize}
   % \item [1.] For $\alpha\in(0,2]$,
    %\[
 %\Psi_1\prec n^{-1/2-2/\alpha},\quad \Psi_2\prec n^{-3/2-2/\alpha},\quad \Psi_3\prec n^{-3/2}.
%\]
%\item [1.] For $\alpha\in(2,+\infty)$,
%\begin{gather*}
 %\Psi_1\prec n^{-1/2-2/\alpha},\quad \Psi_2\prec n^{-1/2-4/\alpha},\quad \Psi_3\prec n^{-1/2-2/\alpha}.   
%\end{gather*}
%\item [2.] For $\alpha=+\infty$,
%\begin{gather*}
% \Psi_1\prec n^{-1/2}\log^{-2K}n,\quad \Psi_2\prec n^{-1/2}\log^{-2K}n,\quad \Psi_3\prec n^{-1/2}\log^{-K}n. 
%\end{gather*}
%\end{itemize}
\end{lemma}
\begin{proof}
As  before, due to similarity, we focus on the study of the separable covariance i.i.d. model as in Case (2) of Assumption \ref{assum_model} when $\xi^2$ satisfies (\ref{ass3.1}).

Note that according to the priori control (\ref{eq_entrywiselocallaw}), we have that for $1 \leq i \neq j \leq n$ 
\begin{gather}\label{eq_expansionone}
\mathcal G_{ii}=\frac{1}{z(1+\xi^2_im_{1n}(z)}+\rO_{\prec}(n^{-1/2-1/\alpha}), \quad \mathcal G_{ij} =\rO_{\prec}(n^{-1/2-1/\alpha}).
\end{gather}
For the diagonal entries, when $i=1,$ using (\ref{eq: def of mu_1}), we observe that
\begin{gather}\label{eq_g11control}
    \begin{split}
        \mathcal G_{11}&=-\frac{1}{z(1+\xi^2_{1}m_{1n}(z))}+\rO_{\prec}(n^{-1/2-1/\alpha})\\
        &=-\frac{1}{z(1+\xi^2_{1}m_{1n}(\mu_1))}+\frac{z\xi_{1}^2(m_{1n}(z)-m_{1n}(\mu_1))}{(z(1+\xi^2_{1}m_{1n}(\mu_1)))(z(1+\xi^2_{1}m_{1n}(z)))}+\rO_{\prec}(n^{-1/2-1/\alpha})\\
        &=\frac{1}{zd_1m_{1n}(\mu_1)}-\frac{z\xi_{1}^2(m_{1n}(z)-m_{1n}(\mu_1))}{zd_1m_{1n}(\mu_1)}\left(\mathcal G_{11}+\rO_{\prec}(n^{-1/2-2/\alpha})\right)+\rO_{\prec}(n^{-1/2-1/\alpha})\\
        &=\frac{1}{zd_1m_{1n}(\mu_1)}-\frac{z\xi_{1}^2}{zd_1m_{1n}(\mu_1)}(\mathcal G_{11}+\rO_{\prec}(n^{-1/2-1/\alpha}))\times \rO_{\prec}(n^{-1/\alpha})+\rO_{\prec}(n^{-1/2-1/\alpha}),
    \end{split}
\end{gather}
where in the fourth step we used Lemma \ref{lem: basic bounds}. By Lemma \ref{lem: basic bounds}, (\ref{eq_mu1part}) and (\ref{def1}),  this yields that  for some constant $C>0$
\begin{gather}\label{eq_g11}
\begin{split}
    |\mathcal G_{11}|&=\frac{1}{|zd_1m_{1n}(\mu_1)|}+\rO_{\prec}(n^{-1/2-1/\alpha})\\
    &=\frac{C}{d_1}+\rO_{\prec}(n^{-1/2-1/\alpha}).
\end{split}
\end{gather}

Similarly, when $2 \leq i \leq n,$ by (\ref{def1}) and the definition of $d_1,$ using Lemma \ref{lem: basic bounds}, we see that
\begin{equation}\label{eq_expansiontwo}
\mathcal{G}_{ii}=\rO_{\prec}(n^{-1/\alpha}). 
\end{equation}

We also provide some basic controls for the matrix $\mathcal G^{(i)}$ for all $1 \leq i \leq n.$ By definition and an elementary calculation, it is not hard to see that
\begin{equation}\label{eq_trivialcontrol}
\mathcal G_{ii}^{(i)}=-z^{-1}; \ \mathcal{G}_{i k}^{(i)}=0, \  1 \leq k \neq i \leq n. 
\end{equation}  
Moreover, using (\ref{eq_expansionone}), (\ref{eq_expansiontwo}) and the third identity of Lemma \ref{lem: Resolvent}, we find that for $1 \leq i \leq n,$ 
\begin{equation}\label{eq_nontrivialcontrol}
\mathcal G_{kk}^{(i)}=\rO_{\prec}(n^{-1/\alpha}),  \ k \neq i; \ \mathcal G_{kl}^{(i)}=\rO_{\prec}(n^{-1/2-1/\alpha}), \ k,l \neq i;
\end{equation}

With the above preparation, we now proceed to the control of $Z_i$ in (\ref{eq: decomp m_2}), unlike in (\ref{eq: bound for Z}), since $Q$ and $\mathcal{Q}$ have the same non-zero eigenvalues, we control it as follows using the above bounds 
\begin{equation}\label{eq_standarddiscussion}
\begin{split}
     Z_i&\prec\frac{\xi^2_i}{n}\| G^{(i)}\Sigma\|_F\prec\frac{\xi^2_i}{n}\|G^{(i)}\|_F=\frac{\xi^2_i}{n}(\operatorname{tr}((G^{(i)})^2))^{1/2}\le\frac{\xi^2_i}{n}\operatorname{tr}((\mathcal G^{(i)})^2)^{1/2}+\frac{\xi_i^2}{n}\frac{\sqrt{|n-p|}}{|z|}\\
    & \asymp \frac{\xi^2_i}{n}\|\mathcal G^{(i)}\|_F+\frac{\xi^2_i}{n}\frac{n^{1/2}}{n^{1/\alpha}}\\
    &=\frac{\xi^2_i}{n} \left((\mathcal G_{ii}^{(i)})^2+\sum_{j\neq i}(\mathcal G^{(i)}_{jj})^2+\sum_{j\neq k\neq i}(\mathcal G^{(i)}_{jk})^2 \right)^{1/2}+\frac{\xi^2_i}{n}\frac{n^{1/2}}{n^{1/\alpha}}  \\
& \prec        \frac{\xi_1^2}{n}\left( |z|^{-2}+n n^{-2/\alpha}+n^2 n^{-1-2/\alpha}\right)^{1/2}+\frac{\xi^2_i}{n}\frac{n^{1/2}}{n^{1/\alpha}}\\
      &\prec\frac{\xi_i^2}{n^{1/2+1/\alpha}},
\end{split}
\end{equation}
where in the third and fourth steps we used (\ref{eq_trivialcontrol}) and (\ref{eq_nontrivialcontrol}). 

Besides, unlike in (\ref{eq_m1approximate}), by a discussion similar to (\ref{eq_standarddiscussion}),  we now have from Lemma \ref{lem:large deviation} that 
\begin{equation}\label{eq_cccccc11111}
\begin{split}
   T_i:=\frac{1}{n}{\rm tr} G^{(i)}\Sigma-m_1(z)&=\frac{1}{n}\mathbf{y}_i^{*}G \Sigma G^{(i)}\mathbf{y}_i=\frac{1}{n}\frac{\mathbf{y}^{*}_i G^{(i)}\Sigma G^{(i)}\mathbf{y}_i}{1+\mathbf{y}^{*}_i G^{(i)}\mathbf{y}_i}\\
    &\prec\frac{\xi^2_i}{n}\frac{n^{-1}\| G^{(i)}\|^2_F}{|1+\mathbf{y}^{*}_i G^{(i)}\mathbf{y}_i|}\\
    &\prec\frac{\xi^2_i}{n^2}|z||\mathcal G_{ii}|\left(\|\mathcal G^{(i)}\|^2_F+\frac{n}{n^{2/\alpha}} \right), \\
    & \prec \frac{\xi_i^2}{n^{1+2/\alpha}}. 
\end{split}
\end{equation}
where in the second step we used the relation (\ref{eq_usefulformula}), in the third step we used Lemma \ref{lem: Resolvent} and in the last two steps we used a discussion similar to (\ref{eq_standarddiscussion}). 

With the above control, we now use an idea similar to the proof of  Lemma \ref{lem: average local law for large eta} to conclude the proof. The key ingredient is to explore the relation of $m_1$ and $m_2.$ We start with $m_2.$ Using the above notations and Lemma \ref{lem: Resolvent}, we find that  
\begin{equation}\label{eq_eeeone}
\frac{1}{-z(1+\xi_i^2 m_1(z))}=\frac{1}{\mathcal{G}_{ii}^{-1}+z(Z_i+T_i)}.
\end{equation}
Consequently, by (\ref{eq_g11}), (\ref{eq_expansiontwo}), (\ref{eq_standarddiscussion}) and (\ref{eq_cccccc11111}), we see that 
\begin{equation}\label{eq_eeetwo}
\frac{1}{-z(1+\xi_i^2 m_1(z))} \prec n^{-1/\alpha}. 
\end{equation}
Then using the decomposition as in (\ref{eq: decomp m_2}), we have that 
\begin{gather}\label{eq_modeltemp}
    \begin{split}
        m_2&=\frac{1}{n}\frac{\xi^2_1}{-z(1+\xi^2_1n^{-1}\operatorname{tr} G^{(1)}\Sigma+Z_1)}+\frac{1}{n}\sum_{i=2}^n\frac{\xi^2_i}{-z(1+\xi^2_in^{-1}\operatorname{tr} G^{(i)}\Sigma+Z_i)}\\
        &=\frac{1}{n}\frac{\xi^2_1}{-z(1+\xi^2_1m_1(z)+\xi_1^2 n^{-1-2/\alpha}+Z_1)}+\frac{1}{n}\sum_{i=2}^n\frac{\xi^2_i}{-z(1+\xi^2_im_1(z)+\xi_i^2 n^{-1-2/\alpha}+Z_i)}\\
        &=\frac{1}{n}\sum_{i=1}^n\frac{\xi^2_i}{-z(1+\xi^2_im_1(z))}+\rO_{\prec}\left(n^{-3/2}+\frac{1}{n}\sum_{i=2}^n\frac{\xi^4_i}{|z|n^{1/2+1/\alpha}}\right)\\
 %       &=\frac{1}{n}\sum_{i=1}^n\frac{\xi^2_i}{-z(1+\xi^2_im_1(z))}+O_{\prec}(\frac{1}{np}n^{-2/\alpha+\epsilon}+\frac{1}{n^{3/2}}+n^{-1/2-2/\alpha})\\
        &=\frac{1}{n}\sum_{i=1}^n\frac{\xi^2_i}{-z(1+\xi^2_im_1(z))}+\rO_{\prec}(n^{-1/2-1/\alpha}),
    \end{split}
\end{gather}
where in the second step we used (\ref{eq_cccccc11111}), in the third step we used a discussion similar to (\ref{eq_g11control}) and in the last step we used (\ref{def1}). Then we study $m_1$ using the arguments as between (\ref{eq: decomp of mathcal_G}) and (\ref{eq: m_1 by m_2}). We provide the key ingredients as follows. First, for $\mathtt{R}_{11}$ in (\ref{eq: decomp R_1}), note that by definition of $m_2^{(i)}$ and (\ref{eq_trivialcontrol}), we have that   
\begin{equation*}
\begin{split}
   \left |m^{(i)}_2 \right |&\le\frac{1}{n} \left(\sum_{j\neq i}\xi^2_j| \mathcal G^{(i)}_{jj}|+|z|^{-1} \right)\\
    &\le\frac{1}{n} \left[\sum_{j\neq 1}\xi^2_j\left( | \mathcal G_{jj}|+\frac{|\mathcal G_{j1}||\mathcal G_{1j}|}{|\mathcal G_{11}|} \right)+|z|^{-1} \right]\\
    &\prec n^{-1/\alpha},
    %\frac{1}{n}\sum_{j\neq1}\frac{\xi^2_j}{n^{2/\alpha}}\rightarrow0,
\end{split}
\end{equation*}
where in the second step we used Lemma \ref{lem: Resolvent} and in the last step we used (\ref{eq_expansionone}), (\ref{eq_expansiontwo}) and (\ref{def1}). Moreover, by Lemma \ref{lem: Resolvent} and (\ref{eq_expansiontwo}), we find that $(1+\mathbf{y}_i^* G^{(i)} \mathbf{y}_i)^{-1} \prec 1.$ Therefore, we conclude that for all $1 \leq i \leq n,$
\begin{equation*}
\begin{split}
  &   {\rm tr}\left(\frac{ G^{(i)}(\mathbf{y}_i\mathbf{y}_i^{*}-n^{-1}\xi^2_i\Sigma)}{1+\mathbf{y}_i^{*} G^{(i)}\mathbf{y}_i}(I+m_2^{(i)}\Sigma)^{-1}\Sigma\right)  \asymp {\rm tr} \left(\xi^2_i G^{(i)}(\mathbf{u}_i\mathbf{u}_i^{*}-n^{-1}I) (I+m_2^{(i)}\Sigma)^{-1} \Sigma^2 \right)\\
   &= \xi_i^2 \left( \mathbf{u}_i^* G^{(i)} (I+m_2^{(i)}\Sigma)^{-1} \Sigma^2 \mathbf{u}_i -n^{-1} {\rm tr} \left( G^{(i)} (I+m_2^{(i)}\Sigma)^{-1} \Sigma^2 \right) \right)\\
 & \prec \frac{\xi_i^2}{n}\|G^{(i)}\|_F\prec\frac{\xi_i^2}{n^{1/2+1/\alpha}},
\end{split}
\end{equation*}
where in the last step we used a discussion similar to (\ref{eq_standarddiscussion}). Together with (\ref{def1}), we can conclude that $\mathtt R_{11} \prec n^{-1/2-1/\alpha}.$ For $\mathtt R_{12},$ using (\ref{eq_decompositionleavoneout}), (\ref{eq_expansionone}) and (\ref{eq_expansiontwo}) 
\begin{equation*}
m_2-m^{(i)}_2 \prec\frac{1}{n}\sum_{j\neq i}\frac{\xi^2_in^{-1-2/\alpha}}{n^{-1/\alpha}} + \frac{1}{n}\prec n^{-1}.
\end{equation*}
Then by an argument similar to (\ref{eq_simimimimimi}), we can conclude that $\mathtt R_{12} \prec n^{-1-1/\alpha}.$ Similarly, for $R_2,$ we have that
\[
\begin{split}
    &\frac{1}{n} \left|{\rm tr}\left(\frac{(G^{(i)}-G)\Sigma(I+m_2\Sigma)^{-1}\Sigma}{1+\mathbf{y}_i^{*} G^{(i)}\mathbf{y}_i}\right) \right|=\frac{1}{n}\left|\frac{\mathbf{y}_i^{*} G^{(i)}\Sigma(I+m_2\Sigma)^{-1}\Sigma G\mathbf{y}_i}{1+\mathbf{y}_i^{*} G^{(i)}\mathbf{y}_i}\right|\\
    &\asymp\frac{1}{n}\left|\mathbf{y}_i^{*} G^{(i)}\Sigma(I+m_2\Sigma)^{-1}\Sigma G^{(i)}\mathbf{y}_i \right | \prec\frac{\xi_i^2}{n^2}\| G^{(i)}\|_F \| G \|_F \prec\frac{\xi_i^2}{n^{1+2/\alpha}}.
\end{split}
\]
Consequently, we have that 
\begin{equation*}
\frac{z}{n} \left|{\rm tr}(R_2\Sigma) \right|\prec\frac{1}{n}\sum_i\frac{\xi_i^2}{n^{1+1/\alpha}}\prec n^{-1-1/\alpha}.
\end{equation*}

Combining all the above arguments, we find that \eqref{eq: m_1 by m_2} still holds true. Armed with all the above controls, using an argument similar to the discussions between (\ref{eq_closenessm2m2n}) and (\ref{eq_verbalttt}), we can conclude the proof. 
\end{proof}

%{\color{blue}
%We summarise the average local law for both polynomial case and exponential case on $z\in\tilde{\mathbf{D}}$ below.
%\begin{lemma}[Average local law on $\tilde{\mathbf{D}}$]\label{lem: average local law}
%Suppose Assumptions \ref{ass1},\ref{ass2},\ref{ass3.1},\ref{ass3.2} and events $\Omega_n$ hold. Then it holds uniformly for $z\in\tilde{\mathbf{D}}$ that
%\begin{itemize}
%   % \item [1.] For $\alpha\in(0,2]$,
%    %\begin{gather*}
%     %   \Psi_1\prec n^{-1/2-2/\alpha},\quad\Psi_2\prec n^{-3/2-2/\alpha},\quad \Psi_3\prec n^{-3/2}.
%    %\end{gather*}
%    \item [1.] For $\alpha\in(2,+\infty)$,
%    \begin{gather}
%        \max_{ij}(G+z^{-1}(I+m_{1n}(z)D^2))\prec n^{-1/2-2/\alpha},\\
%        |m_1-m_{1n}|+|m_{\mathcal{Q}}-m_n|\prec n^{-1/2-4/\alpha},\quad |m_2-m_{2n}|\prec n^{-1/2-2/\alpha}.
%    \end{gather}
%    \item [2.] For $\alpha=+\infty$,
%    \begin{gather}
%        \max_{ij}(G+z^{-1}(I+m_{1n}(z)D^2))\prec n^{-1/2}\log^{-K}n,\\
%        |m_1-m_{1n}|+|m_{\mathcal{Q}}-m_n|\prec n^{-1/2}\log^{-2K}n,\quad |m_2-m_{2n}|\prec n^{-1/2}\log^{-K}n.
%    \end{gather}
%\end{itemize}
%\end{lemma}
%}

\quad Combining the above two lemmas, we now proceed to the proof of Proposition \ref{eq_propooursidebulk}.  We will use a continuity argument as in \cite[Lemma A.12]{Ding&Yang2018} or \cite[Section 4.1]{Alex2014}. In fact, our discussion is easier since the real part in the spectral domain $\widetilde{\mathbf{D}}_{u}$ is divergent so that the rate is independent of $\eta$. Due to similarity, we focus on explaining the key ingredients. 

\begin{proof}[\bf Proof of Proposition \ref{eq_propooursidebulk}]
As before, due to similarity, we focus on the study of the separable covariance i.i.d. model as in Case (2) of Assumption \ref{assum_model} when $\xi^2$ satisfies (\ref{ass3.1}).

 For each $z=E+\ri\eta\in\widetilde{\mathbf{D}}_{u}$, we fix the real part and construct a sequence $\{\eta_j\}$ by setting $\eta_j=\mathtt C \mu_1-j n^{-3}$. Then it is clear that $\eta$ falls in an interval $[\eta_{j-1},\eta_j]$ for some $0\le j\le Cn^{1/\alpha+3}, C>0$ is some constant.

In Lemma \ref{lem: average local law for large eta}, we have proved that the results hold  for $\eta_0$. Now we assume (\ref{eq_entrywiselocallaw}) holds for some $\eta_j$. Then according to Lemma \ref{lem: self improvement}, we have that
\begin{equation*}
|m_1(z_j)-m_{1n}(z_j)|+|m_Q(z_j)-m_n(z_j)| \prec n^{-1/2-2/\alpha}, \ |m_2(z_j)-m_{2n}(z_j)| \prec n^{-1/2-1/\alpha}. 
\end{equation*}
%{\color{red} here}
%{\color{blue} Now assume $\max_{ij}(G+z^{-1}(I+m_{1n}(z)D^2))\prec n^{-1/2-2/\alpha}$ holds } 
For any $\eta^{\prime}$ lying in the interval $[\eta_{k-1},\eta_k]$, denote $z^{\prime}=E+\ri\eta^{\prime}$ and $z_j=E+\ri\eta_j$. According to the first resolvent identity, we have that 
 \begin{equation} \label{eq: Lip for mathcal G}
 \| \mathcal G(z')-\mathcal{G}(z_j) \| \leq n^{-3} \| \mathcal{G}(z') \| \| \mathcal{G}(z_j) \| \prec n^{-11/6-1/\alpha},
\end{equation}  
where in the second step we used the basic bound $\|\mathcal G(z^{\prime})\|\leq n^{2/3},$ (\ref{eq_expansionone}), (\ref{eq_expansiontwo}) and Gershgorin circle theorem.

On the one hand, according to the definitions in (\ref{eq_m1m2}), using the first resolvent identity, we have that  
\begin{equation*}
\begin{split}
    m_1(z^{\prime})-m_1(z_j)&=\frac{1}{n}{\rm tr}[(G(z^{\prime})-G(z_j))\Sigma]=\frac{1}{n^4}{\rm tr}(G(z^{\prime})G(z_j)\Sigma) \prec\frac{1}{n^4 \eta_j}\|G(z_j)\|_F \prec n^{-17/6-1/\alpha},
\end{split}
\end{equation*}
where in the last step we used a discussion similar to (\ref{eq_standarddiscussion}). Similarly, combining (\ref{eq: Lip for mathcal G}) and (\ref{def1}),  we have that  
\begin{equation*}
    m_2(z^{\prime})-m_2(z_j)=\frac{1}{n}\sum_{i=1}^n\xi^2_i(\mathcal G_{ii}(z^{\prime})-\mathcal G_{ii}(z_j))
\prec n^{-11/6-1/\alpha},
\end{equation*}
and by a discussion similar to (\ref{eq_standarddiscussion})
\begin{align*}
|m_Q(z')-m_Q(z_j)|& =\frac{1}{p} \left| \operatorname{tr} \left( G(z')-G(z_j) \right) \right|  \leq n^{-4} \|G(z') \|_F \| G(z_j) \|_F \\
& \prec n^{-4} (n^{1-2/\alpha})^{1/2} n^{1/2+2/3}=n^{-7/3-1/\alpha}. 
\end{align*}

% we have the Lipschitz properties
%{\color{blue}
%\begin{equation}
%    \max_{ik}(|G(z^{\prime})-G(z_j)|)\le n^{-3}\max_{ik}(|G(z^{\prime})||G(z_j)|)\prec n^{-3/2-2/\alpha},
%\end{equation}
% and the fact
%\[
%G(z^{\prime})=\sum_{k=0}^{N-1}G(z_j)(n^{-3}G(z_j))^k+G(z^{\prime})(n^{-3}G(z_j))^N,
%\]}
%for sufficient large $N$.
 On the other hand, using the definitions in (\ref{eq_systemequationsm1m2}), we decompose that 
\[
\begin{split}
    m_{1n}(z^{\prime})-m_{1n}(z_j)
%    &=\frac{1}{n}\sum_i\Big(\frac{\sigma_i}{-z^{\prime}(1+\sigma_im_{2n}(z^{\prime}))}-\frac{\sigma_i}{-z_j(1+\sigma_im_{2n}(z_j))}\Big)\\
    &=\frac{1}{n}\sum_i\left(\frac{\sigma_i}{-z^{\prime}(1+\sigma_im_{2n}(z^{\prime}))}-\frac{\sigma_i}{-z^{\prime}(1+\sigma_im_{2n}(z_j))}\right)\\
    &+\frac{1}{n}\sum_i\left(\frac{\sigma_i}{-z^{\prime}(1+\sigma_im_{2n}(z_j))}-\frac{\sigma_i}{-z_j(1+\sigma_im_{2n}(z_j))}\right)\\
    &:=\mathcal{M}_{11}+\mathcal{M}_{12}.
\end{split}
\]
For $\mathcal{M}_{11}$, according to Lemma \ref{lem: basic bounds} and (\ref{eq_systemequationsm1m2}), we readily obtain that
\[
\begin{split}
    \mathcal{M}_{11}&=\frac{1}{n}\sum_{i}\frac{\sigma^2_i}{-z^{\prime}(1+\sigma_im_{2n}(z^{\prime}))(1+\sigma_im_{2n}(z_j))}\big(m_{2n}(z_j)-m_{2n}(z^{\prime})\big)\\
%    &=\rO(|z|^{-1})\times\frac{1}{n}\sum_i\Big(\frac{\xi^2_i}{-z_j(1+\xi^2_im_{1n}(z_j))}-\frac{\xi^2_i}{-z^{\prime}(1+\xi^2_im_{1n}(z^{\prime}))}\Big)\\
    &=\rO(|z'|^{-1})\times\frac{1}{n}\sum_i\Big(\frac{\xi^2_i}{-z_j(1+\xi^2_im_{1n}(z_j))}-\frac{\xi^2_i}{-z_j(1+\xi^2_im_{1n}(z^{\prime}))}\Big)\\
    &+\rO(|z'|^{-1})\times\frac{1}{n}\sum_i\left(\frac{\xi^2_i}{-z_j(1+\xi^2_im_{1n}(z^{\prime}))}-\frac{\xi^2_i}{-z^{\prime}(1+\xi^2_im_{1n}(z^{\prime}))}\right)\\
    &=\rO(|z'|^{-1}) \times \left( \mathtt{M}_{11,1}+ \mathtt{M}_{11,2} \right). 
%&    \rO(|z'|^{-1})\left(\frac{1}{n}\frac{\xi^4_1}{z_jd^2_1|m_{1n}(z_j)||m_{1n}(z^{\prime})|}+\frac{1}{n}\sum_{i\ge2}\frac{\xi^4_i}{z_j(1+\xi^2_im_{1n}(z_j))(1+\xi^2_im_{1n}(z^{\prime})}\right)\times(m_{1n}\left(z^{\prime})-m_{1n}(z_j)\right)\\
%    &+\rO_{\prec}(\frac{z_j-z^{\prime}}{ z_j(z^{\prime})^2 })\\
%    &=\ro(1)\times(m_{1n}(z^{\prime})-m_{1n}(z_j))+\rO_{\prec}(n^{-3-6/\alpha}).
\end{split}
\]
%where in the third and fourth steps we used a discussion similar to (\ref{eq_g11}) and (\ref{eq_expansiontwo}).

For $\mathtt{M}_{11,1},$ by a discussion similar to (\ref{eq_g11}) and (\ref{eq_expansiontwo}), we find that 
\[
\begin{split}
    \mathtt{M}_{11,1}&=\frac{1}{n}\sum_{i=1}^n\frac{\xi_i^4}{-z_j (1+\xi^2_im_{1n}(z^{\prime}))(1+\xi^2_im_{1n}(z_j))}(m_{1n}(z')-m_{1n}(z_j))\\
    &\prec\left(\frac{\xi_1^4}{-nz_j d^2_1 |m_{1n}(z_j)||m_{1n}(z')|}+\frac{1}{n}\sum_{i\ge2}\frac{\xi_i^4}{|z^{\prime}(1+\xi^2_1m_{1n}(z^{\prime}))(1+\xi^2_im_{1n}(z_j))|}\right)\times|m_{1n}(z_j)-m_{1n}(z^{\prime})|\\
    &\prec \ro(1) \times |m_{1n}(z_j)-m_{1n}(z^{\prime})|.
\end{split}
\]
Similarly, for $\mathtt{M}_{11,2},$ we have that 
\[
\begin{split}
    \mathtt{M}_{11,2}=\frac{1}{n}\sum_{i}\frac{\xi_i^2 (z_j-z^{\prime})}{z^{\prime}z_j(1+\xi_i^2 m_{1n}(z'))} \prec \frac{z^{\prime}-z_j}{z^{\prime}z_j} \prec n^{-3-2/\alpha}.
\end{split}
\]
Analogously, we can prove that $\mathcal{M}_{12} \prec n^{-3-4/\alpha}$. 
Therefore, combining the above bounds with (\ref{def1}), we see that 
\begin{equation*}
    |m_{1n}(z^{\prime})-m_{1n}(z_j)|\prec  n^{-3-2/\alpha}.
\end{equation*}

\quad By similar procedures and arguments, we can also prove that 
%{\color{red}[FROM HERE]}
%
%
%Similarly, 
%\[
%\begin{split}
%    m_{2n}(z^{\prime})-m_{2n}(z_j)&=\frac{1}{n}\sum_{i=1}^n\Big(\frac{\xi^2_i}{-z^{\prime}(1+\xi^2_im_{1n}(z^{\prime}))}-\frac{\xi^2_i}{-z_j(1+\xi^2_im_{1n}(z_j))}\Big)\\
%    &=\frac{1}{n}\sum_{i=1}^n\Big(\frac{\xi^2_i}{-z^{\prime}(1+\xi^2_im_{1n}(z^{\prime}))}-\frac{\xi^2_i}{-z^{\prime}(1+\xi^2_im_{1n}(z_j))}\Big)\\
%    &+\frac{1}{n}\sum_{i=1}^n\Big(\frac{\xi^2_i}{-z^{\prime}(1+\xi^2_im_{1n}(z_j))}-\frac{\xi^2_i}{-z_j(1+\xi^2_im_{1n}(z_j))}\Big)\\
%    &:=\mathcal{M}_{21}+\mathcal{M}_{22}.
%\end{split}
%\]
%For $\mathcal{M}_{21}$,
%
%also
%\[
%\begin{split}
%    \mathcal{M}_{22}&=\frac{1}{n}\sum_{i=1}^n\frac{\xi_i^2(z^{\prime}-z_j)}{z^{\prime}z_j(1+\xi^2_im_{1n}(z_j))}\\
%    &{\color{blue}=\frac{1}{n}\frac{\xi^2_1(z^{\prime}-z_j)}{z^{\prime}z_jd_1|m_{1n}(z_j)|}|+\frac{1}{n}\sum_{i\ge2}\frac{\xi_i^2(z^{\prime}-z_j)}{z^{\prime}z_j(1+\xi^2_im_{1n}(z_j))}}\\
%    &{\color{blue}=O_{\prec}(n^{-3-4/\alpha})}.
%\end{split}
%\]
%Consequently, we have 
\begin{equation*}
    |m_{2n}(z^{\prime})-m_{2n}(z_j)|\prec n^{-3-2/\alpha}, \  |m_{n}(z^{\prime})-m_{n}(z_j)|\prec n^{-3-2/\alpha}. 
\end{equation*}
and
\begin{gather*}
    \left\|(z^{\prime})^{-1}(I+m_{1n}(z^{\prime})D^2)^{-1}-(z_j)^{-1}(I+m_{1n}(z_j)D^2)^{-1}\right\|\prec n^{-3-2/\alpha}.
\end{gather*}

Therefore, combining all the above bounds with triangle inequality, we see that the results of part 1 of Theorem \ref{thm_unboundedcaselocallaw} hold for $z'.$ Using an induction procedure and a standard lattice argument (for example, see \cite{Alex2014,Ding&Yang2018}), we find that the results hold for all $z \in \widetilde{\mathbf{D}}_{u}$ and conclude the proof of Proposition \ref{eq_propooursidebulk}. 
%{\color{red} from here: triangle inequality. }

%Then combining with \eqref{eq: Lip for mathcal G}, we can conclude 
%{\color{blue}
%\[
%\begin{split}
%    &\max_{ik}\big(G(z^{\prime})-(z^{\prime})^{-1}(I+m_{1n}(z^{\prime})D^2)^{-1}\big)\\
%    &\le\max_{ik}(|G(z^{\prime})-G(z_j)|)+\max_{ik}(G(z_j)-(z_j)^{-1}(I+m_{1n}(z_j)D^2)^{-1})\\
%    &+\|(z^{\prime})^{-1}(I+m_{1n}(z^{\prime})D^2)^{-1}-(z_j)^{-1}(I+m_{1n}(z_j)D^2)^{-1}\|\\
%    &\prec n^{-1/2-2/\alpha}.
%\end{split}
%\]
%}
%Similarly, because
%
%and
%
%Consequently, we conclude that 
%\begin{gather*}
%    |m_1(z^{\prime})-m_{1n}(z^{\prime})|\prec n^{-1/2-4/\alpha},\quad |m_2(z^{\prime})-m_{2n}(z^{\prime})|\prec n^{-1/2-2/\alpha},\\ |m_{\mathcal{Q}}(z^{\prime})-m_n(z^{\prime})|\prec n^{-1/2-4/\alpha}.
%\end{gather*}
%Then by induction and Lemma \ref{lem: self improvement} combining with the standard lattice argument, {\color{blue} uniformly for any $z\in\tilde{\mathbf{D}}$, we have $\max_{1j}(G+z^{-1}(I+m_{1n}(z)D^2))\prec n^{-1/2-2/\alpha}$, $|m_1-m_{1n}|+|m_{\mathcal{Q}}-m_n|\prec n^{-1/2-4/\alpha}$ and $|m_2-m_{2n}|\prec n^{-1/2-2/\alpha}$.}
\end{proof}
%{\color{blue}
%\begin{remark}
%The above procedures can be generalised to the case $\alpha=+\infty$ with minor changes, since here we mainly rely on the Lipschitz properties of the quantities of $G_{ij}, m_{1n},m_{2n}$.
%\end{remark}
%}

\subsubsection{Proof of Theorem \ref{thm_unboundedcaselocallaw}} \label{sec_proofa2}

Once Proposition \ref{eq_propooursidebulk} is proved, we can roughly locate the edge eigenvalues of $Q$ as in Lemma \ref{lem: upper bound for eigenvalues} so that we can expand the spectral domain from $\widetilde{\mathbf{D}}_{u}$ to $\mathbf{D}_{u}$ for $Q^{(1)}$ and conclude the proof of Theorem \ref{thm_unboundedcaselocallaw}.

\quad As before, we focus on the study of the separable covariance i.i.d. model as in Case (2) of Assumption \ref{assum_model} when $\xi^2$ satisfies (\ref{ass3.1}). Recall  the definitions of $\mu_2$ and $\lambda_1^{(1)}$ around (\ref{eq_keylocationdefinition}) and in Figure \ref{fig_locationtizubounded}. By Lemma \ref{lem: upper bound for eigenvalues} and an analogous argument, as well as  
Weyl's inequality, we find that conditional on the event $\Omega,$ with high probability, 
\begin{equation}\label{eq_keyused}
    \mu_1>\lambda_1>\mu_2>\lambda_1^{(1)}.
\end{equation}
By (\ref{eq_mu1part}) and a similar argument, we see that $\mu_k \asymp \xi_k^2, k=1,2.$ Together with (\ref{def1}), we have that $\mu_1-\mu_2\ge C_1 n^{1/\alpha}\log^{-1}n$ for some constant $C_1>0$ on the event $\Omega.$ This implies for some constant $C>0,$ for $z \in \mathbf{D}_{u},$ 
\begin{gather}\label{eq: eigen gap}
    |\lambda_1^{(1)}-z|\ge Cn^{1/\alpha}\log^{-1}n,
\end{gather}
Now we proceed to the proof of Theorem \ref{thm_unboundedcaselocallaw}. Recall (\ref{eq_defnminor}) and (\ref{eq_defnminorG}).

%..... {\color{red} from here}
%
%
%As the Average local law \ref{lem: average local law} in hand, we can locate the largest eigenvalues of $\lambda_1$ and $\lambda_1^{(1)}$ by the procedures in Lemma \ref{lem: upper bound for eigenvalues}. It can be seen that by and Lemma \ref{lem: upper bound for eigenvalues}, 
%Together with the observation that , we conclude that for any $z\in\mathbf{D}_{u1}$,
%
%for some positive constant $C$.

%We have the following inner region local law
%\begin{lemma}\label{lem: average local law on D_u1}
%Suppose Assumptions \ref{ass1},\ref{ass2},\ref{ass3.1},\ref{ass3.2} and events $\Omega_n$ hold, then uniformly for $z\in\mathbf{D}_{u1}$, we have
%\begin{itemize}
%    \item When $\alpha\in(2,\infty)$,
%    \begin{gather}
%    |m_1^{(1)}(z)-m_{1n}^{(1)}(z)|\prec n^{-1/2-2/\alpha}.
%\end{gather}
%\item When $\alpha=+\infty$,
%\begin{gather}
%    |m_1^{(1)}(z)-m_{1n}^{(1)}(z)|\prec n^{-1/2}\log^{-K}n.
%\end{gather}
%\end{itemize}
%\end{lemma}
\begin{proof}[\bf Proof of Theorem \ref{thm_unboundedcaselocallaw}]  Again, we focus on the study of the separable covariance i.i.d. model as in Case (2) of Assumption \ref{assum_model} when $\xi^2$ satisfies (\ref{ass3.1}).

Observe  by \eqref{eq: eigen gap} that it holds uniformly for $z\in\mathbf{D}_{u}$ and $\mathcal{T}\subset\{2,\dots,n\}$, for some constant $C_1>0$
\begin{gather}\label{eq: est of G^(1)}
    \|G^{(1\mathcal{T})}\|\leq C_1 n^{-1/\alpha} \log n.
\end{gather}
By the definition of $m_2^{(1)}$ and a decomposition similar to (\ref{eq: decomp m_2}),  we have that 
\begin{gather*}
        m_2^{(1)}=\frac{1}{n}\sum_{i=2}^n\frac{\xi^2_i}{-z-z\mathbf{y}_i^{*}G^{(1i)}\mathbf{y}_i}=\frac{1}{n}\sum_{i=2}^n\frac{\xi^2_i}{-z(1+\xi^2_in^{-1}\operatorname{tr}G^{(1i)}\Sigma+Z_i^{(1)})},\\
        Z_i^{(1)}=\mathbf{y}_i^{*}G^{(1i)}\mathbf{y}_i-\xi^2_in^{-1}\operatorname{tr}G^{(1i)}\Sigma.
\end{gather*}
By arguments similar to (\ref{eq: bound for Z}) and (\ref{eq_m1approximate}) but with \eqref{eq: est of G^(1)}, we obtain that 
\begin{gather*}
    Z_i^{(1)}\prec\frac{\xi^2_i}{n}\|G^{(1i)}\Sigma\|_F\le\frac{\xi^2_i}{n}\|G^{(1i)}\|\|\Sigma\|_F\prec\frac{\xi^2_i}{n^{1/2}}n^{-1/\alpha},\\
    \frac{1}{n}\operatorname{tr}(G^{(1i)}\Sigma)-m_1^{(1)}(z)=\frac{1}{n}\mathbf{y}^{*}G^{(1)}\Sigma G^{(1i)}\mathbf{y}_i\prec \frac{\xi^2_i}{n}n^{-2/\alpha}.
\end{gather*}
In addition, using (\ref{eq: est of G^(1)}) and a discussion similar to (\ref{eq_eeeone})--(\ref{eq_modeltemp}), we readily see that 
\begin{gather*}
    m^{(1)}_2=\frac{1}{n}\sum_{i=2}^n\frac{\xi^2_i}{-z(1+\xi^2_im^{(1)}_1)}+\rO_{\prec}\left(n^{-1/2-1/\alpha}\right).
\end{gather*}
Using the decomposition 
\begin{gather*}
    Q^{(1)}-zI=\sum_{i=2}^n\mathbf{y}_i\mathbf{y}_i^{*}+zm^{(1)}_2(z)\Sigma-z(I+m_2^{(1)}(z)\Sigma),
\end{gather*}
by arguments similar to (\ref{eq: decomp of mathcal_G})--(\ref{eq: m_1 by m_2}) with $\|\mathcal{G}^{(1\mathcal{T})}\|=\|G^{(1 \mathcal{T})}\|\prec n^{-1/\alpha}$, we conclude that 
\begin{gather*}
\begin{split}
    m_1^{(1)}&=-z^{-1}\frac{1}{n}\operatorname{tr}((I+m_2^{(1)}(z)\Sigma)^{-1}\Sigma)+\rO_{\prec}\left(n^{-1/2-2/\alpha}\right)\\
    &=-\frac{1}{n}\sum_{i=1}^p\frac{\sigma_i}{z(1+m_2^{(1)}\sigma_i)}+\rO_{\prec}\left(n^{-1/2-2/\alpha}\right).
\end{split}
\end{gather*}
Combining with the definitions in (\ref{eq_systemequationsm1m2}), we see that
\begin{gather*}
    \begin{split}
        m_1^{(1)}(z)-m_{1n}^{(1)}(z)&=-\frac{1}{n}\sum_{i=1}^p\frac{\sigma_i}{z(1+\sigma_im_2^{(1)}(z))}+\frac{1}{n}\sum_{i=1}^p\frac{\sigma_i}{z(1+\sigma_im_{2n}^{(1)}(z))}+\rO_{\prec}(n^{-1/2-2/\alpha})\\
        &=\frac{1}{n}\sum_{i=1}^p\frac{\sigma_i^2(m_2^{(1)}(z)-m_{2n}^{(1)}(z))}{z(1+\sigma_im_{2n}^{(1)}(z))(1+\sigma_im_2^{(1)}(z))}+\rO_{\prec}(n^{-1/2-2/\alpha})\\
        &=\left(\frac{1}{n}\sum_{i=1}^p\frac{\sigma_i^2}{z(1+\sigma_im_{2n}^{(1)}(z))(1+\sigma_im_2^{(1)}(z))}\right)\left(\frac{1}{n}\sum_{i=2}^n\frac{\xi^4_i(m_1^{(1)}(z)-m_{1n}^{(1)}(z))}{z(1+\xi^2_im_1^{(1)}(z))(1+\xi^2_im_{1n}^{(1)}(z))} \right)+\rO_{\prec}(n^{-1/2-2/\alpha})\\
        &=\ro(1)(m_1^{(1)}(z)-m_{1n}^{(1)}(z))+\rO_{\prec}(n^{-1/2-2/\alpha}),
    \end{split}
\end{gather*} 
where in the third step we used a discussion similar to (\ref{eq_eeetwo}) and (\ref{def1}). This completes our proof. 

%{\color{red} Above procedures hold for $\alpha=+\infty$ as well, where in that case $\|G^{(i\mathcal{T})}\|\prec \log^{-1/\beta}$. And then a parallel argument conclude the final results.}

%{\color{red}[from here]}
%
%{\color{blue}
%%
%%, we have
%%
%%Using , we see that
%%
%%Together with the fact $|1+\xi^2_im_1^{(1)}(z)|>0,\;i=2,\dots,n$ {\color{red} [need more detail here]}, we obtain 
%
%Taking decomposition for $Q^{(1)}$ as 
%
%By similar procedures for $m_1^{(1)}$ as in Lemma \ref{lem: average local law for large eta} with the fact , we may obtain that
%
%
%Combining the estimation of $m_1^{(1)}(z)$ and $m_2^{(1)}(z)$, we have
%
%Therefore, we conclude that
%\begin{gather}
%    |m_1^{(1)}(z)-m_{1n}^{(1)}(z)|\prec n^{-1/2-4/\alpha}.
%\end{gather}
%}
\end{proof}

\subsection{Bounded support setting: proof of Theorem \ref{thm_boundedcaselocallaw}}
In this section, we will prove Theorem \ref{thm_boundedcaselocallaw}. In Section \ref{sec_proofpartiboundedlocallaw}, we study $m_{1n,c}(z)-m_{1n}(z)$ and $m_{n,c}-m_n,$ which is a counterpart of Lemma 4.4 of \cite{Kwak2021}. Then in Section \ref{sec_proofpartiiboundedlocallaw}, we study $m_{1n}(z)-m_{1}(z)$ and $m_Q-m_n,$ which is a counterpart of Proposition 5.1 of \cite{Kwak2021}. Due to similarity, we only provide the details for the separable covariance i.i.d data as in (2) of Assumption \ref{assum_model}. The elliptical data can be handled similarly.  

\subsubsection{Control of $m_{1n,c}(z)-m_{1n}(z)$ and $m_{n,c}(z)-m_n(z)$}\label{sec_proofpartiboundedlocallaw}

Due to similarity, we focus on $|m_{1n,c}-m_{1n}|$ and briefly discuss $|m_{n,c}-m_n|$ in the end. The proof ideas follow  Lemma 4.5 of \cite{lee2016extremal} or Lemma 4.4 of \cite{Kwak2021}. We focus on explaining the parts deviates the most. 

\begin{proof}
According to the definitions of $m_{1n,c}$ and $m_{1n}$ in (\ref{eq_systemequationsm1m2intergrate}) and (\ref{eq_systemequationsm1m2}), we observe that 
\begin{align}\label{eq_p1plusp2}
    |m_{1n,c}(z)-m_{1n}(z)|
    & \leq \left|\frac{1}{n}\sum_{i=1}^p\frac{\sigma_i}{-z+\sigma_i\int\frac{s}{1+sm_{1n,c}(z)}\mathrm{d}F(s)}-\frac{1}{n}\sum_{i=1}^p\frac{\sigma_i}{-z+\frac{\sigma_i}{n}\sum_{j=1}^n\frac{\xi^2_j}{1+\xi^2_jm_{1n,c}(z)}}\right| \nonumber \\
    &+|m_{1n,c}(z)-m_{1n}(z)|\left|\frac{1}{n}\sum_{i=1}^p\frac{\frac{\sigma_i^2}{n}\sum_{j=1}^n\frac{\xi^4_j}{(1+\xi^2_jm_{1n}(z))(1+\xi^2_jm_{1n,c}(z))}}{(-z+\frac{\sigma_i}{n}\sum_{j=1}^n\frac{\xi^2_j}{1+\xi^2_jm_{1n}(z)})(-z+\frac{\sigma_i}{n}\sum_{j=1}^n\frac{\xi^2_j}{1+\xi^2_jm_{1n,c}(z)})}\right| \\ \nonumber
    &:=  \mathsf{P}_1+\mathsf{P}_2. 
\end{align}

On the one hand, for $\mathsf{P}_1, $ we have that 
\begin{equation*}
\mathsf{P}_1=\frac{1}{n}\sum_{i=1}^p\frac{\sigma_i^2|n^{-1}\sum_j\frac{\xi^2_j}{1+\xi^2_jm_{1n,c}(z)}-\int\frac{s}{1+sm_{1n,c}(z)}\mathrm{d}F(s)|}{|(-z+\frac{\sigma_i}{n}\sum_{j=1}^n\frac{\xi^2_j}{1+\xi^2_jm_{1n,c}(z)})(-z+\sigma_i\int\frac{s}{1+sm_{1n,c}(z)}\mathrm{d}F(s))|}.
\end{equation*}
Since $z \in \mathbf{D}_b^\prime \subset \mathbf{D}_b,$ according to Assumption \ref{assum_additional_techinical} and the continuity of $m_{1n,c}$, 
we conclude that $|-z+\sigma_i\int\frac{s}{1+sm_{1n,c}(z)}\mathrm{d}F(s))| \geq c$ for some constant $c>0.$ Together with (\ref{def4}), we show that $|-z+\frac{\sigma_i}{n}\sum_{j=1}^n\frac{\xi^2_j}{1+\xi^2_jm_{1n,c}(z)}| \geq c'$ for some $c'>0$ when $n$ is sufficiently large. Using (\ref{def4}) again, we conclude that on $\Omega,$ for some small constant $\epsilon>0$ and some constant $C>0$ 
\begin{equation}\label{eq_defnp1}
\mathsf{P}_1 \leq  C n^{-1/2+\epsilon}. 
\end{equation} 

\quad On the other hand, for $\mathsf{P}_2,$ for notional convenience, we further write it as $\mathsf{P}_2=|m_{1n,c}(z)-m_{1n}(z)| \times |\mathsf{T}|.$ For $\mathsf{T},$ by Cauchy-Schwarz inequality, we have that
\begin{equation}\label{eq_tintot1t2}
|\mathsf{T}| \leq \mathsf{E}_1 \mathsf{E}_2,
\end{equation}
where $\mathsf{E}_k, k=1,2,$ are defined as 
\begin{gather}
    \begin{split}
        &\mathsf{E}_1:= \left(\frac{1}{n}\sum_{i=1}^p\frac{\frac{\sigma^2_i}{n}\sum_{j=1}^n\frac{\xi^4_j}{(1+\xi^2_jm_{1n}(z))^2}}{|-z+\frac{\sigma_i}{n}\sum_{j=1}^n\frac{\xi^2_j}{1+\xi^2_jm_{1n}(z)}|^2}\right)^{1/2},\\
        &\mathsf{E}_2:=\left(\frac{1}{n}\sum_{i=1}^p\frac{\frac{\sigma^2_i}{n}\sum_{j=1}^n\frac{\xi^4_j}{(1+\xi^2_jm_{1n,c}(z))^2}}{|-z+\frac{\sigma_i}{n}\sum_{j=1}^n\frac{\xi^2_j}{1+\xi^2_jm_{1n,c}(z)}|^2}\right)^{1/2}.
    \end{split}
\end{gather}
Together with the identity (\ref{eq_expansionusefullessorequaltoone}) below and the fact $m_{1n}(z) \asymp 1$ for $z \in \mathbf{D}_b^\prime$, we find that $\mathsf{E}_1  \leq 1.$ For the term $\mathsf{E}_2,$ we first consider a closely related quantity $\mathsf{W}(z)$ defined as 
\begin{equation*}
\mathsf{W}(z):=\frac{1}{n} \sum_{i=1}^p \frac{\sigma_i^2 \int \frac{s^2}{|1+s m_{1n,c}(z)|^2} \mathrm{d} F(s) }{|-z+\sigma_i \int \frac{s^2}{(1+s m_{1n,c}(z))^2} \mathrm{d} F(s)|^2}=1-\eta \frac{|m_{1n,c}(z)|^2}{\operatorname{Im} m_{1n,c}(z)}.
\end{equation*}
By assumption that $\phi^{-1}>\varsigma_3,$ (\ref{eq_onlyoneequationdecide}) that $m_{1n,c}(L_+)=-l^{-1}$ and recall the notations in (\ref{eq_phasetransition}), we see that
\begin{equation}\label{eq_edgeresults}
\mathsf{W}(L_+)<1.
\end{equation}
Armed with (\ref{eq_edgeresults}), using  (\ref{def4}) and Assumption \ref{assum_additional_techinical}, we can apply  an argument similar to Lemma A.6 of \cite{lee2016extremal} or Lemma A.7 of \cite{Kwak2021}  to conclude that  when $n$ is sufficiently large, for $z \in \mathbf{D}_b,$
\begin{equation}\label{eq_defnmathsfW}
\mathsf{E}^2_2=\mathsf{W}(L_+)+\ro(1)<\mathfrak{c}',
\end{equation}
for some constant $0<\mathfrak{c}'<1.$
%\begin{equation}\label{eq_defnmathsfW}
%\mathsf{W}:=\frac{1}{n} \sum_{i=1}^p \frac{\sigma_i^2 \int \frac{s^2}{|1+s m_{1n,c}(z)|^2} \mathrm{d} F(s) }{|-z+\ \int \frac{s^2}{(1+s m_{1n,c}(z))^2} \mathrm{d} F(s)|^2}=1-\eta \frac{|m_{1n,c}(z)|^2}{\operatorname{Im} m_{1n,c}(z)}<\mathfrak{c}',
%\end{equation}
Consequently, we find that when $n$ is sufficiently large, $\mathsf{E}_2<1.$ Together with (\ref{eq_tintot1t2}), we can conclude that $|\mathsf{T}|<1$. This yields that 
\begin{equation*}
\mathsf{P}_2 = \mathfrak{c} |m_{1n,c}-m_{1n}|,
\end{equation*} 
for some constant $0<\mathfrak{c}<1.$ 

\quad Inserting the above control back into (\ref{eq_p1plusp2}), using (\ref{eq_defnp1}), we can conclude our proof that
\begin{equation}\label{eq_deterministiclose}
m_{1n,c}=m_{1n}(z)+\rO(n^{-1/2+\epsilon}). 
\end{equation} 
The proof of $m_{n,c}-m_n$ follows from an argument similar to (\ref{eq_p1plusp2}) using (\ref{eq_systemequationsm1m2}) and (\ref{eq_systemequationsm1m2intergrate}) that
\begin{equation*}
m_n(z)=\frac{1}{p} \sum_{i=1}^p \frac{1}{-z+\sigma_i n^{-1} \sum_{j=1}^n \frac{\xi_j^2}{1+\xi_j^2 m_{1n}}}, \ m_{n,c}(z)=\frac{1}{p} \sum_{i=1}^p \frac{1}{-z+\sigma_i \int_0^l \frac{s}{1+s m_{1n,c}(z)} \mathrm{d} F(s)},
\end{equation*}
the results of $m_{1n,c}-m_{1n}$ and (\ref{def4}). We omit the details. 
\end{proof}

\subsubsection{Control of $m_{1n}(z)-m_1(z)$ and $m_n(z)-m_Q(z)$}\label{sec_proofpartiiboundedlocallaw}
Due to similarity, we focus on $|m_{1n,c}-m_{1n}|$ and will briefly discuss how to study $|m_Q-m_n|$ from line to line. 
The proof ideas follow  Proposition 5.1 of \cite{lee2016extremal} or Proposition 5.1 of \cite{Kwak2021}. We focus on explaining the parts deviates the most.  The proof relies on the following two lemmas. 

\begin{lemma}\label{lem: first bound for m_1n-m_1}
Conditional on the event $\Omega$ in Theorem \ref{thm_boundedcaselocallaw}, for all $z=E+\ri\eta\in\mathbf{D}_b^{\prime}$ with $n^{-1/2+\epsilon_d}\le\eta\le n^{-1/(d+1)+\epsilon_d}$, we have 
\begin{gather*}
    |m_{1n}(z)-m_1(z)|\prec\frac{1}{n\eta_0}, \   |m_{n}(z)-m_Q(z)|\prec\frac{1}{n\eta_0}.
\end{gather*}
\end{lemma}
%{\color{red}[from here]}
\begin{lemma}\label{lem: second bound for m_1n-m_1}
 Assuming that $|m_{1n}(z)-m_1(z)|\prec n^{\epsilon_d}(n\eta_0)^{-1}$, then conditional on the event $\Omega,$ we have that for all $z \in \mathbf{D}_b^\prime$ 
\begin{gather*}
    |m_{1n}(z)-m_1(z)|\prec\frac{1}{n\eta_0}, \   |m_{n}(z)-m_Q(z)|\prec\frac{1}{n\eta_0}.
\end{gather*}
\end{lemma}

\quad Armed with the above two lemmas, we now proceed to the control of $m_{1n}(z)-m_1(z).$ 

\begin{proof}[\bf Proof: control of $m_{1n}(z)-m_1(z)$ and $m_n(z)-m_Q(z)$] Due to similarity, we only prove $m_{1n}(z)-m_1(z).$  We prove this by mathematical induction as that of Proposition 5.1 of \cite{lee2016extremal}. Fix $E$ such that $z=E+\ri\eta_0\in\mathbf{D}_b^{\prime}$, we consider a sequence $(\eta_j)$ defined by $\eta_j=\eta_0+jn^{-2}$. Let $K$ be the smallest positive integer such that $\eta_K\ge n^{-1/2+\epsilon_d}$. Note that for $j=K$, by Lemma \ref{lem: first bound for m_1n-m_1}, we have that  
\begin{gather*}
    |m_{1n}(z_j)-m_1(z_j)|\prec\frac{1}{n\eta_0}.
\end{gather*}
Then for any $z=E+\ri\eta$ with $\eta_{j-1}\le\eta\le\eta_j$, we have that for some constant $C>0$
\begin{align*}
 |m_1(z_j)-m(z)|=\frac{1}{n} \operatorname{tr}\left[(G(z_j)-G(z)) \Sigma\right]=\frac{|z_j-z|}{n} \operatorname{tr}(G(z_j)G(z) \Sigma) \le C \frac{|z_j-z|}{\eta_{j-1}^2}\le C\frac{n^{2\epsilon_d}}{n},
\end{align*}
where we used the first resolvent identity and the trivial bound $|G(z)| \leq \eta^{-1}$, and similarly 
\begin{align*}
|m_{1n}(z_j)-m_{1n}(z)|=\left| \int \left[\frac{1}{x-z_j}-\frac{1}{x-z} \right] \rho(x) \mathrm{d} x   \right|\le\frac{|z_j-z|}{\eta_{j-1}^2}\le\frac{n^{2\epsilon_d}}{n}. 
\end{align*}
Thus we find that if $|m_{1n}(z_j)-m_1(z_j)|\prec\frac{1}{n\eta_0}$, then by Lemma \ref{lem: second bound for m_1n-m_1}, for some constant $C'>0$
\begin{gather}\label{eq_latticeargumentbelow}
    |m_{1n}(z)-m_1(z)|\le|m_{1n}(z_j)-m_1(z_j)|+\frac{C' n^{2\epsilon_d}}{n}\prec\frac{n^{\epsilon_d}}{n\eta_0}.
\end{gather}
This gives the result that $|m_{1n}(z)-m_1(z)|\prec(n\eta_0)^{-1}$ for $z=E+\ri\eta$ with $\eta_{j-1}\le\eta\le\eta_j$. The proof for each $z$ can be completed by an induction on  $j.$ Finally, using an induction procedure and a standard lattice argument (for example, see \cite{Alex2014,Ding&Yang2018}), we find that the results hold for all $z \in \mathbf{D}_{b}^\prime$. More specifically, we construct a lattice $\mathcal{L}$ from $z^{\prime}=E^{\prime}+\ri\eta_0\in\mathbf{D}_b^{\prime}$ with $|z-z^{\prime}|\le n^{-3}$. It is obvious that the bound holds uniformly on $\mathcal{L}$. For any $z=E+\ri\eta_0\notin\mathcal{L}$, we find a $z^{\prime}\in\mathcal{L}$ and then $|z-z^{\prime}|\le n^{-3}$. Moreover, using resolvent identity, we can conclude that  $|m_1(z)-m_1(z^{\prime})|\le\eta^{-2}_0|z-z^{\prime}| \ll (n \eta_0)^{-1}$. Therefore, we conclude the proof.

\end{proof}

\quad In what follows, we prove lemmas \ref{lem: first bound for m_1n-m_1} and \ref{lem: second bound for m_1n-m_1}. The proofs are similar to those of Lemmas 5.6 and 5.7 of \cite{lee2016extremal}, except that we will need a weak local law as follows.

\begin{proposition}[Weak averaged local law]  \label{proposition_boundedweaklocallaw}
Suppose the assumptions of Theorem \ref{thm_boundedcaselocallaw} hold. We have that for $z \in \mathbf{D}_b^\prime$ 
\begin{equation*}
|m_Q(z)-m_n(z)|+|m_1(z)-m_{1n}(z)|+|m_2(z)-m_{2n}(z)|=\rO_{\prec}\left( (n \eta)^{-1/4} \right).
\end{equation*}
\end{proposition} 
\begin{proof}
The proof of Proposition \ref{proposition_boundedweaklocallaw} is relatively standard in the random matrix literature, for example, see Section 4.1 of \cite{Alex2014} or Section 3.6 of \cite{erdHos2013spectral} or Appendix A.2 of \cite{Ding&Yang2018} or Section 5.2 of \cite{yang2019edge}. Due to similarity, as in Lemma 5.12 of \cite{yang2019edge},  we only provide the key ingredients. Define the $z$-dependent parameter
\begin{gather}\label{eq_defnpsi}
    \Psi(z):=\sqrt{\frac{\operatorname{Im}m_{1}(z)}{n\eta}}+\frac{1}{n \eta}.
\end{gather}      
Recall (\ref{eq: decomp m_2}). By Lemma \ref{lem:large deviation} and (\ref{lem:Wald}), we find that 
\begin{gather}\label{eq_zibound}
    Z_i\prec\frac{\xi^2_i}{n}\|G^{(i)}\Sigma^{1/2}\|_F \leq l\sqrt{\frac{\operatorname{Im}m^{(i)}_{1}(z)}{n\eta}}\asymp \Psi,
\end{gather}
where in the last step we used (\ref{lem:trace_difference}).  Together with (\ref{lem:trace_difference}) and the first equation of (\ref{eq: decomp m_2}), we conclude that 
\begin{gather}\label{eq_m2uboundeddecomposition}
    m_2=\frac{1}{n}\sum_{i=1}^n\frac{\xi^2_i}{-z(1+\xi^2_im_1(z)+\rO_{\prec}(\Psi))}.
\end{gather} 

\quad For $m_1(z),$ recall (\ref{eq: decomp of mathcal_G}). According to the definition of $m_1(z)$ in (\ref{eq_m1m2}), we have that 
\begin{gather}\label{eq_m1decompositionfinafinalfinal}
    m_1(z)=\frac{1}{n}\operatorname{tr}(G(z)\Sigma)=-\frac{1}{n}\sum_{i=1}^p\frac{\sigma_i}{z(1+m_2(z)\sigma_i)}+\frac{1}{n}\operatorname{tr}(R_1\Sigma)+\frac{1}{n}\operatorname{tr}(R_2\Sigma).
\end{gather}
Similarly, for $m_Q(z)$ in (\ref{eq_mq}), we have that
\begin{equation}\label{eq_mqdecomposition}
  m_{Q}(z)=\frac{1}{p}\operatorname{tr}(G(z))=-\frac{1}{p}\sum_{i=1}^p\frac{1}{z(1+m_2(z)\sigma_i)}+\frac{1}{p}\operatorname{tr}(R_1)+\frac{1}{p}\operatorname{tr}(R_2).
\end{equation}

On the one hand, when $\eta \asymp 1,$ by a discussion similar to (5.45) of \cite{yang2019edge} , we find that $\| (I+m_2^{(i)} \Sigma)^{-1} \|<\infty.$ Then using (\ref{eq_zibound}), by a discussion similar to the equations between (\ref{eq: decomp R_1}) and (\ref{eq: m_1 by m_2}), we find that 
\begin{gather}\label{eq: est of m_1 and m_2}
m_1(z)=-\frac{1}{n}\sum_{i=1}^p\frac{\sigma_i}{z(1+m_2(z)\sigma_i)}+\rO_{\prec}\Big(\frac{1}{n}\sum_i\frac{\xi^2_i\Psi}{z(1+\xi^2_im_1(z)+\rO_{\prec}(\Psi))}\Big).
\end{gather}
Similarly, we have 
\begin{equation*}
  m_{Q}(z)=-\frac{1}{p}\sum_{i=1}^p\frac{1}{z(1+m_2(z)\sigma_i)}+\rO_{\prec}\Big(\frac{1}{p}\sum_i\frac{\xi^2_i\Psi}{z(1+\xi^2_im_1(z)+\rO_{\prec}(\Psi))}\Big).
\end{equation*}

On the other hand, denote $\Xi:=\{|\mathcal{G}_{ij}(z)+\delta_{ij}(z(1+m_{1n}(z) \xi_i^2))^{-1}|+|m_2(z)-m_{2n}(z)| \leq (\log n)^{-1} \}.$ When restricted on $\Xi,$ by Assumption \ref{assum_additional_techinical}, we also have that $\| (I+m_2^{(i)} \Sigma)^{-1} \|<\infty.$ By an analogous argument, we find that (\ref{eq: est of m_1 and m_2}) also holds true. By an argument similar to (\ref{eq_nontrivialcontrol}) using Lemma \ref{lem: Resolvent}, we have that for $i \neq j,$ $\mathbf{1}(\eta \geq 1) \mathcal{G}_{ij} \prec \Psi, \ \mathbf{1}(\Xi) \mathcal{G}_{ij} \prec \Psi. $

The above arguments show that the counterparts of Lemmas 5.9 and 5.10 of \cite{yang2019edge} hold. Therefore, by the same arguments as in Lemma 5.12 of \cite{yang2019edge}, we can conclude the proof.  

\end{proof} 

\quad Next we provide some useful controls whose proofs and results will be used in the proof of Lemmas \ref{lem: first bound for m_1n-m_1} and \ref{lem: second bound for m_1n-m_1}. The results and arguments are analogous to Lemma 5.4 of \cite{lee2016extremal}. We only point out the key ingredients in our proof and refer the readers to \cite{lee2016extremal} for more details. 

\begin{lemma}\label{lem: est for Im m_1}
Suppose the assumptions of Theorem \ref{thm_boundedcaselocallaw} hold. Then we have that on $\Omega$ and for all $z=E+\ri\eta_0\in\mathbf{D}_b^{\prime}$
\begin{gather*}
    \operatorname{Im}m_1(z)\prec\frac{1}{n\eta_0}, \    \  \operatorname{Im}m_Q(z)\prec\frac{1}{n\eta_0}.
\end{gather*}
\end{lemma}
\begin{proof}
Due to similarity, we focus our proof on $\operatorname{Im} m_Q(z).$ We prove by contradiction. Given some $\epsilon>0,$ conditional on $\Omega,$ for some sufficiently small constants $0<c_1, c_2<1,$ we first introduce a probability event $\Xi_1 \equiv \Xi_1(\epsilon)$ so that the followings holds:
\begin{enumerate}
\item For $z \in \mathbf{D}_b^\prime,$
\begin{equation*}
|m_Q(z)-m_n(z)|+|m_1(z)-m_{1n}(z)|+|m_2(z)-m_{2n}(z)| \leq  (n \eta)^{-1/4+c_1\epsilon}.
\end{equation*}
\item For $z \in \mathbf{D}_b^\prime$ and $Z_i$ in (\ref{eq: decomp m_2}),
\begin{equation*}
\max_i Z_i \leq n^{c_2 \epsilon}\Psi.
\end{equation*}
\end{enumerate}
According to Proposition \ref{proposition_boundedweaklocallaw} and (\ref{eq_zibound}), we find that there exists some large constant $D>0$ so that $\mathbb{P}(\Xi_1)=1-n^{-D}.$ In what follows, we restrict ourselves on $\Xi_1$ so that the discussions are purely deterministic.  

\quad  Assuming that 
\begin{equation*}
 \operatorname{Im}m_1(z)>n^{\epsilon}\frac{1}{n\eta_0}.
\end{equation*}
We then conclude from the definition of $\Psi$ in (\ref{eq_defnpsi}) that
\begin{equation}\label{eq_Psicontrol}  
\Psi=\ro(\operatorname{Im}m_{1}(z)).
\end{equation}
 Moreover, by (\ref{eq_zopointrate}) and (\ref{eq_oneregimeedgecontrol}), we readily see that    $\operatorname{Im}m_{1n}(z)\ll \operatorname{Im}m_1(z) $. This implies that for some constant $C>0$  
 \begin{gather}\label{eq_tobeusedasacontradiction}
    |m_{1n}(z)-m_1(z)|\ge|\operatorname{Im}m_{1n}-\operatorname{Im}m_1|>Cn^{\epsilon}\frac{1}{n\eta_0}.
\end{gather}

\quad On the other hand, by Proposition \ref{proposition_boundedweaklocallaw} and Assumption \ref{assum_additional_techinical}, we see that (\ref{eq: est of m_1 and m_2}) still holds. Together with $m_{1n}$ in (\ref{eq_systemequationsm1m2}), using (\ref{eq: decomp m_2}), we find that 
\begin{align}
m_{1n}-m_1 & =\frac{1}{n}\sum_{i=1}^p\frac{\frac{\sigma_i^2}{n}\sum_j\frac{\xi^4_j(m_{1n}-m_1)+\xi^2_j\rO(n^{c_2 \epsilon}\Psi)}{(1+\xi^2_jm_{1n})(1+\xi^2_jm_1+\rO(n^{c_2 \epsilon}\Psi))}}{(z-\frac{\sigma_i}{n}\sum_j\frac{\xi^2_j}{1+\xi^2_jm_{1n}})(z-\frac{\sigma_i}{n}\sum_j\frac{\xi^2_j}{1+\xi^2_jm_{1}+\rO(n^{c_2\epsilon}\Psi)})}+\rO(n^{c_2 \epsilon}\Psi) \nonumber \\
&=\mathsf{C}_1(m_{1n}-m_1)+\mathsf{C}_2+\rO(n^{c_2 \epsilon} \Psi), \label{eq_decompositiondeomdeomdeodemdem}
\end{align}
where $\mathsf{C}_1, \mathsf{C}_2$ are defined as 
\begin{align*}
& \mathsf{C}_1:=\frac{1}{n}\sum_{i=1}^p\frac{\frac{\sigma_i^2}{n}\sum_j\frac{\xi^4_j}{(1+\xi^2_jm_{1n})(1+\xi^2_jm_1+\rO(n^{c_2 \epsilon}\Psi))}}{(z-\frac{\sigma_i}{n}\sum_j\frac{\xi^2_j}{1+\xi^2_jm_{1n}})(z-\frac{\sigma_i}{n}\sum_j\frac{\xi^2_j}{1+\xi^2_jm_{1}+\rO(n^{c_2\epsilon}\Psi)})}, \\
& \mathsf{C}_2:=\frac{1}{n}\sum_{i=1}^p\frac{\frac{\sigma_i^2}{n}\sum_j\frac{\xi^2_j\rO(n^{c_2 \epsilon}\Psi)}{(1+\xi^2_jm_{1n})(1+\xi^2_jm_1+\rO(n^{c_2 \epsilon}\Psi))}}{(z-\frac{\sigma_i}{n}\sum_j\frac{\xi^2_j}{1+\xi^2_jm_{1n}})(z-\frac{\sigma_i}{n}\sum_j\frac{\xi^2_j}{1+\xi^2_jm_{1}+\rO(n^{c_2\epsilon}\Psi)})}.
\end{align*}
We first control $\mathsf{C}_2.$ It is easy to see that $m_1\sim 1$ by contradiction. If $m_1\ll 1$, one can observe from (\ref{eq_m2uboundeddecomposition}) that $m_2\sim 1$ which yields $m_1 \asymp 1$ by (\ref{eq: est of m_1 and m_2}). If $m_1\gg 1$, we have $m_2\ll 1$ from (\ref{eq_m2uboundeddecomposition}), and then it gives that $m_1 \asymp 1$ by (\ref{eq: est of m_1 and m_2}). Similarly, we can show that $m_2 \asymp 1.$ Together with Proposition \ref{proposition_boundedweaklocallaw}, we find that $m_{1n}, m_{2n} \asymp 1.$ Since $z \asymp 1,$ using the definition of $m_{2n}$ in (\ref{eq_systemequationsm1m2}) and Proposition \ref{proposition_boundedweaklocallaw}, we find that 
\begin{equation*}
\frac{1}{n}\sum_j\frac{\xi^2_j\rO(n^{c_2 \epsilon}\Psi)}{(1+\xi^2_jm_{1n})(1+\xi^2_jm_1+\rO(n^{c_2 \epsilon}\Psi))}=\rO(n^{c_2 \epsilon} \Psi ).  
\end{equation*}
Moreover, by Proposition \ref{proposition_boundedweaklocallaw}, (\ref{eq_deterministiclose}) and Assumption \ref{assum_additional_techinical}, we find that 
\begin{equation*}
\frac{1}{n}\sum_i \frac{1}{(z-\frac{\sigma_i}{n}\sum_j\frac{\xi^2_j}{1+\xi^2_jm_{1n}})(z-\frac{\sigma_i}{n}\sum_j\frac{\xi^2_j}{1+\xi^2_jm_{1}+\rO(n^{c_2\epsilon}\Psi)})} \asymp 1. 
\end{equation*}  
This yields that 
\begin{equation}\label{eq_mathsfc2control}
\mathsf{C}_2=\rO(n^{c_2 \epsilon} \Psi). 
\end{equation}

\quad For $\mathsf{C}_1,$ using Proposition \ref{proposition_boundedweaklocallaw} and Assumption \ref{assum_additional_techinical}, by an argument similar to (\ref{eq_tintot1t2}), we can conclude that when $n$ is sufficiently large, for some constant $0<\mathfrak{c}<1,$ 
\begin{equation}\label{eq_mathsfc1control}
\mathsf{C}_1  \leq \mathfrak{c}.
\end{equation}
Combining (\ref{eq_decompositiondeomdeomdeodemdem}), (\ref{eq_mathsfc2control}) and (\ref{eq_mathsfc1control}), we conclude that  
\begin{gather*}
    |m_{1n}-m_1|=\rO(n^{c_2 \epsilon}\Psi),
\end{gather*}
which contradicts (\ref{eq_tobeusedasacontradiction}) since $c_2<1$ is sufficiently small. This completes our proof for each fixed $z.$ For uniformity in $z,$ we can follow a standard lattice argument as discussed below (\ref{eq_latticeargumentbelow}). This finishes the proof. The discussion for $m_Q$ follows from an analogous discussion with the help of (\ref{eq_m2uboundeddecomposition}) and (\ref{eq: est of m_1 and m_2}). 
\end{proof}

\begin{remark}\label{remk_zibound}
Two remarks are in order. First, it is easy to see that repeating the proof of Lemma \ref{lem: est for Im m_1}, we can prove the results for all $\eta$ as specified in (\ref{eq_spectraldomaintwo}). Second, we remark that combining (\ref{eq_zibound}) and Lemma \ref{lem: est for Im m_1}, when $z=E+\ri \eta_0 \in \mathbf{D}_b^\prime,$ we have that conditional on $\Omega$ 
\begin{equation*}
Z_i \prec  \frac{1}{n \eta_0}. 
\end{equation*}
\end{remark}

\quad Armed with the above discussions and results, following the strategies of Lemmas 5.6 and 5.7 of \cite{lee2016extremal} or \cite{Kwak2021},  we prove Lemmas \ref{lem: first bound for m_1n-m_1} and \ref{lem: second bound for m_1n-m_1} using similar arguments as in Lemma \ref{lem: est for Im m_1}. Due to similarity, we only provide the key ingredients. 

\begin{proof}[\bf Proof of Lemmas \ref{lem: first bound for m_1n-m_1} and \ref{lem: second bound for m_1n-m_1}] Due to similarity, we focus our proof on Lemma \ref{lem: first bound for m_1n-m_1} and briefly mention that of Lemma \ref{lem: second bound for m_1n-m_1} in the end. Due to similarity, we only explain $|m_{1n}(z)-m_1(z)|.$

\quad The proof is similar to that of Lemma \ref{lem: est for Im m_1} and we prove by contradiction. We also restrict ourselves on the event $\Xi_1$ in Lemma \ref{lem: est for Im m_1}. We assume that $|m_{1n}(z)-m_1(z)|>n^{\epsilon}(n\eta_0)^{-1}$. To see a contraction,  in addition to the the arguments of Lemma \ref{lem: est for Im m_1}, we need to provide a finer control for $\Psi$ since in the current proof it depends on $\eta$ instead of $\eta_0.$ Note that
\begin{gather}\label{eq_modificationkeykey}
    \begin{split}
       n^{c_2 \epsilon} \Psi&=n^{c_2 \epsilon}\sqrt{\frac{|\operatorname{Im}m_1-\operatorname{Im}m_{1n}+\operatorname{Im}m_{1n}|}{n\eta}}+n^{c_2 \epsilon} \frac{1}{n\eta}\\
        &\le   n^{c_2 \epsilon} \sqrt{\frac{|\operatorname{Im}m_1-\operatorname{Im}m_{1n}|}{n\eta}} +  n^{c_2 \epsilon} \sqrt{\frac{\operatorname{Im}m_{1n}}{n\eta}}+n^{c_2 \epsilon} \frac{1}{n\eta}\\
        &=\ro(|m_1-m_{1n}|),
    \end{split}
\end{gather}
where in the last step we used (\ref{eq_zopointrate11}) and the assumption  $|m_{1n}(z)-m_1(z)|>n^{\epsilon}(n\eta_0)^{-1} \gg (n \eta)^{-1}$ when $n^{-1/2+\epsilon_d} \leq \eta \leq n^{-1/(d+1)+\epsilon_d}.$ Replacing $  n^{c_2 \epsilon} \Psi$ with $\ro(|m_1-m_{1n}|)$ in the arguments between (\ref{eq_decompositiondeomdeomdeodemdem}) and (\ref{eq_mathsfc1control}), we find that $|m_{1n}-m_1|=\ro(|m_{1n}-m_1|)$ which is a contraction. This proves the result for each fixed $z.$ The uniformity follows from the same lattice argument as mentioned in the end of the proof of Lemma \ref{lem: est for Im m_1}. 

The proof of Lemma \ref{lem: second bound for m_1n-m_1} is similar. We also prove by contradiction and assume that $n^{\epsilon}(n\eta_0)^{-1}<|m_1(z)-m_{1n}(z)|\le n^{\epsilon+\epsilon_d}(n\eta_0)^{-1}.$  Under this assumption, together with (\ref{eq_zopointrate}) and (\ref{eq_oneregimeedgecontrol}), we find that (\ref{eq_modificationkeykey}) still holds true. The rest of the arguments are similar and we omit the details.

\end{proof}

%\begin{remark}\label{rem_extremevalue_appendix}
%{\color{red} add remark on exponential case and other cases}
%\end{remark}

\section{Locations for extreme eigenvalues and proof of the main results}\label{sec_locationofeigenvalues}
In this section, we study the first order convergent limits of the largest eigenvalues of $Q,$ i.e., $\lambda_1(Q).$ In Section \ref{sec_sub_unbounded1st}, we investigate the case when $\{\xi_i^2\}$ have unbounded support as in (i) of Assumption \ref{assum_D}. In Section \ref{sec_sub_bounded1st}, we study the bounded support case as in (ii) of Assumption \ref{assum_D}.

\subsection{The unbounded support case}\label{sec_sub_unbounded1st}
%{\color{red} [need to be careful here, separate the polynomial case and exponential case; mention that we focus on the polynomial case. ] [Also mention we focus on the i.i.d. case and the elliptical case can be handled similarly. ]} 

In order to quantify the location of $\lambda_1 \equiv \lambda_1(Q),$ we need to introduce several auxiliary quantities. Recall $\mu_1$ defined in (\ref{eq: def of mu_1}). Similarly, we denote $\mu_2$ by replacing $\xi^2_{(1)}$ with $\xi^2_{(2)}$ in (\ref{eq: def of mu_1}).  Moreover,  for $d_1$ and the sufficiently small constant $\epsilon>0$ in (\ref{eq_firstddefinition}), we denote 
\begin{equation}\label{eq_keylocationdefinition}
\mu_1^{\pm}:=\mu_1 \pm n^{-1/2+2 \epsilon} d_1 ,
%\left( \log n \right)^{\alpha \delta(\alpha-4)-1} n^{\epsilon},
\end{equation}
and recall that 
\begin{equation}\label{eq_qremoveonecolumn}
Q^{(1)}:=Q-\mathbf{y}_{(1)}\mathbf{y}^{*}_{(1)},
\end{equation}
where $\mathbf{y}_{(1)}$ is the column of $Y$ in (\ref{eq_datamatrix}) associated with $\xi^2_{(1)}$. Accordingly, we denote the largest eigenvalue of $Q^{(1)}$ as $\lambda_1^{(1)} \equiv \lambda_1(Q^{(1)}).$ Throughout this section, we shall prove Figure \ref{fig_locationtizubounded} so that the location of $\lambda_1$ can be quantified with high probability on the event $\Omega$. 

\begin{figure}[ht]
\centering
\begin{tikzpicture}
\foreach \x/\y in {0/\lambda_2, 0.8/\lambda_1^{(1)}, 1.5/\mu_2, 3/\mu_1^-, 4/\mu_1, 5/\mu_1^+}
 \draw[ultra thick] (\x,0.1) -- (\x,-0.1) node[below]{$\y$};
\draw[ultra thick] (-0.5,0)--(6,0);
\foreach \mycoord in {(3,0)}
    \draw [mybrace] \mycoord -- node[above, yshift=2mm]{$\lambda_1$ is here} ++(2,0); 
\end{tikzpicture}
\caption{Location of the largest eigenvalue of $Q$.} \label{fig_locationtizubounded} 
\end{figure}
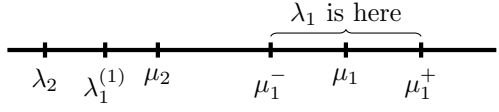
%\begin{tikzpicture}
%\def\labelz{{"Judicial Independence","Long Label 2","Label 3","$\mu_1^+$"}}
%\draw [<->,>=triangle 45] (0,0) -- (10,0);
%\foreach \x in {1,...,4} {%
%        \draw [<->,>=stealth] (2*\x,-.4) -- (2*\x,.4);
%        \node at (2*\x,2.4-.4*\x){\pgfmathparse{\labelz[\x-1]}\pgfmathresult};
%}
%\end{tikzpicture}

More formally, the main result is summarized in Proposition \ref{lem: eigenvalue rigidity} below. 
\begin{proposition}\label{lem: eigenvalue rigidity}
Suppose Assumptions \ref{assum_model}, \ref{assumption_techincial} and (i) of Assumption \ref{assum_D} hold. For some sufficiently small constant $\epsilon>0$ and $\mu_1^{\pm}$ defined in (\ref{eq_keylocationdefinition}), condition on the probability event $\Omega$ in Lemma \ref{lem_probabilitycontrol}, with high probability, we have that 
\begin{equation*}
\lambda_1 \in [\mu_1^-, \mu_1^+]. 
\end{equation*}
% Then for some small enough positive $c$ such that $ 0\le c<|2/\alpha-1/2|,$ we have that 
%\begin{equation}
%|\lambda_1-\mu_1|\prec n^{-1/2+2/\alpha+c}\log^{\alpha\delta(\alpha-4)-1}n
%\end{equation}
%where $\delta(x)$ is the Dirac function with $\delta(0)=1$.
\end{proposition}

%\begin{remark}\label{rem: rescaling}
%{\color{red}[revise here]} A few remarks are in order. First, Proposition \ref{lem: eigenvalue rigidity} unify cases (a) and (b) in (i) of Assumption \ref{assum_D}. More specifically, as will be seen in the proof, case (b) corresponds to that $\alpha=\infty.$ Second, when $2<\alpha<4,$ we find that $\alpha\in(2,4]$, $n^{-1/2+2/\alpha}\log^{\alpha\delta(\alpha-4)-1}n\gg1.$ However,  comparing the original rate of $\mu_1$ which is roughly $n^{2/\alpha}$ with the distance of $|\mu_1-\mu_2|=\rO(n^{2/\alpha})$ with high probability (see (\ref{def1})), this result shows that the leading order part in (\ref{eq_keylocationdefinition}) is still $\mu_1$  and also significantly narrow the range around $\mu_1.$ 
%
%{\color{red} [revised this paragraph later.]
% Actually, since the typical rate of the extreme value is pretty large of the heavy-tailed distributed random variables, we usually use $a^2_{np}\log^{\alpha\delta(\alpha-4)}n$ rather than $n$ to scale $\xi^2$. Then one can see that 
%\[
%na^{-2}_{np}\log^{-\alpha\delta(\alpha-4)}n|\lambda_1-\mu_1|\prec n^{1/2-2/\alpha+c}\log^{-1}n\prec 1,
%\]
%which gives the convergence for $\alpha\in(2,4]$.}
%\end{remark}

%In can be seen that for case , we then need a stronger re-scaling factor for this case. One can refer to  at the end of this subsection for more details.

%Before proceeding to the proof of Proposition \ref{lem: eigenvalue rigidity}, we first provide an useful estimate. 

We now proceed to the proof of Proposition \ref{lem: eigenvalue rigidity} following the structure described in Figure \ref{fig_locationtizubounded}. 
\begin{proof}[\bf Proof of Proposition \ref{lem: eigenvalue rigidity}]
%Due to similarity, we focus on the setting $2\leq \alpha<\infty$ {\color{red} [revised] and will only discuss the case $\alpha=\infty$ briefly in the end. Almost mention how we condition on the probability event}. 
Due to similarity, we focus on the study of the separable covariance i.i.d. model as in Case (2) of Assumption \ref{assum_model} when $\xi^2$ satisfies (\ref{ass3.1}). The main differences from the other cases will be explained in the end of the proof. 

In what follows, we restrict the discussion on the probability event $\Omega$ in Lemma \ref{lem_probabilitycontrol}.
By Weyl's inequality, we have that $\lambda_2 \leq \lambda_1^{(1)}.$ Moreover, by (\ref{eq_keyused}), we see that with high probability $\lambda_1^{(1)}<\mu_2.$ The rest of the proof leaves to prove that the following two claims: 
\begin{equation}\label{eq_mu1mu2bound}
\mu_1-\mu_2 \geq n^{1/\alpha} \log^{-1}n,
\end{equation}
and for $Q^{(1)}$ in (\ref{eq_qremoveonecolumn}) and 
\begin{equation}\label{eq_defnmlambda}
M(\lambda)=1+\mathbf{y}^{*}_{(1)}G_1^{(1)}(\lambda)\mathbf{y}_{(1)}, \ G^{(1)}_1(\lambda):=(Q^{(1)}-\lambda I)^{-1},
\end{equation}
$M(\lambda)$ changes sign with high probability at $\mu_1^-$ and $\mu_1^+.$  In fact, for $\lambda_1,$ it should satisfy the following equation with high probability
\begin{equation}\label{eq_masterequation}
{\rm det}(\lambda_1I-\mathbf{y}_{(1)}\mathbf{y}^{*}_{(1)}-Q^{(1)})=0\Rightarrow M(\lambda_1)=0,
\end{equation}
as long as $\lambda_1>\lambda_1^{(1)}.$ On the other hand, if $M(\lambda)$ changes sign at $\mu_1^{\pm},$ by continuity,  there must at least be an eigenvalue of $Q$ in the interval $[\mu_1^-, \mu_1^+].$ If (\ref{eq_mu1mu2bound}) holds, combining the above arguments, we see that the only possibility is $\lambda_1$ and it is also true that $\lambda_1>\lambda_1^{(1)}.$

We first justify (\ref{eq_mu1mu2bound}). Recall that $\mu_1$ is defined in (\ref{eq: def of mu_1}) according to 
\begin{equation*}
1+(\xi^2_{(1)}+d_{1})m_{1n}(\mu_1)=0.
\end{equation*}
Together with (\ref{eq_functionFequal}) and (\ref{eq:F(m,z)}), we readily obtain that 
\begin{equation*}
1=\frac{1}{n}\sum_{i=1}^p\frac{\sigma_i}{\frac{\mu_1}{\xi^2_{(1)}+d_{1}}-\frac{\sigma_i}{n}\sum_{j=1}^n\frac{\xi_{(j)}^2}{\xi^2_{(1)}+d_{1}-\xi^2_{(j)}}}.
\end{equation*}
Recall (\ref{e2_definition}). Using the definition of $d_1$ and \eqref{def1}, we see that on $\Omega,$ for some constant $C>0$
%We will show that $M(\lambda)$ changes sign when $\lambda$ grows from $\mu_1-n^{-1/2+2/\alpha+c}\log^{\alpha\delta(\alpha-4)-1}n$ to $\mu_1+n^{-1/2+2/\alpha+c}\log^{\alpha\delta(\alpha-4)-1}n$.
%It is essential to find an approximation of $\mathbf{y}^{*}_{(1)}(\mathcal{Q}^{(1)}-\lambda I)^{-1}\mathbf{y}_{(1)}$ at the first step. Since we want to obtain a bound for $\|(\mathcal{Q}^{(1)}-\lambda I)^{-1}\|_2$, we need a upper bound for the largest eigenvalue of $\mathcal{Q}^{(1)}$. Actually, from the procedures we reach Lemma \ref{lem: upper bound for eigenvalues}, a direct upper bound for $\|\mathcal{Q}^{(1)}\|_2$ is $\mu_2$. Recall $\mu_2$ satisfies the equation,
%
%where $m^{(1)}_{1n}(z)$ is the solution of 
%\[
%m^{(1)}_{1n}(z)=\frac{1}{n}\sum_{i=1}^p\frac{\sigma_i}{-z+\frac{\sigma_i}{n}\sum_{j=2}^n\frac{\xi^2_{(j)}}{1+\xi^2_{(j)}m^{(1)}_{1n}(z)}},\quad z\in\mathbb{C}^{+}.
%\]
%
%Now, we consider the gap between $\mu_1$ and $\mu_2$. By definition of $\mu_2$,
%together with $d_{2}$ and \eqref{def1},
\begin{equation}\label{eq: 2.2 beta=2}
\begin{split}
\frac{1}{n}\sum_{j=1}^n\frac{\xi_{(j)}^2}{\xi^2_{(1)}+d_{1}-\xi^2_{(j)}}&=\frac{1}{n}\frac{\xi_{(1)}^2}{d_{1}}+\frac{1}{n}\sum_{j=2}^{n}\frac{\xi_{(j)}^2}{\xi^2_{(1)}+d_{1}-\xi^2_{(j)}}\\
   &\le \frac{Cn^{\epsilon}\log n}{n}+\frac{1}{n}\sum_{j=2}^{n}\frac{\xi_{(j)}^2}{\xi^2_{(1)}+d_{1}-\xi^2_{(2)}}\\
   &\le\frac{Cn^{\epsilon}\log n}{n}+\frac{C\log n}{n^{1/\alpha}}=C e.
\end{split}
\end{equation}

Using the definition of $\varphi$ in (\ref{eq_defnvarphi}), 
the above arguments imply that on $\Omega$
\begin{equation}\label{eq_mu1part}
\frac{\mu_1}{\xi_{(1)}^2+d_1}=\varphi+\rO(e).  
\end{equation}

By an analogous argument, we have that for some constants $C_k>0, k=1,2,3,$ 
%
%
%we conclude that 
%\[
%\frac{\mu_2}{\xi^2_{(2)}+d_{2}}=C\phi\bar{\sigma}.
%\]
%Moreover, one has
\[
\begin{split}
&\frac{1}{n}\left(\sum_{j=1}^n \frac{\xi^2_{(j)}}{\xi^2_{(1)}+d_{1}-\xi^2_{(j)}}-\sum_{j=2}^n\frac{\xi^2_{(j)}}{\xi^2_{(2)}+d_{1}-\xi^2_{(j)}}\right)\\
&=\frac{1}{n}\left(\frac{\xi^2_{(1)}}{d_{1}}+\sum_{j=2}^n\frac{\xi^2_{(j)}}{\xi^2_{(1)}+d_{1}-\xi^2_{(j)}}-\sum_{j=2}^n \frac{\xi^2_{(j)}}{\xi^2_{(2)}+d_{1}-\xi^2_{(j)}}\right)\\
&=C_1e-\frac{1}{n}\sum_{j=2}^n\frac{\xi^2_{(j)}(\xi^2_{(1)}-\xi^2_{(2)})}{(\xi^2_{(1)}+d_{1}-\xi^2_{(j)})(\xi^2_{(2)}+d_{1}-\xi^2_{(j)})}\\
&\ge C_1 e-\frac{1}{n}\sum_{j=2}^n \frac{\xi^2_{(j)}}{\xi^2_{(2)}+d_{1}-\xi^2_{(j)}}\\
&\geq C_2 e,
\end{split}
\]
and on the other hand
\[
\begin{split}
    &\frac{1}{n}\left(\sum_{j=1}^n \frac{\xi^2_{(j)}}{\xi^2_{(1)}+d_{1}-\xi^2_{(j)}}-\sum_{j=2}^n \frac{\xi^2_{(j)}}{\xi^2_{(2)}+d_{1}-\xi^2_{(j)}}\right)\\
    &\le \frac{1}{n}\left(\frac{\xi^2_{(1)}}{d_{1}}+\sum_{j=2}\frac{\xi^2_{(j)}}{\xi^2_{(1)}+d_{1}-\xi^2_{(j)}}+\sum_{j=2}\frac{\xi^2_{(j)}}{\xi^2_{(2)}+d_{1}-\xi^2_{(j)}}\right)\\
    & \leq C_3e.
\end{split}
\]
Using the above control, the definition of $\mu_2$ and a discussion similar to (\ref{eq_mu1part}), we can prove that 
\begin{equation}\label{eq_mu2part}
\frac{\mu_2}{\xi_{(2)}^2+d_1}=\phi \bar{\sigma}+\rO(e).  
\end{equation}

Combining (\ref{eq_mu1part}) and (\ref{eq_mu2part}), we immediately see that  
\begin{equation}\label{eq_bbbbb}
\frac{\mu_1}{\xi^2_{(1)}+d_{1}}-\frac{\mu_2}{\xi^2_{(2)}+d_{1}}=\rO(e).
\end{equation}
This implies that 
\[
\begin{split}
    \mu_1-\mu_2&=(\xi^2_{(1)}+d_{1})\left(\frac{\mu_1}{\xi^2_{(1)}+d_{1}}-\frac{\mu_2}{\xi^2_{(2)}+d_{1}}\right)+\mu_2\left(\frac{\xi^2_{(1)}+d_{1}}{\xi^2_{(2)}+d_{1}}-1\right)\\
    &=(\xi^2_{(1)}+d_{1})\left(\frac{\mu_1}{\xi^2_{(1)}+d_{1}}-\frac{\mu_2}{\xi^2_{(2)}+d_{1}}\right)+\frac{\mu_2}{\xi^2_{(2)}+d_1}\left(\xi^2_{(1)}-\xi^2_{(2)}\right)\\
    &\ge n^{2/\alpha}\log^{-1}n,
\end{split}
\]
where in the third step we used (\ref{eq_bbbbb}), (\ref{def1}) and the definition of $e_2$ in (\ref{e2_definition}). This completes the proof of (\ref{eq_mu1mu2bound}).

Next, we will show that 
\begin{equation*}
M(\mu_1^-)<0,  \ M(\mu_1^+)>0. 
\end{equation*}
Due to similarity, in what follows, we focus on the first inequality. Note that 
\begin{equation*}
\mathbf{y}_1^* G_1^{(1)}(\mu_1^-) \mathbf{y}_1= \xi_{(1)}^2  \ub_1^*\Sigma^{1/2} G_1^{(1)}(\mu_1^-) \Sigma^{1/2} \ub_1. 
\end{equation*}
Moreover, recall that $m_1^{(1)}(\mu_1^-)=n^{-1} \operatorname{tr}(\Sigma^{1/2}G_1^{(1)}(\mu_1^-)\Sigma^{1/2}).$ Then according to Lemma \ref{lem:large deviation}, we have that 
\begin{align}\label{eq_conconconone}
\mathbf{y}_1^* G_1^{(1)}(\mu_1^-) \mathbf{y}_1& =\xi^2_{(1)} m_1^{(1)}(\mu_1^-)+\rO_{\prec} \left( \frac{\xi^2_{(1)}}{n}\| G_1^{(1)}(\mu_1^-)\|_F \right)  \nonumber \\
& =\xi^2_{(1)} m_1^{(1)}(\mu_1^-)+\rO_{\prec} \left( \frac{\xi^2_{(1)}}{n^{1/2+1/\alpha}}\right), 
\end{align}
where in the second step we used (\ref{eq_mu1mu2bound}) and the fact $\mu_2>\lambda_1^{(1)}$.  Moreover, for some sufficiently small constant $\epsilon_0>0$ and $z_0=\mu_1^-+\mathrm{i} n^{-1/2-\epsilon_0},$ we can decompose that  
%{\color{red} cannot obtain the control below using the equation above.}
%{\color{blue}
% Observe that 
%\[
%\begin{split}
%   \mathbf{y}^{*}_{(1)}(\mathcal{Q}^{(1)}-\lambda I)^{-1}\mathbf{y}_{(1)}-\xi^2_{(1)}m_1^{(1)}(\lambda)\prec\frac{\xi^2_{(1)}}{n}\|(\mathcal{Q}^{(1)}-\lambda I)^{-1}\|_F\prec \frac{\xi^2_{(1)}}{n^{1/2+2/\alpha}},
%\end{split}
%\]
%where we use the lower bound for $\mu_1-\mu_2$. }
%
%
%The discussion contains two steps. {\color{red} from here} Now, we set the target interval around $\mu_1$ at $\mu_1-n^{-1/2+2/\alpha+c}\log^{\alpha\delta(\alpha-4)-1}n$ and $\mu_1+n^{-1/2+2/\alpha+c}\log^{\alpha\delta(\alpha-4)-1}n$ which is strictly on the right of $\mu_2$ by the above estimation.  Let $\lambda=\mu_1-n^{-1/2+2/\alpha+c}\log^{\alpha\delta(\alpha-4)-1}n$.
%
%Denote $z=\lambda+i n^{-1/2-\epsilon}$, we have 
\begin{align}\label{eq_bigexpansion}
m_1^{(1)}(\mu_1^-)& =\left[m_1^{(1)}(\mu_1^-)-m_1^{(1)}(z_0) \right]+\left[m_1^{(1)}(z_0)-m_{1n}^{(1)}(z_0)\right]+\left[m_{1n}^{(1)}(z_0)-m_{1n}^{(1)}(\mu_1^-)\right]+m_{1n}^{(1)}(\mu_1^-) \nonumber \\
&=\mathsf{P}_1+\mathsf{P}_2+\mathsf{P}_3+m_{1n}^{(1)}(\mu_1^-).
\end{align}
First, by Theorem \ref{thm_unboundedcaselocallaw}, we have that $\mathsf{P}_2 \prec n^{-1/2-2/\alpha}$. Second, let $\{\mathbf{v}^{(1)}_i\}$ be the eigenvectors of $Q^{(1)}$ associated with the eigenvalues $\{\lambda_i^{(1)}\},$ then we have that
%
%
%
%Suppose $\lambda^{(1)}_i$ be the ordered eigenvalues of $\mathcal{Q}^{(1)}$ with the corresponding eigenvectors $\mathbf{v}^{(1)}_i$, we may rewrite the first distinction as
\[
\begin{split}
|\mathsf{P}_1|&\le\frac{1}{n}\sum_{i=1}^p|T\mathbf{v}^{(1)}_i|^2\left|\frac{1}{\lambda^{(1)}_i-\mu_1^-}-\frac{1}{\lambda^{(1)}_i-z_0}\right|\\
  &=\frac{1}{n}\sum_{i=1}^p|T\mathbf{v}^{(1)}_i|^2\left|\frac{\mathrm{i}n^{-1/2-\epsilon_0}}{(\lambda^{(1)}_i-\mu_1^-)(\lambda^{(1)}_i-z_0)}\right|\\
  &\le\frac{1}{n}\sum_{i=1}^p|T\mathbf{v}^{(1)}_i|^2\Big|\frac{\mathrm{i}n^{-1/2-\epsilon_0}+\rO_{\prec}(n^{-2/\alpha-1/2}\log^2n)}{|\lambda^{(1)}_i-z_0|^2}\Big|\\
  &\prec\operatorname{Im}(m_1^{(1)}(z_0))\times \rO_{\prec}(n^{-1/2-\epsilon_0})
%  &\prec\operatorname{Im}(m_{1n}^{(1)}(z_0))\times O_{\prec}(n^{-1/2-\epsilon})+O_{\prec}(n^{-1-4/\alpha})\\
\prec n^{-1-1/\alpha},
\end{split}
\]
where in the third step we used (\ref{eq_mu1mu2bound}) and the fact $\mu_2>\lambda_1^{(1)}$ and in the last step we used Lemma \ref{lem: basic bounds}, (\ref{eq_mu1part}) and (\ref{def1}). Third, according to Lemma \ref{lem_solutionsystem},  we can decompose that 
\[
\begin{split}
   \mathsf{P}_3 &=\frac{1}{n}\sum_{i=1}^p\frac{\sigma_i}{-z_0(1+\frac{\sigma_i}{n}\sum_{j=2}^n\frac{\xi^2_{(j)}}{-z_0(1+m^{(1)}_{1n}(z_0)\xi^2_{(j)})})}-\frac{1}{n}\sum_{i=1}^p\frac{\sigma_i}{-\mu_1^-(1+\frac{\sigma_i}{n}\sum_{j=2}^n \frac{\xi^2_{(j)}}{-\mu_1^-(1+m^{(1)}_{1n}(\mu_1^-)\xi^2_{(j)})})}\\
    &=\left[\frac{1}{n}\sum_{i=1}^p\frac{\sigma_i}{-z_0(1+\frac{\sigma_i}{n}\sum_{j=2}^n \frac{\xi^2_{(j)}}{-z_0(1+m^{(1)}_{1n}(z_0)\xi^2_{(j)})})}-\frac{1}{n}\sum_{i=1}^p\frac{\sigma_i}{-z_0(1+\frac{\sigma_i}{n}\sum_{j=2}^n\frac{\xi^2_{(j)}}{-\mu_1^-(1+m^{(1)}_{1n}(\mu_1^-)\xi^2_{(j)})})} \right]\\
    &+\left[\frac{1}{n}\sum_{i=1}^p\frac{\sigma_i}{-z_0(1+\frac{\sigma_i}{n}\sum_{j=2}^n \frac{\xi^2_{(j)}}{-\mu_1^-(1+m^{(1)}_{1n}(\mu_1^-)\xi^2_{(j)})})}-\frac{1}{n}\sum_{i=1}^p\frac{\sigma_i}{-\mu_1^-(1+\frac{\sigma_i}{n}\sum_{j=2}^n\frac{\xi^2_{(j)}}{-\mu_1^-(1+m^{(1)}_{1n}(\mu_1^-)\xi^2_{(j)})})} \right]\\
    &:=\mathcal{M}^{(1)}_{31}+\mathcal{M}^{(1)}_{32}.
\end{split}
\]
Note that according to (\ref{eq_mu1part}), (\ref{eq_mu1mu2bound}), (\ref{def1}) and Lemma \ref{lem: basic bounds}, we conclude that with high probability $|1+\sigma_i m_{2n}^{(1)}(z_0)|, |1+\sigma_i m_{2n}^{(1)}(\mu_1^-)|  \sim 1.$ For $\mathcal{M}^{(1)}_{31}$, using the definitions in  (\ref{eq_systemequationsm1m2}) and the above bounds,  we have that with high probability
\begin{align} \label{eq_m11control}
    \mathcal{M}^{(1)}_{31}&=\frac{1}{n}\sum_{i=1}^p \frac{\sigma^2_i}{-z_0(1+\sigma_im^{(1)}_{2n}(z))(1+\sigma_im^{(1)}_{2n}(\mu_1^-))}\big(m^{(1)}_{2n}(\mu_1^-)-m^{(1)}_{2n}(z)\big)\\
    &=\rO(|z_0|^{-1})\times\frac{1}{n}\sum_{j=2}^n \left(\frac{\xi^2_{(j)}}{-\mu_1^-(1+\xi^2_{(j)}m_{1n}^{(1)}(\mu_1^-))}-\frac{\xi^2_{(j)}}{-z(1+\xi^2_{(j)}m_{1n}^{(1)}(z_0))}\right) \nonumber \\
    &=\rO(|z_0|^{-1})\times\frac{1}{n}\sum_{j=2}^n \left(\frac{\xi^2_{(j)}}{-\mu_1^-(1+\xi^2_{(j)}m^{(1)}_{1n}(\mu_1^-))}-\frac{\xi^2_{(j)}}{-\mu_1^-(1+\xi^2_{(j)}m^{(1)}_{1n}(z_0))}\right) \nonumber \\
    &+\rO(|z_0|^{-1})\times\frac{1}{n}\sum_{j=2}^n \left(\frac{\xi^2_{(j)}}{-\mu_1^-(1+\xi^2_{(j)}m^{(1)}_{1n}(z))}-\frac{\xi^2_{(j)}}{-z(1+\xi^2_{(j)}m_{1n}^{(1)}(z_0))}\right) \nonumber \\ 
    &=\rO(|z_0|^{-1})\times\frac{1}{n}\sum_{j=2}^n \frac{\xi^4_{(j)}(m^{(1)}_{1n}(z_0)-m_{1n}^{(1)}(\mu_1^-))}{-\mu_1^-(1+\xi^2_{(j)}m^{(1)}_{1n}(\mu_1^-))(1+\xi^2_{(j)}m^{(1)}_{1n}(z_0))}+\rO(|z_0|^{-1})\times\frac{1}{n}\sum_{j=2}^n \frac{\xi^2_{(j)}(\mu_1^--z_0)}{z\mu_1^-(1+\xi^2_{(j)}m^{(1)}_{1n}(z_0))} \nonumber \\
    &=\ro(1)\times(m^{(1)}_{1n}(z_0)-m^{(1)}_{1n}(\mu_1^-))+\rO(n^{-2/\alpha-1/2-\epsilon_0}). \nonumber 
\end{align}
where in the second to last step we used Lemma \ref{lem: basic bounds} and in the last step we used Lemma \ref{lem: basic bounds}  and (\ref{def1}). This implies with high probability 
\begin{equation*}\label{control_m11}
\mathcal{M}_{31}^{(1)}=\rO \left( n^{-2/\alpha-1/2-\epsilon_0} \right).  
\end{equation*}

Similarly, for $\mathcal{M}^{(1)}_{32}$, we have that 
\begin{equation}\label{eq_controlm12}
\begin{split}
    \mathcal{M}^{(1)}_{32}&=\frac{1}{n}\sum_{i=1}^p \frac{\sigma_i(z_0-\mu_1^-)}{z_0\mu_1^-(1+\sigma_im^{(1)}_{2n}(\mu_1^-))}\\
    &=\rO(n^{-2/\alpha-1/2-\epsilon_0}).
\end{split}
\end{equation}
Combining the above arguments, we have that $\mathsf{P}_3=\rO\left( n^{-2/\alpha-1/2-\epsilon_0} \right). $

Inserting the bounds of $\mathsf{P}_k, 1 \leq k \leq 3$ into (\ref{eq_bigexpansion}), we conclude that 
\[
|m^{(1)}_1(\mu_1^-)-m^{(1)}_{1n}(\mu_1^-)|\prec n^{-1/\alpha-1/2-\epsilon_0}.
\]
Together with (\ref{eq_conconconone}) and (\ref{def1}), it yields that  
\begin{equation}\label{eq_M1reducedcase}
  M(\mu_1^-)=1+(\xi^2_{(1)}+d_1)m^{(1)}_{1n}(\mu_1^-)+\rO_{\prec}(n^{-1/2-\epsilon_0}).
\end{equation}
%{\color{red}[from here]}where we use the fact $m^{(1)}_{1n}(\lambda)\prec \|(\mathcal{Q}^{(1)}-\lambda I)^{-1}\|_2\prec n^{-2/\alpha}$.

In what follows, we study $1+(\xi^2_{(1)}+d_1)m^{(1)}_{1n}(\mu_1^-)$. We rewrite that
\begin{equation}\label{eq_reducedcontrol}
\begin{split}
    1+(\xi^2_{(1)}+d_1)m^{(1)}_{1n}(\mu_1^-)&=1+(\xi^2_{(1)}+d_1)m^{(1)}_{1n}(\mu_1)-(\xi^2_{(1)}+d_1)\big(m^{(1)}_{1n}(\mu_1)-m^{(1)}_{1n}(\mu_1^-)\big).
\end{split}
\end{equation}
Using the definition for $\mu_1$ that $1+(\xi^2_{(1)}+d_1)m_{1n}(\mu_1)=0$ and the definitions in (\ref{eq_systemequationsm1m2}), by a discussion similar to (\ref{eq_m11control}),  we have that for some constant $C>0$
\begin{equation}\label{eq: m^s_1_1n-m_1n}
\begin{split}
    &1+(\xi^2_{(1)}+d_1)m^{(1)}_{1n}(\mu_1)\\
    &=1+(\xi^2_{(1)}+d_1)m_{1n}(\mu_1)+(\xi^2_{(1)}+d_1)\big(m_{1n}^{(1)}(\mu_1)-m_{1n}(\mu_1)\big)\\
    &=(\xi^2_{(1)}+d_1)\frac{1}{n}\sum_{i=1}^p\left(\frac{\sigma_i}{-\mu_1(1+\sigma_i m_{2n}^{(1)}(\mu_1))}-\frac{\sigma_i}{-\mu_1(1+\sigma_i m_{2n}(\mu_1))}\right)\\
    &=(\xi^2_{(1)}+d_1) \left[\frac{1}{n}\sum_{i=1}^p\frac{\sigma^2_i}{-\mu_1(1+\sigma_i m_{2n}^{(1)}(\mu_1))(1+\sigma_i m_{2n}(\mu_1))}\right]\big(m_{2n}(\mu_1)-m_{2n}^{(1)}(\mu_1)\big)\\
    & \leq C \frac{(\xi^2_{(1)}+d_1)}{\mu_1}\times \left|\frac{1}{n}\sum_{j=1}^n\frac{\xi^2_{(j)}}{-\mu_1(1+\xi^2_{(j)}m_{1n}(\mu_1))}-\frac{1}{n}\sum_{j=2}^n\frac{\xi^2_{(j)}}{-\mu_1(1+\xi^2_{(j)}m^{(1)}_{1n}(\mu_1))}\right|\\
    &\leq C \frac{(\xi^2_{(1)}+d_1)}{\mu_1}\times\xi^2_{(2)}|m_{1n}^{(1)}(\mu_1)-m_{1n}(\mu_1)|+n^{-1}.
\end{split}
\end{equation}
where in the last step we again used (\ref{def1}). This yields that 
\begin{equation*}
|(\xi^2_{(1)}+d_1)\big(m_{1n}^{(1)}(\mu_1)-m_{1n}(\mu_1)\big)| \leq C \frac{(\xi^2_{(1)}+d_1)}{\mu_1}\times\xi^2_{(2)}|m_{1n}^{(1)}(\mu_1)-m_{1n}(\mu_1)|+n^{-1}, 
\end{equation*}
which implies that $(\xi^2_{(1)}+d_1)\big(m^{(1)}_{1n}(\mu_1)-m_{1n}(\mu_1)\big)=\rO(n^{-1})$. Together with (\ref{eq_reducedcontrol}), we have 
\begin{equation}\label{eq_finalpartone}
1+(\xi^2_{(1)}+d_1)m^{(1)}_{1n}(\mu_1^-)=-(\xi^2_{(1)}+d_1)\big(m^{(1)}_{1n}(\mu_1)-m^{(1)}_{1n}(\mu_1^-)\big)+\rO(n^{-1}).
\end{equation}

Recall that we have proved the facts that $\mu_1, \mu_1^->\lambda_1^{(1)},$ by Theorem \ref{thm_boundedcaselocallaw} and the monotonicity of $m_1^{(1)}$ outside the bulk,  the first term on the right-hand side of (\ref{eq_finalpartone}) is negative. In order to show $M(\mu_1^-)<0,$ in light of (\ref{eq_M1reducedcase}), it suffices to show that its magnitude is much larger than $\rO(n^{-1/2-\epsilon_0}).$  To see this, we decompose that
\[
\begin{split}
    &m^{(1)}_{1n}(\mu_1)-m^{(1)}_{1n}(\mu_1^-)\\
    &=\frac{1}{n}\sum_{i=1}^p\frac{\sigma_i}{-\mu_1(1+\sigma_im^{(1)}_{2n}(\mu_1))}-\frac{1}{n}\sum_{i=1}^p\frac{\sigma_i}{-\mu_1^-(1+\sigma_im^{(1)}_{2n}(\mu_1^-))}\\
    &=\left[\frac{1}{n}\sum_{i=1}^p\frac{\sigma_i}{-\mu_1(1+\frac{\sigma_i}{n}\sum_{j=2}^n\frac{\xi^2_{(j)}}{-\mu_1(1+\xi^2_{(j)}m_{1n}^{(1)}(\mu_1)})}-\frac{1}{n}\sum_{i=1}^p\frac{\sigma_i}{-\mu_1(1+\frac{\sigma_i}{n}\sum_{j=2}^n\frac{\xi^2_{(j)}}{-\mu_1(1+\xi^2_{(j)}m_{1n}^{(1)}(\mu_1^-))})} \right]\\
    &+\left[\frac{1}{n}\sum_{i=1}^p\frac{\sigma_i}{-\mu_1(1+\frac{\sigma_i}{n}\sum_{j=2}^n\frac{\xi^2_{(j)}}{-\mu_1(1+\xi^2_{(j)}m_{1n}^{(1)}(\mu_1^-))})}-\frac{1}{n}\sum_{i=1}^p\frac{\sigma_i}{-\mu_1^-(1+\frac{\sigma_i}{n}\sum_{j=2}^n\frac{\xi^2_{(j)}}{-\mu_1^-(1+\xi^2_{(j)}m_{1n}^{(1)}(\mu_1^-))})} \right]\\
    &:=\tilde{\mathcal{M}}_{11}^{(1)}+\tilde{\mathcal{M}}_{12}^{(1)}.
\end{split}
\]
Similar to the discussion of (\ref{eq_m11control}), we have that $\tilde{\mathcal{M}}_{11}^{(1)}$,
\[
\begin{split}
    \tilde{\mathcal{M}}_{11}^{(1)}&=\frac{1}{n}\sum_{i=1}^p\frac{\sigma^2_i\Big(\frac{1}{n}\sum_{j=2}^n\frac{\xi^2_{(j)}}{-\mu_1(1+\xi^2_{(j)}m_{1n}^{(1)}(\lambda))}-\frac{1}{n}\sum_{j=2}^n\frac{\xi^2_{(j)}}{-\mu_1(1+\xi^2_{(j)}m_{1n}^{(1)}(\mu_1))}\Big)}{-\mu_1(1+\frac{\sigma_i}{n}\sum_{j=2}^n\frac{\xi^2_{(j)}}{-\mu_1(1+\xi^2_{(j)}m_{1n}^{(1)}(\mu_1))})(1+\frac{\sigma_i}{n}\sum_{j=2}^n\frac{\xi^2_{(j)}}{-\mu_1(1+\xi^2_{(j)}m_{1n}^{(1)}(\mu_1^-))})}\\
    &=\rO(\frac{1}{\mu_1})\times\frac{1}{n}\sum_{j=2}^n\frac{\xi^4_{(j)}(m^{(1)}_{1n}(\mu_1)-m^{(1)}_{1n}(\mu_1^-))}{-\mu_1(1+\xi^2_{(j)}m_{1n}^{(1)}(\mu_1^-))(1+\xi^2_{(j)}m_{1n}^{(1)}(\mu_1))}\\
    &=\ro(1)\times(m^{(1)}_{1n}(\mu_1)-m^{(1)}_{1n}(\mu_1^-)).
\end{split}
\]
Moreover, similar to (\ref{eq_controlm12}), for $\tilde{\mathcal{M}}_{12}^{(1)}$ we have that with high probability
\[
\begin{split}
    \tilde{\mathcal{M}}_{12}^{(1)}&=\frac{1}{n}\sum_{i=1}^p\frac{\sigma_i(\mu_1-\mu_1^-)}{\mu_1\mu_1^-(1+\frac{\sigma_i}{n}\sum_{j=2}^n\frac{\xi^2_{(j)}}{-\mu_1(1+\xi^2_{(j)}m_{1n}^{(1)}(\mu_1^-))})(1+\frac{\sigma_i}{n}\sum_{j=2}^n\frac{\xi^2_{(j)}}{-\mu_1^-(1+\xi^2_{(j)}m_{1n}^{(1)}(\mu_1^-))})}\\
    & \asymp \frac{\mu_1-\mu_1^-}{\mu_1^-\mu_1}  \asymp n^{-1/2-1/\alpha+\epsilon}.
\end{split}
\]
This implies that 
\[
m^{(1)}_{1n}(\mu_1)-m^{(1)}_{1n}(\lambda) \asymp n^{-1/2-1/\alpha+\epsilon}. 
\]
Together with (\ref{eq_finalpartone}), the definition of $d_2$ and (\ref{def1}), we readily see that 
\begin{equation}\label{eq: lambda-lambda_(1)}
   1+(\xi^2_{(1)}+d_2)m^{(1)}_{1n}(\mu_1^-)) \asymp -n^{-1/2+\epsilon},
\end{equation}
which concludes the proof of $M(\mu_1^-)<0$ when $n$ is sufficiently large. Similarly, we can prove that $M(\mu_1^+)>0.$

Before concluding the proof, we briefly discuss the proof of the other cases. For the case (\ref{ass3.2}), the main difference is  
to utilize the second part of Theorem \ref{thm_unboundedcaselocallaw}. For elliptical model, we  repeat the proof verbatim with some minor modification. For example, in (\ref{eq: 2.2 beta=2}), we need to apply (\ref{def3}) and in (\ref{eq_mu1part}) we need to use the definition in (\ref{eq:F(m,z)1}) and for (\ref{eq_conconconone}) we need to use (2) of Lemma \ref{lem:large deviation}. We omit further details. 

\end{proof}
%\begin{remark}{\label{rem: rescaling}}
%
%\end{remark}

\subsection{The bounded support case}\label{sec_sub_bounded1st} 
In this section, we study the first order convergence of $\lambda_1$ under the assumptions of part (1) of Theorem \ref{thm_main_bounded}. The other parts will be discussed in Section \ref{sec_proofofmainresults}. The main result of this section can be summarized in the following proposition. As before, due to similarity, we only prove for the separable covariance i.i.d. data as in (2) of Assumption \ref{assum_model}. 

\begin{proposition}\label{prop_boundedsetting} Suppose the assumptions of part (1) of Theorem \ref{thm_main_bounded} hold, then conditional on the probability event as in Lemma \ref{localestimate2}, we have that 
   \begin{gather*}
       \left |\lambda_{1}-\left(\widehat{L}_{+}-\frac{  1-\phi \widehat{\varsigma}_3}{\widehat{\varsigma}_4}\frac{l-\xi^2_{(1)}}{l\xi^2_{(1)}}\right) \right |=\rO_{\mathbb{P}} \left[\frac{1}{n^{1/(d+1)}} \left(\frac{n^{3 \epsilon_d}}{n^{-1/(d+1)+1/2}}+\frac{\log n}{n^{1/(d+1)}} \right) \right].
    \end{gather*}
\end{proposition}

The proof of Proposition \ref{prop_boundedsetting} relies crucially on the following lemma whose justification will be provided in the end of this section.  

\begin{lemma}\label{lem_connection} Suppose the assumptions of Proposition \ref{prop_boundedsetting} hold. Recall $E_0$ defined in (\ref{eq: def of gamma}). Conditional on the probability event $\Omega$ as in Lemma \ref{localestimate2},  we have that
\begin{equation}\label{eq_holdspartoneoneone}
\lambda_1=E_0+\rO_{\mathbb{P}} \left( n^{-1/2+3\epsilon_d} \right), 
\end{equation}
and
\begin{gather}\label{eq_secondconnect}
    \operatorname{Re}m_{1n}(\lambda_{1}+\ri\eta_0)=-\frac{1}{\xi^2_{(1)}}+\rO_{\mathbb{P}}(n^{-1/2+3\epsilon_d}).
\end{gather}

\end{lemma}

\quad Armed with Lemma \ref{lem_connection}, we proceed to the proof of Proposition \ref{prop_boundedsetting}.

\begin{proof}[\bf Proof of Proposition \ref{prop_boundedsetting}] Recall (\ref{eq_conditionaledgedefinition}), we have that conditional on $\Omega,$ $m_{1n}(\widehat{L}_+)=-l^{-1}.$  Together with (\ref{eq_expansionlinear}), we conclude that
\begin{equation*}
\operatorname{Re}m_{1n}(\widehat{L}_++\ri n^{-1/2-\epsilon_d})=-l^{-1}+\rO(n^{-1/2-\epsilon_d}).
\end{equation*}
Moreover, according to the definition in (\ref{eq: def of gamma}), we have that
\begin{equation*}
\operatorname{Re} m_{1n}(E_0+\ri n^{-1/2-\epsilon_d})=-\xi_{(1)}^2. 
\end{equation*}
By (\ref{def4}), we obtain that conditional on $\Omega$
\begin{equation*}
\operatorname{Re}m_{1n}(\widehat{L}_++\ri n^{-1/2-\epsilon_d})=\operatorname{Re} m_{1n}(E_0+\ri n^{-1/2-\epsilon_d})+\rO \left( \frac{\log n}{n^{d+1}} \right),
\end{equation*}
which implies that $\widehat{L}_+=E_0+\rO(\log n/n^{d+1}).$ Let $\Xi$ be the probability event that (\ref{eq_holdspartoneoneone}) holds. We therefore conclude from (\ref{eq_closenessequation}) that  when restricted on $\Xi$ and $n$ is sufficiently large, $\lambda_1+\ri \eta_0 \in \mathbf{D}_b.$ Consequently, 
%Observe that $\lambda_1 \leq l \sigma_1 \lambda_1(XX^*).$ Since $\lambda_1 (XX^*) \rightarrow (1+\sqrt{\phi})^2$ almost surely \cite{bai2009spectral}, we have that $\lambda_1$ is bounded almost surely.  
by part I of Lemma \ref{localestimate2}, we find the following holds on $\Xi$
 \begin{gather}\label{eq_keykeyequationequation}
     m_{1n}( \widehat{L}_{+})-m_{1n}(\lambda_{1}+\ri\eta_0)=\frac{\widehat{\varsigma}_4}{(1-\phi\widehat{\varsigma}_3)}\left(\widehat{L}_{+}-\lambda_1-\ri \eta_0\right)+\rO\left( (\log n) (n^{-1/(d+1)})^{\min\{d,2\}}\right).
 \end{gather}
Again by  $m_{1n}(\widehat{L}_+)=-l^{-1},$ Together with the second part of Lemma \ref{lem_connection},   considering the real parts of both sides of (\ref{eq_keykeyequationequation}), we obtain that on $\Xi$
\begin{equation*}
-l^{-1}+\xi_{(1)}^{-2}+\rO_{\mathbb{P}}(n^{-1/2+3\epsilon_d})=\frac{\widehat{\varsigma}_4}{(1-\phi\widehat{\varsigma}_3)}\left(\widehat{L}_{+}-\lambda_1\right)+\rO\left( (\log n) (n^{-1/(d+1)})^{\min\{d,2\}}\right).
\end{equation*}
This completes our proof. 

\end{proof}

\quad The rest of this section leaves to the proof of Lemma \ref{lem_connection}.  We first prove the following lemma which will be used in the proof of Lemma \ref{lem_connection}. It essentially locates the points in $\mathbf{D}_b^\prime$ for which $\operatorname{Im} m_Q(z) \gg  \eta_0$ near the edge. It is a counterpart of \cite[Lemmas 5.12, 5.13 and 5.15]{lee2016extremal} and \cite[Lemmas 5.13, 5.14 and 5.16]{Kwak2021}. Due to similarity, we only sketch the key points of the proof. Recall $z_0$ defined in (\ref{eq: def of gamma}).

\begin{lemma}\label{lem_keycomponentend}  Suppose the assumptions of Lemma \ref{lem_connection}, we have that the followings holds with high probability 
\begin{enumerate}
\item[(1).] For any $z=E+\ri \eta_0 \in \mathbf{D}_b^{\prime}$ satisfying that $|z-z_0| \geq n^{-1/2+3\epsilon_d},$ 
\begin{equation}\label{eq_outsideetazero}
\operatorname{Im} m_1(z) \asymp \eta_0, \ \operatorname{Im} m_Q(z) \asymp \eta_0.
\end{equation}
\item[(2).] For $m_1^{(1)}(z)$ and $m_Q^{(1)}(z)$ defined around (\ref{eq_defnminorG}), we have that for all $z=E+\ri \eta_0 \in \mathbf{D}_b^{\prime},$
\begin{equation}\label{eq_outsideetazerominor}
\operatorname{Im} m_1^{(1)}(z) \asymp \eta_0, \ \operatorname{Im} m_Q^{(1)}(z) \asymp \eta_0.
\end{equation}
\item[(3).] There exists some $E_0' \in \mathbb{R}$ such that for $z_0'=E_0'+\ri \eta_0,$ the followings holds simultaneously 
\begin{equation}\label{eq_insideeta}
|z_0'-z_0| \leq n^{-1/2+3 \epsilon_d}, \ \text{and} \ \operatorname{Im} m(z_0') \gg \eta_0. 
\end{equation}
\end{enumerate}

\end{lemma}

\begin{proof}
Due to similarity, we focus our discussion on $m_1(z)$ and will explain the minor differences for $m_Q(z)$ from line to line. Recall (\ref{eq: decomp m_2}). Our proof relies on the following fluctuation average which provides a stronger control on $n^{-1} \sum_{i=1}^p Z_i$ than the one in Remark \ref{remk_zibound}. They are counterparts of Lemmas 5.8 and 5.9 and Corollary 5.10 of \cite{lee2016extremal}. We deter its proof to Appendix \ref{sec_FAlemma}. 
\begin{lemma}\label{lem_fa} Suppose the assumptions of Lemma \ref{lem_keycomponentend} hold, we have that the followings holds on $\Omega$ 
\begin{enumerate}
\item[(1).] For all $i \neq 1$ and all $z=E+\ri \eta_0 \in \mathbf{D}_b^\prime,$ we have that 
\begin{equation}\label{eq_c5first}
|m_2-m_2^{(1)}| \prec \frac{1}{n \eta_0}, \ \ |m_2-m_2^{(i)}|+|m_2^{(i)}-m_2^{(i1)}| \prec n^{-1+1/(d+1)+4 \epsilon_d}.  
\end{equation}
\item[(2).] For all $z \in \mathbf{D}_b^\prime,$
\begin{equation*}
\left| \frac{1}{n} \sum_{i=2}^n Z_i \right|+\left| \frac{1}{n} \sum_{i=2}^n Z_i^{(1)} \right| \prec n^{-1/2-\frac{1}{2}(\frac{1}{2}-\frac{1}{d+1})+2 \epsilon_d}.
\end{equation*}
\item[(3).] For all $z \in \mathbf{D}_b^\prime,$
\begin{equation*}
\left| \frac{1}{n} \sum_{i=2}^n \frac{(\xi_i^2+\xi_i^4) Z_i}{(1+\xi_i^2 m_{1n}(z))^2 }  \right| \prec n^{-1/2-\frac{1}{2}(\frac{1}{2}-\frac{1}{d+1})+2 \epsilon_d}. 
\end{equation*}
\end{enumerate}
\end{lemma}

\quad Armed with Lemma \ref{lem_fa}, we proceed to finish our proof.  The proof of part (1) is similar to that of (\ref{eq_oneregimeedgecontrol}) by using the local law Theorem \ref{thm_boundedcaselocallaw}. We only provide the key arguments. Using (\ref{eq_m1decompositionfinafinalfinal}), 
(\ref{eq_zibound}), and (\ref{eq_m2uboundeddecomposition}), we see that 
\begin{gather*}
     m_1=-\frac{1}{n}\sum_{i=1}^p \frac{\sigma_i}{z-\frac{\sigma_i}{n}\sum_{j=1}^n\frac{\xi^2_j}{1+\xi_j^2m_1+Z_j}}+\frac{1}{n} \operatorname{tr}(R_1 \Sigma)+\frac{1}{n} \operatorname{tr}(R_2 \Sigma).
\end{gather*}
According to Theorem \ref{thm_boundedcaselocallaw}, by a discussion similar to (\ref{eq: est of m_1 and m_2}) using Remark \ref{remk_zibound}, we have that 
\begin{gather}\label{eq_decomposition}
\begin{split}
    \operatorname{Im}m_1(z)
    &=\frac{1}{n}\sum_{i=1}^p \frac{\sigma_i\eta_0}{|z-\frac{\sigma_i}{n}\sum_{j=1}^n \frac{\xi^2_j}{1+\xi^2_jm_{1n}+	Z_j}|^2}+\frac{1}{n}\sum_{i=1}^p \frac{\frac{\sigma_i^2}{n}\sum_{j=1}^n\frac{\xi^4_j\operatorname{Im}m_{1}+\xi^2_j \operatorname{Im} Z_j}{|1+\xi^2_jm_{1n}+Z_j|^2}}{|z-\frac{\sigma_i}{n}\sum_j\frac{\xi^2_j}{1+\xi^2_jm_{1n}+Z_j}|^2}\\
    &+\rO_{\prec}\left(\frac{1}{n}\sum_{j=1}^n \frac{Z_j}{|1+\xi^2_j m_{1n}+Z_j|^2}+\frac{1}{(n \eta_0)^2}\right) \\
    & :=\mathsf{R}_1+\mathsf{R}_2+\mathsf{R}_3.
\end{split}
\end{gather}
Together with Assumption \ref{assum_additional_techinical}, (\ref{def4}) and Remark \ref{remk_zibound}, we find that for some small constant $c'>0,$ when $n$ is sufficiently large, 
\begin{equation}\label{eq_lowerbound}
\left| z-\frac{\sigma_i}{n} \sum_{j=1}^n \frac{\xi_j^2}{1+\xi_j^2 m_{1n}(z)+Z_j} \right| \geq c'. 
\end{equation}
This implies that 
\begin{equation}\label{eq_R1bound}
\mathsf{R}_1 \asymp \eta_0. 
\end{equation}
For $\mathsf{R}_2,$ on the one hand, by Theorem \ref{thm_boundedcaselocallaw} and (\ref{eq_defnmathsfW}), we can conclude that there exists some constant $0<\mathsf{c}'<1,$
\begin{equation*}
\frac{1}{n}\sum_{i=1}^p \frac{\frac{\sigma_i^2}{n}\sum_{j=1}^n\frac{\xi^4_j}{|1+\xi^2_jm_{1n}+Z_j|^2}}{|z-\frac{\sigma_i}{n}\sum_j\frac{\xi^2_j}{1+\xi^2_jm_{1n}+Z_j}|^2} \leq \mathsf{c}'<1.
\end{equation*} 
Moreover, as $|z-z_0| \geq n^{-1/2+3 \epsilon_d},$ according to (\ref{eq_a11coro}) and Remark \ref{remk_zibound}, we find that $1+\xi_j^2 m_{1n}(z)+Z_j \geq C n^{-1/2+3 \epsilon_d}$ for some constant $C>0.$ Together with (\ref{eq_lowerbound}) and (3) of Lemma \ref{lem_fa}, we find that  with high probability  
\begin{equation*}
\frac{1}{n}\sum_{i=1}^p \frac{\frac{\sigma_i^2}{n}\sum_{j=1}^n\frac{\xi^2_j \operatorname{Im} Z_j}{|1+\xi^2_jm_{1n}+Z_j|^2}}{|z-\frac{\sigma_i}{n}\sum_j\frac{\xi^2_j}{1+\xi^2_jm_{1n}+Z_j}|^2} \ll \eta_0.
\end{equation*}
Consequently, we have that with high probability 
\begin{equation}\label{eq_r2control}
\mathsf{R}_2=\mathsf{c}' \operatorname{Im} m_{1}+\ro(\eta_0). 
\end{equation}
Similarly, we can prove that with high probability $\mathsf{R}_3=\ro(\eta_0).$ Together with (\ref{eq_R1bound}), (\ref{eq_r2control}) and (\ref{eq_decomposition}), we can conclude the prove of $m_1(z).$ The discussion for $m_Q(z)$ is similar except we need to use (\ref{eq_m2uboundeddecomposition}) and (\ref{eq_mqdecomposition}).

The proof of part (2) is similar to that of part (1) using Lemma \ref{lem_fa}, Theorem \ref{thm_boundedcaselocallaw} and Remark \ref{remk_zibound}. The idea is analogous to the proof of Lemma 5.13 of \cite{lee2016extremal} or Lemma 5.14 of \cite{Kwak2021}. We omit further details.  

Finally, for part (3), we find from Lemma \ref{lem:large deviation}, (\ref{eq_outsideetazerominor}) and (\ref{lem:Wald}) that for some small constant $\epsilon'<\epsilon_d/2,$ with high probability, 
\begin{equation}\label{eq_boundboundbound121212121}
|Z_1| \leq n^{-1/2+\epsilon'}.  
\end{equation} 
Without loss of generality, we assume that $\xi_1^2 \geq \xi_2^2 \geq \cdots \geq \xi_n^2.$ On the one hand, by the definition of $m_2(z)$ and (\ref{eq: decomp m_2}), we have 
\begin{equation*}
   m_2=\frac{ \xi_1^2 \mathcal{G}_{11}}{n}+\frac{1}{n}\sum_{i=2}^{n}\frac{\xi^2_i}{-z(1+\xi^2_im_1^{(i)}+Z_i)}.
\end{equation*}
Together with Lemma \ref{lem: Resolvent}, we see that 
\begin{gather}\label{eq_keyrepresentationG11}
        \frac{1}{\xi_1^2 \mathcal{G}_{11}}=-z(1+\xi^2_{1}m_1^{(1)}+Z_{1}). 
\end{gather} 
Denote
\begin{equation*}
 z^{\pm}_{0}=z_0\pm n^{-1/2+3\epsilon_d}.  
\end{equation*}
Recall (\ref{eq: def of gamma}) that $1+\xi_1^2  \operatorname{Re} m_{1n}(E_0+\ri \eta_0)=0.$ Using Theorem \ref{thm_boundedcaselocallaw}, (\ref{lem:trace_difference}) and (\ref{eq_a11coro}), together with (\ref{eq_outsideetazerominor}) and Remark \ref{remk_zibound}, we conclude that for some constant $C>0$
\begin{equation*}
\frac{1}{\xi_1^2 \mathcal{G}_{11}(z_0^-)} \geq  Cn^{-1/2+3\epsilon_d}, \ \text{and} \  \ \frac{1}{\xi_1^2 \mathcal{G}_{11}(z_0^+)} \leq -Cn^{-1/2+3\epsilon_d}. 
\end{equation*}
Consequently, by continuity, we find that there exists $z_{1}=E_{1}+\ri\eta_0$ with $E_{1}\in(E_0-n^{-1/2+3\epsilon_d},E_0+n^{-1/2+3\epsilon_d})$ that $\operatorname{Re}\mathcal{G}_{11}(z_{1})=0$. For the choice of $z_1,$ together with (\ref{eq_keyrepresentationG11}), we find that 
\begin{gather}\label{eq_form}
    |\operatorname{Im}(z_1 \xi_1^2 \mathcal G_{11}(z_{1}))|=\frac{1}{|\operatorname{Im}m^{(1)}_1(z_{1})+\operatorname{Im}Z_{1}|}\ge n^{1/2-\epsilon_d/2},
\end{gather}
where we used (\ref{eq_boundboundbound121212121}) with the assumption $\epsilon'<\epsilon_d/2$ and (\ref{eq_outsideetazerominor}).  On the other hand, following lines of the proof of \cite[Lemma 5.15]{lee2016extremal}, by a decomposition similar to (\ref{eq: m_1 by m_2}) and a discussion similar to (\ref{eq_decomposition}), using Lemma \ref{lem_fa}, we find that 
\begin{gather*}
        \operatorname{Im}m_1(z_{1}) \asymp \eta_0+\frac{\operatorname{Im} (\xi_1^2 z_1 \mathcal{G}_{11}(z_1))}{n}.
\end{gather*}
Together with (\ref{eq_form}), we conclude that
\begin{equation*}
\operatorname{Im} m_1(z_1) \gg \eta_0. 
\end{equation*}
The discussion for $m_Q$ is similar and we omit the details. This completes our proof. 

\end{proof}

\quad Finally, armed with Lemma \ref{lem_keycomponentend}, we proceed to the proof of Lemma \ref{lem_connection}. Since the details are similar to those of Proposition 4.6 of \cite{lee2016extremal} or Proposition 4.7 of \cite{Kwak2021}, we only provide the key ingredients.  

\begin{proof}[\bf Proof of Lemma \ref{lem_connection}]
We first prove (\ref{eq_holdspartoneoneone}). Using the spectral decomposition of $Q,$ for $m_Q(z)$ in (\ref{eq_mq}), we find that  
\begin{gather}\label{eq_spectraldecompositionuseful}
    \operatorname{Im}m_Q(E+\ri\eta_0)=\frac{1}{n}\sum_{i=1}^n\frac{\eta_0}{(\lambda_i-E)^2+\eta_0^2}.
\end{gather}
This yields that 
\begin{equation*}
\operatorname{Im}m_Q(\lambda_1+\ri\eta_0)\ge(n\eta_0)^{-1}\gg\eta_0,
\end{equation*}
where we used the definition of $\eta_0$ in (\ref{eq_eta0definition}). It is clear that $\lambda_1=\rO_{\mathbb{P}}(1).$ First, if $\lambda_1 \in \mathbf{D}_b^\prime,$ then the proof follows directly from (\ref{eq_outsideetazero}). Second, if $\lambda_1 \notin \mathbf{D}_b^{\prime},$ on the other hand, for the upper bound, 
 %{\color{red} need to add more here, the discussion for the upper bound needs to be revised. we do not know whether it is in $\mathbf{D}_b^\prime$ yet. }
by (3) of Lemma \ref{lem_keycomponentend} and (\ref{eq_zopointrate}), using the definition of $\mathbf{D}_b^{\prime}$ in (\ref{eq_spectralparameterprime}), with an argument similar to Proposition 4.7 of \cite{lee2016extremal}, we have that on $\Omega,$  $\lambda_1<E_0+n^{-1/2+3\epsilon_d}$ holds with $1-\ro(1)$ probability. On the other hand, for the lower bound, we prove by contradiction.  We assume that $\lambda_1<E_0(1)-n^{-1/2+3\epsilon_d}$. Then we see from (\ref{eq_spectraldecompositionuseful}) that $\operatorname{Im}m_Q(E+\ri\eta_0)$ is a decreasing function of $E$ on the interval $(E_0-n^{-1/2+3\epsilon_d},E_0+n^{-1/2+3\epsilon_d})$. However, from Lemma \ref{lem_keycomponentend} and its proof (recall that $z_0^-=E_0-n^{-1/2+3\epsilon_d}$), we have seen that, $\operatorname{Im}m_Q(z_0)\gg\eta_0$, $\operatorname{Im}m_Q(z_0^{-})\sim\eta_0,$ which is a contradiction. It implies that $\lambda_1\ge E_0(1)-n^{-1/2+3\epsilon_d}$ and completes the proof of (\ref{eq_holdspartoneoneone}). 

\quad Then we prove (\ref{eq_secondconnect}). By (\ref{eq_a11coro}), (\ref{eq_zopointrate}) and (\ref{eq_oneregimeedgecontrol}), we find that 
\begin{equation*}
\operatorname{Re} m_{1n}(\lambda_1+\ri \eta_0)=\operatorname{Re} m_{1n}(z_0)+\rO_{\mathbb{P}}(n^{-1/2+3 \epsilon_d})=-\frac{1}{\xi_{(1)}^2}+\rO_{\mathbb{P}}(n^{-1/2+3 \epsilon_d}),
\end{equation*}
where in the last step we used (\ref{eq: def of gamma}). This completes our proof. 
\end{proof}

%{\color{red}
%\begin{remark}
%By the proof of the Lemma \ref{lem: Eigenvalue rigidity csc}, one can check that if $d<1$, there is no eigenvalue lies in the interval $[\mu_1-n^{-1/2+\epsilon_0},\mu_1+n^{-1/2+\epsilon_0}]$. In fact, for the case in Assumption \ref{ass3.4}, one can check that $\xi^2$ satisfies the following conditions
%\begin{gather*}
%    \mathbb{E}\xi^2=\phi;\quad \mathbb{E}\xi^{2p}<\infty;\\
%    \operatorname{Im}m_{\phi_0}(z)\asymp\sqrt{\kappa+\eta},
%\end{gather*}
%where $\kappa=|E-L_{+}|$ for $z=E+i\eta$ and re call that $L_{+}$ is the rightmost edge of $\tilde{\rho}_{\mathcal{W}}$. Then by Theorem 3.6 in \cite{Wen2021}, we have
%\begin{lemma}
%Suppose Assumptions \ref{ass1},\ref{ass2},\ref{ass3.4} hold. Then 
%\[
%\lim_{n\rightarrow\infty}\mathbb{P}(\gamma n^{2/3}(\lambda_1-L_{+})\le s)=F_1(s),
%\]
%where $F_1(s)$ is the type-1 Tracy-Widom distribution and $\gamma$ is some constant. 
%\end{lemma}
%In another word, the largest eigenvalue of $\mathcal{W}$ recovers to the classic universality case.
%\end{remark}
%}

\subsection{Proof of main results of Section \ref{sec_mainresults}}\label{sec_proofofmainresults}

In this section, we prove Theorems \ref{thm_main_unbounded} and \ref{thm_main_bounded} using the results in Sections \ref{sec_sub_unbounded1st} and \ref{sec_sub_bounded1st}.  

\begin{proof}[\bf Proof of Theorem \ref{thm_main_unbounded}]
For the first part,  according to (\ref{eq_mu1part}), by Lemma \ref{lem_probabilitycontrol}, when $n$ is sufficiently large, 
\begin{equation*}
\mu_1=(\xi_{(1)}^2+d_1)\left(\phi \bar{\sigma} +\rO_{\mathbb{P}}(e) \right). 
\end{equation*} 
Together with Proposition \ref{lem: eigenvalue rigidity}, we find that 
\begin{equation*}
\lambda_1=(\xi_{(1)}^2+d_1)\left(\phi \bar{\sigma} +\rO_{\mathbb{P}}(e) \right)+\rO_{\mathbb{P}}\left( n^{-1/2+2\epsilon}d_1 \right). 
\end{equation*}
Using (\ref{def1}) and (\ref{def3}) and the definition of $d_1$ in (\ref{eq_firstddefinition}), we can complete the proof for the first part. 

%{\color{red} from here}
%Similarly, when case (b) holds, we can prove the following results hold with high probability using {\color{red} [add the citation here]} and (\ref{def3}) and Lemma \ref{lem: ratio}
%\begin{equation*}
%\lambda_1=\phi \bar{\sigma} \xi_{(1)}^2 \left(1+\rO(\log^{-1} n)  \right). 
%\end{equation*} 

For the second part, it follows directly from the results in Lemma \ref{lem_summaryevt}, (\ref{def1}) and (\ref{def3}). This completes our proof. 
\end{proof}

\begin{proof}[\bf Proof of Theorem \ref{thm_main_bounded}]
Due to similarity, we focus our discussion on the separable covariance i.i.d. data as in case (2) of Assumption \ref{assum_model}. The elliptical case can be handled similarly. 
For part (1), (\ref{eq_onlyoneequationdecide}) has been proved in (II) of Lemma \ref{localestimate2}. For (\ref{eq_boundednnnnn}), the proofs follow from Proposition \ref{prop_boundedsetting}, II of Lemma \ref{localestimate2} with the fact $d>1$ and (\ref{def4}). Then (\ref{eq_distributionresultweibull}) follows from (\ref{eq_boundednnnnn}) and Lemma \ref{lem_summaryevt}. 

Then we proceed to the proof of parts (2) and (3). Following \cite[Lemma 2.5]{ding2021spiked}, we see that conditional on $\Omega$ in Lemma \ref{lem_probabilitycontrol}, $(\widehat{L}_+, m_{1n}(\widehat{L}_+))$ should satisfy the following systems of equations 
\begin{gather}\label{eq_edgeequationstwodecide}
    m_{1n}(\widehat{L}_{+})=\frac{1}{n}\sum_{i=1}^p \frac{\sigma_i}{-\widehat{L}_{+}+\frac{\sigma_i}{n}\sum_j\frac{\xi^2_j}{1+\xi^2_jm_{1n}(\widehat{L}_{+})}},\quad 1=\frac{1}{n}\sum_{i=1}^p \frac{\frac{\sigma_i^2}{n}\sum_j\frac{\xi^4_j}{|1+\xi^2_jm_{1n}(\widehat{L}_{+})|^2}}{\left|\widehat{L}_{+}-\frac{\sigma_i}{n}\sum_j\frac{\xi^2_j}{1+\xi^2_jm_{1n}(\widehat{L}_{+})}\right|^2}.
\end{gather}
Similarly, $(L_+, m_{1n,c}(L_+))$ should satisfy the following equations
\begin{gather}\label{eq_edgeequationstwodecide2}
    m_{1n,c}(L_{+})=\frac{1}{n}\sum_i\frac{\sigma_i}{-L_{+}+\sigma_i\int\frac{s}{1+sm_{1n,c}(L_{+})}\mathrm{d}F(s)},\quad 1=\frac{1}{n}\sum_i\frac{\sigma^2_i\int\frac{s^2}{|1+sm_{1n,c}(L_{+})|^2}\mathrm{d}F(s)}{|L_{+}-\sigma_i\int\frac{s}{1+sm_{1n,c}(L_{+})}\mathrm{d}F(s)|^2}.
\end{gather}
Using the definitions of $\varsigma_k$ and $\widehat{\varsigma}_k, 1 \leq k \leq 3,$ by (\ref{def4}) and an argument similar to II of Lemma \ref{localestimate2}, when $n$ is sufficiently large, we see that our assumption $\phi^{-1}<\varsigma_3$ implies $\phi^{-1}<\widehat{\varsigma}_3$ on $\Omega.$ This yields that for some constant $\delta>0$
\begin{gather}\label{eq_>1plusdelta}
    \frac{1}{n} \sum_i\frac{\sigma_i^2 \widehat{\varsigma}_1}{(\widehat{L}_{+}-\sigma_i\widehat{\varsigma}_2)^2} >1+\delta,\quad \frac{1}{n}\sum_i\frac{\sigma_i^2\varsigma_1}{(L_{+}-\sigma_i\varsigma_2)^2}>1+\delta.
\end{gather}
where the first inequality is restricted on the event $\Omega$. From now on, for notional simplicity, we always restrict ourselves on $\Omega$ so that the discussion is purely deterministic. Recall (\ref{eq_phasetransition}) and (\ref{eq_finitesample123}). Combining the second equations in (\ref{eq_edgeequationstwodecide}) and (\ref{eq_edgeequationstwodecide2}) with (\ref{eq_>1plusdelta}), we readily obtain that 
\begin{equation}\label{rem1nbound}
m_{1n}(\widehat{L}_{+})>-l^{-1}, \ \ m_{1n,c}(L_{+})>-l^{-1}.
\end{equation}
Together with (\ref{def4}), we  have that for all $1 \leq j \leq n$ and some constant $\delta'>0$ 
\begin{equation}\label{eq_keyboundhastobeenused}
\frac{1}{\left|1+\xi_j^2 m_{1n}(\widehat{L}_+)\right|} \geq \delta', \   \frac{1}{\left|1+s m_{1n,c}(L_+)\right|} \geq \delta' \ \text{for any} \ 0<s \leq l. 
\end{equation}
We now proceed to the proof. The proof consists of two steps. In the first step, we prove the results assuming that
\begin{equation}\label{eq_keyanasztaz}
\left| m_{1n,c}(L_+)-m_{1n}(\widehat{L}_+) \right|=\mathrm{O}_{\mathbb{P}}(n^{-1/2}), \ |L_+-\widehat{L}_+|=\rO_\mathbb{P}( n^{-1/2}). 
\end{equation}
In the second step, we justify (\ref{eq_keyanasztaz}). We start with step one. 

\vspace{3pt}

\noindent{\bf Step one:} Under the assumption \ref{eq_keyanasztaz}, the key component of the proof is the following lemma. Denote
\begin{equation*}
\mathsf{C}_1:=\frac{1}{n} \sum_{i=1}^p \frac{\sigma_i}{ \left( L_+-\sigma_i \int \frac{s}{1+sm_{1n,c}} \mathrm{d} F(s) \right)^2}, 
\end{equation*}
and
\begin{equation*}
 \mathcal{X}:=\frac{1}{n} \sum_{j=1}^n \left( \frac{\xi_j^2}{1+\xi_j^2 m_{1n,c}(L_+)}-\int \frac{s}{1+sm_{1n,c}(L_+)} \mathrm{d} F(s) \right).
\end{equation*}
According to Assumption \ref{assum_additional_techinical}, we have that 
\begin{equation*}
\mathsf{C}_1 \asymp 1. 
\end{equation*}

\begin{lemma}\label{lem_keyfinalfinalkey} Under the assumptions of Theorem \ref{thm_main_bounded} and (\ref{eq_keyanasztaz}), we have that 
\begin{equation*}
\mathsf{C}_1(\widehat{L}_+-L_+)=\mathsf{C}_1\mathcal{X}+\rO_{\mathbb{P}}(n^{-1}). 
\end{equation*} 
\end{lemma}

Armed with Lemma \ref{lem_keyfinalfinalkey}, we can easily prove parts (2) and (3). Recall (\ref{eq_cltvariance}).  It is clear from (\ref{def4}) $\mathcal{X}=\rO_{\mathbb{P}}(n^{-1/2}),$ and from central limit theorem that $\mathcal{X}$ is asymptotically Gaussian with variance $n^{-1} \vartheta$.   
We decompose that 
\begin{equation*}
\lambda_1-L_+=\lambda_1-\widehat{L}_++\widehat{L}_+-L_+.
\end{equation*}
According to \cite{DX_ell,9779233} and Lemma \ref{localestimate2}, we find that on $\Omega,$ $|\lambda_1-\widehat{L}_+| \prec n^{-2/3}$ and $n^{2/3} \gamma(\lambda_1-\widehat{L}_+)$ follows type-1 Tracy-Widom law. This concludes the general results in part (3). Moreover, for part (2), it is easy to see that when $d>1,$ by Cauchy-Schwarz inequality, $\vartheta$ is bounded from blow so that the Gaussian part dominates the Tracy-Widom part and we hence conclude the proof.

\quad  To complete Step one, we now prove Lemma \ref{lem_keyfinalfinalkey}.

\begin{proof}[\bf Proof of Lemma \ref{lem_keyfinalfinalkey}]

Using the first parts in equations (\ref{eq_edgeequationstwodecide}) and (\ref{eq_edgeequationstwodecide2}), we see that 
    \begin{align}\label{eq_generalgeneraldecomposition}
       m_{1n,c}(L_{+})-m_{1n}(\widehat{L}_{+})&=\frac{1}{n}\sum_i\frac{\sigma_i}{\widehat{L}_{+}-\frac{\sigma_i}{n}\sum_j\frac{\xi^2_j}{1+\xi^2_jm_{1n}(\widehat{L}_{+})}}-\frac{1}{n}\sum_i\frac{\sigma_i}{L_{+}-\frac{\sigma_i}{n}\sum_j\frac{\xi^2_j}{1+\xi^2_jm_{1n,c}(L_{+})}} \nonumber \\
       &+\frac{1}{n}\sum_i\frac{-\sigma_i\int\frac{s}{1+sm_{1n,c}(L_{+})}\mathrm{d}F(s)+\frac{\sigma_i}{n}\sum_j\frac{\xi^2_j}{1+\xi^2_jm_{1n,c}(L_{+})}}{(L_{+}-\frac{\sigma_i}{n}\sum_j\frac{\xi^2_j}{1+\xi^2_jm_{1n,c}(L_{+})})(L_{+}-\sigma_i\int\frac{s}{1+sm_{1n,c}(L_{+})}\mathrm{d}F(s))} \nonumber \\
        &=\frac{1}{n}\sum_i\frac{\sigma_i(L_{+}-\widehat{L}_{+})}{(\widehat{L}_{+}-\frac{\sigma_i}{n}\sum_j\frac{\xi^2_j}{1+\xi^2_jm_{1n,c}(L_{+})})(L_{+}-\frac{\sigma_i}{n}\sum_j\frac{\xi^2_j}{1+\xi^2_jm_{1n,c}(L_{+})})} \nonumber \\
        &+\frac{1}{n}\sum_i\frac{\sigma_i}{\widehat{L}_{+}-\frac{\sigma_i}{n}\sum_j\frac{\xi^2_j}{1+\xi^2_jm_{1n}(\widehat{L}_{+})}}-\frac{1}{n}\sum_i\frac{\sigma_i}{\widehat{L}_{+}-\frac{\sigma_i}{n}\sum_j\frac{\xi^2_j}{1+\xi^2_jm_{1n,c}(L_{+})}} \nonumber \\
        &+\frac{1}{n}\sum_i\frac{\sigma_i\int\frac{s}{1+sm_{1n,c}(L_{+})}\mathrm{d}F(s)-\frac{\sigma_i}{n}\sum_j\frac{\xi^2_j}{1+\xi^2_jm_{1n,c}(L_{+})}}{(L_{+}-\frac{\sigma_i}{n}\sum_j\frac{\xi^2_j}{1+\xi^2_jm_{1n,c}(L_{+})})(L_{+}-\sigma_i\int\frac{s}{1+sm_{1n,c}(L_{+})}\mathrm{d}F(s))} \nonumber \\
        &=\frac{1}{n}\sum_i\frac{\sigma_i(L_{+}-\widehat{L}_{+})}{(\widehat{L}_{+}-\frac{\sigma_i}{n}\sum_j\frac{\xi^2_j}{1+\xi^2_jm_{1n,c}(L_{+})})(L_{+}-\frac{\sigma_i}{n}\sum_j\frac{\xi^2_j}{1+\xi^2_jm_{1n,c}(L_{+})})} \nonumber \\
        &+\frac{1}{n}\sum_i\frac{-\frac{\sigma_i^2}{n}\sum_j\frac{\xi^4_j(m_{1n}(\widehat{L}_{+})-m_{1n,c}(L_{+}))}{(1+\xi^2_jm_{1n}(\widehat{L}_{+}))(1+\xi^2_jm_{1n,c}(L_{+}))}}{(\widehat{L}_{+}-\frac{\sigma_i}{n}\sum_j\frac{\xi^2_j}{1+\xi^2_jm_{1n}(\widehat{L}_{+})})(\widehat{L}_{+}-\frac{\sigma_i}{n}\sum_j\frac{\xi^2_j}{1+\xi^2_jm_{1n,c}(L_{+})})} \nonumber \\
        &+\frac{1}{n}\sum_i\frac{-\sigma_i\int\frac{s}{1+sm_{1n,c}(L_{+})}\mathrm{d}F(s)+\frac{\sigma_i}{n}\sum_j\frac{\xi^2_j}{1+\xi^2_jm_{1n,c}(L_{+})}}{(L_{+}-\frac{\sigma_i}{n}\sum_j\frac{\xi^2_j}{1+\xi^2_jm_{1n,c}(L_{+})})(L_{+}-\sigma_i\int\frac{s}{1+sm_{1n,c}(L_{+})}\mathrm{d}F(s))} \nonumber \\
        &:=\mathsf{T}_1+\mathsf{T}_2+\mathsf{T}_3. 
\end{align}
For the term $\mathsf{T}_1,$ by (\ref{eq_keyanasztaz}), Assumption \ref{assum_additional_techinical} and (\ref{def4}), we can see that 
\begin{equation}\label{eq_T1controlcontrol}
\mathsf{T}_1=\mathsf{C}_1(L_+-\widehat{L}_+)+\rO_{\mathbb{P}}(n^{-1}).
\end{equation}
%By a similar discussion using (\ref{eq_keyboundhastobeenused}) and Theorem \ref{thm_boundedcaselocallaw}, we have that
%\begin{equation}\label{eq_T2controlcontrol}
%\mathsf{T}_2=\rO_{\prec}(n^{-1}). 
%\end{equation}
For the term $\mathsf{T}_2,$ we see that 
\begin{align}\label{eq_T2controlcontrol}
        &\mathsf{T}_2=\frac{1}{n}\sum_i\frac{-\frac{\sigma_i^2}{n}\sum_j\frac{\xi^4_j(m_{1n}(\widehat{L}_{+})-m_{1n,c}(L_{+}))}{(1+\xi^2_jm_{1n}(\widehat{L}_{+}))(1+\xi^2_jm_{1n,c}(L_{+}))}}{(\widehat{L}_{+}-\frac{\sigma_i}{n}\sum_j\frac{\xi^2_j}{1+\xi^2_jm_{1n}(\widehat{L}_{+})})^2} \nonumber \\
        &+\frac{1}{n}\sum_i\frac{\big(\frac{\sigma_i^2}{n}\sum_j\frac{\xi^4_j}{(1+\xi^2_jm_{1n}(\widehat{L}_{+}))(1+\xi^2_jm_{1n,c}(L_{+}))}\big)^2(m_{1n}(\widehat{L}_{+})-m_{1n,c}(L_{+}))^2}{(\widehat{L}_{+}-\frac{\sigma_i}{n}\sum_j\frac{\xi^2_j}{1+\xi^2_jm_{1n}(\widehat{L}_{+})})^2(\widehat{L}_{+}-\frac{\sigma_i}{n}\sum_j\frac{\xi^2_j}{1+\xi^2_jm_{1n,c}(L_{+})})} \nonumber \\
        &=\frac{1}{n}\sum_i\frac{-\frac{\sigma_i^2}{n}\sum_j\frac{\xi^4_j(m_{1n}(\widehat{L}_{+})-m_{1n,c}(L_{+}))}{(1+\xi^2_jm_{1n}(\widehat{L}_{+}))^2}}{(\widehat{L}_{+}-\frac{\sigma_i}{n}\sum_j\frac{\xi^2_j}{1+\xi^2_jm_{1n}(\widehat{L}_{+})})^2}+\frac{1}{n}\sum_i\frac{-\frac{\sigma_i^2}{n}\sum_j\frac{\xi^4_j(m_{1n}(\widehat{L}_{+})-m_{1n,c}(L_{+}))^2}{(1+\xi^2_jm_{1n}(\widehat{L}_{+}))^2(1+\xi^2_jm_{1n,c}(L_{+}))}}{(\widehat{L}_{+}-\frac{\sigma_i}{n}\sum_j\frac{\xi^2_j}{1+\xi^2_jm_{1n,c}(L_{+})})^2} \nonumber \\
        &+\frac{1}{n}\sum_i\frac{\big(\frac{\sigma_i^2}{n}\sum_j\frac{\xi^4_j}{(1+\xi^2_jm_{1n}(\widehat{L}_{+}))(1+\xi^2_jm_{1n,c}(L_{+}))}\big)^2(m_{1n}(\widehat{L}_{+})-m_{1n,c}(L_{+}))^2}{(\widehat{L}_{+}-\frac{\sigma_i}{n}\sum_j\frac{\xi^2_j}{1+\xi^2_jm_{1n}(\widehat{L}_{+})})(\widehat{L}_{+}-\frac{\sigma_i}{n}\sum_j\frac{\xi^2_j}{1+\xi^2_jm_{1n,c}(L_{+})})^2}  \\
        &=-(m_{1n}(\widehat{L}_{+})-m_{1n,c}(L_{+}))+\rO_{\mathbb{P}}(n^{-1}), \nonumber
\end{align}
where in the last step we used the second equation of (\ref{eq_edgeequationstwodecide}) for the first term of (\ref{eq_T2controlcontrol}), and (\ref{eq_keyanasztaz}), Theorem \ref{thm_boundedcaselocallaw}, (\ref{def4}) and Assumption \ref{assum_additional_techinical}  for the second and third terms.  Similarly, for $\mathsf{T}_3,$ we have that
\begin{equation}\label{eq_T3controlcontrol}
\mathsf{T}_3=\mathsf{C}_1 \mathcal{X}+\rO_{\mathbb{P}}(n^{-1}). 
\end{equation}

Insert (\ref{eq_T1controlcontrol}), (\ref{eq_T2controlcontrol}) and (\ref{eq_T3controlcontrol}) into (\ref{eq_generalgeneraldecomposition}), we can conclude the proof. 
\end{proof}

Then we prove (\ref{eq_keyanasztaz}) to complete step two and the proof of the theorem. 

\vspace{3pt} 

\noindent{\bf Step two:} To prove (\ref{eq_keyanasztaz}), we first rewrite (\ref{eq_edgeequationstwodecide}) and (\ref{eq_edgeequationstwodecide2}) a little bit. Recall (\ref{eq:F(m,z)}). We find that (\ref{eq_edgeequationstwodecide}) can be rewritten as 
\begin{gather*}
    F_n(m_{1n}(\widehat{L}_{+}),\widehat{L}_{+})=0,\quad \frac{\partial F_n}{\partial x}(m_{1n}(\widehat{L}_{+}),\widehat{L}_{+})=0,
\end{gather*}
where we denote 
\begin{gather}\label{eq_Fnxyoriginaldefinition}
    F_n(x,y)=\frac{1}{n}\sum_{i=1}^p\frac{\sigma_i}{-y+\frac{\sigma_i}{n}\sum_{j=1}^n\frac{\xi^2_j}{1+x\xi^2_j}}-x.
\end{gather}
Similarly, (\ref{eq_edgeequationstwodecide2}) can be rewritten as 
\begin{gather*}
    F_{n,c}(m_{1n,c}(L_{+}),L_{+})=0,\quad \frac{\partial F_{n,c}}{\partial x}(m_{1n,c}(L_{+}),L_{+})=0,
\end{gather*}
where we denote
\begin{gather*}
    F_{n,c}(x,y)=\frac{1}{n}\sum_{i=1}^p\frac{\sigma_i}{-y+\sigma_i\int\frac{s}{1+xs}dF(s)}-x.
\end{gather*}  
For pair $(\widetilde{x}, \widetilde{y})$ so that $\widetilde{x}>-l^{-1}$ (recall (\ref{rem1nbound})), as long as they satisfy Assumption \ref{assum_additional_techinical} in the sense that $\min_{1 \leq i \leq p} |\Tilde{y}-\sigma_i\int\frac{s}{1+\Tilde{x}s}dF(s)| \geq \tau, $ by (\ref{def4}), we find that 
\begin{gather}\label{eq_controlboundveryveryuseful}
    \left|F_{n,c}(\Tilde{x},\Tilde{y})-F_n(\Tilde{x},\Tilde{y})\right|+\left|\frac{\partial F_{n,c}}{\partial x}(\Tilde{x},\Tilde{y})-\frac{\partial F_{n}}{\partial x}(\Tilde{x},\Tilde{y})\right|+
    \left|\frac{\partial F_{n,c}}{\partial y}(\Tilde{x},\Tilde{y})-\frac{\partial F_{n}}{\partial y}(\Tilde{x},\Tilde{y})\right|=\rO_{\mathbb{P}}(n^{-1/2}).
\end{gather}
Set $(x_0, y_0)=(m_{1n,c}(L_+), L_+).$ Then we have that 
\begin{gather}\label{eq_assumptionequation}
    F_{n,c}(x_0,y_0)=0,\quad \frac{\partial F_{n,c}}{\partial x}(x_0,y_0)=0,\quad  0<\frac{\partial F_{n,c}}{\partial y}(x_0,y_0)<\infty, \quad \frac{\partial^2 F_{n}}{\partial y^2}(x_0,y_0)<0.
\end{gather}
It suffices to prove the following lemma. 
\begin{lemma}\label{lem_stabilityargument}
 There exists a pair $(x_1,y_1)$ with condition $|x_1-x_0|+|y_1-y_0|=\mathrm{O}_{\mathbb{P}}(n^{-1/2})$ such that with probability $1-\ro(1)$
    \begin{gather}\label{eq_final_result}
    F_n(x_1,y_1)=0,\quad \frac{\partial F_{n}}{\partial x}(x_1,y_1)=0. 
\end{gather}
\end{lemma}  
  
With Lemma \ref{lem_stabilityargument}, according to (\ref{eq_edgeequationstwodecide}) and Theorem \ref{lem_solutionsystem}, we see that (\ref{eq_keyanasztaz}) holds. In the rest, we prove Lemma \ref{lem_stabilityargument} using (\ref{eq_controlboundveryveryuseful}). 

\begin{proof}[\bf Proof of Lemma \ref{lem_stabilityargument}]For some small $\epsilon>0,$ we consider the probability event $\Xi$ so that (\ref{def4}) holds and (\ref{eq_controlboundveryveryuseful}) reads as
\begin{equation}\label{eq_reinterpretteresult}
   \left|F_{n,c}(\Tilde{x},\Tilde{y})-F_n(\Tilde{x},\Tilde{y})\right|+\left|\frac{\partial F_{n,c}}{\partial x}(\Tilde{x},\Tilde{y})-\frac{\partial F_{n}}{\partial x}(\Tilde{x},\Tilde{y})\right|+
    \left|\frac{\partial F_{n,c}}{\partial y}(\Tilde{x},\Tilde{y})-\frac{\partial F_{n}}{\partial y}(\Tilde{x},\Tilde{y})\right|=\rO(n^{-1/2+\epsilon}).
\end{equation}
We have seen that $\mathbb{P}(\Xi)=1-\mathrm{o}(1).$ Now we fix a realization $\{\xi_i^2\} \in \Xi$ so that the discussions below are purely deterministic.

For the above fixed constant $\epsilon>0,$ we set the region
\begin{gather*}
    \mathcal{N}(x,y):=\{(x,y): |x-x_0|+|y-y_0|\leq n^{-1/2+\epsilon}\},
\end{gather*}
To prove the first part of (\ref{eq_final_result}), it suffices to prove that there exists a solution of $F_n(x,y)=0$ in the region $\mathcal{N}(x,y).$ By Bolzano's theorem, we see that for sufficiently large $n,$ we can find two points $(x_{11},y_{11})$ and $(x_{12},y_{12})$ on $\mathcal{N}(x,y)$ so that $F_{n,c}(x_{11},y_{11})<0, \ F_{n,c}(x_{12},y_{12})>0.$ Together with (\ref{eq_reinterpretteresult}), we see that $ F_{n}(x_{11},y_{11})<0,\ F_{n}(x_{12},y_{12})>0.$ Therefore, by continuity, we can find some point $(x',y')$ so that $F_n(x',y')=0.$ Repeating the above procedure, by implicit function theorem, we find that there exists a curve $x \equiv x(y)$ on $\mathcal{N}(x,y)$ so that $F_n(x,y)=0.$ Similarly, we can show that there exists another curve $\widehat{x} \equiv \widehat{x}(\widehat{y})$ on $\mathcal{N}(\widehat{x},\widehat{y})$ so that the second part of (\ref{eq_final_result}) holds in the sense that $\partial F_n (\widehat{x}, \widehat{y})/\partial \widehat{x}=0.$ 

In order to show (\ref{eq_final_result}), we need to prove that the curves $(x,y)$ and $(\widehat{x}, \widehat{y})$ must have at least one intersection in the region $\mathcal{N}(x,y).$ We prove by contradiction. Otherwise, the curve $(x,y)$ will lie in one of the areas separated by $(\widehat{x},\widehat{y})$ with strictly $\partial F_n(x,y)/ \partial x<0$ or $\partial F_n(x,y)/ \partial x>0$. By (\ref{eq_assumptionequation}), we see that $\mathcal{N}(x,y),$  $\partial F_n(x,y)/\partial y>0$. Without loss of generality, we assume $\partial F_n(x,y)/ \partial x<0$. On the one hand, as $F_n(x,y)=0$, one may conclude that for small neighbor around the points on $(x,y)$, it holds that $\mathrm{d} x/\mathrm{d} y>0$. On the other hand, taking the derivative $F_n(x,y)$ with respect to $y$, we obtain that 
\begin{gather*}
    \frac{\mathrm{d} x}{\mathrm{d} y}\times \left(\frac{1}{n}\sum_{i=1}^p\frac{\frac{ \sigma_i^2}{n}\sum_j\frac{\xi^4_j}{(1+x\xi^2_j)^2}}{(-y+\frac{\sigma_i}{n}\sum_j\frac{\xi^2_j}{1+x\xi^2_j})^2}-1\right)=0,
\end{gather*}
which implies $\partial F_n(x,y)/\partial x=0$ and gives the contradiction. This concludes our proof.

\end{proof}

\end{proof}

\subsection{Proof of the results of Section \ref{sec_statapplication}}\label{sec_proofsectionapplication}
In this section, we prove the results of Section \ref{sec_statapplication} which are related to our statistical applications. 

\begin{proof}[\bf Proof of Theorem \ref{thm_main_spike}]  We start with the proof of part (1). For the data matrix $\widetilde{Y}$ defined using (\ref{eq_datamodel}) around (\ref{eq_eigenvalueintheendspiked}), we denote  $\widetilde{Q}:=\widetilde{Y} \widetilde{Y}^*$ and $\widetilde{\mathcal Q}:=\widetilde{Y}^* \widetilde{Y}.$ Since these two matrices have the same non-zero eigenvalues, we focus on the later one for convenience. For the spiked covariance matrix model in (\ref{eq_truemodelspiked}), we decompose it as follows 
\begin{equation*}
 \widetilde{\Sigma}:=\Sigma_s+\Sigma_o,
\end{equation*}
where we denote the two $p \times p$ matrices as 
\begin{equation}\label{eq_twomatricesdecomposition}
\Sigma_s:=\sum_{i=1}^r\Tilde{\sigma}_i\mathbf{v}_i\mathbf{v}_i^{*} \equiv V_1\Lambda_sV_1^{*},\quad \Sigma_o:=\sum_{i=r+1}^p \sigma_i\mathbf{v}_i\mathbf{v}_i^{*} \equiv V_2\Lambda_oV_2^{*}.
\end{equation}
Consequently, we can decompose $\widetilde{\mathcal Q}$ as follows
\begin{equation*}
\widetilde{\mathcal Q}=DX^{*}\Tilde{\Sigma}XD=DX^{*}\Sigma_sXD+DX^{*}\Sigma_oXD.
\end{equation*}
Note that with high probability 
 \begin{equation*}
        \|DX^{*}\Sigma_oXD\|=\| D^2 X^* \Sigma_0 X \| \leq \| D^2 \| \|X^{*} \Sigma_0 X\| \le \sigma_r \xi_{(1)}^2 \| X^* X \| \sim \xi^2_{(1)},
    \end{equation*}
    where in the last step we used \cite{Wen2021} that $\|X^* X \|$ is bounded from above with high probability. Using (\ref{def1}) and (\ref{def3}) as well as Weyl's inequality, we see that from the assumption of (\ref{spiked_assumption}) that, for $1 \leq i \leq r,$ 
    \begin{equation}\label{eq_proof4.1firstpartone}
   \frac{ \mu_i-\lambda_i(DX^* \Sigma_sXD)}{\widetilde{\sigma}_i}=\ro_{\mathbb{P}}(1). 
    \end{equation}
Then we consider the first few largest  eigenvalues of $DX^* \Sigma_s XD,$ or equivalently those of $\Sigma_s^{1/2} XD^2X^* \Sigma_s^{1/2}.$ By a discussion similar to Lemma D.1 of \cite{ding2021spiked}, we find that if $\lambda$ is an eigenvalues of $\Sigma_s^{1/2} XD^2X^* \Sigma_s^{1/2},$ recalling (\ref{eq_twomatricesdecomposition}), we have that 
    \begin{gather}\label{eq_masterequationeigenvalues}
        \operatorname{det}(V_1^{*}XD^2X^{*}V_1-\lambda\Lambda_s^{-1})=0.
    \end{gather}     
Moreover, due to the rotational invariant property of $X$, without loss of generality, we can assume the columns of $V_1$ are standard basis in $\mathbb{R}^p.$ Consequently, we observe that 
    \begin{gather*}
        V_1^{*}XD^2X^{*}V_1=
        \begin{pmatrix}
        \sum_i\xi^2_iu_{i1}^2 &\sum_i\xi^2_iu_{i1}u_{i2} &\cdots &\sum_i\xi^2_iu_{i1}u_{i r}\\
        \sum_i\xi^2_iu_{i2}u_{i1}&\sum_i\xi^2_iu_{i2}^2 &\cdots &\sum_i\xi^2_iu_{i2}u_{ir}\\
        \vdots &\vdots &\ddots &\vdots\\
        \sum_i\xi^2_iu_{ir}u_{i1} &\sum_i\xi^2_iu_{ir}u_{i2} &\cdots &\sum_i\xi^2_iu_{ir}^2
        \end{pmatrix}_{r\times r},
    \end{gather*}    
 where we used the fact that the $i$-th column of $X$ is $\mathbf{u}_i=(u_{i1}, u_{i2}, \cdots, u_{ip})^*.$  Recall that $u_{i1}$ has the same distribution as the self-normalized random variable $\mathsf{g}_{i1}:=g_{i1}/\sqrt{\sum_{j=1}^p g_{ij}^2},$ where $\{g_{ij}\}$ are i.i.d. standard Gaussian random variables. Let $\operatorname{Kol}$ be the  Kolmogorov distance.  According to the discussions in \cite[Section 4.2.1]{pinelis2009optimal}, we find that for some standard Gaussian random variable $\mathsf{g}$ independent of $u_{i1},$ we have that 
 \begin{equation}\label{eq_koldistancecontrol}
 \operatorname{Kol}(\mathsf{g}, \sqrt{p} u_{i1})=\rO(p^{-1}). 
 \end{equation}
 Since $r$ is finite and $\{\xi_i^2\}$ and $X$ are independent, using the assumptions in (i) of Assumption \ref{assum_D}, by straightforward calculations using Markov inequality, we conclude that 
 \begin{equation*}
   V_1^{*}XD^2X^{*}V_1=\phi^{-1} \mathbb{E} \xi^2 I_r+\mathrm{o}_{\mathbb{P}}(1), 
\end{equation*}    
where $I_r$ is a $r \times r$ identity matrix. Together with (\ref{eq_masterequationeigenvalues}), we conclude that for $1 \leq i \leq r$
\begin{equation}\label{eq_proof.1firstparttwo}
\frac{\lambda_i(DX^* \Sigma_s XD)}{\widetilde{\sigma}_i}=\phi^{-1} \mathbb{E} \xi^2+\mathrm{o}_{\mathbb{P}}(1). 
\end{equation}   
Combining (\ref{eq_proof4.1firstpartone}), we have completed the proof of part (1). 

Then we proceed with part (2). The proof follows closely from a discussion similar to  the proof of \cite[Theorem 3.7]{ding2021spiked}, or \cite[Theorem 2.7]{knowles2013isotropic}, or \cite[Theorem 2.7]{bloemendal2016principal}, or \cite[Theorem 3.6]{DJ22}. Due to similarity, we only sketch the proof strategies, provide the key ingredients and point out the main differences. In fact, our proof will be easier since the spikes are much larger than the edges and we only consider the first few extremal non-outlier eigenvalues. As discussed in   \cite[Appendix D]{ding2021spiked}, or \cite[Section 6]{knowles2013isotropic}, or \cite[Section 4]{bloemendal2016principal}, the proof consists of the following three steps.
\begin{enumerate}
\item[(i).] We first find the permissible regions in which contain the eigenvalues of $\widetilde{\mathcal{Q}}$ with high probability.
\item[(ii).] Then we apply a counting argument to a special case (where all the spikes are well-separate), and show that the results hold under this special case. 
\item[(iii).] Finally we use a continuity argument to extend the results in (ii) to the general case using the gaps in the permissible regions.   
\end{enumerate}

In what follows, we choose a realization $\{\xi_i^2\} \in \Omega$ so that (\ref{def4}) holds with $1-\ro(1)$ probability as in Lemma \ref{lem_probabilitycontrol}. With this restriction, $m_{1n}$ and $\mu_1$ in (\ref{eq: def of mu_1}) are purely deterministic. Recall $d_1$ in (\ref{eq_firstddefinition}) and the $\epsilon$ used therein. Due to similarity, we focus on the polynomial decay setting (\ref{ass3.1}). The exponential decay case can be handled similarly. 

\quad For Step (i), to find the permissible region, for some large constant $\mathsf{C}>0,$ we denote the set  for $1 \leq i \leq k$
\begin{equation}\label{eq_gamm1set}
\Gamma_i:=\left\{ x\in [\lambda_i, \mu_1+ n^{-1/2+2 \epsilon} d_1]: \ \operatorname{dist}(x, \operatorname{spec}({Q}))> \mathsf{C} n^{-1/2+ 2\epsilon} d_1 \right\},
\end{equation}
where $\text{spec}(\widetilde{Q})$ stands for the spectrum of ${Q}.$  The results of Step (i) can be summarized as follows. 
\begin{lemma} There exists some constant $\mathsf{C}>0$ so that the set $\cup_i \Gamma_i$ contains no eigenvalue of $\widetilde{{Q}}.$ 
\end{lemma}
\begin{proof}
The proof is similar to that of Lemma D.4 of \cite{ding2021spiked} or Lemma 5.4 of \cite{DJ22} and we only sketch the key points here. Due to the rational invariance of $X,$ we can without of generality assume that $\widetilde{\Sigma}$ is diagonal and decompose that  $\widetilde{\Sigma}=\Sigma_1+\Sigma $ with $\Sigma_1=\widetilde{\Sigma}-\Sigma=U^* D U$, where we recall that $\widetilde{\Sigma}$ is constructed based on $\Sigma.$ Here $D$ is an $r \times r$  matrix containing the nonzero eigenvalues of $\Sigma_1$ and $U^*$ is the $p \times r$ matrix containing the first $r$ standard basis in $\mathbb{R}^p.$ Under the assumption of (\ref{spiked_assumption}) and (\ref{ass2}), we see that $D$ is invertible when $n$ is sufficiently large.  By a discussion similar to (\ref{eq_masterequationeigenvalues}), we see that $x$ is an eigenvalue of $\widetilde{\mathcal Q}$ but not $\mathcal{Q}$ if and only if
 \begin{gather}\label{eq_newmasterequation}
    \operatorname{det}(I-D^{1/2}U^{*}XD(x I-DX^{*}\Sigma XD)^{-1}DX^{*}UD^{1/2})=0.
    \end{gather}
Moreover, for $x \in \Gamma_i,  1 \leq i \leq k$ and $\eta:= n^{-1/2},$ we define $z_x=x+\ri \eta.$ According to Proposition \ref{eq_propooursidebulk}, we have that with high probability 
\begin{equation*}
   \left \| (DX^{*}\Sigma XD-z_x I)^{-1}+z^{-1}(I+m_{1n}(z_x)D^2)^{-1} \right \|_{\infty}=\rO(n^{-1/2-1/\alpha+\epsilon}). 
\end{equation*} 
Together with the arguments around (\ref{eq_koldistancecontrol}), using the fact $r$ is finite, $U$ contains the standard basis and the definition of $m_{2n}(z)$ in (\ref{eq_systemequationsm1m2elliptical}), we find that 
\begin{align*}
\left\|z^{-1}U^{*}XD(I+m_{1n}(z_x)D^2)^{-1}DX^{*}U-m_{2n}(z_x)I \right \|_{\infty}=\rO_{\mathbb{P}}(n^{-1/2-1/\alpha+2\epsilon}),
\end{align*}
where we used the fact that $|z_x| \asymp \mu_1$ and (\ref{eq_lowerboundmu1}).  According to a discussion similar to  equation (D.29) of \cite{ding2021spiked} and Lemma \ref{lem: basic bounds}, we see that for $t=1,2,$
\begin{equation*}
m_{tn}(z_x)-m_{tn}(x) \asymp \operatorname{Im} m_{tn}(z)=\rO(n^{-1/2-1/\alpha+\epsilon}).
\end{equation*}  
Combining all the above controls, we find that for some constant $C>0$
\begin{equation*}
\left \| I-D^{1/2}U^{*}XD(x I-DX^{*}\Sigma XD)^{-1}DX^{*}UD^{1/2} \right \|_{\infty}=C\max_{1 \leq j \leq r} |m_{2n}(z_x)+\widetilde{\sigma}_j^{-1}|+\rO_{\mathbb{P}} \left(n^{-1/2-1/\alpha+\epsilon} \right).
\end{equation*}
Since $x \in \Gamma_i,$ together with Lemma \ref{lem: basic bounds}, we see that $|m_{2n}(z_x)+\widetilde{\sigma}_j^{-1}| \gg n^{-1/2-1/\alpha+\epsilon}.$ This implies that $x$ is not an eigenvalue of $\widetilde{\mathcal{Q}}$ and completes the proof. 
\end{proof}

\quad As mentioned in the proof of Theorem 2.7 of \cite{bloemendal2016principal}, once Step (i) is done, Steps (ii) and (iii) are more standard.  For Step (ii), together with the interlacing results as in Lemma C.3 of \cite{ding2021spiked}, we perform the counting argument to prove (\ref{eq_sticking}) for a special case assuming $\widetilde{\sigma}_1>\widetilde{\sigma}_2 > \cdots >\widetilde{\sigma}_r.$ The details can be found in Lemma D.5 of \cite{ding2021spiked} or Lemma 5.5 of \cite{DJ22}.    
For Step (iii), we use a continuity argument for all possible configurations $\{\widetilde{\sigma}_i\}_{1 \leq i \leq r}.$ The details can be found in the proof of Theorem 3.7 of \cite{ding2021spiked}. Since most of the arguments can be made verbatim following lines of the counterparts of \cite{ding2021spiked} or \cite{DJ22} or \cite{bloemendal2016principal} or \cite{knowles2013isotropic}, we omit further details. This completes the proof.      
    
\end{proof}

\begin{proof}[\bf Proof of Corollary \ref{coro_distribution}]
(\ref{eq_universalityresult}) follows directly from (\ref{eq_sticking}), Theorem \ref{thm_main_unbounded}, and the discussions in Remark \ref{rem_multiplevaluespikedremark}. On the other hand, when $\mathbf{H}_a$ and the assumption that $r_*>r$ hold, we have that 
\begin{equation*}
\mathbb{T} \geq  \frac{\mu_r-\mu_{r+1}}{\mu_{r+1}-\mu_{r+2}}.
\end{equation*}
Moreover, according to part (2) of Theorem \ref{thm_main_spike} and (\ref{def1}) or (\ref{def3}), we see that with $1-\ro(1)$ probability, for some constant $C>0$
\begin{equation*}
(\mu_{r+1}-\mu_{r+2})^{-1} \geq C \mathsf{T}^{-1}.
\end{equation*}
In addition, according to part (1) of Theorem \ref{thm_main_spike}, we see from the assumption of (\ref{spiked_assumption}) and (\ref{def1}) or (\ref{def3}) that  
\begin{equation*}
\mu_{r}-\mu_{r+1} \geq C \widetilde{\sigma}_r. 
\end{equation*}
Combining the above two controls, we can prove (\ref{eq_deltan(1)ha}) and (\ref{eq_deltan(1)harproperty}) for the statistic $\mathbb{T}.$ The other two equations can be proved similarly for the statistic $\mathbb{T}_{r_0}$. This completes our proof.   
\end{proof}

\begin{proof}[\bf Proof of Theorem \ref{thm_boostrappingdiscussion}] The proof follows from  strategies similar to Theorem 2.4 of \cite{CHP} or Theorem 3.5 of \cite{2022arXiv220206188Y}. We focus on explaining the main ideas and omit the details. The core of the proof is to introduce the auxiliary  quantities $\theta_i, 1 \leq i \leq r,$ where for each $1 \leq i \leq r,$ $\theta_i$ satisfies the equation 
\begin{gather*}
    \frac{\theta_i}{\Tilde{\sigma}_i}=\Big(1-\frac{1}{n\theta_i}\sum_{j=r+1}^p\frac{\sigma_j}{1-\Tilde{\sigma}_i^{-1}\sigma_j}\Big)^{-1}.
\end{gather*}  
With the restriction that $\theta_i \in [\widetilde{\sigma}_i, 2 \widetilde{\sigma}_i],$ the existence of uniqueness of $\theta_i$ have been justified in \cite{CHP,2022arXiv220206188Y}. Furthermore, under the assumption of (\ref{spiked_assumption}), we can conclude from equation (2.10) of \cite{CHP} that 
\begin{equation}\label{eq_thetaipriorbound}
\theta_i/\widetilde{\sigma}_i=1+\mathrm{o}(1), 1 \leq i \leq r.
\end{equation}
 Moreover, following lines of the proof of Theorem 2.2 of \cite{CHP} or Lemma 3.4 of \cite{2022arXiv220206188Y}, we can conclude that for $1 \leq i \leq r,$ we can obtain that 
 \begin{gather*}
        \frac{\widehat{\lambda}_i}{\theta_i}=1+\mathbf{k}_i^{*}(V_1^{*}XX^{*}V_1-I_r)\mathbf{k}_i+\ro_{\mathbb{P}}(\frac{1}{\sqrt{n}}),
    \end{gather*}   
 where we recall (\ref{eq_twomatricesdecomposition}) for the definition of $V_1$ and $\mathbf{k}_i, 1 \leq i \leq r,$ are the standard basis in $\mathbb{R}^r.$ Together with Lemma \ref{lem:large deviation}, we see that 
 \begin{equation}\label{eq_firstreductiondistribution}
 \frac{\mu_i}{\widehat{\lambda}_i}= \frac{\mu_i}{\theta_i} \left(1+\rO_{\mathbb{P}}(n^{-1/2}) \right). 
\end{equation} 
The rest of the proof leaves to establish the asymptotics of $\mu_i/\theta_i.$ In the actual proof, it is more convenient to work with random quantities $\widehat{\theta}_i, 1 \leq i \leq r,$ defined according to 
\begin{gather}\label{eq_defnwidehattheta}
    \frac{\widehat{\theta}_i}{\Tilde{\sigma}_i}=\left(1-\frac{1}{n \widehat{\theta}_i}\sum_{j=r+1}^p\frac{\sigma_j}{1-\Tilde{\sigma}_i^{-1}\sigma_j}+\frac{1}{n \widehat{\theta}_i}\sum_{k=r+1}^p\frac{\sigma_k}{1-\frac{\sigma_k}{n \widehat{\theta}_i}\sum_{j=1}^n\xi^2_j/\mathbb{E}\xi^2}\right)^{-1}. 
\end{gather} 
Similar to $\theta_i,$ with the restriction $\widehat{\theta}_i \in [\widetilde{\sigma}_i, 2 \widetilde{\sigma}_i],$ we can obtain the uniqueness and existence of the solutions with high probability.

We first summarize some important properties of $\widehat{\theta}_i.$ On the one hand, it is easy to use the definitions of $\theta_i, \widehat{\theta}_i,$ the assumption of (\ref{spiked_assumption}) and (\ref{eq_thetaipriorbound}), we can obtain that 
\begin{equation*}
\frac{\theta_i}{\widehat{\theta}_i}=1+\mathrm{o}_{\mathbb{P}}(1).
\end{equation*}
In light of (\ref{eq_firstreductiondistribution}), we find that it suffices to work with $\mu_i/\widehat{\theta}_i.$ On the other hand,  using the above controls with (\ref{eq_defnwidehattheta}) and (\ref{def1}) or (\ref{def3}), we find that 
\begin{align*}
        \frac{1}{n}\sum_{j=1}^n\frac{\xi^2_j}{\mathbb{E}\xi^2}-\frac{\widehat{\theta}_i}{\Tilde{\sigma}_i}&=\frac{\widehat{\theta}_i}{\Tilde{\sigma}_i}\times\left(\frac{1}{n}\sum_{j=1}^n\frac{\xi^2_j-\mathbb{E}\xi^2-\frac{\xi^2_j}{n \widehat{\theta}_i}\sum_{k=r+1}^p\frac{\sigma_k}{1-\Tilde{\sigma}_1^{-1}\sigma_k}+\frac{\xi^2_i}{n \widehat{\theta}_i}\sum_{k=r+1}^p\frac{\sigma_k}{1-\frac{\sigma_k}{n \widehat{\theta}_i}\sum_{j=1}^n\xi^2_j/\mathbb{E}\xi^2}}{\mathbb{E}\xi^2}\right)\\
%        &=\frac{\theta_1}{\Tilde{\sigma}_1}\times\big(\frac{1}{n}\sum_{j=1}^n\frac{\xi^2_j-\mathbb{E}\xi^2_1}{\mathbb{E}\xi^2_1}\big)-\frac{\frac{1}{n}\sum_{j=1}^n\xi^2_j/\mathbb{E}\xi^2_1-\frac{\theta_1}{\Tilde{\sigma}_1}}{\Tilde{\sigma}_1}\times\frac{1}{n}\sum_{j=1}^n\frac{\frac{\xi^2_j}{n\theta_1}\sum_{k=r+1}^p\frac{\Tilde{\sigma}^2_k}{(1-n^{-1}\theta_1^{-1}\Tilde{\sigma}_k\sum_j\xi^2_j/\mathbb{E}\xi^2_1)(1-\Tilde{\sigma}_1^{-1}\Tilde{\sigma}_k)}}{\mathbb{E}\xi^2_1}\\
        &=\frac{1}{n}\sum_{j=1}^n\frac{\xi^2_j-\mathbb{E}\xi^2}{\mathbb{E}\xi^2}-\frac{\frac{1}{n}\sum_{j=1}^n\xi^2_j/\mathbb{E}\xi^2-\frac{\widehat{\theta}_i}{\Tilde{\sigma}_i}}{\Tilde{\sigma}_i}\times\ro_{\mathbb{P}}(1)+\ro_{\mathbb{P}}(n^{-1/2}),
\end{align*}
which implies that 
\begin{equation}\label{eq_finercontrolend}
\frac{\widehat{\theta}_i}{\widetilde{\sigma}_i}=1+\mathrm{o}_{\mathbb{P}}(n^{-1/2}). 
\end{equation}

Now we proceed to complete the proof following that of Theorem 2.4 of \cite{CHP}. For notational convenience, we now work with the rescaled matrix 
\begin{equation*}
\check{\mathcal Q}:= \check{D} X^* \widetilde{\Sigma} X \check{D}, \ \check{D}^2:=(\mathbb{E} \xi^2)^{-1} D^2, 
\end{equation*}
whose eigenvalues are denoted as $\check{\lambda}_1 \geq \check{\lambda}_2 \geq \cdots \geq \check{\lambda}_{ \{p \wedge n\}}>0.$ Note that $\mu_i=\mathbb{E} \xi^2 \check{\lambda}_i.$ By a discussion similar to (\ref{eq_masterequationeigenvalues}) and (\ref{eq_newmasterequation}), using (\ref{eq_twomatricesdecomposition}), we find that  $\check{\lambda}_i, 1 \leq i \leq r$ satisfy the equation
    \begin{gather*}
        \operatorname{det}(\Lambda_s^{-1}-V_1^{*}X\check{D}(\check{\lambda}_i I-\check{D}X^{*}\Sigma_oX\check{D})^{-1}\check{D}X^{*}V_1)=0.
    \end{gather*}
Denote $\mathbf{B}(x):=xI-\check{D}X^{*}\Sigma_oX\check{D}$ and $\delta_i=(\check{\lambda}_i-\widehat{\theta}_i)/\widehat{\theta}_i$, the above determinant can be rewritten into 
\begin{gather}\label{eq_generalldeterminant}
    \operatorname{det}(\widehat{\theta}_i\Lambda_s^{-1}-\widehat{\theta}_iV_1^{*}X\check{D}\mathbf{B}^{-1}(\widehat{\theta}_i)\check{D}X^{*}V_1+\delta_i\widehat{\theta}_i^2V_1^{*}X\check{D}\mathbf{B}^{-1}(\check{\lambda}_i)\mathbf{B}^{-1}(\widehat{\theta}_i)\check{D}X^{*}V_1)=0.
\end{gather}
Following the procedure in Section 7.1 of \cite{CHP} or Lemma C.5 of \cite{2022arXiv220206188Y}, we find that for $1 \leq i,l \leq r$ (recall that $\widetilde{\Sigma}$ is assumed to be diagonal)
   \begin{gather*}
        \widehat{\theta}_i\mathbf{e}_i^{*}V_1^{*}X\check{D}\mathbf{B}^{-1}(\widehat{\theta}_i)\check{D}X^{*}V_1\mathbf{e}_l=\mathbf{1}(l=i)\sum_{j=1}^nx_{kj}^2\xi^2_j/\mathbb{E}\xi^2+\ro_{\mathbb{P}}(n^{-1/2}).
    \end{gather*}
Similarly, by a discussion similar to Lemma C.6 of \cite{2022arXiv220206188Y}, we conclude that 
    \begin{gather*}
        \delta_i \widehat{\theta}_i^2[V_1^{*}X\check{D}\mathbf{B}^{-1}(\check{\lambda}_i)\mathbf{B}^{-1}(\widehat{\theta}_i)\check{D}X^{*}V_1]_{il}= \delta_i(\mathbf{1}(l=i)+\ro_{\mathbb{P}}(1)).
    \end{gather*}
Inserting the above two controls into (\ref{eq_generalldeterminant}), by the assumption of (\ref{eq_separation}), using Leibniz’s formula for determinant, one has that 
    \begin{gather}\label{eq: spike_det_est}
        \delta_i(1+\ro_{\mathbb{P}}(1))=\sum_{j=1}^nx_{ij}^2\xi^2_j/\mathbb{E}\xi^2-\frac{\widehat{\theta}_i}{\Tilde{\sigma}_i}+\ro_{\mathbb{P}}(n^{-1/2}).
    \end{gather}
    Combining (\ref{eq_finercontrolend}), we conclude that 
    \begin{equation*}
    \frac{\check{\lambda}_i-\widehat{\theta}_i}{\widehat{\theta}_i}=\sum_{j=1}^nx_{ij}^2\xi^2_j/\mathbb{E}\xi^2-1+\ro_{\mathbb{P}}(n^{-1/2}).
    \end{equation*}
Together with (\ref{eq_firstreductiondistribution}) and (\ref{eq_finercontrolend}), using central limit theorem, we can conclude the proof. 

\end{proof}

%\subsection{The spiked eigenvalues}

\section{Proof of some auxiliary lemmas}\label{appendix_last}

\subsection{Preliminary estimates: Proof of Lemmas \ref{lem: basic bounds} and \ref{localestimate2}}\label{sec_proofpreliminarystieltjes}

\subsubsection{Proof of Lemma \ref{lem: basic bounds}}

Due to similarity, we only prove the results for the separable covariance i.i.d. data model when $\xi^2$ decays  polynomially, i.e., when (\ref{eq:F(m,z)}) and (\ref{ass3.1}) hold. The other cases can be proved analogously and we omit the details.    

\begin{proof}[{\bf Proof}]
We start with the first statement.   We now abbreviate $F_n(m_{1n}(z)) \equiv F_{n}(m_{1n}(z),z)$ throughout the proof. 
For the real part,  it suffices to prove that with high probability for some $0<C_2<1<C_1$
\begin{equation}\label{eq_realcontrol}
\operatorname{Re} m_{1n}(z) \in \left[-C_1  \frac{\phi \bar{\sigma} E}{E^2+\eta^2}, -C_2  \frac{\phi \bar{\sigma} E}{E^2+\eta^2}\right].
\end{equation}
Moreover, by continuity and Theorem \ref{lem_solutionsystem}, it suffices to prove the following inequalities
\begin{equation}\label{eq)ccc}
 \operatorname{Re}F_n(-C_2 \phi \bar{\sigma} E (E^2+\eta^2)^{-1}+\ri \operatorname{Im} m_{1n}(z))<0, \  \operatorname{Re}F_n(-C_1 \phi \bar{\sigma} E (E^2+\eta^2)^{-1}+\ri \operatorname{Im} m_{1n}(z))>0.
\end{equation} 
We only focus on the first part. By definition, we have that   
 \begin{align}\label{eq_expansion}
 \operatorname{Re} & F_n(m_{1n}(z))=-\operatorname{Re}m_{1n}(z) \\
&  -\frac{1}{n}\sum_{i=1}^p\frac{\sigma_i\operatorname{Re}(z-\frac{\sigma_i}{n}\sum_{j=1}^n\frac{\xi^2_j}{1+m_{1n}(z)\xi^2_j})}{\operatorname{Re}^2(z-\frac{\sigma_i}{n}\sum_{j=1}^n\frac{\xi^2_j}{1+m_{1n}(z)\xi^2_j})+\operatorname{Im}^2(z-\frac{\sigma_i}{n}\sum_{j=1}^n\frac{\xi^2_j}{1+m_{1n}(z)\xi^2_j})}. \nonumber
 \end{align}
Note that 
 \begin{gather*}
     \operatorname{Re}\left(z-\frac{\sigma_i}{n}\sum_{j=1}^n\frac{\xi^2_j}{1+m_{1n}\xi^2_j} \right)=E-\frac{\sigma_i}{n}\sum_{j=1}^n\frac{\xi^2_j(1+\xi^2_j\operatorname{Re}m_{1n})}{(1+\xi^2_j\operatorname{Re}m_{1n})^2+\xi^4_j\operatorname{Im}^2m_{1n}},\\
     \operatorname{Im}\left(z-\frac{\sigma_i}{n}\sum_{j=1}^n\frac{\xi^2_j}{1+m_{1n}\xi^2_j}\right)=\eta+\frac{\sigma_i}{n}\sum_{j=1}^n\frac{\xi^4_j\operatorname{Im}m_{1n}}{(1+\xi^2_j\operatorname{Re}m_{1n})^2+\xi^4_j\operatorname{Im}^2m_{1n}}.
 \end{gather*}
By a discussion similar to (\ref{eq: 2.2 beta=2}), 
% If $\operatorname{Re}m_{1n}=m_{1n}(E)$, since $m_{1n}(E)\ge-(\xi^2_{(1)}+d_{\beta})^{-1}$, we have 
% \[
% |\frac{1}{n}\sum_{j=1}^n\frac{\xi^2_j(1+\xi^2_j\operatorname{Re}m_{1n})}{(1+\xi^2_j\operatorname{Re}m_{1n})^2+\xi^4_j\operatorname{Im}^2m_{1n}}|\le|\frac{1}{n}\sum_{j=1}^n\frac{\xi^2_j(\xi^2_{(1)}+d_{\beta})}{\xi^2_{(1)}+d_{\beta}-\xi^2_j}|\le \mu_1\times o(1),
% \]
% which further indicates that 
% \[
% \operatorname{Re}(z-\frac{\sigma_i}{n}\sum_{j=1}^n\frac{\xi^2_j}{1+m_{1n}\xi^2_j})>0,
% \]
% and then 
% \[
% \begin{split}
%     \operatorname{Re}F_n(m_{1n}(z),z)&\ge-m_{1n}(E)-\frac{1}{n}\sum_{i=1}^p\frac{\sigma_i}{E-\frac{\sigma_i}{n}\sum_{j=1}^n\frac{\xi^2_j(1+\xi^2_j\operatorname{Re}m_{1n})}{(1+\xi^2_j\operatorname{Re}m_{1n})^2+\xi^4_j\operatorname{Im}^2m_{1n}}}\\
%     &>-m_{1n}(E)-\frac{1}{n}\sum_{i=1}^p\frac{\sigma_i}{E-\frac{\sigma_i}{n}\sum_{j=1}^n\frac{\xi^2_j}{1+\xi^2_jm_{1n}(E)}}=0.
% \end{split}
% \]
if $\operatorname{Re}m_{1n}=-C_2 (\phi \bar{\sigma} E)/(E^2+\eta^2)$,
% 
% 
%  observe that 
% \[
% \frac{E}{(1+c)(E^2+\eta^2)}<\frac{1}{\mu_1},
% \]
we have that
% \[
% \begin{split}
%     \frac{1}{n}\sum_{j=1}^n\frac{\xi^{2}_j}{1+\xi^2_j\operatorname{Re}m_{1n}}&\le \frac{1}{n}\sum_{j=1}^n\frac{\xi^{2}_j}{1-\xi^2_j\phi\bar{\sigma}\mu_1^{-1}}\\
%     &\le C\frac{1}{n}\sum_{j=1}^n\frac{\xi^{2}_j}{1-\xi^2_j(\xi^2_{(1)}+d_{\beta})^{-1}}\\
%     &\le E\times o(1).
% \end{split}
% \]
%It follows that 
 \begin{equation}\label{eq_controlcontrolddd}
 \begin{split}
    \operatorname{Re}\left(z-\frac{\sigma_i}{n}\sum_{j=1}^n\frac{\xi^2_j}{1+m_{1n}\xi^2_j} \right)
%    &=E-\frac{\sigma_i}{n}\sum_{j=1}^n\frac{\xi^2_j(1+\xi^2_j\operatorname{Re}m_{1n})}{(1+\xi^2_j\operatorname{Re}m_{1n})^2+\xi^4_j\operatorname{Im}^2m_{1n}}\\
%    &\ge E-\frac{\sigma_i}{n}\sum_{j=1}^n\frac{\xi^2_j}{1+\xi^2_j\operatorname{Re}m_{1n}}\\
    &\ge E(1-\ro(1)),\\
     \operatorname{Im}\left(z-\frac{\sigma_i}{n}\sum_{j=1}^n\frac{\xi^2_j}{1+m_{1n}\xi^2_j} \right)
%     &=\eta+\frac{\sigma_i}{n}\sum_{j=1}^n\frac{\xi^4_j\operatorname{Im}m_{1n}}{(1+\xi^2_j\operatorname{Re}m_{1n})^2+\xi^4_j\operatorname{Im}^2m_{1n}}\\
%     &\le \eta+\frac{\sigma_i}{2n}\sum_{j=1}^n\frac{\xi^2_j}{1+\xi^2_j\operatorname{Re}m_{1n}}\\
     &\le\eta+ E\times \ro(1).
 \end{split}
 \end{equation}
% where we use the inequality $a^2+b^2\ge2ab$ for positive constants $a,b$. 
Therefore, together with (\ref{eq_expansion}), we see that 
\begin{equation}\label{eq: 2.3}
  \begin{split}
   \operatorname{Re}F_n(-C_2 \phi \bar{\sigma} E (E^2+\eta^2)^{-1}+\ri \operatorname{Im} m_{1n}(z))
  % &=-\operatorname{Re}m_{1n}-\frac{1}{n}\sum_{i=1}^p\frac{\sigma_i\operatorname{Re}(z-\frac{\sigma_i}{n}\sum_{j=1}^n\frac{\xi^2_j}{1+m_{1n}\xi^2_j})}{\operatorname{Re}^2(z-\frac{\sigma_i}{n}\sum_{j=1}^n\frac{\xi^2_j}{1+m_{1n}\xi^2_j})+\operatorname{Im}^2(z-\frac{\sigma_i}{n}\sum_{j=1}^n\frac{\xi^2_j}{1+m_{1n}\xi^2_j})}\\
   &\le C_2\phi\bar{\sigma}\frac{E}{(E^2+\eta^2)}-\frac{1}{n}\sum_{i=1}^p\frac{\sigma_i\times E(1-\ro(1))}{E^2+(\eta+E \times \ro(1))^2}\\
   &\le (C_2-1+\ro(1)) \frac{\phi\bar{\sigma}E}{(E^2+\eta^2)}<0,
\end{split}  
\end{equation}
for sufficient large $n$. This completes the discussion for the real part. For the complex part, the idea is similar and it suffices to prove that when $z \in \mathbf{D}_{u},$ for some constants $C_1, C_2>0$ 
\begin{equation}\label{eq_imaginarycontrol}
\operatorname{Im} m_{1n}(z) \in \left[ C_1 \frac{ \eta \phi \bar{\sigma}}{E^2+\eta^2}, C_2 \eta \left| \operatorname{Re} m_{1n}(z) \right| \right].
\end{equation} 
Equivalently, it suffices to prove that
\begin{equation*}
\operatorname{Im} F_n(\operatorname{Re} m_{1n}(z)+\ri C_2 \eta \left| \operatorname{Re} m_{1n}(z) \right|)<0,  \  \operatorname{Im} F_n\left(\operatorname{Re} m_{1n}(z)+\ri C_1 \frac{ \eta \phi \bar{\sigma}}{E^2+\eta^2} \right)>0,  
\end{equation*}
where by definition
\[
\operatorname{Im}F_n(m_{1n},z)=-\operatorname{Im}m_{1n}+\frac{1}{n}\sum_{i=1}^p\frac{\sigma_i\operatorname{Im}(z-\frac{\sigma_i}{n}\sum_{j=1}^n\frac{\xi^2_j}{1+m_{1n}\xi^2_j})}{\operatorname{Re}^2(z-\frac{\sigma_i}{n}\sum_{j=1}^n\frac{\xi^2_j}{1+m_{1n}\xi^2_j})+\operatorname{Im}^2(z-\frac{\sigma_i}{n}\sum_{j=1}^n\frac{\xi^2_j}{1+m_{1n}\xi^2_j})}.
\]
The proof of the above inequalities is similar to (\ref{eq)ccc}) using (\ref{eq_realcontrol}). We briefly discuss the proof of the first inequality in which case by a discussion similar to (\ref{eq_controlcontrolddd}) 
\[
\begin{split}
   \operatorname{Im}F_n(m_{1n},z)&=\eta\operatorname{Re}m_{1n}+\frac{\phi\bar{\sigma}\eta(1+E)}{E^2+\eta^2(1+E\times \ro(1))^2}\\
   &\le-C_2\frac{\phi\bar{\sigma}\eta E}{(E^2+\eta^2)}+\frac{\phi\bar{\sigma}\eta(1+E)}{E^2(1+\eta^2\times \ro(1))+\eta^2(1+2E\times \ro(1))}<0.
\end{split}
\]
This completes our proof.

For the second statement, from the first statement, we see that it is valid to write $m_{1n}(E).$ Since $\mu_1 \gg d_1$ holds with high probability (see (\ref{eq_lowerboundmu1})), it suffices to prove that for some constants $0<C_2<1<C_1,$ when $|E-\mu_1| \leq C d_1,$  
\begin{equation}\label{eq_boundbound}
m_{1n}(E) \in \left[-C_1 \frac{\phi \bar{\sigma}}{E},-C_2 \frac{\phi \bar{\sigma}}{E}\right].
\end{equation}
Due to simplicity, we again only focus on the proof of the upper bound. According to Theorem \ref{lem_solutionsystem} and (\ref{eq_functionFequal}), we shall have that $F_n(m_{1n}(E))=0.$ Moreover, since $F_{n}(m_{1n}(\mu_1))=0,$ to prove (\ref{eq_boundbound}), it suffices to prove
\begin{equation}\label{eq_reduceddiscussion}
F_{n}(-C_2 \phi \bar{\sigma}/E)<0, \ F_{n}(-C_1 \phi \bar{\sigma}/E)>0.
\end{equation}
Due to similarity, we focus our discussion on the first inequality. By definition, we have that 
\[
 \begin{split}
     F_n(-C_2 \phi \bar{\sigma}/E)&=C_2\frac{\phi\bar{\sigma}}{E}-\frac{1}{n}\sum_{i=1}^p\frac{\sigma_i}{E-\frac{\sigma_i}{n}\sum_{j=1}^n\frac{\xi^2_j}{1-\xi^2_j(C_2\frac{\phi\bar{\sigma}}{E})}}\\
     &=C_2 \frac{\phi\bar{\sigma}}{E}-\frac{1}{n}\sum_{i=1}^p\frac{\sigma_i}{E(1-\frac{\sigma_i}{n} \sum_{j=1}^n \frac{\xi_j^2}{E-C_2\xi_j^2 \phi \bar{\sigma}})}.
 \end{split}
 \]
By a discussion similar to (\ref{eq: 2.2 beta=2}), we further have that 
\begin{equation}\label{eq_similarcontrolone}
 F_n(-C_2 \phi \bar{\sigma}/E)=C_2 \frac{\phi \bar{\sigma}}{E}-\frac{1}{n} \sum_{i=1}^p \frac{\sigma_i}{E(1-\ro(1))}=(C_2-1+\ro(1)) \frac{\phi \bar{\sigma}}{E}<0.  
\end{equation}
This completes the proof of the first statement.

Finally, we prove the third statement using the first two statements. For $z \in \mathbf{D}_{u},$ by definition, we have that
\begin{equation}\label{eq: m_2n}
   \begin{split}
   m_{2n}(z)&=\frac{1}{n}\sum_{j=1}^n\frac{\xi^2_j}{-E-\ri\eta-(E+\ri\eta)\xi^2_j(\operatorname{Re}m_{1n}+\ri \operatorname{Im}m_{1n})}\\
   &=\frac{1}{n}\sum_{j=1}^n\frac{\xi^2_j \left[(-E-\xi^2_j(E\operatorname{Re}m_{1n}-\eta\operatorname{Im}m_{1n})+\ri\eta+\ri\xi^2_j(E\operatorname{Im}m_{1n}+\eta\operatorname{Re}m_{1n})) \right]}{(-E-\xi^2_j(E\operatorname{Re}m_{1n}-\eta\operatorname{Im}m_{1n}))^2+(-\eta-\xi^2_j(E\operatorname{Im}m_{1n}+\eta\operatorname{Re}m_{1n}))^2}.
\end{split} 
\end{equation}
According to the results in the first two statements, the definition of $\mathbf{D}_{u}$ in (\ref{eq_spectraldomainone}) and the elementary relation that $|m_{1n}(z)|=\rO(1),$ we find that for some constants $C_1, C_2>0,$ when $n$ is sufficiently large 
 \begin{equation*}
 |(-E-\xi^2_j(E\operatorname{Re}m_{1n}-\eta\operatorname{Im}m_{1n})+\ri\eta+\ri\xi^2_j(E\operatorname{Im}m_{1n}+\eta\operatorname{Re}m_{1n}))| \leq C_1 E, 
 \end{equation*}
and
\begin{equation*}
(-E-\xi^2_j(E\operatorname{Re}m_{1n}-\eta\operatorname{Im}m_{1n}))^2+(-\eta-\xi^2_j(E\operatorname{Im}m_{1n}+\eta\operatorname{Re}m_{1n}))^2 \geq C_2 (E+\xi_j^2)^2. 
\end{equation*}
Together with a discussion similar to (\ref{eq: 2.2 beta=2}), we readily see that for some large constant $C>0$
%then when $\eta\le O(1)$, by \eqref{eq: interval for Re} and \eqref{eq: interval for Im} we have,
\begin{equation}\label{eq_rem2n}
\begin{split}
    |m_{2n}(z)|&\le \frac{C}{n}\sum_{j=1}^n\frac{E \xi^2_j}{(E+\xi_j^2)^2} \le \frac{C}{n}\sum_{j=1}^n\frac{\xi^2_j}{E-\xi^2_j}=\rO(e_{2}).
\end{split}
\end{equation}
Then together with the definition of $m_n(z),$ we see that 
%On the other hand, if $\eta\rightarrow\infty$,
%\[
%|m_{2n}(z)|\le\frac{1}{n}\sum_{j=1}^n\frac{\xi^2_j}{E+\eta-\xi^2_j}=o(1).
%\]
%Therefore, we always have $|m_{2n}(z)|=O(e_{\beta})$ for $z\in\mathbf{D}$. 
%As for $|m_n(z)|$, we may observe that 
\[
\begin{split}
   |m_n(z)| \le\frac{1}{p|z|}\sum_{i=1}^p\frac{1}{|1+\sigma_i m_{2n}(z)|} \le\frac{1}{p |z|}\sum_{i=1}^p\frac{1}{1+\sigma_i|m_{2n}(z)|}
\le \frac{1}{|z|}=\rO(E^{-1}).
\end{split}
\]

To control the imaginary part, by \eqref{eq: m_2n}, we can write
\[
\operatorname{Im}m_{2n}(z)=\frac{1}{n}\sum_{j=1}^n\frac{\xi^2_j(\eta+\xi^2_j(E\operatorname{Im}m_{1n}+\eta\operatorname{Re}m_{1n}))}{(-E-\xi^2_j(E\operatorname{Re}m_{1n}-\eta\operatorname{Im}m_{1n}))^2+(-\eta-\xi^2_j(E\operatorname{Im}m_{1n}+\eta\operatorname{Re}m_{1n}))^2}.
\]
Combining with (\ref{eq_realcontrol}) and (\ref{eq_imaginarycontrol}), we see that for some constant $C>0$ 
\begin{equation}\label{eq_imm2n}
\begin{split}
   \operatorname{Im}m_{2n}(z)& \leq \frac{C}{n}\sum_{j=1}^n\frac{\xi^2_j(\eta+\xi^2_j\eta)}{(E+\xi^2_j(\rO(1)+\eta^2\times \rO(E^{-1})))^2+(\eta+\xi^2_j\times \rO(1)+\eta\times \rO(E^{-1}))^2}\\
   &=\rO\left(\frac{1}{n}\sum_{j=1}^n\frac{\eta\xi^4_j}{\rO(E^2)} \right)= \rO\left(\frac{\eta}{E} \right),
\end{split}
\end{equation}
where in the last step we used (\ref{def1}). 
Moreover, using the definition of $m_n(z)$ in (\ref{eq_systemequationsm1m2}), we can write
\[
\operatorname{Im}m_n(z)=\frac{1}{p}\sum_{i=1}^p\frac{\eta+\sigma_i\eta\operatorname{Re}m_{2n}+\sigma_i E\operatorname{Im}m_{2n}}{(E+\sigma_i E\operatorname{Re}m_{2n}-\sigma_i\eta\operatorname{Im}m_{2n})^2+(\eta+\sigma_i\eta\operatorname{Re}m_{2n}+\sigma_i E\operatorname{Im}m_{2n})^2}.
\]
Together with (\ref{eq_rem2n}) and (\ref{eq_imm2n}), we can easily see that 
\begin{equation*}
\operatorname{Im} m_n(z)=\rO\left( \frac{\eta }{E^2} \right).
\end{equation*}
This completes our proof. 
%\end{proof}
\end{proof}

\begin{remark}\label{rem_JHX}
We may observe from the proof of Lemma \ref{lem: basic bounds} that in many cases we can directly write $m_{1n}(\mu_1)$ without considering its imaginary part. For example, when (\ref{ass3.1}) holds, for $z_0=\mu_1+\ri\eta$, using (\ref{eq:F(m,z)}) and (\ref{eq: def of mu_1}), we see that 
\begin{align*}
    \lim_{\eta\downarrow0}\operatorname{Im}m_{1n}(z_0)&=\lim_{\eta\downarrow0}\frac{1}{n}\sum_i\frac{\sigma_i(\eta+\frac{\sigma_i}{n}\sum_{j=1}^n\frac{\xi^4_j\operatorname{Im}m_{1n}(z_0)}{|1+\xi^2_jm_{1n}(z_0)|^2})}{|z_0-\frac{\sigma_i}{n}\sum_{j=1}^n\frac{\xi^2_j}{1+\xi^2_jm_{1n}(z_0)}|^2}\\
    &=\lim_{\eta\downarrow0}\frac{1}{n}\sum_i\frac{\frac{\sigma_i^2}{n}\sum_{j=1}^n\frac{\xi^4_j\operatorname{Im}m_{1n}(z_0)}{|1+\xi^2_jm_{1n}(z_0)|^2}}{|\mu_1-\frac{\sigma_i}{n}\sum_{j=1}^n\frac{\xi^2_j}{1+\xi^2_jm_{1n}(z_0)}|^2}\\
    &=\Big(\frac{1}{n}\sum_i\frac{\frac{\sigma_i^2}{n}\sum_{j=1}^n\frac{(\xi^2_{(1)}+d_2)^2\xi^4_j}{|\xi^2_{(1)}+d_2-\xi^2_j|^2}}{|\mu_1-\frac{\sigma_i}{n}\sum_{j=1}^n\frac{(\xi^2_{(1)}+d_2)\xi^2_j}{\xi^2_{(1)}+d_2-\xi^2_j}|^2}\Big)\times\lim_{\eta\downarrow0}\operatorname{Im}m_{1n}(z_0).
%    &=\rO\Big(\frac{1}{n}\sum_i\frac{\sigma_i^2(\xi^2_{(1)}+d_2)^2e_2^2}{|\mu_1-(\xi^2_{(1)}+d_2)e_2|^2}\Big)\times\lim_{\eta\rightarrow0}\operatorname{Im}m_{1n}(z_0),
\end{align*}
Then by a discussion similar to (\ref{eq_imm2n}), using the fact that $\alpha \geq 2$, we find that for any $\mu_1\gtrsim\xi^2_{(1)}$ 
\begin{gather*}
    \lim_{\eta\downarrow0}\operatorname{Im}m_{1n}(z_0)=\mathrm{o}(1)\times \lim_{\eta\downarrow0}\operatorname{Im}m_{1n}(z_0),
\end{gather*}
which holds true if and only if $\lim_{\eta\downarrow0}\operatorname{Im}m_{1n}(z_0)=0$. This shows that $m_{1n}(\mu_1)$ is well-defined for $\mu_1\gtrsim\xi^2_{(1)}$.

\end{remark}

\subsubsection{Proof of Lemma \ref{localestimate2}}

Due to similarity, we focus our discussion on the separable covariance i.i.d. data model. The elliptical data model can be handled analogously using (\ref{eq:F(m,z)1}) instead of (\ref{eq:F(m,z)}) whenever it is necessary. We omit the details due to similarity.  

\begin{proof}[\bf Proof] We first prove the results when conditionally. Let $\Omega$ be the event satisfying (c) of Definition \ref{defn_probset}. According to Lemma \ref{lem_probabilitycontrol}, we find that $\mathbb{P}(\Omega)=1-\mathrm{o}(1).$ Now we choose a realization $\{\xi_i^2\} \in \Omega$ so that the proofs of parts (a) and (b)
are purely deterministic. 

\begin{proof}[\bf Proof of (a)] We start with (\ref{eq_conditionaledgedefinition}). According to (\ref{eq_systemequationsm1m2}), we have that
\begin{gather*}
    m_{1n}(z)=\frac{1}{n}\sum_{i=1}^p\frac{\sigma_i}{-z+\frac{\sigma_i}{n} \sum_{j=1}^n\frac{\xi^2_j}{1+\xi^2_jm_{1n}(z)}}.
\end{gather*}
To characterize the bulk of the spectrum, we take the imaginary part on the both sides of the above equation and let $\eta \downarrow 0$ to obtain that 
\begin{gather}
\begin{split}
    \operatorname{Im}m_{1n}(z)
%    &=\frac{1}{n}\sum_i\frac{-\sigma_i^2\operatorname{Im}(\frac{1}{n}\sum_j\frac{\xi^2_j}{1+\xi^2_jm_{1n}})}{\operatorname{Re}^2(E-\frac{\sigma_i}{n}\sum_j\frac{\xi^2_j}{1+\xi^2_jm_{1n}})+\operatorname{Im}^2(\frac{\sigma_i}{n}\sum_j\frac{\xi^2_j}{1+\xi^2_jm_{1n}})}\\
%    &=\frac{1}{n}\sum_i\frac{\sigma_i^2(\frac{1}{n}\sum_j\frac{\xi^4_j\operatorname{Im}(m_{1n})}{\operatorname{Re}^2(1+\xi^2_jm_{1n})+\xi^4_j\operatorname{Im}^2(m_{1n})})}{\operatorname{Re}^2(E-\frac{\sigma_i}{n}\sum_j\frac{\xi^2_j}{1+\xi^2_jm_{1n}})+\operatorname{Im}^2(\frac{\sigma_i}{n}\sum_j\frac{\xi^2_j}{1+\xi^2_jm_{1n}})}\\
    &=\frac{1}{n}\sum_{i=1}^p \frac{\sigma_i^2 \left(\frac{1}{n}\sum_j\frac{\xi^4_j\operatorname{Im}(m_{1n})}{\operatorname{Re}^2(1+\xi^2_jm_{1n})+\xi^4_j\operatorname{Im}^2(m_{1n})}\right)}{(E-\operatorname{Re}(\frac{\sigma_i}{n}\sum_j\frac{\xi^2_j}{1+\xi^2_jm_{1n}}))^2+\operatorname{Im}^2(\frac{\sigma_i}{n}\sum_j\frac{\xi^2_j}{1+\xi^2_jm_{1n}})}.
\end{split}
\end{gather}
The above equation can be further rewritten as 
%We may rewrite the above equation as
\begin{gather}\label{eq: supp for m1phi}
    0=\operatorname{Im}m_{1n}(z)(1-g (m_{1n},E)),   
%    \big(1-\frac{1}{n}\sum_i\frac{\sigma_i^2(\frac{1}{n}\sum_j\frac{\xi^4_j}{\operatorname{Re}^2(1+\xi^2_jm_{1n})+\xi^4_j\operatorname{Im}^2(m_{1n})})}{(E-\operatorname{Re}(\frac{\sigma_i}{n}\sum_j\frac{\xi^2_j}{1+\xi^2_jm_{1n}}))^2+\operatorname{Im}^2(\frac{\sigma_i}{n}\sum_j\frac{\xi^2_j}{1+\xi^2_jm_{1n}})}\big),
\end{gather}
where $g(m_{1n},E)$ is denoted as 
\begin{gather*}
    g(m_{1n},E):=\frac{1}{n}\sum_{i=1}^p\frac{\sigma_i^2 \left(\frac{1}{n}\sum_j\frac{\xi^4_j}{\operatorname{Re}^2(1+\xi^2_jm_{1n})+\xi^4_j\operatorname{Im}^2(m_{1n})}\right)}{(E-\operatorname{Re}(\frac{\sigma_i}{n}\sum_j\frac{\xi^2_j}{1+\xi^2_jm_{1n}}))^2+\operatorname{Im}^2(\frac{\sigma_i}{n}\sum_j\frac{\xi^2_j}{1+\xi^2_jm_{1n}})}.
\end{gather*}
Similar to the arguments used in \cite{Kwak2021,lee2016extremal}, it is easy to see that for any fixed $\operatorname{Re}m_{1n}<-l^{-1}$ and $E$, $g(m_{1n},E)\rightarrow0$ when $|\operatorname{Im}m_{1n}|\rightarrow\infty,$ and $g(m_{1n},E)\rightarrow+\infty$ in order to satisfy (\ref{eq: supp for m1phi}) when $|\operatorname{Im}m_{1n}|\rightarrow 0.$ 
%
%
%
%, we have , as the following term tends to $+\infty$,
%\[
%\frac{1}{n}\sum_i\frac{\frac{\sigma_i^2}{n}\sum_j\frac{\xi^4_j}{(1+\xi^2_j\operatorname{Re}m_{1n})^2}}{(E-\frac{\sigma_i}{n}\sum_j\frac{\xi^2_j}{(1+\xi^2_j\operatorname{Re}m_{1n})})^2},
%\]
%note that 
%\[
%\frac{1}{n}\sum_j\frac{\xi^4_j}{(1+\xi^2_j\operatorname{Re}m_{1n})^2}/(\frac{1}{n}\sum_j\frac{\xi^2_j}{(1+\xi^2_j\operatorname{Re}m_{1n})})^2\rightarrow+\infty,
%\]
%as $n^{-1}\sum_j\frac{\xi^2_j}{(1+\xi^2_j\operatorname{Re}m_{1n})}\rightarrow\infty$.
Therefore, by monotonicity, there exists a unique $\operatorname{Im}m_{1n}>0$ such that \eqref{eq: supp for m1phi} holds, which corresponds to the bulk of the spectrum. 

Furthermore, for any fixed $\operatorname{Re} m_{1n} >-l^{-1}$ and fixed  $E$ so that Theorem \ref{lem_solutionsystem} holds, we have that $g(m_{1n},E)$ is monotone decreasing in terms of $|\operatorname{Im} m_{1n}|.$
%
%
%outside the bulk of the spectrum (i.e., $\operatorname{Im} m_{1n}(z)=0$), using the definitions in (\ref{eq_finitesample123}), it is elementary to see that  
%\begin{gather*}
%    \begin{split}
%        \frac{\partial g(-l^{-1},E)}{\partial E}&=\frac{1}{n}\sum_{i=1}^p \frac{-2\sigma_i^2\widehat{\varsigma}_1}{(E-\sigma_i\widehat{\varsigma}_2)^3}.
%    \end{split}
%\end{gather*}
Let $E_+$ be defined according to $m_{1n}(E_+)=-l^{-1}.$ In view of (\ref{eq:F(m,z)}) and (\ref{eq_functionFequal}), we have that 
\begin{equation*}
l^{-1}=\frac{1}{n} \sum_{i=1}^p  \frac{ \sigma_i}{E_+-\sigma_i \widehat{\varsigma}_2}. 
\end{equation*}
Let $\widetilde{\varsigma}_3$ be defined similarly as $\widehat{\varsigma}_3$ in (\ref{eq_finitesample123}) by replacing $\widehat{L}_+$ with $E_+.$ Based on the above arguments and definitions, it is easy to see that  
%In addition, since $\varsigma_1>0,$ we have that 
%\begin{align*}
%\frac{-2}{n} \sum_{i=1}^p  \frac{\sigma_i^2 \widehat{\varsigma}_1}{(E_+-\sigma_i \widehat{\varsigma}_2)^3} \asymp -\frac{1}{n} \sum_{i=1}^p \frac{\sigma_i}{E_+-\sigma_i \widehat{\varsigma}_2}=-l^{-1}<0. 
%\end{align*}
%
%
%{\color{red} from here}
%
%As we can see that $\hat{\varsigma}_1,\hat{\varsigma}_2>0$, the above quantity will greater than zero when $E<\sigma_p\hat{\varsigma}_2$ while it will less than zero when $E>\sigma_1\hat{\varsigma}_2$. Now, we set $E_{\phi}$ such that 
%\begin{equation}
%    0=\frac{1}{n}\sum_i\frac{2\sigma_i^2\hat{\varsigma}_1}{(E_{\phi}-\sigma_i\hat{\varsigma}_2)^3}.
%\end{equation}
%Then one has for $E>E_{\phi}$, $\frac{\partial g(-l^{-1},E)}{\partial E}<0$.  {\color{red} from here}. 
%
%
%For any fixed $E>E_{\phi}$, we observe that $g(m_{1n},E)$ is a monotone decrease function of $|\operatorname{Im}m_{1n}|$ so that 
\[
\begin{split}
    &\sup_{\operatorname{Re}m_{1n}\in(-l^{-1},\infty)} g(m_{1n}, E)=g(-l^{-1}, E_+)=\phi \widetilde{\varsigma}_3.
%    \frac{1}{n}\sum_i\frac{\sigma_i^2(\frac{1}{n}\sum_j\frac{\xi^4_j}{\operatorname{Re}^2(1+\xi^2_jm_{1n})+\xi^4_j\operatorname{Im}^2(m_{1n})})}{(E-\operatorname{Re}(\frac{\sigma_i}{n}\sum_j\frac{\xi^2_j}{1+\xi^2_jm_{1n}}))^2+\operatorname{Im}^2(\frac{\sigma_i}{n}\sum_j\frac{\xi^2_j}{1+\xi^2_jm_{1n}})}
    \\
 %   &=\sup_{E\in(E_{\phi},\infty)}\frac{1}{n}\sum_i\frac{\frac{\sigma_i^2}{n}\sum_j\frac{l^2\xi^4_j}{(l-\xi^2_j)^2}}{(E-\frac{\sigma_i}{n}\sum_j\frac{l\xi^2_j}{(l-\xi^2_j)})^2}\\
 %   &=\frac{1}{n}\sum_i\frac{\frac{\sigma_i^2}{n}\sum_j\frac{l^2\xi^4_j}{(l-\xi^2_j)^2}}{(E_{\phi}-\frac{\sigma_i}{n}\sum_j\frac{l\xi^2_j}{(l-\xi^2_j)})^2}\\
 %   &=\phi\hat{\varsigma}_3<1.
\end{split}
\]
Assuming that $\phi \widetilde{\varsigma}_3<1,$ we conclude that (\ref{eq: supp for m1phi}) holds only if $\operatorname{Im} m_{1n}(z)=0$ which corresponds to the outside part of the spectrum. This shows that 
%
%
%Moreover, we can show that $\hat{L}_{+}\ge E_{\phi}$. In fact, we only need to show that
%\begin{gather}
%    \frac{1}{n}\sum_i\frac{-2\sigma_i^2\varsigma_1}{(\hat{L}_{+}-\sigma_i\varsigma_2)^3}<0.
%\end{gather}
%Since we may find constants $C_{\pm}>0$ such that $C_{+}\le\frac{t\varsigma_2}{(\hat{L}_{+}-t\hat{\varsigma}_2)^2}\ge C_{-}$ for $t$ in the range of support of $dF^n_{\Sigma}(t)$, it follows that
%\begin{gather}
%    \frac{C_{+}}{n}\sum_i\frac{-2\sigma_i}{\hat{L}_{+}-\sigma_i\hat{\varsigma}_2}\le\frac{1}{n}\sum_i\frac{-2\sigma_i^2\hat{\varsigma}_1}{(\hat{L}_{+}-\sigma_i\hat{\varsigma}_2)^3}\le \frac{C_{-}}{n}\sum_i\frac{-2\sigma_i}{\hat{L}_{+}-
%    \sigma_i\hat{\varsigma}_2}=-2C_{-}<0.
%\end{gather}
$m_{1n}=-l^{-1}$ is at the right edge of the spectrum and gives the expression of the end point $\widehat{L}_{+}$ as in (\ref{eq_conditionaledgedefinition}). Therefore, $\widehat{L}_+=E_+$ and $\widehat{\varsigma}_3=\widetilde{\varsigma}_3.$ This completes the proof. 

% in \eqref{eq: def of L+}.

\quad Second, the proof of \eqref{eq: concave decay of rho_Q} follows from an argument similar to Lemma 8.4 of \cite{lee2016extremal} utilizing the estimate (\ref{ass3.4}), we omit the details. 

\quad Third, we prove (\ref{eq_closenessequation}). The closeness of $\varsigma_k$ and $\widehat{\varsigma}_k, k=1,2,3,4,$  follows from arguments similar to the last equation of (\ref{def4}). Now we proceed to the proof of the second equation.
According to the proof of (\ref{eq_conditionaledgedefinition}) and an analogous argument, we found that $m_{1n}(\widehat{L}_+)=-l^{-1}$ and $m_{1n,c}(L_+)=-l^{-1}.$ Together with the definitions of $\varsigma_2, \widehat{\varsigma}_2,$ $m_{2n}$ and $m_{2n,c},$ we find that 
\begin{equation}\label{eq_deterministicrelation}
\varsigma_2=-m_{2n,c}(L_+) L_+, \ \widehat{\varsigma}_2=-m_{2n}(\widehat{L}_+) \widehat{L}_+. 
\end{equation} 
Next, by (\ref{eq_conditionaledgedefinition}) and an analogous argument for $L_+$ (see (\ref{eq_onlyoneequationdecide}) and the proof of part II below), we have that 
\begin{gather}\label{eq_denomimatorcontrol}
    \begin{split}
        0&=\frac{1}{n}\sum_{i=1}^p \frac{-l\sigma_i}{(-L_{+}+\sigma_i\varsigma_2)}+\frac{1}{n}\sum_{i=1}^p \frac{l\sigma_i}{(-\widehat{L}_{+}+\sigma_i\widehat{\varsigma}_2)}\\
        &=\frac{1}{n}\sum_i\frac{-l\sigma_i}{(-L_{+}+\sigma_i\varsigma_2)}+\frac{1}{n}\sum_i\frac{l\sigma_i}{(-L_{+}+\sigma_i\widehat{\varsigma}_2)}+\frac{1}{n}\sum_i\frac{l\sigma_i(\widehat{L}_{+}-L_{+})}{(-\widehat{L}_{+}+\sigma_i \widehat{\varsigma}_2)(-L_{+}+\sigma_i\widehat{\varsigma}_2)}.
    \end{split}
\end{gather} 
By first equation of (\ref{eq_closenessequation}), (\ref{eq_deterministicrelation}) and Assumption \ref{assum_additional_techinical}, we can conclude our proof.

\quad Fourth, we work with (\ref{eq_expansionlinear}) and (\ref{eq_a11coro}). Due to similarity, we only prove (\ref{eq_expansionlinear}). According to (\ref{lem_solutionsystem}), we have that 
\begin{gather*}
    m_{1n}(z)=\frac{1}{n}\sum_{i=1}^p \frac{\sigma_i}{-z+\frac{\sigma_i}{n}\sum_{j=1}^n\frac{\xi^2_j}{1+\xi^2_jm_{1n}(z)}}.
\end{gather*}
Consequently, it is easy to see that for $z= \widehat{L}_{+}-\kappa+\ri\eta\in\mathbf{D}_b,$ 
\begin{gather}\label{eq_decompositionmmmmmm}
\begin{split}
    m_{1n}(\widehat{L}_{+})-m_{1n}(z)
%    &=\frac{1}{n}\sum_i\frac{\sigma_i(-z+\frac{\sigma_i}{n}\sum_j\frac{\xi^2_j}{1+\xi^2_jm_{1n}(z)}+\widehat{L}_{+}-\frac{\sigma_i}{n}\sum_j\frac{\xi^2_j}{1+\xi^2_jm_{1n}(\widehat{L}_{+})})}{(-\widehat{L}_{+}+\frac{\sigma_i}{n}\sum_j\frac{\xi^2_j}{1+\xi^2_jm_{1n}(\widehat{L}_{+})})(-z+\frac{\sigma_i}{n}\sum_j\frac{\xi^2_j}{1+\xi^2_jm_{1n}(z)})}\\
%    &=\frac{(\widehat{L}_{+}-z)}{n}\sum_i\frac{\sigma_i}{(-\widehat{L}_{+}+\frac{\sigma_i}{n}\sum_j\frac{\xi^2_j}{1+\xi^2_jm_{1n}(\widehat{L}_{+})})(-z+\frac{\sigma_i}{n}\sum_j\frac{\xi^2_j}{1+\xi^2_jm_{1n}(z)})}\\
%    &+\frac{1}{n}\sum_i\frac{\frac{\sigma_i^2}{n}\sum_j\frac{\xi^4_j(m_{1n}(\widehat{L}_{+})-m_{1n}(z))}{(1+\xi^2_jm_{1n}(\widehat{L}_{+}))(1+\xi^2_jm_{1n}(z))}}{(-\widehat{L}_{+}+\frac{\sigma_i}{n}\sum_j\frac{\xi^2_j}{1+\xi^2_jm_{1n}(\widehat{L}_{+})})(-z+\frac{\sigma_i}{n}\sum_j\frac{\xi^2_j}{1+\xi^2_jm_{1n}(z)})}\\
    =R_1(\widehat{L}_{+}-z)+R_2(m_{1n}(\widehat{L}_{+})-m_{1n}(z)),
\end{split}
\end{gather}
where we denote 
\begin{gather*}
    R_1:=\frac{1}{n}\sum_{i=1}^p \frac{\sigma_i}{(-\widehat{L}_{+}+\frac{\sigma_i}{n}\sum_j\frac{\xi^2_j}{1+\xi^2_jm_{1n}(\widehat{L}_{+})})(-z+\frac{\sigma_i}{n}\sum_j\frac{\xi^2_j}{1+\xi^2_jm_{1n}(z)})}\\
    R_2:=\frac{1}{n}\sum_{i=1}^p \frac{\frac{\sigma_i^2}{n}\sum_j\frac{\xi^4_j}{(1+\xi^2_jm_{1n}(\widehat{L}_{+}))(1+\xi^2_jm_{1n}(z))}}{(-\widehat{L}_{+}+\frac{\sigma_i}{n}\sum_j\frac{\xi^2_j}{1+\xi^2_jm_{1n}(\widehat{L}_{+})})(-z+\frac{\sigma_i}{n}\sum_j\frac{\xi^2_j}{1+\xi^2_jm_{1n}(z)})}.
\end{gather*}
To study the terms $R_1$ and $R_2,$ we will need the following control whose proof follows from equations (4.24)-(4.28) of \cite{lee2016extremal} 
\begin{equation}\label{eq_boundused}
\frac{1}{n} \sum_{j=1}^n \frac{\xi_j^4}{(1-\xi_j^2 l^{-1})(1+\xi_j^2 m_{1n}(z))}=
\begin{cases}
\rO\left( \log n \right),  & d \geq 2; \\
\rO \left( |l^{-1}+m_{1n}(z)|^{d-2}\log n \right), & 1<d \leq 2.
\end{cases}
\end{equation}
For the denominator of $R_1$, since $z \in \mathbf{D}_b,$ by a discussion similar to (\ref{eq_denomimatorcontrol}), we find they are bounded from below so that $R_1=\rO(1).$  Furthermore, since $m_{1n}(\widehat{L}_+)=-l^{-1},$ by a straightforward calculation, using the definition of $\widehat{\varsigma}_2$ in (\ref{eq_finitesample123}) and the control (\ref{eq_boundused}), we observe that 
 \begin{gather}\label{eq_R1control}
    \begin{split}
        R_1&=\widehat{\varsigma}_4-\frac{1}{n}\sum_{i=1}^p \frac{\sigma_i (z-\widehat{L}_{+})+\frac{\sigma_i^2}{n}\sum_j\frac{\xi^4_j(m_{1n}(z)+l^{-1})}{(1-\xi^2_jl^{-1})(1+\xi^2_jm_{1n}(s))}}{(-\widehat{L}_{+}+\sigma_i\widehat{\varsigma}_2)^2(-z+\frac{\sigma_i}{n}\sum_j\frac{\xi^2_j}{1+\xi^2_jm_{1n}(z)})}\\
        &=\widehat{\varsigma}_4+\rO(|z-\widehat{L}_{+}|)+\rO(|m_{1n}(z)+l^{-1}|^{ \min\{d-1,1\}}\log n),
%        &=\frac{ \widehat{\varsigma}_3}{ \widehat{\varsigma}_1}+\rO(|z-\widehat{L}_{+}|)+\rO(|m_{1n}(z)+l^{-1}|^{\min\{d-1,1\}}\log n).
    \end{split}
\end{gather}
where in the second step we again used an argument similar to (\ref{eq_denomimatorcontrol}). Similarly, for $R_2,$ we find that 
\begin{gather}\label{eq_R2control}
\begin{split}
R_2&=\phi \widehat{\varsigma}_3+\rO(|z-\widehat{L}_{+}|)+\rO(|m_{1n}(z)+l^{-1}|^{\min\{d-1,1\}}\log n).
\end{split}
\end{gather}
We next provide a useful deterministic control. Using a discussion similar to \cite[Lemma A.4]{Kwak2021}, we find from (\ref{eq_systemequationsm1m2}) that 
\begin{equation}\label{eq_expansionusefullessorequaltoone}
\frac{1}{n} \sum_{i=1}^p \frac{\sigma_i^2 \frac{1}{n} \sum_{j=1}^n \frac{\xi_j^4}{(1+\xi_j^2 m_{1n}(z))^2} }{|-z+\frac{\sigma_i}{n} \sum_{j=1}^n \frac{\xi_j^2}{1+\xi_j^2 m_{1n}(z)}|^2}=1-\frac{1}{n}\sum_i\frac{\sigma_i\eta/\operatorname{Im}m_{1n}(z)}{|-z+\frac{\sigma_i}{n}\sum_j\frac{\xi^2_j}{1+\xi^2_jm_{1n}(z)}|^2}=1-\eta \frac{|m_{1n}(z)|^2}{\operatorname{Im} m_{1n}(z)}.
\end{equation}
Since $\operatorname{Im} m_{1n}(z)>0,$ this implies that
\begin{equation*}
0 \leq 1-\frac{1}{n}\sum_i\frac{\sigma_i\eta/\operatorname{Im}m_{1n}(z)}{|(-z+\frac{\sigma_i}{n}\sum_j\frac{\xi^2_j}{1+\xi^2_jm_{1n}(z)})|^2} \leq 1. 
\end{equation*} 
Together with Cauchy-Schwarz inequality, we see that 
\begin{gather*}
    \begin{split}
       |R_2| \leq (\phi\widehat{\varsigma}_3)^{1/2}\left(1-\frac{1}{n}\sum_i\frac{\sigma_i\eta/\operatorname{Im}m_{1n}(z)}{|(-z+\frac{\sigma_i}{n}\sum_j\frac{\xi^2_j}{1+\xi^2_jm_{1n}(z)})|^2}\right)^{1/2} <1,
    \end{split}
\end{gather*} 
where we used the fact $\operatorname{Im}m_{1n}(z)>0$ and $\eta>0.$ Using (\ref{eq_decompositionmmmmmm}), we find that $m_{1n}(\widehat{L}_+)-m_{1n}(z) \asymp \widehat{L}_+-z.$ Then we can conclude our proof using (\ref{eq_decompositionmmmmmm}), (\ref{eq_R1control}) and (\ref{eq_R2control}).

\quad Finally, we prove the controls for the imaginary parts. For (\ref{eq_zopointrate11}), the discussion is  similar to that of Lemma 4.5 of \cite{Kwak2021}. According to (\ref{eq_functionFequal}), we find see that 
\begin{gather}\label{eq_diudiudiudiudiu}
   \begin{split}
       -m_{1n}(z)&=\frac{1}{n}\frac{\sigma_{1}}{z-\frac{\sigma_{1}}{n}\sum_{j=1}^n\frac{\xi^2_j}{1+\xi^2_jm_{1n}(z)}}+\frac{1}{n}\sum_{i=2}^{p}\frac{\sigma_i}{z-\frac{\sigma_i}{n}\sum_{j=1}^n\frac{\xi^2_j}{1+\xi^2_jm_{1n}(z)}}\\
       &=\rO(\frac{1}{n\eta})+\frac{1}{n}\sum_{i=2}^{p}\frac{\sigma_i}{z-\frac{\sigma_i}{n}\sum_{j=1}^n\frac{\xi^2_j}{1+\xi^2_jm_{1n}(z)}},
   \end{split} 
\end{gather}
where in the step we the fact that $\operatorname{Im} m_{1n}(z) \geq 0$ and the trivial bound that 
\begin{equation}\label{eq_trivialcontroleta}
\frac{1}{n} \left|\left(z-\frac{\sigma_{1}}{n}\sum_{j=1}^n\frac{\xi^2_j}{1+\xi^2_jm_{1n}(z)}\right)^{-1} \right| \leq n^{-1} \left(\eta+\frac{\sigma_1}{n} \sum_{j=1}^n \frac{\xi_j^4 \operatorname{Im} m_{1n}(z)}{|1+\xi_j^2 m_{1n}(z)|^2} \right)^{-1} \leq (n \eta)^{-1}. 
\end{equation}
Taking the imaginary part on both sides of (\ref{eq_diudiudiudiudiu}), we see that for some constant $0<c<1,$
\begin{gather*}
\begin{split}
    \operatorname{Im}m_{1n}(z)&=\frac{1}{n}\sum_{i=2}^{p}\frac{\sigma_i(\eta+\frac{\sigma_i}{n}\sum_{j=1}^n\frac{\xi^4_j\operatorname{Im}m_{1n}(z)}{|1+\xi^2_jm_{1n}(z)|})}{|z-\frac{\sigma_i}{n}\sum_{j=1}^n\frac{\xi^2_j}{1+\xi^2_jm_{1n}(z)}|^2}+\rO(\frac{1}{n\eta})\\
    &=\frac{1}{n}\sum_{i=2}^{p}\frac{\sigma_i\eta}{|z-\frac{\sigma_i}{n}\sum_{j=1}^n\frac{\xi^2_j}{1+\xi^2_jm_{1n}(z)}|^2}+\frac{1}{n}\sum_{i=2}^{p}\frac{\frac{\sigma_i^2}{n}\sum_{j=1}^n\frac{\xi^4_j\operatorname{Im}m_{1n}(z)}{|1+\xi^2_jm_{1n}(z)|}}{|z-\frac{\sigma_i}{n}\sum_{j=1}^n\frac{\xi^2_j}{1+\xi^2_jm_{1n}(z)}|^2}+\rO(\frac{1}{n\eta})\\
    &=\rO(\eta)+\rO(\frac{1}{n\eta})+ c \operatorname{Im} m_{1n}(z),
    \end{split}
\end{gather*}
where in the last step we used discussions similar to (\ref{eq_L1bound}) and (\ref{eq_L3BOUND}) below. This concludes the proof. Then we prove (\ref{eq_oneregimeedgecontrol}) and (\ref{eq_zopointrate}) following \cite[Lemma 5.2]{lee2016extremal}. Due to similarity, we focus our analysis on $m_{1n}(z)$ and discuss $m_n(z)$ briefly in the end. In what follows, for notational simplicity, without loss of generality, we assume that on $\Omega,$ $\xi_{(i)}^2=\xi_i^2.$ In what follows, we identify $\eta \equiv \eta_0$ till the end of the proof of the lemma.  According to (\ref{eq_systemequationsm1m2}), we find that
\begin{align}\label{eq_decompositionl1l2l3}
    \operatorname{Im}m_{1n}(z)&=\frac{1}{n}\sum_{i=1}^p\frac{\sigma_i(\eta+\frac{\sigma_i}{n}\sum_{j=1}^n\frac{\xi^4_j\operatorname{Im}m_{1n}(z)}{|1+\xi^2_jm_{1n}(z)|^2})}{|z-\frac{\sigma_i}{n}\sum_{j=1}^n\frac{\xi^2_j}{1+\xi^2_jm_{1n}(z)}|^2} \nonumber \\
    &=\frac{1}{n}\sum_{i=1}^p \frac{\sigma_i\eta}{|z-\frac{\sigma_i}{n}\sum_{j=1}^n\frac{\xi^2_j}{1+\xi^2_jm_{1n}(z)}|^2}+\frac{1}{n}\sum_{i=1}^p \frac{\frac{\sigma_i^2}{n}\frac{\xi^4_1\operatorname{Im}m_{1n}(z)}{|1+\xi^2_1m_{1n}(z)|^2}}{|z-\frac{\sigma_i}{n}\sum_{j=1}^n\frac{\xi^2_j}{1+\xi^2_jm_{1n}(z)}|^2}+ \frac{1}{n} \sum_{i=1}^p \frac{\frac{\sigma_i^2}{n}\sum_{j=2}^n\frac{\xi^4_j\operatorname{Im}m_{1n}(z)}{|1+\xi^2_jm_{1n}(z)|^2}}{|z-\frac{\sigma_i}{n}\sum_{j=1}^n\frac{\xi^2_j}{1+\xi^2_jm_{1n}(z)}|^2} \nonumber \\
    &=\mathsf{L}_1+\mathsf{L}_2+\mathsf{L}_3.  
    \end{align}
For the denominator, by the results and arguments in Section \ref{sec_proofpartiboundedlocallaw}, we observe that when $z\in \mathbf{D}_b^\prime$  
\begin{align*}
z-\frac{\sigma_i}{n} \sum_{j=1}^n \frac{\xi_j^2}{1+\xi_j^2 m_{1n}(z)}& = z-\frac{\sigma_i}{n} \sum_{j=2}^n \frac{\xi_j^2}{1+\xi_j^2 m_{1n}(z)}+\frac{\sigma_i}{n} \frac{\xi_1^2}{1+\xi_1^2 m_{1n}(z)} \\
&=z-\frac{\sigma_i}{n} \sum_{j=2}^n \frac{\xi_j^2}{1+\xi_j^2 m_{1n,c}(z)}+\rO((n \eta)^{-1}+n^{-1/2-1/(d+1)}). 
\end{align*} 
Together with Assumption \ref{assum_additional_techinical} and (\ref{def4}), we find that for some small constant $c'>0,$ when $n$ is sufficiently large, 
\begin{equation*}
\left| z-\frac{\sigma_i}{n} \sum_{j=1}^n \frac{\xi_j^2}{1+\xi_j^2 m_{1n}(z)} \right| \geq c'. 
\end{equation*}
This implies that 
\begin{equation}\label{eq_L1bound}
\mathsf{L}_1 \asymp \eta. 
\end{equation}

For $\mathsf{L}_2,$ on the one hand, when $|z-z_0| \geq Cn^{-1/2+3 \epsilon_d},$ by (\ref{eq_a11coro}), we conclude that on $\Omega,$ for some constant $C>0$
\begin{equation}\label{eq_L2bound1}
|\mathsf{L}_2| \leq n^{-1/2-3 \epsilon_d} \operatorname{Im} m_{1n}(z).  
\end{equation}
On the other hand, when $z=z_0$ so that $\operatorname{Re} m_{1n}(z)=-\xi_1^2,$ we can rewrite $\mathsf{L}_2$ as 
\begin{equation}\label{eq_L2bound2}
\mathsf{L}_2=\frac{1}{n}\sum_{i=1}^ p\frac{\frac{\sigma_i}{n\operatorname{Im}m_{1n}(z)}}{|z-\frac{\sigma_i}{n}\sum_{j}\frac{\xi^2_j}{1+\xi^2_jm_{1n}(z)}|^2}. 
\end{equation}

Next, for $\mathsf{L}_3,$ by (\ref{eq_defnmathsfW}), (\ref{def4}) and the results and arguments in Section \ref{sec_proofpartiboundedlocallaw}, using the trivial bound for $\mathsf{L}_2$ that $|\mathsf{L}_2|=\rO((n\eta)^{-1}),$ we conclude that when $z \in \mathbf{D}_b^{\prime},$ for some constant $0<\mathfrak{c}<1$
\begin{equation}\label{eq_L3BOUND}
|\mathsf{L}_3| \leq \mathfrak{c} \operatorname{Im} m_{1n}(z). 
\end{equation}

\quad Consequently, we find that (\ref{eq_oneregimeedgecontrol}) follows from (\ref{eq_decompositionl1l2l3}), (\ref{eq_L1bound}), (\ref{eq_L2bound1}) and (\ref{eq_L3BOUND}). Moreover, (\ref{eq_zopointrate}) follows from (\ref{eq_decompositionl1l2l3}), (\ref{eq_L1bound}), (\ref{eq_L2bound2}) and (\ref{eq_L3BOUND}) by solving the associated quadratic equation. Finally, we mention that the results for $\operatorname{Im}m_n(z)$ essentially follows from  (\ref{eq_systemequationsm1m2}) that
\begin{gather*}
    \operatorname{Im}m_n(z)=\frac{1}{n}\sum_{i=1}^p \frac{\eta}{|z-\frac{\sigma_i}{n}\sum_j\frac{\xi^2_j}{1+\xi^2_jm_{1n}(z)}|^2}+\frac{1}{n}\sum_{i=1}^p \frac{\frac{1}{n}\sum_{j}\frac{\xi^4_j\operatorname{Im}m_{1n}(z)}{|1+\xi^2_jm_{1n}(z)|^2}}{|z-\frac{\sigma_i}{n}\sum_j\frac{\xi^2_j}{1+\xi^2_jm_{1n}(z)}|^2},
\end{gather*}
with the results for $\operatorname{Im}m_{1n}(z)$. This completes our proof.

\end{proof}

\begin{proof}[\bf Proof of Part (b)] 

For (\ref{eq_boundedfrombelowimportant}), on the one hand, when when $d>1$ and $\phi^{-1}<\widehat{\varsigma}_3,$ the result has been proved in (\ref{rem1nbound}).  On the other hand, when $-1<d \leq 1,$  we employ the proof idea as in the proof of Lemma A.3 of \cite{lee2013local11111} using a continuity argument. Recall (\ref{eq_Fnxyoriginaldefinition}). Denote 
\begin{equation*}
   g(x,y) \equiv \frac{\partial F_n(x,y)}{\partial x}+1=\frac{1}{n}\sum_{i=1}^p\frac{\frac{\sigma_i^2}{n}\sum_{j=1}^n\frac{\xi^4_j}{(1+x\xi^2_j)^2}}{(-y+\frac{\sigma_i}{n}\sum_{j=1}^n\frac{\xi^2_j}{1+x\xi^2_j})^2}.
\end{equation*}
From our assumption that $-1<d \leq 1$ and (\ref{ass3.4}), we find that there exist constants $C,C_0>0$ such that $\mathrm{d}F(x)\ge C(l-x)^{d}\ge C_0(l-x)$ for $x\in(0,l)$. Let $D$ be a sufficiently large constant and choose a sufficiently small constant $0<\epsilon<D^{-1}$, we have that when $n$ is sufficiently large, there exists some constants $C_1, C_2, C_3>0$
    \begin{align*}
            g(-(l+\epsilon)^{-1},\widehat{L}_{+})&=\frac{1}{n}\sum_{i=1}^p\frac{\frac{\sigma_i^2}{n}\sum_{j=1}^n\frac{(l+\epsilon)^2\xi^4_j}{(l+\epsilon-\xi^2_j)^2}}{\left(\widehat{L}_{+}-\frac{\sigma_i}{n}\sum_{j=1}^n\frac{(l+\epsilon)\xi^2_j}{l+\epsilon-\xi^2_j} \right)^2}\\
            &\ge\frac{C_1}{n}\sum_{i=1}^p\frac{\sigma_i^2\int_{l-(D-1)\epsilon}^l\frac{(l+\epsilon)^2x^2}{(l+\epsilon-x)^2}\mathrm{d}F(x)}{(\widehat{L}_{+}-\sigma_i\rO(1))^2} \ge\frac{1}{n}\sum_{i=1}^p\frac{C_2\int_{l-(D-1)\epsilon}^l\frac{(l-x)}{(l+\epsilon-x)^2}\mathrm{d}x}{(\widehat{L}_{+}-\sigma_i\rO(1))^2}\\
            &=\frac{1}{n}\sum_{i=1}^p\frac{C_2\int_{\epsilon}^{D\epsilon}\frac{(t-\epsilon)}{t^2}\mathrm{d}t}{(\widehat{L}_{+}-\sigma_i\rO(1))^2} \geq C_3(\log D-1+\frac{1}{D})>1,
    \end{align*}
    for sufficiently large $D>0.$ Similar arguments apply to $g(-(l-\epsilon),\widehat{L}_+).$ Consequently, by the continuity of $g(x,y)$, we obtain that $ \partial F_n (-l^{-1},\widehat{L}_{+}) /\partial x>0. $ Since $\partial F_n (m_{1n}(\widehat{L}_+),\widehat{L}_{+}) /\partial x=0,$ we can conclude that (\ref{rem1nbound}) still holds. That is, $m_{1n}(\widehat{L}_{+})>-l^{-1}$. This finishes the proof of (\ref{eq_boundedfrombelowimportant}).

For (\ref{thm_main_squared_root_bounded}) and (\ref{eq_gammadefinition}), using (\ref{eq_boundedfrombelowimportant}),   by  a discussion similar to (\ref{eq_assumptionequation}), we see that 
\begin{equation}\label{eq_secondmomentbounded}
\frac{\partial^2 F_n(m_{1n}(\widehat{L}_+), \widehat{L}_+)}{\partial x^2} \asymp 1.
\end{equation} 
Armed with this input, the square root behavior of $\rho$ at $\widehat{L}_+$ can be obtained in the same way as Lemma A.1 of \cite{lee2013local11111}. Due to similarity, we omit the details. This completes our proof of Part (b).  
\end{proof}

Finally, it is easy to check that we can follow lines of the proofs of parts (a) and (b) to prove the unconditional results by replacing the related quantities verbatim. We omit further details. 

\end{proof}

%{\color{red} To replace $ \widehat{L}_{+}$ with $L_{+}$ and $ \widehat{\varsigma}_3$, $ \widehat{\varsigma}_1$ with $\varsigma_3$ and $\varsigma_1$ respectively,  we only need to recall the first result in Theorem \ref{thm_boundedcaselocallaw} and \eqref{def4}. Those procedures will give additional errors with the rate $n^{-1/2+\epsilon_d}$ which much faster than $n^{-1/(d+1)}$ since $d>1$. In conclusion, we obtain the local behavior of $m_{1n}(z)$ at the right edge. }

%\begin{proof}[\bf Proof of Part (III)]
%
%\end{proof}

\subsection{Control of some bad probability events: proof of Lemma \ref{lem_probabilitycontrol}}\label{sec_appendxi_goodevent}

In this subsection, we prove Lemma \ref{lem_probabilitycontrol} case by case. We first prove Case (a) in Definition \ref{defn_probset}. 

\begin{proof}[\bf Proof of Case (a).] First, the last statement of (\ref{def1}) follows directly from strong law of large number. In fact, the result holds almost surely. 
%from Markov's inequality, the i.i.d.  assumption and  $\mathbb{E} \xi^2<\infty$ , as well as the assumption (\ref{ass1}).

%{\color{red} (with the rate $1-\rO(\log^{-2/\alpha}n)$)}

Then,  we prove the second statement of (\ref{def1}). For the upper bound, since $\{\xi_i^2\}$ are independent, we readily see that when $n$ is sufficiently large, for some constant $C'>0,$  
%For any positive constants $\kappa_0$ and $\delta$,
\begin{equation*}
\begin{split}
    \mathbb{P}(\xi_{(1)}^2 \leq Cn^{1/\alpha}
    \log n)&=\left(1-\mathbb{P}(\xi^2> Cn^{1/\alpha}
    \log n) \right)^n \ge\left(1-\frac{L(C n^{1/\alpha}\log n)}{(C n^{1/\alpha}\log n)^{\alpha}}\right)^n\\
    &\ge(1-C' n^{-1}\log^{-\alpha} n)^n \asymp \exp\left( -1/(C' \log^{\alpha} n)\right) \asymp 1-\rO(\log^{-\alpha}n). 
\end{split}
\end{equation*}
where in the second step we used the assumption (\ref{ass3.1})  and in the third step we used the assumption that $L(\cdot)$ is a slowly varying function. This proves the upper bound. Similarly, for the lower bound, we can show that  for some large constant $C>0$
\begin{equation}\label{eq_bbbbboneoneoneone}
    \mathbb{P}(\xi^2_{(1)}\le n^{1/\alpha}\log^{-1}n)=\rO(\log n/n^C).
\end{equation}
This concludes the proof of the second statement. 

Next, we prove the first statement using the second one. Note that 
\begin{equation*}
\begin{split}
    \mathbb{P}(\xi^2_{(1)}-\xi^2_{(2)}<n^{1/\alpha}\log^{-1}n)&=\mathbb{P}(\xi^2_{(1)}<n^{1/\alpha}\log^{-1}n+\xi^2_{(2)}) =\mathbb{P}(\xi^2_{(1)}<Cn^{1/\alpha}\log^{-1}n)=\rO(\log n/n^C),
\end{split}
\end{equation*}
where the second and third steps we used the results of the second statement. 

Then we justify the  fourth statement. In what follows, without loss of generality, we assume that $n^b$ is an integer.  For $c>1$ and $b>1/2$, we notice that for some large constant $C>0$
\begin{equation}\label{eq_controlkey}
\begin{split}
   \mathbb{P}( & \xi^2_{(1)}-\xi^2_{(n^b)}<c^{-1}n^{1/\alpha}\log^{-1}n)\le\mathbb{P}(\xi^2_{(n^b)}\ge(1-c^{-1})n^{1/\alpha}\log^{-1}n)\\
   &=\sum_{k= n^{b} }^{n}\binom{n}{k}\left[\mathbb{P}(\xi^2\ge(1-c^{-1})n^{1/\alpha}\log^{-1}n)\right]^k \left[\mathbb{P}(\xi^2\le(1-c^{-1})n^{1/\alpha}\log^{-1}n)\right]^{n-k}\\
   &=\sum_{k= n^{b} }^{n}\binom{n}{k}\left(\frac{\log^{\alpha}n}{n} \right)^k \left(1-\frac{\log^{\alpha}n}{n}\right)^{n-k}\\
   &\le\sum_{k= n^{b} }^{n}\left(\frac{en}{k} \right)^k \left(\frac{\log^{\alpha}n}{n}\right)^k \left(1-\frac{\log^{\alpha}n}{n} \right)^{n-k}\\
   &\le\sum_{k= n^{b} }^{n}\left(\frac{e}{k}\right)^k\log^{\alpha k}n e^{-\log^{\alpha}n(1-k/n)}=\rO(\log n/n^C),
\end{split}
\end{equation}
where in the first step we used (\ref{eq_bbbbboneoneoneone}), in the third step we used (\ref{ass3.1}) and in the fourth step we used Stirling's formula. This concludes the proof.  

Finally, we proceed to the proof of the third statement. Define a sequence of intervals $I_k:=\{Cn^{1/\alpha}\log^{-1}n+kn^{\epsilon},Cn^{1/\alpha}\log^{-1}n+(k+1)n^{\epsilon}\},k=[\![1,n^{1/\alpha-\epsilon}]\!]$. It is easy to see that if $\xi^2_{(i)}-\xi^2_{(i+1)}<n^{\epsilon}$ when $\xi^2_{(i)},\xi^2_{(i+1)}\in I_k$ for some $k$. Setting $\mathsf p_k:=\mathbb{P}(\xi^2\in I_k)$, we see that
\begin{equation*}
\mathbb{P}(|j\in[\![1,n]\!]:\xi^2_j\in I_k|=0)=(1- \mathsf p_k)^n,\quad \mathbb{P}(|j\in[\![1,n]\!]:\xi^2_j\in I_k|=1)=n \mathsf p_k(1-\mathsf p_k)^{n-1}.
\end{equation*}
We now provide an estimate for $p_k$ using (\ref{ass3.1}). Note that
%\begin{equation}
\begin{align}\label{eq_pkbound}
  \mathsf p_k&=\mathbb{P}(Cn^{1/\alpha}\log^{-1}n+kn^{\epsilon}\le\xi^2\le Cn^{1/\alpha}\log^{-1}n+(k+1)n^{\epsilon}) \nonumber \\
%   &=\mathbb{P}(\xi^2\ge Cn^{1/\alpha}\log^{-1}n+kn^{\epsilon})-\mathbb{P}(\xi^2\ge Cn^{1/\alpha}\log^{-1}n+(k+1)n^{\epsilon})\\
   &\le \frac{1}{(Cn^{1/\alpha}\log^{-1}n+kn^{\epsilon})^{\alpha}}-\frac{1}{(Cn^{1/\alpha}\log^{-1}n+(k+1)n^{\epsilon})^{\alpha}} \nonumber \\
   &\le Cn^{-1}\log^{\alpha}n\frac{(1+(k+1)n^{-1/\alpha+\epsilon}\log n)^{\alpha}-(1+kn^{-1/\alpha+\epsilon}\log n)^{\alpha}}{(1+kn^{-1/\alpha+\epsilon}\log n)^{\alpha}} \nonumber \\
   &\le Cn^{-1}\log^{\alpha}n^{-1/\alpha+\epsilon}\log n=Cn^{-(1+1/\alpha)+\epsilon}\log^{\alpha}n.
\end{align}
%\end{equation}
Armed with the above estimate, we see that when $n$ is sufficiently large, 
\begin{equation*}
\mathbb{P}(\xi^2_{(i)},\xi^2_{(i+1)}\in I_k)\le\mathbb{P}(|j\in[\![1,n]\!]:\xi^2_j\in I_k|\ge2)=1-(1-\mathsf p_k)^n-n \mathsf p_k(1-\mathsf p_k)^{n-1}\le n^2\mathsf p_k^2.
\end{equation*}
Consequently, together with (\ref{eq_pkbound}), we have that for some constant $C_1>0$
\begin{equation}\label{eq_controldddd}
\mathbb{P}(\xi^2_{(i)}-\xi^2_{(i+1)}\le n^{\epsilon})\le\sum_{k=1}^{n^{1/\alpha-\epsilon}}n^2\mathsf p_k^2\le C n^{1/\alpha+2-\epsilon}n^{-(2+2/\alpha)+2\epsilon}\log^{\alpha}n=n^{-1/\alpha+\epsilon}\log^{\alpha}n=\ro(1),
\end{equation}
as long as $\epsilon<1/\alpha$. This finishes the proof of the third statement.

\end{proof}

%For the case $\alpha\in(0,2]$, the situation becomes a little bit different, since the second moment or even first moment of $\xi$ can be infinite. Without loss of generality, we assume that $\xi$ is continuous. Actually, one can see that the first four statements in Definition \ref{def1} still hold. Now we reconsider the statement $5$. Now, suppose $\alpha\in(0,2]$, 
%\[
%\begin{split}
%    \sum_{i=1}^n\xi^2_i&=\sum_{i=1}^n\big(\xi^2_i\mathbf{1}_{\{\xi^2_i>n^{2/\alpha}\log n\}}+\xi^2_i\mathbf{1}_{\{\xi^2_i\le n^{2/\alpha}\log n\}}\big)\\
%    &:=\mathcal{L}_{11}+\mathcal{L}_{12}.
%\end{split}
%\]
%We have for any $c>0$,
%\[
%\begin{split}
%   \mathbb{P}(\mathcal{L}_{11}>cn^{2/\alpha}\log n)&\le n\mathbb{P}(\xi^2_1>n^{2/\alpha}\log n)\\
%   &=\frac{n}{n\log^{\alpha/2} n}\rightarrow 0,
%\end{split}
%\]
%Besides, an application of Markov's inequality and Lyapunov's moment inequality with $\gamma>\alpha/2$ if $\alpha=2$ and $\gamma=1$ otherwise shows that 
%\[
%\begin{split}
%   \mathbb{P}(\mathcal{L}_{12}>cn^{2/\alpha}\log n)&\le c\frac{n}{n^{2/\alpha}\log n}\big(\mathbb{E}(\xi^2_1)^{\gamma}\mathbf{1}_{\{\xi^2_1\le n^{2/\alpha}\log n\}}\big)^{1/\gamma}\\
%   &\le c\frac{n}{n^{2/\alpha}\log n}\frac{(n^{2/\alpha}\log n)^{1-1/\gamma}}{(n\log^{\alpha/2} n)^{1/\gamma}}\rightarrow 0,
%\end{split}
%\]
%where we use the Karamata's theorem and the fact $\mathbb{P}(\xi^2>x)$ varies regularly and bounded on right half real line, see \cite{Bingham1987}. We then have 
%\[
%\frac{1}{n}\sum_{i=1}^n\xi^2_i\le C\frac{n^{2/\alpha}\log n}{n},
%\]
%with probability tending to one.

Then we prove Case (b) of Definition \ref{defn_probset}. 
%{\color{red} from here}
%\begin{assumption}\label{ass4.2}
%Let $\Omega_n$ be the events define in Definition \ref{def3}, we have,
%\[
%\mathbb{P}(\Omega_n)\rightarrow1,\quad n\rightarrow\infty.
%\]
%\end{assumption}
\begin{proof}[\bf Proof of Case (b).] Due to similarity and for notational simplicity, we focus on the case $\beta=1.$ The general setting can be proved analogously and we omit the details. 
 
We start with the second statement of (\ref{def3}). For the upper bound, for any $C>1,$ following Markov inequality, we have that for some universal constant $C'>0$ when $n$ is sufficiently large, by (\ref{ass3.2}),  
\[
\begin{split}
    \mathbb{P}(\xi^2_{(1)}<C \log n)&=(1-\mathbb{P}(\xi^2\ge C\log n))^n\ge \left(1-\frac{\mathbb{E}e^{t\xi^2}}{e^{tC\log n}}\right)^n\\
    &=\left(1-\frac{C^{\prime}}{n^{tC}} \right)^n \asymp \exp(-1/n^{{tC}-1}) \asymp 1-\rO(n^{-({tC}-1)}).
\end{split}
\]
We can therefore conclude our proof using $t=1.$ Similarly, we can prove the lower bound that for some large constant $C_1>1$
\begin{align*}
\mathbb{P}(\xi^2_{(1)} \leq C^{-1} \log n)=\rO(n^{-C_1}). 
\end{align*} 
This completes the proof of the first statement. 

%On the other hand, for any fixed $c\in\mathbb{Z}^{+}$ greater than one, by Chernoff bound, we can choose a $t>0$ such that $t\times C^{-1}>1$ with
%\[
%\begin{split}
%    \mathbb{P}(\xi^2_{(c)}<C^{-1}\log n)&=\sum_{k=0}^{c-1}\binom{n}{k}\mathbb{P}(\xi^2>C^{-1}\log n)^k\mathbb{P}(\xi^2\le C^{-1}\log n)^{n-k}\\
%    &\le(\mathbb{P}(\xi^2\le C^{-1}\log n))^n\sum_{k=0}^{c-1}\frac{1}{k!}(\frac{n\mathbb{P}(\xi^2>C^{-1}\log n)}{\mathbb{P}(\xi^2\le C^{-1}\log n)})^{n-k}\\
%    &\le (\mathbb{P}(\xi^2\le C^{-1}\log n))^n\sum_{k=0}^{c-1}\frac{1}{k!}(\frac{C^{\prime}n^{1-tC^{-1}}}{\mathbb{P}(\xi^2\le C^{-1}\log n)})^{n-k}\\
%    &\le C^{\prime}(\mathbb{P}(\xi^2\le C^{-1}\log n))^n\times o(1)\rightarrow 0{\color{red} =O(n^{-D})},
%\end{split}
%\]
%where we also use the fact $\xi^2$ spreads among the right half real line, then $\mathbb{P}(\xi^2\le C^{-1}\log n)<1$.

For the first statement, the discussion is similar to (\ref{eq_pkbound}). The main difference is that the sequence of intervals are defined as $I_k:=\{C^{-1}\log n+k\times C^{-1}\log n,C^{-1}\log n+(k+1)\times C^{-1}\log n\},k\in[\![1,(C-C^{-1})/C^{-1}]\!].$ By Chernoff bound, we can control $\mathsf p_k:=\mathbb{P}(\xi^2\in I_k)$ as follows 
\[
\begin{split}
   \mathsf p_k&=\mathbb{P}(C^{-1}\log n+k\times C^{-1}\log n\le\xi^2\le C^{-1}\log n+(k+1)\times C^{-1}\log n)\\
    &=\mathbb{P}(\xi^2\ge C^{-1}\log n+k\times C^{-1}\log n)-\mathbb{P}(\xi^2\ge C^{-1}\log n+(k+1)\times C^{-1}\log n)\\
    &\le C^{\prime}(n^{-t(k+1)C^{-1}}-\inf_{t^{\prime}>0}n^{-t^{\prime}(k+2)C^{-1})})\\
    &\le C^{\prime}n^{-tC^{-1}},
\end{split}
\]
where $C'>0$ is some universal constant and in the third step we used (\ref{ass3.2}). Now we choose $t$ so that $tC^{-1}>2.$ Then by a discussion similar to (\ref{eq_controldddd}), we have
\[
\mathbb{P}(\xi^2_{(1)}-\xi^2_{(2)}\le C^{-1}\log n)\le n^2 \mathsf p_k^2\le C^{\prime}n^{2-tC^{-1}}.
\]
This completes the proof of the first statement.

Finally,  the last statement follows directly from the strong law of large number. In fact, the result holds almost surely. 
%for the third statement, using the independence of $\xi_i$ and applying the Chernoff bound, we readily see from (\ref{ass3.2}) that for some constant $C>0$
%\[
%\begin{split}
%    \mathbb{P}\left(\frac{1}{n}\sum_{i=1}^n\xi^2_i> \mathsf{c}_n \log n\right)\le\inf_{t>0}\frac{\mathbb{E}\exp\{\frac{t}{n}\sum_{i=1}^n\xi^2_i\}}{n^{t \mathsf{c}_n}}\leq C n^{\mathsf -c_n}.
%\end{split}
%\]
%Therefore, by choose $\mathsf{c}_n:=(\log n)^{-\ell}$ for some $0<\ell<1$, we can conclude the proof. 
%
%Now, taking interception finishes the proof.{\color{red} (only in this case the decayed rates satisfy the def of high prob.)}
\end{proof}

%\begin{assumption}\label{ass4.3}
%Let $\tilde{\Omega}_n$ be the events define in Definition \ref{def4}, we have,
%\[
%\mathbb{P}(\tilde{\Omega}_n)\rightarrow1,\quad n\rightarrow\infty.
%\]
%\end{assumption}

Finally we prove Case (c) of Definition \ref{defn_probset}. 

\begin{proof}[\bf Proof of Case (c).] 
Note that the fourth statement holds trivially and surely. 

For the first statement, under the assumption of  (\ref{ass3.4}), we see that the lower bound follows from that  
\[
\begin{split}
   \mathbb{P}(l-\xi^2_{(1)}>n^{-1/(d+1)-\epsilon_d})&=\big(1-\mathbb{P}(l-\xi^2_{(1)}<n^{-1/(d+1)-\epsilon_d})\big)^n\\
&\ge (1-Cn^{-\epsilon_d(d+1)-1})^n\\
&\ge 1-Cn^{-\epsilon_d(d+1)}. 
\end{split}
\]
Similarly, for the upper bound, we find that when $n$ is sufficiently large, for some constant $C'>0$
\[
\begin{split}
    \mathbb{P}(l-\xi^2_{(1)}>n^{-1/(d+1)}\log n)&\le n\big(1-\mathbb{P}(l-\xi^2\le n^{-1/(d+1)}\log n)\big)^{n-1}\\
    &\le n\big(1-C^{-1}n^{-1}\log^{d+1}n\big)^{n-1}\\
    &\le ne^{-C^{-1}\log^{d+1}n} \leq n^{-C'}.
\end{split}
\]
This completes the proof of the first statement. 

For the third statement, we prove by contradiction, i.e., there exists some sequence $\mathtt{a}_n=\ro(1),$ $l-\xi_{\lfloor b n \rfloor} \leq \mathtt{a}_n$ holds with high probability.  In fact, by a discussion similar to (\ref{eq_controlkey}) using (\ref{ass3.4}), we have that as long as $c \equiv c_n > n/\mathtt{a}_n,$
\[
\begin{split}
    \mathbb{P}(l-\xi^2_{(c)} \leq \mathtt{a}_n)&=\mathbb{P}(\xi^2_{(c)} \geq l-\mathtt{a}_n)\\
    &=\sum_{k=c+1}^{n}\binom{n}{k}\mathbb{P}(\xi^2>l-\mathtt{a}_n)^k\mathbb{P}(\xi^2\le l-\mathtt{a}_n)^{n-k}=\rO(n^{-C}),
\end{split}
\]
for some constant $C>0$ when $n$ is sufficiently large. This completes our proof for the third statement.

For the second statement, its discussion is similar to (\ref{eq_pkbound}). In this case, we will define the partition of the intervals as $I_k=[l-(k+1)n^{-1/(d+1)-\epsilon_d},l-kn^{-1/(d+1)-\epsilon_d}]$ for $k=[\![1,n^{\epsilon_d}\log n]\!].$  Analogous to the arguments of (\ref{eq_pkbound}), we have that 
\[
\mathsf p_k=\mathbb{P}(\xi^2\in I_k)\le Cn^{-\epsilon_d}n^{-1/(d+1)}(n^{-1/(d+1)}\log n)^d=Cn^{-1-\epsilon_d}\log^dn. 
\]  
Using the above control with (\ref{eq_controldddd}), we readily obtain that 
\[
\mathbb{P}(\xi^2_{(1)}-\xi^2_{(2)}\le n^{-1/(d+1)-\epsilon_d})\le n^2\mathsf p_k^2\le Cn^{-2\epsilon_d}\log^{2d}n.
\]
This completes the proof of the second statement. 

Finally, we proceed to the proof of the last statement. Denote the random variable $\tau_{\xi_i}$ as follows
\begin{gather*}
    \tau_{\xi^2_i}:=\frac{\xi^2_i}{1+\xi^2_im_{1n,c}(z)}-\int\frac{t}{1+tm_{1n,c}(z)}\mathrm{d}F(t).
\end{gather*}
By definition $\mathbb{E}\tau_{\xi^2_i}=0$. On the one hand, according to the discussion around (\ref{eq_defnmathsfW}), we find that 
\begin{gather*}
    \frac{1}{n}\sum_{i=1}^p\frac{\sigma_i^2\int\frac{t^2}{|1+tm_{1n,c}(z)|^2}\mathrm{d}F(t)}{|z-\sigma_i\int\frac{t}{1+tm_{1n,c}(z)}\mathrm{d}F(t)|^2}<1.
\end{gather*}
Together with Assumption \ref{assum_additional_techinical} and the continuity of $m_{2n,c},$ we can therefore conclude that for some constant $C_0>0,$
\begin{gather*}
    \int\frac{t^2}{|1+tm_{1n,c}(z)|^2}\mathrm{d}F(t)<C_0.
\end{gather*}
As a consequence, by Cauchy-Schwarz inequality, we find that for some constants $C_1, C_2>0$
\begin{gather*}
  \mathbb{E}|\tau_{\xi^2}|^2\le C_1  \int\frac{t^2}{|1+tm_{1n,c}(z)|^2}\mathrm{d}F(t) < C_2<\infty. 
\end{gather*}
Since $\tau_{\xi_i^2}, 1 \leq i \leq n,$ are independent, we can conclude our proof using Markov inequality. 
\end{proof}

%\subsection{Proof of Lemmas...}\label{sec_additionalprooflemmas}

\subsection{Fluctuation averaging arguments: Proof of Lemma \ref{lem_fa}}\label{sec_FAlemma}
In this section, we prove the fluctuation averaging results in Lemma \ref{lem_fa} following the strategies of Section 6 of \cite{lee2016extremal}. Fluctuation averaging is a common step in the proof of local laws for random matrix models, especially when the LSD has a square root decay behavior near the edge so that the entries of the resolvents can be controlled under some ansatz; see the monograph \cite{erdHos2017dynamical} for a review. However, in our setting, due to the lack of square root decay as in (\ref{eq: concave decay of rho_Q}), many entries of the resovelents, even the off-diagonal ones can be large when $\eta \sim n^{-1/2}.$ To address this issue, we will follow the strategies of \cite{lee2016extremal} to focus on the resolvent fractions instead of the entries themselves; see the discussion above Sections 6.1 of \cite{Kwak2021,lee2016extremal}. In what follows, due to similarity, we focus on the parts which deviate from \cite[Section 6]{lee2016extremal} the most.

   \begin{proof}[\bf Proof of Lemma \ref{lem_fa}]
In what follows, with loss of generality, we assume that $\xi_1^2 \geq \xi_2^2 \geq \cdots \geq \xi_n^2.$   
   
  We start with part (1). Recall (\ref{eq: decomp m_2}). Using Theorem \ref{thm_boundedcaselocallaw} and Remark \ref{remk_zibound}, we have that 
 \begin{gather*}
\begin{split}
    |m_2-m_2^{(1)}| \le \left|\frac{1}{n}\frac{\xi^2_{1}}{z(1+\xi^2_{1}m_{1n}+\rO_{\prec}((n \eta_0)^{-1}))} \right|+\left|\frac{1}{n}\sum_{i=2}^{p}\frac{ \rO_{\prec}((n\eta_0)^{-1})}{z(1+\xi^2_im_{1n}+\rO_{\prec}((n \eta_0)^{-1}))(1+\xi^2_im_{1n}+\rO_{\prec}((n\eta_0)^{-1}))}\right|.
\end{split}
\end{gather*} 
For the first term on the right-hand side of the equation,  it can be trivially bounded by $(n \eta_0)^{-1}$ by a discussion similar to (\ref{eq_trivialcontroleta}) using (\ref{eq_zopointrate11}).  The second term can also be controlled by $(n \eta_0)^{-1}$ using a discussion similar to (\ref{eq_L1bound}). The proves the first equation in (\ref{eq_c5first}).  For the second equation, due to similarity, we focus on $|m_2-m_2^{(i)}|.$ Using  (\ref{eq_decompositionleavoneout}), we have that
\begin{gather}\label{eq_differenceshouldbehere}
    |m_{2}-m_2^{(i)}|\le\frac{|\mathcal{G}_{ii}|}{n}+\frac{1}{n}\sum_{j \neq i}|\mathcal{G}_{jj}-\mathcal{G}_{jj}^{(i)}|.
\end{gather}
For $\mathcal{G}_{ii},$ by  Lemma \ref{lem: Resolvent}, Theorem \ref{thm_boundedcaselocallaw} and the assumption that $z \in \mathbf{D}_b^\prime$ in (\ref{eq_spectralparameterprime}), we conclude that with high probability, for some constant $C>0$
\begin{gather}\label{eq_giiboundinverse}
    |\mathcal{G}_{ii}|=\frac{1}{|z(1+\xi^2_im^{(i)}_1+Z_i)|}\le Cn^{1/(d+1)+\epsilon_d}. 
\end{gather} 
For $\mathcal{G}_{jj}-\mathcal{G}_{jj}^{(i)},$ by  Lemma \ref{lem: Resolvent}, (\ref{lem:Wald}) and Lemma \ref{lem:large deviation}, 
\begin{gather*}
    |\mathcal{G}_{jj}-\mathcal{G}_{jj}^{(i)}|=|\frac{\mathcal{G}_{ij}\mathcal{G}_{ji}}{\mathcal{G}_{ii}}| \prec |\mathcal{G}_{ii}||\mathcal{G}_{jj}^{(i)}|^2\frac{\operatorname{Im}m_1^{(ij)}}{n\eta_0} \prec \frac{n^{2 \epsilon_d}}{n }|\mathcal{G}_{ii}||\mathcal{G}_{jj}^{(i)}|^2,
\end{gather*}  
where in the last step we used (\ref{eq_zopointrate11}) and Theorem \ref{thm_boundedcaselocallaw}. Inserting all the above bounds back to (\ref{eq_differenceshouldbehere}) and use the trivial bound that $|\mathcal{G}_{jj}^{(i)}| \leq \eta_0,$ we can conclude the proof. This completes the proof of part (1). 

We now proceed to the proof of parts (2) and (3). Due to similarity, we focus on the details of part (2) and briefly mention how to prove (3) in the end. For simplicity, following the conventions in \cite{Kwak2021,lee2016extremal}, we denote the operator
$$P_i:=\mathbf{1}-\mathbb{E}_i,$$
where $\mathbb{E}_i$ is the conditional expectation with respect to $\mathbf{y}_i$. Using Lemma \ref{lem: Resolvent}, we see that on $\Omega$
\begin{gather}\label{eq_pigiiinverse}
    \frac{1}{n}\sum_{i=2}^n P_i(\frac{1}{\mathcal{G}_{ii}})=\frac{1}{n}\sum_{i=2}^n P_i(-z-z\mathbf{y}_i^{*}G^{(i)}(z)\mathbf{y}_i)=-\frac{z}{n}\sum_{i=2}^n Z_i.
\end{gather}
Consequently, it suffices to show that 
\begin{gather*}
    \left|\frac{1}{n}\sum_{i}P_i(\frac{1}{\mathcal{G}_{ii}})\right|\prec n^{-1/2-\frac{1}{2}(\frac{1}{2}-\frac{1}{d+1})+2 \epsilon_d}.
\end{gather*}  
By Chebyshev's inequality, it suffices to prove the following lemma. 
\begin{lemma}\label{lem_fakeycomponents}
Under the assumptions of Lemma \ref{lem_fa}, for any $z\in\mathbf{D}_b^{\prime}$ and fixed even number $M\in\mathbb{N}$, we have
\begin{gather*}
    \mathbb{E}^X\left|\frac{1}{n}\sum_{i=2}^n P_i(\frac{1}{\mathcal{G}_{ii}(z)})\right|^M\prec n^{M(-1/2-\frac{1}{2}(\frac{1}{2}-\frac{1}{d+1})+2 \epsilon_d)}.
\end{gather*}
\end{lemma} 
\begin{proof}
The proof strategy and technique follows closely from Section 6 of \cite{lee2016extremal}. In what follows, we adopt the way how \cite[Section 6.3]{Kwak2021} generalizes \cite[Section 6.2]{lee2016extremal} and only check the core estimates that have been used in \cite{lee2016extremal}. We first provide some notations following the conventions of \cite[Section 6.1]{lee2016extremal}. For any subset $\mathcal{T},\mathcal{T}^{\prime}\subset\{1,\dots, n\}$ with $i,j\notin\mathcal{T}$ and $j\notin\mathcal{T}^{\prime}$, we set
\begin{gather*}
  F_{ij}^{(\mathcal{T},\mathcal{T}^{\prime})} \equiv   F_{ij}^{(\mathcal{T},\mathcal{T}^{\prime})}(z):=\frac{\mathcal{G}_{ij}^{(\mathcal{T})}(z)}{\mathcal{G}_{jj}^{(\mathcal{T}^{\prime})}(z)}. 
\end{gather*}
In case $\mathcal{T}=\mathcal{T}^\prime=\emptyset$, we simply write $F_{ij}=F_{ij}^{(\mathcal{T},\mathcal{T}^{\prime})}$. With Lemma \ref{lem: Resolvent}, according to \cite[Lemma 6.1]{lee2016extremal}, we have that for any subset $\mathcal{T},\mathcal{T}^{\prime}\subset\{1,\dots, p\}$ with $i,j\notin\mathcal{T}$ and $j \notin\mathcal{T}^{\prime}$, and $\gamma\notin\mathcal{T}\bigcup\mathcal{T}^{\prime}$
\begin{equation*}
F_{ij}^{(\mathcal{T},\mathcal{T}^{\prime})}=F_{ij}^{(\mathcal{T}\gamma,\mathcal{T}^{\prime})}+F_{i\gamma}^{(\mathcal{T},\mathcal{T}^{\prime})}F_{\gamma j}^{(\mathcal{T},\mathcal{T}^{\prime})}, \ 
\end{equation*}
and 
\begin{equation*}
F_{ij}^{(\mathcal{T},\mathcal{T}^{\prime})}=F_{ij}^{(\mathcal{T},\mathcal{T}^{\prime}\gamma)}-F_{ij}^{(\mathcal{T},\mathcal{T}^{\prime}\gamma)}F_{j\gamma}^{(\mathcal{T},\mathcal{T}^{\prime})}F_{\gamma j}^{(\mathcal{T},\mathcal{T}^{\prime})}.
\end{equation*}
Moreover, we have that for $\gamma \notin \mathcal{T}$
\begin{equation*}
 \frac{1}{\mathcal{G}_{ii}^{(\mathcal{T})}}=\frac{1}{\mathcal{G}_{ii}^{(\mathcal{T\gamma})}} \left(1-F_{i\gamma}^{(\mathcal{T},\mathcal{T})}F_{\gamma i}^{(\mathcal{T},\mathcal{T})}\right).
\end{equation*}
In order to apply the techniques of \cite[Section 6.2]{lee2016extremal}, we need to prove the following estimates 
\begin{gather}\label{eq_keybounds}
\begin{split}
   & |m_1(z)-m_{1n}(z)|\prec (n \eta_0)^{-1},\ \quad \operatorname{Im}m_1(z)\prec (n \eta_0)^{-1}, \ \left|P_i(\frac{1}{\mathcal{G}_{ii}}) \right|\prec (n \eta_0)^{-1}, \ i \neq 1, \\
   & \max_{i\neq j}|F_{ij}(z)|\prec n^{-(1/2-1/(d+1))/2+\epsilon_d},\quad i,j \neq 1,\\
   & \max_{i\neq j} \left|\frac{F_{ij}^{(\emptyset,i)}(z)}{\mathcal{G}_{ii}(z)} \right|\prec (n \eta_0)^{-1},\quad i, j\neq 1,
    \end{split}
\end{gather}
First, the first part of (\ref{eq_keybounds}) follows from Theorem \ref{thm_boundedcaselocallaw}, Lemma \ref{lem: est for Im m_1} and Remark \ref{remk_zibound} (recall (\ref{eq_pigiiinverse})). Second, for the second part of (\ref{eq_keybounds}), by a discussion similar to (\ref{eq_giiboundinverse}), for $i \neq j$ and $i,j \neq 1,$ we have that for some constant $C>0,$ with high probability
\begin{equation}\label{eq_priorpriorbound}
 |\mathcal{G}_{ii}^{(j)}| \leq Cn^{1/(d+1)+\epsilon_d}.
\end{equation}
Together with Lemma \ref{lem: Resolvent}, we see that for some constant $C>0$
\begin{gather*}
    \begin{split}
        |F_{ij}|&=|z\mathcal{G}_{ii}^{(j)}\mathbf{y}_i^{*}G^{(ij)}\mathbf{y}_j| \prec \left|z\mathcal{G}_{ii}^{(j)}\frac{1}{n}\|G^{(ij)}\Sigma\|_F \right| \le C \left|\mathcal{G}_{ii}^{(j)}\left(\frac{\operatorname{Im}m_1^{(ij)}}{n\eta} \right)^{1/2} \right|\\
        &\prec n^{1/(d+1)+\epsilon_d}\frac{1}{n\eta_0}=n^{1/(d+1)-1/2+2\epsilon_d},
    \end{split}
\end{gather*}
where in the second step we used (\ref{lem:Wald}) and in the third step we used (\ref{eq_priorpriorbound}) and the fact $z \in \mathbf{D}_b^\prime.$ Finally, for the third part of (\ref{eq_keybounds}), using Lemma \ref{lem: Resolvent}, Lemma \ref{lem:large deviation} and (\ref{lem:Wald}), we see that 
\begin{equation*}
 \left| \frac{F_{ij}^{(\emptyset,i)}}{\mathcal{G}_{ii}}\right|=\left|\frac{\mathcal{G}_{ij}}{\mathcal{G}_{jj}^{(i)}\mathcal{G}_{ii}} \right|=\left| z\mathbf{y}_i^{*}G^{(ij)}\mathbf{y}_j \right| \prec \sqrt{\frac{\operatorname{Im} m^{(ij)}_1(z)}{n \eta}}.
\end{equation*}
We can therefore conclude our proof using Lemma \ref{lem: est for Im m_1}, Remark \ref{remk_zibound} and (\ref{lem:trace_difference}). 

Using (\ref{eq_keybounds}) and Assumption \ref{assum_model}, we can follow the proof of Corollary 6.4 of \cite{lee2016extremal} verbatim  and conclude that for any $\mathcal{T},\mathcal{T}^{\prime},\mathcal{T}^{\prime\prime}\in\{2,\dots,n\}$ with $|\mathcal{T}|,|\mathcal{T}^{\prime}|,|\mathcal{T}^{\prime\prime}|\le M,$ where $M$ is some large positive even integer, and for $z\in\mathbf{D}_b^{\prime}$, we have that when $i \neq j, i, j \neq 1,$
\begin{gather}\label{eq_keybounds2}
\begin{split}
   & |F_{ij}^{(\mathcal{T},\mathcal{T}^{\prime})}(z)|\prec n^{-(1/2-1/(d+1))/2+\epsilon_d},\\
   & \left|\frac{F_{ij}^{(\mathcal{T}^{\prime},\mathcal{T}^{\prime\prime})}(z)}{G^{(\mathcal{T})}_{ii}(z)} \right|\prec (n \eta_0)^{-1}, \  \ \left|P_i\left(\frac{1}{\mathcal{G}_{ii}^{(\mathcal{T})}} \right)\right|\prec (n \eta_0)^{-1}.
    \end{split}
\end{gather}
Once the key ingredients (\ref{eq_keybounds}) and (\ref{eq_keybounds2}) have been proved, we can follow lines of \cite[Lemma 6.6]{lee2016extremal} or \cite[Lemma 6.11]{Kwak2021} to conclude the proof. Due to similarity, we omit the details.
\end{proof}

\quad This completes the proof of part (2). The proof of part (3) is similar except we need to following the proof of Lemma \ref{lem_fakeycomponents} and \cite[Lemma 6.12]{lee2016extremal}  to show 
\begin{equation*}
   \mathbb{E}^X\left| \frac{1}{n} \sum_{i=2}^n \frac{1}{(1+\xi_i^2 m_{1n}(z))^2} P_i(\frac{1}{\mathcal{G}_{ii}(z)})\right|^M\prec n^{M(-1/2-\frac{1}{2}(\frac{1}{2}-\frac{1}{d+1})+2 \epsilon_d)}.
\end{equation*}
We omit the proof and refer the readers to the proof of \cite[Lemma 6.12]{lee2016extremal} for more details. This completes the proof of Lemma \ref{lem_fa}.

   \end{proof}

%\section{Additional numerical simulations}\label{appendix_lastlast}

\bibliographystyle{abbrv}
\bibliography{randompopulation,randompopulation1}

\begin{thebibliography}{10}

\bibitem{ahn2013eigenvalue}
S.~C. Ahn and A.~R. Horenstein.
\newblock Eigenvalue ratio test for the number of factors.
\newblock {\em Econometrica}, 81(3):1203--1227, 2013.

\bibitem{auffinger2009poisson11}
A.~Auffinger, G.~Ben~Arous, and S.~P{\'e}ch{\'e}.
\newblock Poisson convergence for the largest eigenvalues of heavy tailed
  random matrices.
\newblock {\em Annales de l'IHP Probabilit{\'e}s et statistiques},
  45(3):589--610, 2009.

\bibitem{bai2002determining}
J.~Bai and S.~Ng.
\newblock Determining the number of factors in approximate factor models.
\newblock {\em Econometrica}, 70(1):191--221, 2002.

\bibitem{bai2008large}
J.~Bai, S.~Ng, et~al.
\newblock Large dimensional factor analysis.
\newblock {\em Foundations and Trends{\textregistered} in Econometrics},
  3(2):89--163, 2008.

\bibitem{BY}
Z.~Bai and J.~Yao.
\newblock {Central limit theorems for eigenvalues in a spiked population
  model}.
\newblock {\em Annales de l'Institut Henri Poincaré, Probabilités et
  Statistiques}, 44(3):447 -- 474, 2008.

\bibitem{Baik2005}
J.~Baik, G.~B. Arous, and S.~P{\'e}ch{\'e}.
\newblock Phase transition of the largest eigenvalue for nonnull complex sample
  covariance matrices.
\newblock {\em The Annals of Probability}, 33(5):1643--1697, 2005.

\bibitem{Baik2006}
J.~Baik and J.~W. Silverstein.
\newblock Eigenvalues of large sample covariance matrices of spiked population
  models.
\newblock {\em Journal of multivariate analysis}, 97(6):1382--1408, 2006.

\bibitem{bao2022statistical}
Z.~Bao, X.~Ding, J.~Wang, and K.~Wang.
\newblock Statistical inference for principal components of spiked covariance
  matrices.
\newblock {\em The Annals of Statistics}, 50(2):1144--1169, 2022.

\bibitem{Bao2015}
Z.~Bao, G.~Pan, and W.~Zhou.
\newblock Universality for the largest eigenvalue of sample covariance matrices
  with general population.
\newblock {\em The Annals of Statistics}, 43(1):382--421, 2015.

\bibitem{beirlant2004statistics}
J.~Beirlant, Y.~Goegebeur, J.~Segers, and J.~L. Teugels.
\newblock {\em Statistics of extremes: theory and applications}, volume 558.
\newblock John Wiley \& Sons, 2004.

\bibitem{benaych2011eigenvalues}
F.~Benaych-Georges and R.~R. Nadakuditi.
\newblock The eigenvalues and eigenvectors of finite, low rank perturbations of
  large random matrices.
\newblock {\em Advances in Mathematics}, 227(1):494--521, 2011.

\bibitem{beran1985bootstrap}
R.~Beran and M.~S. Srivastava.
\newblock Bootstrap tests and confidence regions for functions of a covariance
  matrix.
\newblock {\em The Annals of Statistics}, 13(1):95--115, 1985.

\bibitem{bi2022spiked}
D.~Bi, X.~Han, A.~Nie, and Y.~Yang.
\newblock Spiked eigenvalues of high-dimensional sample autocovariance
  matrices: {CLT} and applications.
\newblock {\em arXiv preprint arXiv:2201.03181}, 2022.

\bibitem{Alex2014}
A.~Bloemendal, L.~Erd{\H{o}}s, A.~Knowles, H.-T. Yau, and J.~Yin.
\newblock Isotropic local laws for sample covariance and generalized {W}igner
  matrices.
\newblock {\em Electronic Journal of Probability}, 19:1--53, 2014.

\bibitem{bloemendal2016principal}
A.~Bloemendal, A.~Knowles, H.-T. Yau, and J.~Yin.
\newblock On the principal components of sample covariance matrices.
\newblock {\em Probability theory and related fields}, 164(1):459--552, 2016.

\bibitem{CHP}
T.~T. Cai, X.~Han, and G.~Pan.
\newblock {Limiting laws for divergent spiked eigenvalues and largest nonspiked
  eigenvalue of sample covariance matrices}.
\newblock {\em The Annals of Statistics}, 48(3):1255 -- 1280, 2020.

\bibitem{CAMBANIS1981368}
S.~Cambanis, S.~Huang, and G.~Simons.
\newblock On the theory of elliptically contoured distributions.
\newblock {\em Journal of Multivariate Analysis}, 11(3):368--385, 1981.

\bibitem{coles2001introduction}
S.~Coles.
\newblock {\em An Introduction to Statistical Modeling of Extreme Values}.
\newblock Springer Series in Statistics. Springer, 2001.

\bibitem{couillet2014analysis}
R.~Couillet and W.~Hachem.
\newblock Analysis of the limiting spectral measure of large random matrices of
  the separable covariance type.
\newblock {\em Random Matrices: Theory and Applications}, 3(04):1450016, 2014.

\bibitem{davison1997bootstrap}
A.~C. Davison and D.~V. Hinkley.
\newblock {\em Bootstrap methods and their application}.
\newblock Cambridge university press, 1997.

\bibitem{ding2020high}
X.~Ding.
\newblock High dimensional deformed rectangular matrices with applications in
  matrix denoising.
\newblock {\em Bernoulli}, 26(1):387--417, 2020.

\bibitem{ding2021spiked11}
X.~Ding.
\newblock Spiked sample covariance matrices with possibly multiple bulk
  components.
\newblock {\em Random Matrices: Theory and Applications}, 10(01):2150014, 2021.

\bibitem{ding2020local}
X.~Ding and H.~C. Ji.
\newblock Local laws for multiplication of random matrices.
\newblock {\em The Annals of Applied Probability (in press)}, 2023.

\bibitem{DJ22}
X.~Ding and H.~C. Ji.
\newblock Spiked multiplicative random matrices and principal components.
\newblock {\em arXiv preprint arXiv:2302.13502}, 2023.

\bibitem{DX_ell}
X.~Ding and J.~Xie.
\newblock {T}racy-{W}idom distribution for the edge eigenvalues of elliptical
  model.
\newblock {\em arXiv preprint arXiv 2304.07893}, 2023.

\bibitem{DXYZspiked}
X.~Ding, J.~Xie, L.~Yu, and W.~Zhou.
\newblock Limiting laws for edge eigenvalues under generalized spiked
  elliptical models with applications.
\newblock {\em preprint}, 2023.

\bibitem{Ding&Yang2018}
X.~Ding and F.~Yang.
\newblock A necessary and sufficient condition for edge universality at the
  largest singular values of covariance matrices.
\newblock {\em The Annals of Applied Probability}, 28(3):1679--1738, 2018.

\bibitem{ding2021spiked}
X.~Ding and F.~Yang.
\newblock Spiked separable covariance matrices and principal components.
\newblock {\em The Annals of Statistics}, 49(2):1113--1138, 2021.

\bibitem{9779233}
X.~Ding and F.~Yang.
\newblock Tracy-{W}idom distribution for heterogeneous gram matrices with
  applications in signal detection.
\newblock {\em IEEE Transactions on Information Theory}, 68(10):6682--6715,
  2022.

\bibitem{MR3904784}
E.~Dobriban and A.~B. Owen.
\newblock Deterministic parallel analysis: an improved method for selecting
  factors and principal components.
\newblock {\em Journal of the Royal Statistical Society. Series B. Statistical
  Methodology}, 81(1):163--183, 2019.

\bibitem{boostraporginialpaper}
B.~Efron.
\newblock {Bootstrap Methods: Another Look at the Jackknife}.
\newblock {\em The Annals of Statistics}, 7(1):1 -- 26, 1979.

\bibitem{el2007tracy}
N.~El~Karoui.
\newblock {T}racy--{W}idom limit for the largest eigenvalue of a large class of
  complex sample covariance matrices.
\newblock {\em The Annals of Probability}, 35(2):663--714, 2007.

\bibitem{Karoui2009}
N.~El~Karoui.
\newblock Concentration of measure and spectra of random matrices: Applications
  to correlation matrices, elliptical distributions and beyond.
\newblock {\em The Annals of Applied Probability}, 19(6):2362--2405, 2009.

\bibitem{el2018impact}
N.~El~Karoui.
\newblock On the impact of predictor geometry on the performance on
  high-dimensional ridge-regularized generalized robust regression estimators.
\newblock {\em Probability Theory and Related Fields}, 170(1):95--175, 2018.

\bibitem{el2013robust}
N.~El~Karoui, D.~Bean, P.~J. Bickel, C.~Lim, and B.~Yu.
\newblock On robust regression with high-dimensional predictors.
\newblock {\em Proceedings of the National Academy of Sciences},
  110(36):14557--14562, 2013.

\bibitem{el2019non}
N.~El~Karoui and E.~Purdom.
\newblock The non-parametric bootstrap and spectral analysis in moderate and
  high-dimension.
\newblock In {\em The 22nd International Conference on Artificial Intelligence
  and Statistics}, pages 2115--2124, 2019.

\bibitem{erdHos2013spectral}
L.~Erd{\H{o}}s, A.~Knowles, H.-T. Yau, and J.~Yin.
\newblock Spectral statistics of erd{\H{o}}s--r{\'e}nyi graphs i: local
  semicircle law.
\newblock {\em The Annals of Probability}, 41(3B):2279--2375, 2013.

\bibitem{erdHos2017dynamical}
L.~Erd{\H{o}}s and H.-T. Yau.
\newblock {\em A dynamical approach to random matrix theory}.
\newblock American Mathematical Soc., 2017.

\bibitem{fan2008high}
J.~Fan, Y.~Fan, and J.~Lv.
\newblock High dimensional covariance matrix estimation using a factor model.
\newblock {\em Journal of Econometrics}, 147(1):186--197, 2008.

\bibitem{fan2022estimating}
J.~Fan, J.~Guo, and S.~Zheng.
\newblock Estimating number of factors by adjusted eigenvalues thresholding.
\newblock {\em Journal of the American Statistical Association},
  117(538):852--861, 2022.

\bibitem{fan2019spectral}
Z.~Fan and A.~Montanari.
\newblock The spectral norm of random inner-product kernel matrices.
\newblock {\em Probability Theory and Related Fields}, 173(1):27--85, 2019.

\bibitem{fan2020spectra}
Z.~Fan and Z.~Wang.
\newblock Spectra of the conjugate kernel and neural tangent kernel for
  linear-width neural networks.
\newblock {\em Advances in neural information processing systems},
  33:7710--7721, 2020.

\bibitem{fang1990}
K.-T. Fang and T.~W. Anderson.
\newblock {\em Statistical inference in elliptically contoured and related
  distributions}.
\newblock Allerton Press, 1990.

\bibitem{fisher2016fast}
A.~Fisher, B.~Caffo, B.~Schwartz, and V.~Zipunnikov.
\newblock Fast, exact bootstrap principal component analysis for p$>$1 million.
\newblock {\em Journal of the American Statistical Association},
  111(514):846--860, 2016.

\bibitem{frahm2004generalized}
G.~Frahm.
\newblock {\em Generalized elliptical distributions: theory and applications}.
\newblock PhD thesis, Universit{\"a}t zu K{\"o}ln, 2004.

\bibitem{hachem2016large}
W.~Hachem, A.~Hardy, and J.~Najim.
\newblock Large complex correlated {W}ishart matrices: Fluctuations and
  asymptotic independence at the edges.
\newblock {\em The Annals of Probability}, 44(3):2264--2348, 2016.

\bibitem{han2021eigen}
X.~Han, X.~Tong, and Y.~Fan.
\newblock Eigen selection in spectral clustering: a theory-guided practice.
\newblock {\em Journal of the American Statistical Association}, pages 1--13,
  2021.

\bibitem{hansen2020three}
A.~Hansen.
\newblock The three extreme value distributions: An introductory review.
\newblock {\em Frontiers in Physics}, 8:604053, 2020.

\bibitem{Hu2019}
J.~Hu, W.~Li, Z.~Liu, and W.~Zhou.
\newblock High-dimensional covariance matrices in elliptical distributions with
  application to spherical test.
\newblock {\em The Annals of Statistics}, 47(1):527--555, 2019.

\bibitem{hu2019central}
J.~Hu, W.~Li, and W.~Zhou.
\newblock Central limit theorem for mutual information of large mimo systems
  with elliptically correlated channels.
\newblock {\em IEEE Transactions on Information Theory}, 65(11):7168--7180,
  2019.

\bibitem{Joha2000}
K.~Johansson.
\newblock Shape fluctuations and random matrices.
\newblock {\em Communications in mathematical physics}, 209(2):437--476, 2000.

\bibitem{John2001}
I.~M. Johnstone.
\newblock On the distribution of the largest eigenvalue in principal components
  analysis.
\newblock {\em The Annals of Statistics}, 29(2):295--327, 04 2001.

\bibitem{johnstone2020testing}
I.~M. Johnstone and A.~Onatski.
\newblock Testing in high-dimensional spiked models.
\newblock {\em The Annals of Statistics}, 48(3):1231--1254, 2020.

\bibitem{johnstone2018pca}
I.~M. Johnstone and D.~Paul.
\newblock {PCA} in high dimensions: An orientation.
\newblock {\em Proceedings of the IEEE}, 106(8):1277--1292, 2018.

\bibitem{kariya2014robustness}
T.~Kariya and B.~K. Sinha.
\newblock {\em Robustness of statistical tests}.
\newblock Academic Press, 2014.

\bibitem{ke2021estimation}
Z.~T. Ke, Y.~Ma, and X.~Lin.
\newblock Estimation of the number of spiked eigenvalues in a covariance matrix
  by bulk eigenvalue matching analysis.
\newblock {\em Journal of the American Statistical Association}, pages 1--19,
  2021.

\bibitem{knowles2013isotropic}
A.~Knowles and J.~Yin.
\newblock The isotropic semicircle law and deformation of {W}igner matrices.
\newblock {\em Communications on Pure and Applied Mathematics},
  66(11):1663--1749, 2013.

\bibitem{knowles2017anisotropic}
A.~Knowles and J.~Yin.
\newblock Anisotropic local laws for random matrices.
\newblock {\em Probability Theory and Related Fields}, 169(1):257--352, 2017.

\bibitem{VK}
V.~Koltchinskii and K.~Lounici.
\newblock {Normal approximation and concentration of spectral projectors of
  sample covariance}.
\newblock {\em The Annals of Statistics}, 45(1):121 -- 157, 2017.

\bibitem{Kwak2021}
J.~Kwak, J.~O. Lee, and J.~Park.
\newblock Extremal eigenvalues of sample covariance matrices with general
  population.
\newblock {\em Bernoulli}, 27(4):2740--2765, 2021.

\bibitem{lee2013local11111}
J.~O. Lee and K.~Schnelli.
\newblock Local deformed semicircle law and complete delocalization for wigner
  matrices with random potential.
\newblock {\em Journal of Mathematical Physics}, 54(10):103504, 2013.

\bibitem{lee2016extremal}
J.~O. Lee and K.~Schnelli.
\newblock Extremal eigenvalues and eigenvectors of deformed {W}igner matrices.
\newblock {\em Probability Theory and Related Fields}, 164(1):165--241, 2016.

\bibitem{lee2016tracy}
J.~O. Lee and K.~Schnelli.
\newblock {T}racy--{W}idom distribution for the largest eigenvalue of real
  sample covariance matrices with general population.
\newblock {\em The Annals of Applied Probability}, 26(6):3786--3839, 2016.

\bibitem{li2019local}
H.~Li and P.~Ralph.
\newblock Local pca shows how the effect of population structure differs along
  the genome.
\newblock {\em Genetics}, 211(1):289--304, 2019.

\bibitem{li2018structure}
W.~Li and J.~Yao.
\newblock On structure testing for component covariance matrices of a high
  dimensional mixture.
\newblock {\em Journal of the Royal Statistical Society: Series B (Statistical
  Methodology)}, 80(2):293--318, 2018.

\bibitem{lopes2022improved}
M.~E. Lopes.
\newblock Improved rates of bootstrap approximation for the operator norm: A
  coordinate-free approach.
\newblock {\em arXiv preprint arXiv:2208.03050}, 2022.

\bibitem{lopes2019bootstrapping}
M.~E. Lopes, A.~Blandino, and A.~Aue.
\newblock Bootstrapping spectral statistics in high dimensions.
\newblock {\em Biometrika}, 106(4):781--801, 2019.

\bibitem{MarchenkoandPastur1967}
V.~A. Mar{\v{c}}enko and L.~A. Pastur.
\newblock Distribution of eigenvalues for some sets of random matrices.
\newblock {\em Mathematics of the USSR-Sbornik}, 1(4):457, 1967.

\bibitem{4493413}
R.~R. Nadakuditi and A.~Edelman.
\newblock Sample eigenvalue based detection of high-dimensional signals in
  white noise using relatively few samples.
\newblock {\em IEEE Transactions on Signal Processing}, 56(7):2625--2638, 2008.

\bibitem{nagaraja2015spacings}
H.~N. Nagaraja, K.~Bharath, and F.~Zhang.
\newblock Spacings around an order statistic.
\newblock {\em Annals of the Institute of Statistical Mathematics},
  67(3):515--540, 2015.

\bibitem{naumov2019bootstrap}
A.~Naumov, V.~Spokoiny, and V.~Ulyanov.
\newblock Bootstrap confidence sets for spectral projectors of sample
  covariance.
\newblock {\em Probability Theory and Related Fields}, 174(3):1091--1132, 2019.

\bibitem{onatski2009testing}
A.~Onatski.
\newblock Testing hypotheses about the number of factors in large factor
  models.
\newblock {\em Econometrica}, 77(5):1447--1479, 2009.

\bibitem{onatski2010determining}
A.~Onatski.
\newblock Determining the number of factors from empirical distribution of
  eigenvalues.
\newblock {\em The Review of Economics and Statistics}, 92(4):1004--1016, 2010.

\bibitem{owen1983class}
J.~Owen and R.~Rabinovitch.
\newblock On the class of elliptical distributions and their applications to
  the theory of portfolio choice.
\newblock {\em The Journal of Finance}, 38(3):745--752, 1983.

\bibitem{passemier2014estimation}
D.~Passemier and J.~Yao.
\newblock Estimation of the number of spikes, possibly equal, in the
  high-dimensional case.
\newblock {\em Journal of Multivariate Analysis}, 127:173--183, 2014.

\bibitem{PLone}
L.~Pastur.
\newblock Eigenvalue distribution of large random matrices arising in deep
  neural networks: Orthogonal case.
\newblock {\em Journal of Mathematical Physics}, 63(6):063505, 2022.

\bibitem{PLthree}
L.~Pastur.
\newblock On random matrices arising in deep neural networks: Gaussian case.
\newblock {\em Pure and Applied Functional Analysis}, 5(6):1395--1424, 2022.

\bibitem{PLtwo}
L.~Pastur and V.~Slavin.
\newblock On random matrices arising in deep neural networks: General i.i.d.
  case.
\newblock {\em Random Matrices: Theory and Applications}, page 2250046, 2022.

\bibitem{paul2009no}
D.~Paul and J.~W. Silverstein.
\newblock No eigenvalues outside the support of the limiting empirical spectral
  distribution of a separable covariance matrix.
\newblock {\em Journal of Multivariate Analysis}, 100(1):37--57, 2009.

\bibitem{PillaiandYin2014}
N.~S. Pillai and J.~Yin.
\newblock Universality of covariance matrices.
\newblock {\em The Annals of Applied Probability}, 24(3):935--1001, 2014.

\bibitem{pinelis2009optimal}
I.~Pinelis and R.~Molzon.
\newblock Optimal-order bounds on the rate of convergence to normality in the
  multivariate delta method.
\newblock {\em arXiv preprint arXiv:0906.0177}, 2009.

\bibitem{7347462}
J.~Renard, L.~Lampe, and F.~Horlin.
\newblock Scaled largest eigenvalue detection for stationary time-series.
\newblock {\em IEEE Transactions on Signal Processing}, 64(5):1161--1172, 2016.

\bibitem{resnick2008extreme}
S.~I. Resnick.
\newblock {\em Extreme values, regular variation, and point processes},
  volume~4.
\newblock Springer Science \& Business Media, 2008.

\bibitem{schmidt2003credit}
R.~Schmidt.
\newblock Credit risk modelling and estimation via elliptical copulae.
\newblock In {\em Credit Risk}, pages 267--289. Springer, 2003.

\bibitem{silverstein1992signal}
J.~W. Silverstein and P.~L. Combettes.
\newblock Signal detection via spectral theory of large dimensional random
  matrices.
\newblock {\em IEEE Transactions on Signal Processing}, 40(8):2100--2105, 1992.

\bibitem{stock2016dynamic}
J.~H. Stock and M.~W. Watson.
\newblock Dynamic factor models, factor-augmented vector autoregressions, and
  structural vector autoregressions in macroeconomics.
\newblock In {\em Handbook of macroeconomics}, volume~2, pages 415--525. 2016.

\bibitem{tracy1994level}
C.~A. Tracy and H.~Widom.
\newblock Level-spacing distributions and the airy kernel.
\newblock {\em Communications in Mathematical Physics}, 159(1):151--174, 1994.

\bibitem{usseglio2018estimation}
A.~Usseglio-Carleve.
\newblock Estimation of conditional extreme risk measures from heavy-tailed
  elliptical random vectors.
\newblock {\em Electronic Journal of Statistics}, 12(2):4057--4093, 2018.

\bibitem{wagner2015go}
F.~Wagner.
\newblock {GO-PCA}: An unsupervised method to explore gene expression data
  using prior knowledge.
\newblock {\em PloS one}, 10(11):e0143196, 2015.

\bibitem{Wen2021}
J.~Wen, J.~Xie, L.~Yu, and W.~Zhou.
\newblock Tracy-{W}idom limit for the largest eigenvalue of high-dimensional
  covariance matrices in elliptical distributions.
\newblock {\em Bernoulli}, 28(4):2941--2967, 2022.

\bibitem{xia2021normal}
D.~Xia.
\newblock Normal approximation and confidence region of singular subspaces.
\newblock {\em Electronic Journal of Statistics}, 15(2):3798--3851, 2021.

\bibitem{yang2019edge}
F.~Yang.
\newblock Edge universality of separable covariance matrices.
\newblock {\em Electronic Journal of Probability}, 24:1--57, 2019.

\bibitem{yang2021testing}
X.~Yang, X.~Zheng, and J.~Chen.
\newblock Testing high-dimensional covariance matrices under the elliptical
  distribution and beyond.
\newblock {\em Journal of Econometrics}, 221(2):409--423, 2021.

\bibitem{yao2021rates}
J.~Yao and M.~E. Lopes.
\newblock Rates of bootstrap approximation for eigenvalues in high-dimensional
  {PCA}.
\newblock {\em arXiv preprint arXiv:2104.07328}, 2021.

\bibitem{yao2015sample}
J.~Yao, S.~Zheng, and Z.~Bai.
\newblock {\em Sample covariance matrices and high-dimensional data analysis}.
\newblock Cambridge University Press Cambridge, 2015.

\bibitem{2022arXiv220206188Y}
L.~{Yu}, P.~{Zhao}, and W.~{Zhou}.
\newblock {Testing the number of common factors by bootstrapped sample
  covariance matrix in high-dimensional factor models}.
\newblock {\em arXiv preprint arXiv:2202.06188}, 2022.

\bibitem{zhang2018clt}
B.~Zhang, G.~Pan, and J.~Gao.
\newblock {CLT} for largest eigenvalues and unit root testing for
  high-dimensional nonstationary time series.
\newblock {\em The Annals of Statistics}, 46(5):2186--2215, 2018.

\bibitem{Zhanggeneral}
L.~Zhang.
\newblock Spectral analysis of large dimensional random matrices.
\newblock {\em Ph.D. Thesis, National University of Singapore}, 2006.

\bibitem{zhang2022asymptotic}
Z.~Zhang, S.~Zheng, G.~Pan, and P.-S. Zhong.
\newblock Asymptotic independence of spiked eigenvalues and linear spectral
  statistics for large sample covariance matrices.
\newblock {\em The Annals of Statistics}, 50(4):2205--2230, 2022.

\end{thebibliography}

\end{document}